%% file: dissertation.tex
\documentclass[print,nativefonts]{dissertation}

\usepackage[numbers]{natbib}

\usepackage{etaremune}
\usepackage{graphicx}
\usepackage{array}

\usepackage{amssymb,amsmath, amsthm}
\usepackage{footnote}

\makesavenoteenv{tabular}
\makesavenoteenv{table}
 
\usepackage{subfig} % use subcaption instead
\usepackage{algorithm} %,algorithmic
\usepackage{algorithmicx}
\usepackage{xfrac}
\newtheorem{proposition}{Proposition}

\newcommand{\etal}{et~al.\/~}
\newcommand{\eg}{e.\,g.,\/~}
\newcommand{\ie}{i.\,e.,\/~}
\newcommand{\etc}{etc.\/~}

\newcommand{\cf}{see }

\newcommand{\E}[1]{\mathbf{E}\left[#1\right]}
\newcommand{\Prob}[1]{P_{#1}}
\newcommand{\Eb}{\text{DC}_{max}}

\renewcommand{\paragraph}[1]{\vspace{1em}\noindent\textbf{#1}\;}
\newcommand{\subpar}[1]{\vspace{0.75em}\noindent\emph{#1}}
\newcommand{\fakeparagraph}{\subpar{}}
\newcommand{\parspace}{\vspace{1em}}

\newcommand{\chapref}[1]{Chapter~\ref{#1}}
\newcommand{\chaprefs}[2]{chapters~\ref{#1} and~\ref{#2}}
\newcommand{\secref}[1]{Section~\ref{#1}}

\newcommand{\figref}[1]{Figure~\ref{#1}}

\newcommand{\tabref}[1]{Table~\ref{#1}}

\begin{document}

%% Specify the title and author of the thesis. This information will be used on
%% the title page (in title/title.tex) and in the metadata of the final PDF.
\title[A Step Back Towards the Smart Dust Dream]{Opportunistic Communication in \\Extreme Wireless Sensor Networks}
\author{Marco}{Cattani}

%% Use Roman numerals for the page numbers of the title pages and table of contents.
\frontmatter

\include{title/title}

%\title{Opportunistic communication in Extreme Wireless Sensor Networks:\\Towards the Smart Dust Dream.}

%\maketitle 
%% The (optional) dedication can be used to thank someone or display a significant quotation.
\dedication{\epigraph{``There are always possibilities''}{Spock}}

\tableofcontents

\include{acknowledgements/acknowledgements}

%% Use Arabic numerals for the page numbers of the chapters.
\mainmatter
%% Turn on thumb indices.
\thumbtrue

%------------------------------------------------------------
%----------------------- INTRO CHAPTER ----------------------
%------------------------------------------------------------

%\include{chapter-1.tex}
\include{introduction/introduction}
\include{sofa/sofa}

\include{estreme/estreme}
\include{staffetta/staffetta}
\include{nemo/nemo}
\include{conclusions/conclusions}

\include{summary/summary}

%------------------------------------------------------------
%------------------------- APPENDIX -------------------------
%------------------------------------------------------------

%% Use letters for the chapter numbers of the appendices.

%\appendix

%\include{appendix-a/appendix-a}

%% Turn off thumb indices for unnumbered chapters.
\thumbfalse

%------------------------------------------------------------
%---------------------- REFERENCES --------------------------
%------------------------------------------------------------

\bibliographystyle{dissertation}
\bibliography{library,martella_library}
%\references{library}

%\include{cv/cv}
%\include{publications/publications}

\end{document}

%% file: title/title.tex
%!TEX root = ../dissertation.tex
\begin{titlepage}

\begin{center}

%% The following lines repeat the previous page exactly.

\vspace*{2\bigskipamount}

%% Print the title.
{\makeatletter
\titlestyle\bfseries\LARGE\@title
\makeatother}

%% Print the optional subtitle.
{\makeatletter
\ifx\@subtitle\undefined\else
    \bigskip
    \titlefont\titleshape\Large\@subtitle
\fi
\makeatother}

%% Uncomment the following lines to insert a vertically centered picture into
%% the title page.
%\vfill
%\includegraphics{title}
\vfill

%% Apart from the names and dates, the following text is dictated by the
%% promotieregelement.

{\Large\titlefont\bfseries Proefschrift}

\bigskip
\bigskip

ter verkrijging van de graad van doctor

aan de Technische Universiteit Delft,

op gezag van de Rector Magnificus prof.~ir.~K.~C.~A.~M.~Luyben,

voorzitter van het College voor Promoties,

in het openbaar te verdedigen op 

woensdag 28 september 2016 om 15:00 uur

\bigskip
\bigskip

door

\bigskip
\bigskip

%% Print the full name of the author.
\makeatletter
{\Large\titlefont\bfseries\@firstname\ {\titleshape\@lastname}}
\makeatother

\bigskip
\bigskip

Laurea Informatica, Università degli Studi di Trento, Italië

geboren te Trento, Italië.

%% Extra whitespace at the bottom.
\vspace*{2\bigskipamount}

\end{center}

\clearpage
\thispagestyle{empty}

%% The following line is dictated by the promotieregelement.
\noindent This dissertation has been approved by the %promotor:

%% List the promotors (supervisors).
\medskip\noindent promotor:
%\begin{tabular}{l}
    Prof.\ Dr.\ K.\ G.\ Langendoen and
%\end{tabular}

%% List the (optional) copromotor.
\noindent copromotor: 
    Dr.\ M.\ A.\ Z\'{u}\~{n}iga

\medskip
\noindent Composition of the doctoral committee:

%% List the committee members, starting with the Rector Magnificus and the
%% promotor(s) and ending with the reserve members.
\medskip\noindent
\begin{tabular}{l}
    Rector Magnificus \\
    Prof.\ Dr.\ K.\ G.\ Langendoen, promotor \\
    Dr.\ M.\ A.\ Z\'{u}\~{n}iga, copromotor \\
\end{tabular} 

\medskip
\noindent Independent members:

\medskip\noindent
\begin{tabular}{l}
    Prof.\ Dr.\ D.\ Epema, EEMCS, TU Delft \\ 
    Dr. F.\ Kuipers, EEMCS, TU Delft \\    
    Dr. O.\ Landsiedel, Chalmers University of Technology \\
    Prof.\ Dr. M. van Steen, University of Twente  \\
    Prof.\ Dr.\ T. Voigt, Uppsala University \\
    Prof.\ Dr.\ G.\ J.\ Houben, EEMCS, TU Delft, reserve member \\
\end{tabular}

%% Include the following disclaimer for committee members who have contributed
%% to this dissertation. Its formulation is again dictated by the
%% promotieregelement.
\medskip
%\noindent Prof.\ dr.\ ir.\ J.\ Doe heeft als begeleider in belangrijke mate aan de totstandkoming van het proefschrift bijgedragen.

%%% Here you can include the logos of any institute that contributed financially
%%% to this dissertation.
%\vfill
%\begin{center}
%    \includegraphics[height=0.5in]{title/logos/tudelft}
%    \hspace{2em}
%    \includegraphics[height=0.5in]{title/logos/casimir} \\
%    \includegraphics[height=0.5in]{title/logos/fom}
%    \hspace{2em}
%    \includegraphics[height=0.5in]{title/logos/nwo}
%\end{center}
%\vfill

%\noindent
%\begin{tabular}{@{}p{0.2\textwidth}@{}p{0.8\textwidth}}
%    \textit{Keywords:} & \ldots \\[\medskipamount]
%    \textit{Printed by:} & Johannes Gutenberg \\[\medskipamount]
%    \textit{Front \& Back:} & Beautiful cover art that captures the entire %content of this thesis in a single illustration.
%\end{tabular}

\vspace{4\bigskipamount}

\noindent Copyright \textcopyright\ 2016 by M.~Cattani. All rights reserved. No part of the material protected by this copyright may be reproduced or utilized in any form or by any means, without the prior permission of the author.

%% Uncomment the following lines if this dissertation is part of the Casimir PhD
%% Series, or a similar research school.

\medskip
\noindent ISBN 978-94-6233-386-4

\medskip
\noindent An electronic version of this dissertation is available at \\
\url{http://repository.tudelft.nl/}.

\medskip
\noindent 
\begin{tabular}[width=\textwidth]{c p{0.65\textwidth}}
\raisebox{-.5\height}{\includegraphics[width=0.3\textwidth]{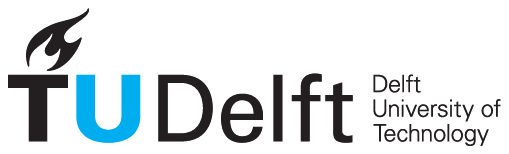}} & This work was carried out in the TU Delft Graduate chool. \\
\end{tabular}

\medskip
\noindent 
\begin{tabular}{c p{0.65\textwidth}}
\raisebox{-.5\height}{\includegraphics[width=0.3\textwidth]{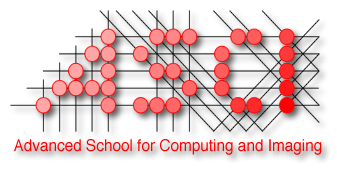}} & This work was carried out in the ASCI graduate school. ASCI dissertation series number 362.
\end{tabular}

\medskip
\noindent 
\begin{tabular}{c p{0.65\textwidth}}
\raisebox{-.5\height}{\includegraphics[width=0.3\textwidth]{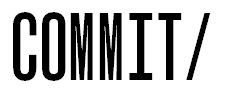}} & This work was funded by the Dutch national program COMMIT/ under the project P09 EWiDS.
\end{tabular}

\end{titlepage}

%% file: acknowledgements/acknowledgements.tex
\chapter*{Acknowledgments}
\setheader{Acknowledgments}

This thesis is the result of few ``\textit{eureka}'' moments, surrounded by much work, fun, coffee and personal experiences. For these vital ingredients, I need to thank a large crowd of extremely special persons. To me, you are all colleagues, family and friends.

First of all, I want to thank my mentor and supervisor Koen Langendoen. Through your \textit{koenification} and Dutch pragmatic attitude, you showed me how to be a better person and a more thorough researcher. You also provided me two inspiring supervisors that were key to my first research successes. 
Matthias gave me the right research directions and taught me how to fight for my ideas. Marco showed me the power of being perseverant and always positive about things. Through his countless comments, Marco patiently taught me how to deeply analyze a problem and clearly present my solution (in proper English!).

Along with my supervisors, I had many brilliant colleagues and students who were always available to listen to my weird ideas and to give me their feedback. I want to especially thank Niels and Przemek for their blunt opinions and advices. Marco, Federico, Carlo and Frederick, whose exceptional research inspired and motivated me throughout my PhD. Andreas, for the countless conversations we had about research, movies, books and food. Platon and Dimitris, students who chose and trusted me as their supervisor and of whom I am really proud. Claudio, who shared with me the experience of a PhD, with all its memorable moments: from wild Dutch parties to boring experiments in deserted museums. Finally, I want to thank Ioannis, who is not just a colleague but also my hacking-mate and trusted confidante. His hands-on experience and problem-solving attitude are something I greatly enjoy and admire. 

Outside of the working environment, my thanks go greatly towards my parents, sister, cousins and acquired relatives. You believed in me and fed me with love, and are responsible for the person I am today. In particular, I want to thank my father, who taught me how to fly,  and my mother, who instead reminded me of keeping my feet on the ground. Finally, my biggest gratitude goes to my wife Daniela, for her ever-present, powerful love. You are my first supporter and life companion, you boost my confidence and point me towards greater goals. You can understand me, help me and make me better, teaching me the power of emotions, empathy, selflessness and life.

Some special thanks go as well to the Netherlands and all its open-minded citizens: a model of what a modern civilization should be. You welcomed me without prejudices and gave me the tools to achieve great \textit{thinks}. In particular, I want to thank the student rugby club S.R.C. Thor and the 50th bestuur, who gifted me with many glorious moments and trusted me to manage their historic clubhouse: a never-ending problem solving exercise that helped me later on, through the harsh times of my PhD.

\begin{flushright}
{\makeatletter\itshape
\@firstname\ \@lastname \\
Delft, September 2016
\makeatother}
\end{flushright}

%% file: introduction/introduction.tex
%!TEX root = ../dissertation.tex

% 2.	When you are happy with the structure, do your best to revise your text so that the answer to each question below is ‘yes’:
%
% a.	Are all the terms sufficiently explained and illustrated by means of examples?
% b.	Is the text organized so that everything is explained only once, i.e. is there no repetition of information across the text?
% c.	Does the first sentence of each paragraph express the main idea of that paragraph, and does the rest of that paragraph develop that main idea?
% d.	Are all paragraphs of roughly the same length?

\chapter{Introduction}
\label{chapter:introduction}

\epigraph{
     ``We often fear what we do not understand. Our best defense is knowledge.''}
	 {Lieutenant Tuvok}
	 
% a.	Present the rationale for your study, by describing
% i.	a problem (or opportunity) in society,
% ii.	a possible solution to this problem,
% iii.	a gap in our knowledge which has prevented  such a solution so far.
% b. Formulate the aim, i.e. rephrase your main question into an aim. This should follow naturally from the gap in knowledge identified above.
% c. Briefly describe how you are going to achieve this aim, for instance, by saying something about your methodology.

\dropcap{W}{hen} the idea of Smartdust was first presented in 1992 both researchers and industry got excited and inspired by the potential that such a technology could achieve. 
A swarm of tiny electromechanical devices such as sensors and robots interconnected via a wireless network that could sense the characteristics of the environment (\eg light, temperature and vibration) and react accordingly.

The hype around Smartdust sparkled the research topic of Wireless Sensor Networks (WSNs): a wide literature of works focused on sensing and communication in resource-constrained devices. 
Nevertheless, compared to the initial conception of Smartdust, these works often tackled easier problems, simplifying some of the original challenges.
% Nevertheless, WSN had only been inspired by the Smartdust idea, simplifying some of its challenges and tackling simpler problems. 
In WSNs sensors are bigger, with fewer constraints and more capabilities. Networks are often static, rather than mobile, with sizes in the order of hundreds of nodes, rather than thousands. Finally, the sensing and actuation tasks are often decoupled, removing the real-time constrains from the problem.

While WSNs' simpler challenges fit many application scenarios such as precision agriculture~\cite{Wark2007} and building monitoring~\cite{Ceriotti2009}, they seldom cope with dense and dynamic scenarios, typical of smart cities~\cite{Zanella2014}, vehicular~\cite{Lee2006} and crowd monitoring~\cite{Martella2014} networks, that closer resemble the Smartdust idea.

% Societal impact
One of the dynamic scenarios, crowd monitoring, has in the recent years shown to be particularly important for people's safety. 
Before a large-scale event (festival, concert, \etc), crowd managers use their experience to list all possible dangerous situations and plan a set of counteractions. 
During the actual gathering, crowds are continuously monitored in order to assess their conditions, while managers decide which planned actions to take. 
In order to be effective, the detection-reaction process for dangerous situations must be as timely as possible. A delayed detection can lead to late reactions and, ultimately, to catastrophic results. 
In 2010, a crowd rush in a popular electronic dance festival in Germany ended up with 21 deaths from suffocation and more than 500 injured people~\cite{Helbing2012}.

This thesis is motivated by such a scenario: the need to monitor a crowd during large-scale events, where several hundreds of people can gather within confined spaces. This work is part of the EWiDS project\footnote{\url{http://www.commit-nl.nl/projects/very-large-wireless-sensor-networks-for-well-being}} aimed at providing participants with wearable devices that can actively monitor the density of their surroundings and issue alerts when crossing dangerous thresholds. According to crowd managers, density is one of the risk factors that helps predicting dangerous situations.

% In practice, recent technology advancements allowed to scale to smartdust scales in a practical way
% Only recently, technology allowed to produce cheap enought devices (few euro) and enable to scale wsn research to thousands of devices

To tackle the aforementioned problem, this thesis makes a step back towards the original Smartdust idea and explores the problem of communication in Extreme Wireless Sensor Networks (EWSNs), where resources are limited, nodes are mobile and densities can drastically fluctuate in space and time.

\section{Wireless Sensor Networks}
\label{sec:introduction_problem_statement}

In order to understand the challenges of Extreme Wireless Sensor Networks (EWSNs), we will first introduce some basic knowledge about traditional Wireless Sensor Networks (WSNs). By highlighting the differences between the two (summarized in \tabref{tab:introduction_wsn}), we define the crucial characteristics of EWSNs and, more importantly, its challenges.

% WSN are set of devices that periodically sense the environment and communicate the information
A Wireless Sensor Network comprises of a small set of battery-powered devices, the so-called \emph{motes}, which periodically monitor the environment and cooperatively communicate the sensed information via wireless links.
WSNs are usually small in size, ranging from tens to hundreds of short-range devices. Enough to cover their usually modest deployment areas. 
From a networking perspective, the maximum neighborhood cardinality is in the order of a few tens of devices.

\paragraph{Data Collection.} 
In WSNs, motes can diffuse and process the sensed information in a distributed fashion (allowing the network to react autonomously) or can forward it to one or more \emph{sink} nodes (\cf \figref{fig:intro_wsn}). Sinks are special nodes that are usually connected to a larger network \eg the Internet, and serve as gateways to store data and forward actuation commands. In some cases, sink nodes act also as coordinators of the other nodes in the network.
Because WSNs usually involve static topologies, shortest-path routing trees can be built in order to deliver the sensed information to sink nodes over a multi-hop wireless connection. 
%In other cases, WSN employs mobile sinks, called \emph{mules}, that ``visit'' each node's position and collect the data with no need of routing structures and multi-hop connectivity.

\paragraph{Energy Efficiency.}
Because in WSNs motes are often powered by small batteries, it is vital to design communication mechanisms that are energy efficient and allow these wireless devices to operate long enough to meet the lifetime requirements imposed by the application (often in the order of days, months, or a few years).

Unfortunately, communication mechanisms are never energy efficient enough. Designers strive to reduce the cost, weight and size of motes, while batteries are among the components that are most expensive, heavy and difficult to miniaturize. Communicating with higher efficiencies than needed allows thus to reduce the capacity of batteries \ie their size, weight and cost, putting pressure on the efficiency of the communication.

The main mechanism employed by WSNs to improve the energy efficiency of wireless devices is called duty cycling and consists in limiting the use of the energy-demanding peripherals \eg the radio, for very short and periodic amounts of time (\cf \figref{fig:intro_dutycycle}). 
The ratio between active and passive periods captures the energy efficiency of a mote \ie its duty cycle $DC$:  
\begin{align} 
	\label{eq:introduction_dc}
	\text{DC} = \frac{\delta}{T},
\end{align}	
where $\delta$ and $T$ are respectively the duration and period of the node's activity. Note that the activity period $T$ can be computed by inverting the mote's activity frequency $\phi$
\begin{align} 
T = 1/\phi, \quad DC = \delta\,\phi.
\end{align} 
As an example, a node being active for 5\,ms every 200\,ms ($\delta$ = 5\,ms, $T$ = 200\,ms, $\phi$ = 5\,Hz) will result in a duty cycle of 0.025. Having such a duty cycle means that the node is active for 2.5\% of the time. 
Compared to a node that is active for 100\,ms every second ($\delta$ = 100\,ms, $T$ = 1000\,ms, $\phi$ = 1\,Hz, $DC$ = 0.1), the lifetime of the former is 4 times longer than the latter, even though it is active 5 times more often. 
As this example shows, energy efficiency is a balancing act between activity frequency $\phi$ and duration $\delta$. A concept that will be further analyzed in \chapref{chapter:staffetta}.

Note that in this thesis, and in most WSN literature, the node's duty cycle is computed solely using the timings of the radio apparatus \ie the \emph{radio duty cycle}, since this peripheral has far higher energy consumption than the other hardware components.

\begin{figure}
	\centering
	\subfloat[A Wireless Sensor Network composed of three wireless nodes and a sink. The latter is connected to a local area network (LAN) to store the sensed data.]{ \label{fig:intro_wsn}  
	\includegraphics[width=0.48\linewidth]{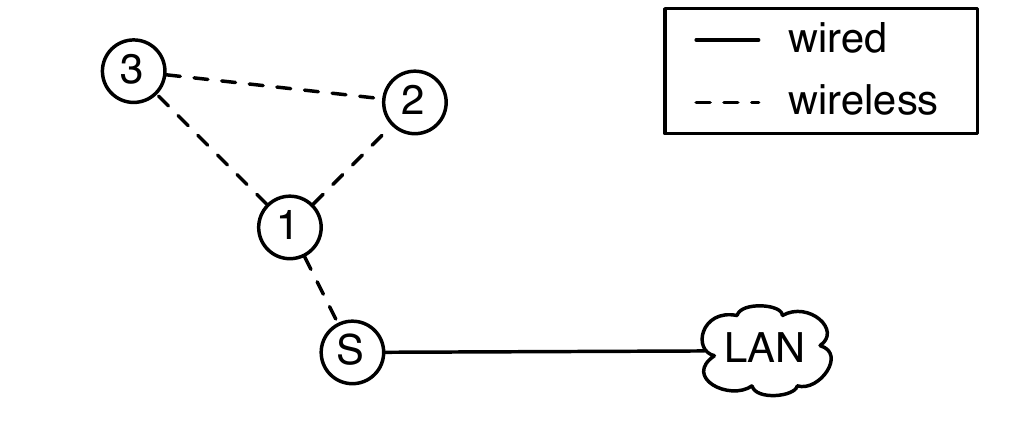} }
	\hspace*{0.03\linewidth}
	\subfloat[In Wireless Sensor Networks nodes wake up periodically ($\phi$) for short amounts of time ($\delta$) in order to save energy (duty cycling).]{ \label{fig:intro_dutycycle}  
	\includegraphics[width=0.48\linewidth]{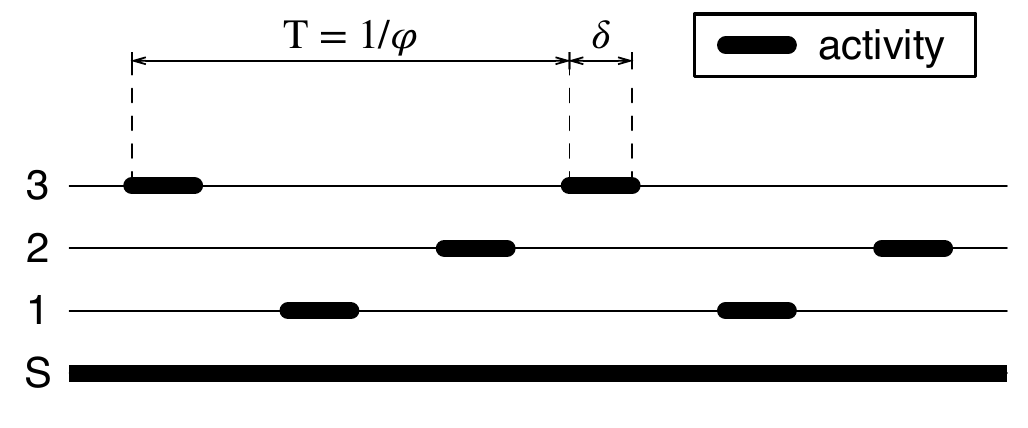} }
	\caption{Representation of a typical Wireless Sensor Network in space and time.}
	\label{fig:intro_wsn_concept} 
\end{figure}

\paragraph{Communication.} 
While duty cycling can easily improve the mote's energy efficiency by two orders of magnitude, the resulting activity periods -- which are short and infrequent -- make
% The effort to achieve low duty cycles can result in unfrequent and short periods of activity, making 
communication between nodes very challenging (\cf \chapref{chapter:sofa}). 
Before communication, two or more nodes need in fact to be active at the same time \ie to \emph{rendezvous}. The smaller the duration of their activity, the harder to rendezvous and communicate.
The solutions to the rendezvous problem
% in duty cycled networks 
can be subdivided in two main categories, \emph{synchronous} and \emph{asynchronous} mechanisms.

\subpar{Synchronous mechanisms.} These mechanisms work by synchronizing the wake-up times of nodes so that all nodes interested in communicating are active at the same time. 
%Because interested in communication to wake up all at the same time. 
This way, motes are able to rendezvous instantly, spending all their energy for communication. 

Unfortunately, node synchronization can be difficult to achieve. 
First, synchronization comes with a communication overhead. In order to synchronize, nodes must first exchange time information to adjust their internal clock (often using asynchronous mechanisms) and decide on a common wake-up schedule. In the case of network-wide communication mechanisms, this overhead can be quite significant. 
Second, smaller activity periods require a more accurate synchronization among the nodes. Thus, extremely efficient communication requires tightly synchronized clocks. 
This feature can be difficult to achieve for wireless sensor nodes, whose time-keeping hardware (oscillators) are cheap and can skew in time \ie running faster or slower than desired. This is due to several reasons such as temperature variations~\cite{Boano2010} and the manufacturing process. 
To address the clock skewing problem, nodes need to continuously exchange their clock information, increasing the communication overhead even more. 
Finally, because all nodes wake up at the same time, the bandwidth available for communication is limited to a small fraction of time, reducing the amount of information that can be communicated. This is particularly problematic in the case of dense networks, where many nodes need to communicate frequently.

\subpar{Asynchronous mechanisms.} Asynchronous mechanisms remove the complexity and overhead of synchronization mechanisms at the cost of longer rendezvous times, effectively trading simplicity and robustness for energy efficiency and communication delay. Because each node has it own knowledge of time, when a devices want to communicate it must first wait for its destination to wake up. 
% As we will see in the next Chapter, the efficiency of rendezvous with these mechanism heavily depends of the activity frequency of the destination node. 
This longer rendezvous increases in turn the latency of the packets and, depending on the medium access control (MAC) implemented, the bandwidth overhead. 
Even though asynchronous mechanisms are far less efficient than their synchronized counterparts they are usually preferred in real-word scenarios, where their simpler mechanisms are less prone to bugs and more robust to network dynamics (expected or not).

% For this reason, almost any mechanism designed for WSN must take into account, to a certain extend, topology dynamics.
\paragraph{Network Dynamics.} 
Even though WSNs' scenarios are commonly static, mechanisms still need to cope to some extent with network dynamics.
%\ie changes in the quality of the wireless links.
For example, 
despite the fact that the device positions do not change, the links' connectivity can drastically fluctuate in time.
This is due to phenomena such as multi-path effects~\cite{Zhou2004a}, which lead to constructive and destructive interference, signal attenuation due to obstacles~\cite{Zhao2003}, and message loss due to external interference~\cite{Brown2014,Hermans2014}.

%\paragraph{Activity Frequency.} 
To cope with network dynamics, nodes must periodically assess the radio channel and update their knowledge about the topology. The higher the dynamics, the more often nodes need to get updates. 
This additional overhead in turn increases the required activity frequency of motes ($\phi$),
% periodic activity of nodes (\emph{activity frequency}), 
and thus, their energy consumption.

\section{Problem Statement}
\begin{table}
	\centering
\begin{tabular}{lcc}
\hline
Characteristic & WSNs & EWSNs\\
\hline
Network size & 10s to 100s & 100s to 1000s\\
Neighborhood size & 10s & 100s\\
Mobility & Quasi-static & Highly mobile\\
Update period & minutes & seconds \\
\hline
\end{tabular}
\caption{Different characteristics between Wireless Sensor Networks (WSNs) and Extreme Wireless Sensor Networks (EWSNs).}
\label{tab:introduction_wsn}
\end{table}

% what makes these networks ``extreme''?
Similar to traditional WSNs, in Extreme Wireless Sensor Networks (EWSNs) devices are equipped with a low-power microcontroller and powered by small batteries or supercapacitors. Thus, they need to be highly energy efficient. Different from WSNs, Extreme Wireless Sensor Networks (EWSNs) raise the challenges of their non-extreme counterparts on four main aspects. 

\subpar{1) Scale.} 
%First, 
EWSNs comprise of thousands of devices, approximately an order of magnitude more devices than traditional WSNs. Because of this, mechanisms that are centralized and require network-wide information have difficulties to scale.  

\subpar{2) Mobility.}
%Second, i
In EWSNs devices are highly mobile, instead of static, resulting in topologies that continuously change over time. In these highly dynamic conditions, mechanisms based on rigid topological structures must strive to timely adapt to the many changes, saturating the already scarce bandwidth and usually taking decisions based on wrong, outdated information.

\subpar{3) Density.}
%Third, b
Because in EWSNs the devices' radio range and the deployment areas are similar in dimensions to the ones in WSNs -- but nodes are far more numerous -- the resulting neighborhood cardinalities are in the order of hundreds of neighbors, rather then tens. 
This increased ``network density'' makes it difficult for communication mechanisms to efficiently share the wireless medium, straining the already limited bandwidth. 

\subpar{4) Variability.}
%Forth, t
The combined effect of high network densities and dynamics makes the network's characteristics 
%conditions
fluctuate drastically both in space and time (\cf \chapref{chapter:nemo}). 
Mechanisms that require complex parameter tuning struggle to	 find the right settings and often under-perform in the highly-variable conditions of EWSNs.

\parspace\noindent Given the aforementioned challenges, this thesis tries to answer the following question 

\begin{center}
\begin{tabular}{|c|}
\hline
\\
\quad \emph{Is it  possible to efficiently communicate in Extreme Wireless Sensor Networks ?} \quad \\
\\
\hline
\end{tabular}
\end{center}

% solution
\parspace Our take is that the conditions of EWSNs require communication mechanisms to reduce their overhead as much as possible to free up the resources for actual data transmissions. 
This idea is distilled in the following four design principles.

% State-less: topology agnostic, robust to reboot
\subpar{1) State-less.}
Due to the scale of EWSNs, nodes cannot rely on methods that are centralized or require up-to-date information from other nodes in the network, either far-away or in the direct neighborhood. In EWSNs failures are the norm, rather than the exception, leading to network partitions due to temporary node unreachability. Protocols should be resilient to these failures and operate 
%as much as possible 
independently from the node and the network states.

% Spontaneous/Passive: listen/observe rather than transmit (based on observation)
\subpar{2) Opportunistic.} 
Communication mechanisms should be opportunistic and exploit existing situations to their own advantage. Information should be extrapolated from passive observations (\cf \chapref{chapter:estreme}) instead of being actively polled from neighbors. 
Mechanisms should not be artificially orchestrated and they should rely as much as possible on emerging behaviors that araise spontaneously with minimal message overhead (\cf \chapref{chapter:staffetta}).

\subpar{3) Anti-fragile.} Communication mechanisms should embrace, rather than fight, the challenges of EWSNs to the point that they perform better in extreme conditions -- where other mechanisms show their fragility (\cf \chapref{chapter:sofa}). 

% Stress-free: works even when bandwidth saturates (backoff, best effort) flexible approximation, rather than optimal
\subpar{4) Robust.} When scarce resources, such as bandwidth, saturate the performance of communication mechanisms must degrade gracefully without drastic disruptions. 
To this end, mechanisms should operate in a best-effort fashion, backing off whenever conditions become too harsh. 
%In mild conditions, instead, a minimal quality of service should be guaranteed. (check this, we do not do it I think)

\parspace 
%Henceforth, providing an energy-efficient stack for communicating in EWSN requires a careful evaluation of the following problem: in duty cycling communication techniques, much of the bandwidth is wasted in coordinating the rendezvous of the (sleeping) nodes, gathering information, adapt to changes. 
Using the four aforementioned principles as guidelines, we present and evaluate three novel communication mechanisms for EWSNs. These mechanisms are able to operate over several hundreds of devices with minimal energy consumption and bandwidth overhead and show a remarkable resilience to network dynamics. 

\section{Thesis Contributions and Outline}
\label{sec:outline}
% d.	Present your outline, i.e. rephrase each key question (see assignment 2) into the objective of a chapter. Spend about one or two sentences on each chapter, making sure that you sufficiently specify the relationship between each chapter and the next.
By using the crowd monitoring scenario as a reference, this thesis tackles the problem of communication in extremely dense and dynamic scenarios, making a step back towards the original Smartdust idea.

\paragraph{Medium Access Control and Data Dissemination -- \chapref{chapter:sofa}}
% Second, we propose a novel communication primitive called \emph{opportunistic anycast},
% explain opp.any. + say it well suits EWSN (performs better in dense rather than sparse networks)
In this chapter we address the cornerstone problem of Extreme Wireless Sensor Networks: efficiently sharing the communication medium among hundreds of continuously changing neighbors. In EWSNs nodes are mobile, and up-to-date network information is limited and costly to obtain. 
To cope with these challenges we propose \emph{SOFA}, a novel medium access control (MAC) protocol that is specifically designed for EWSNs. SOFA does not require any neighborhood information (state-less principle) and performs better in dense rather than sparse networks (anti-fragile principle).
To minimize the delay and overhead of communication SOFA nodes opportunistically forward their information to the first available neighbor \ie the first to wake up while duty cycling its radio. The more neighbors, the shorter the rendezvous time, the lower the communication overhead (opportunistic principle).

The randomized neighbor-selection of SOFA provides is ideal for epidemic mechanisms, such as Gossiping, that are particularly suitable for dense and mobile networks. The combination of gossip mechanisms and SOFA is able to provide a method to process information in EWSNs  
% become the first communication mechanisms that is able to communicate in these networks 
with minimal energy consumption (robustness principle). 

\paragraph{Neighborhood Cardinality Estimation -- \chapref{chapter:estreme}}
% Knowing the neighborhood cardinality in wireless networks is an essential building block of many networking algorithms, such as resource allocation and random-access control. Cardinality estimation is also a valuable tool on itself. It can be used to monitor the surrounding environment \eg crowd density, transforming the radio device into a smart sensor.
% Th this end,
This chapter addresses the problem of estimating the neighborhood cardinality in dense and dynamic networks, where traditional techniques based on neighbor discovery do not work. 
In these extreme conditions, the neighbor discovery process takes a lot of time and neighbors are simply too many to be discovered before the network conditions change.
% \footnote{As we will see in \chapref{chapter:nemo}, this problem can lead to high over-estimations of the neighborhood cardinality.}.
To this end, we propose a mechanism called \emph{Estreme} that estimates, rather than counts, the number of devices in the neighborhood. 
Estreme models the performance of asynchronous duty-cycled MAC protocols, such as SOFA, with respect to the number of active neighbors and uses this model to estimate the neighborhood cardinality with minimal delay and overhead (opportunistic principle). 
%
% In this Chapter we model the performance of duty cycled MAC protocols \eg Sofa, respect to the number of beaconing devices in the neighborood (network density). By exploiting this peculiar relation, we show that is possible to estimate the neighborhood cardinality in a fast and distributed fashion, with minimal message overhead (opportunistic).
% We show that Sofa's performance are agnostic to many network characteristics such as size and mobility, but depend on the local network density. By exploiting this unique characteristic, we build a neighborhood cardinality estimator with minimal overhead, that computes its estimation based only on few observations of the radio traffic (passive) that only requires few observations of the radio traffic.

Estreme improves existing techniques with faster and more accurate estimations and, for the first time, allows all nodes to perform this process simultaneously without interfering with each other. 
The latter characteristic
%timely estimations of \emph{Estreme} 
proves to be particularly useful for crowd-monitoring applications, where crowd managers want, for example, to continuously monitor the changes in the crowd density.

\paragraph{Opportunistic Data Collection -- \chapref{chapter:staffetta}}
% \think{I think this part goes too much into details. Perhaps describe first that given that traditional routing structures would not be applicable to EWSN [this would relate to your prior description of WSN/EWSN], we need an low-overhead opportunistic emergent behaviour.}
This chapter explores the challenges of data collection in EWSNs. Traditional routing mechanisms for WSNs are based on rigid structures and cannot cope with the ever-changing topologies of EWSNs. 
In these dynamic conditions, the communication overhead required to maintain the routing structures increases up to the point that bandwidth saturates and mechanisms (mis)route data based on outdated and wrong information.
To cope with the aforementioned problems, in this chapter we propose a data collection mechanism called \emph{Staffetta}, whose routing structure is highly adaptable and is spontaneously formed from local observations and with minimal overhead (state-less principle).
Staffetta exploits the fact that in opportunistic mechanisms it is more probable and efficient to communicate with the neighbors that wake up more often \ie that are more active. 
In the presence of a fixed energy budget and a sink (which is always active) Staffetta spontaneously creates an activity gradient, where nodes closer to the sink are more active than others and, thus, opportunistically attract more data. 
Compared to existing data collection mechanisms, Staffetta shows higher delivery rates, lower energy consumption and shorter packet latencies. 

\paragraph{Crowd Monitoring in the Wild -- \chapref{chapter:nemo}}
Finally, this chapter explores the deployment challenges of a real EWSN composed of several hundreds of mobile motes. 
This study is motivated by the fact that many protocols for WSNs (extreme or not) are developed and tested under limited, controlled conditions. 
From high-level simulations to large and mobile testbeds. Even though these tools are highly valuable for the first-stage deployment of wireless systems, many researchers stop their evaluation at these controlled conditions, often missing the variety of conditions that only real-word scenarios exhibit. This is particularly true for monitoring applications and EWSNs.

This chapter reports our deployment experience running SOFA and Estreme as a mechanism to monitor the popularity of exhibits in NEMO, a modern, open-space science museum. Different from testbed and previous mobile experiments, the museum experiment exposed our mechanisms to drastic network changes, showing the importance of certain kinds of information, such as the unique node identifier, for the correct sensing and estimation of the crowd parameters. 

\vspace{2em} \noindent Chapters \ref{chapter:sofa}, \ref{chapter:estreme}, and \ref{chapter:staffetta} are based on the following papers:

\begin{itemize}

\item M. Cattani, M. Zuniga, M. Woehrle, K.G. Langendoen, \textit{SOFA: Communication in Extreme Wireless Sensor Networks}, in \textit{11th European Conference on Wireless Sensor Networks (EWSN)} (2014)

\item M. Cattani, M. Zuniga, A. Loukas, K.G. langendoen, \textit{Lightweight neighborhood cardinality estimation in dynamic wireless networks}, in \textit{13th ACM/IEEE International Conference on Information Processing in Sensor Networks (IPSN)} (2014)

\item M. Cattani, A. Loukas,  M. Zimmerling, M. Zuniga, K. Langendoen, \textit{Staffetta: Smart Duty-Cycling for Opportunistic Data Collection}, in \textit{14th ACM Conference on Embedded Networked Sensor Systems (SenSys)} (2016).

\end{itemize}

%% file: sofa/sofa.tex
\chapter{Medium Access Control and Data Dissemination}
%  with Opportunistic Anycast}
\label{chapter:sofa}
% Link to previous chapter
\blfootnote{Parts of this chapter have been published in EWSN'14, Oxford, United Kingdom~\cite{Cattani2014}.}

\epigraph{``It was logical to cultivate multiple options''}{Spock}

%% Start the actual chapter on a new page.
%\newpage

% A problem/gap
%\section{Introduction}
%\label{sec:sofa_introduction}

% Link to previous chapter
\dropcap{T}{he} protocol stack of Wireless Sensor Networks (WSNs) has been mainly designed and optimized for applications satisfying one or more of the following
conditions: (i) low traffic rates (a packet per node every few minutes), (ii) medium sized densities (tens of neighbors), and (iii) static topologies. 
In Extreme Wireless Sensor Networks (EWSNs) however, these relatively mild conditions do not hold and traditional protocol stacks simply collapse: thousands of mobile nodes with hundreds of neighbors that need to disseminate information at a relatively high rate (a packet per node every few seconds, instead of every few minutes). 

From a networking perspective, communication in these extreme scenarios poses three non-trivial technical challenges. 
First, similar to traditional WSNs, EWSNs also work with devices with limited energy resources and need to rely on radio duty-cycling techniques to save energy. 
Second, due to the network scale and node mobility, we cannot rely on methods that combine duty cycling techniques with central coordinators~\cite{Ferrari2012} or that require some level of synchronization between the wake up periods of a given node and its neighbors~\cite{Hoiydi2004,Tang:2011}. The system must be asynchronous and fully distributed. 
Third, due to their inefficient bandwidth utilization, traditional unicast and broadcast primitives -- which are asynchronous, distributed and built on top of duty cycling techniques~\cite{Buettner2006,Sun2008} -- simply collapse under the traffic-demands of EWSNs.

Henceforth, providing energy-efficient communication in EWSNs requires a careful evaluation of the following problem: in asynchronous duty cycling techniques, much of the bandwidth is wasted in coordinating the rendezvous of the (sleeping) nodes. In EWSNs, nodes need to reduce this overhead to free up the channel's bandwidth for the actual \emph{data} transmissions. 

% The SINGLE aim of this chapter
To tackle this problem, we present \emph{SOFA} (Stop On First Ack), a communication stack for EWSNs composed of a medium access control (MAC) layer based on a novel bi-directional communication primitive that we call \emph{opportunistic anycast}. This primitive establishes a data exchange with the {\em first} neighbor to wake up. In this way, SOFA avoids the need for neighborhood discovery and minimizes the inefficient rendezvous time typical of asynchronous MAC protocols. 

By selecting opportunistically the next (random) neighbor to communicate with, SOFA provides an ideal building block for algorithms based on random sampling~\cite{Cattani2011} and gossip~\cite{Boyd2005,Sarwate2012,Shah2009}. The latter ones are particularly suitable to process information in large-scale distributed systems such as EWSNs and offer an alternative to the traditional protocol stack, that cannot operate in extreme network conditions. 

% Outline of the chapter
After exploring the related work in \secref{sec:sofa_relatedwork}, in \secref{sec:sofa_mechanims} we present the design of SOFA, a communication protocol that utilizes \emph{opportunistic anycast} to overcome the limitations of inefficient rendezvous mechanisms. To scale in EWSNs, SOFA combines the energy efficiency typical of low-power MAC protocols with the robustness and versatility of gossip-like communication.

In \secref{sec:sofa_implementation} we explain how to efficiently implement SOFA on low-cost sensor nodes \ie motes. Implementation is particularly important, since opportunistic mechanisms are known to suffer 
%from excessive packet collisions and duplicates 
when too many neighbors are available to forward the same packet (acknowledgments collision and duplicate packets).

Finally, in \secref{sec:sofa_evaluation}, we extensively evaluate SOFA both in simulations and testbed experiments. Results show that SOFA can successfully deliver messages, regardless of mobility, in networks with densities of hundreds of nodes while maintaining the duty cycle at approximately 2\%.

\section{Related Work}
\label{sec:sofa_relatedwork}
The constrained energy resources of WSNs led to a first generation of protocols that traded bandwidth utilization for lower energy consumption. 
Such protocols are based on asynchronous node operations, which implies that senders need to wait for their receiver to wake up (\emph{rendezvous phase}) before sending their data. 
While in low power listening (LPL)~\cite{Buettner2006} nodes send a beaconing sequence until the receiver wakes up, in low power probing (LPP)~\cite{Dutta2010,Sun2008}, the sender waits for a wake-up beacon from the receiver (\cf \figref{fig:sofa_lpl_lpp}). 

\begin{figure}
	\centering
		\subfloat[Low Power Listening (LPL).]{ \label{fig:sofa_lpl} 
		\includegraphics[width=0.55 
		\textwidth]{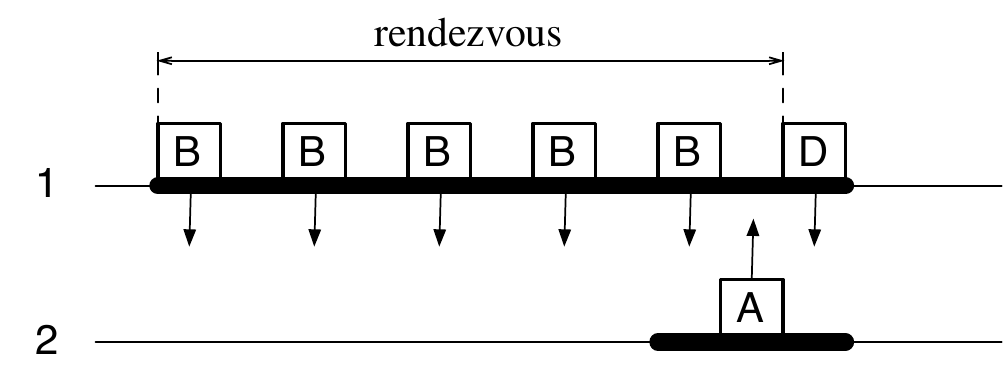} } \\
		\subfloat[Low Power Probing (LPP).]{ \label{fig:sofs_lpp} 
		\includegraphics[width=0.55 
		\textwidth]{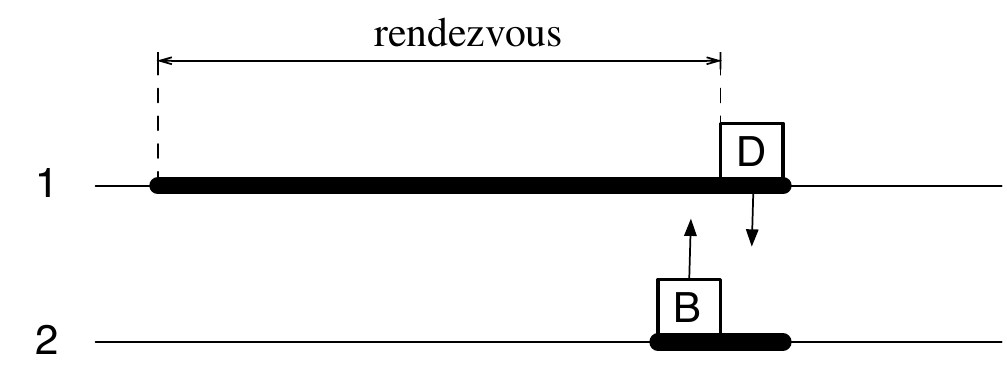} } 
	\caption{Asynchronous communication over a duty cycled mechanism. In LPL, the initiator (node 1) sends a strobe of beacons (B) until the destination wakes up and send an acknowledgment (A). In LPP, instead, the initiator simply waits for a wake-up beacon (B) from its destination. For both mechanisms, data (D) is sent after the rendezvous phase.} 
	\label{fig:sofa_lpl_lpp}
\end{figure}

Despite their high overhead -- in asynchronous mechanisms, the expected rendezvous takes half of the wake-up period~\cite{Buettner2006,Moss:2008} -- these protocols are good enough for traditional WSNs, where the data rate is low (few senders, many receivers) and most of the bandwidth can be used for coordination purposes, rather than data transmission. With a receiver's wake-up period $T$ = 100\,ms, for example, the rendezvous phase will last on average 50\,ms, a lot more than the few millisecond typically needed in WSNs for the data communication. 
In the best scenario, this will allow no more than 20 nodes to transmit their data every minute. Beyond this threshold, a typical case for EWSNs, the channel saturates. 

The WSNs community is well aware of the limitations of the first generation of low-power protocols, which are significant even for the milder condition of non-extreme WSNs. 
%and several notable contributions have improved their performance. 
To reduce the overhead of the rendezvous phase, protocols such as WiseMAC~\cite{Hoiydi2004} and ContikiMAC~\cite{Dunkels2011} keep track of the wake-up periods of their neighbors and use this information to wake up just a few instants before the intended receiver (phase lock). 
This type of protocols works very well on networks with stable topologies, where the overhead of estimating the neighbors' waking periods is seldom done. 
The highly mobile scenarios of EWSNs, however, reduce significantly the efficacy of these methods. 

Higher efficiency can be achieved with global synchronization \ie having all nodes to (briefly) wake up at the same moments in time. 
Even though these mechanisms limit the available bandwidth to a small fraction of time, they practically eliminates the ``rendezvous'' problem and therefore are widely present in WSNs literature. A first family of protocols based on time slots (TDMA)~\cite{Hoesel2004,VanDam2003,Ye2002} proved to achieve high throughputs, but struggle when topologies are too dynamic and slot must often be re-allocated. On the other hands, synchronized protocols based on random access (CDMA)~\cite{Hoiydi2004,Halkes2007} are more robust to dynamics but only work in sparse networks, due to their inefficient bandwidth utilization.

Finally, a new, growing family of protocols exploits physical phenomena of the radio such as the capture effect and the constructive interference to avoid collisions and let nodes to efficiently flood the network.
These mechanisms~\cite{Ferrari2012,Landsiedel2013} proved to be agnostic to mobility and resilient to interference, while consuming low amounts of energy. Unfortunately, the radio phenomena they are based upon are very susceptible to large networks and high densities, and significantly depend on the involved hardware communication layer.

% connect to previous challenges
% To summarize, from a networking perspective communication in mobile scenarios with high traffic rates and high densities poses three non-trivial technical challenges. First, similar to traditional WSNs, these EWSNs also work with devices with limited energy resources and need to rely on radio duty-cycling techniques to save energy (energy efficiency). Second, due to the network scale and node mobility, we cannot rely on methods that require some level of synchronization between the wake up periods of a given node and its neighbors. The system must be asynchronous and fully distributed. Third, due to their inefficient bandwidth utilization, traditional communication primitives simply collapse under the traffic-demands of EWSNs.
%
% Henceforth, providing an energy-efficient stack for communicating in EWSN requires a careful evaluation of the following problem: in asynchronous duty cycling techniques, much of the bandwidth is wasted in coordinating the rendezvous of the (sleeping) nodes. In EWSNs nodes need to reduce this overhead to free up the channel’s bandwidth for actual data transmissions.

\section{Problem Statement}
Several notable studies have identified the important role that opportunistic communication has on improving the performance of low-power WSNs. In essence, the key idea of these studies is the following: instead of waiting for a pre-defined node to wake up, opportunistically transmit to who is available now. In ORW~\cite{Landsiedel2012}, the authors propose to use anycast communication to improve the performance of CTP~\cite{Gnawali2009}, the de-facto data collection protocol in WSNs. In Backcast~\cite{Dutta2008}, the authors show that by using anycast communication, the capture effect can be leveraged to increase the probability of receiving an ack from a viable receiver. While SOFA is motivated and inspired by these studies, there is an important difference. We do not use opportunistic anycast to improve the performance of traditional network protocols under mild conditions, but to enable a new communication protocol that scales to EWSNs. 

\section{Mechanism}
\label{sec:sofa_mechanims}
The design of SOFA follows two main goals: reduce the  inefficient rendezvous phase of low-power MAC protocols, and guarantee that the dissemination of data is performed in an efficient and reliable way. To satisfy these goals, SOFA implements an efficient communication primitive, called \emph{opportunistic anycast}, that minimizes the rendezvous overhead and natively supports Gossip, a robust data dissemination technique created for large-scale networks.

Before proceeding it is important to remark that SOFA focuses on maximizing the messages exchanged locally among neighbors (1-hop), leaving the multi-hop dissemination and aggregation of information to the upper layer \eg Gossip. 

\begin{figure}
	\begin{center}
		\subfloat[Normal conditions]{ \label{fig:sofa_mechanism_normal} 
		\includegraphics[width=0.55 
		\textwidth]{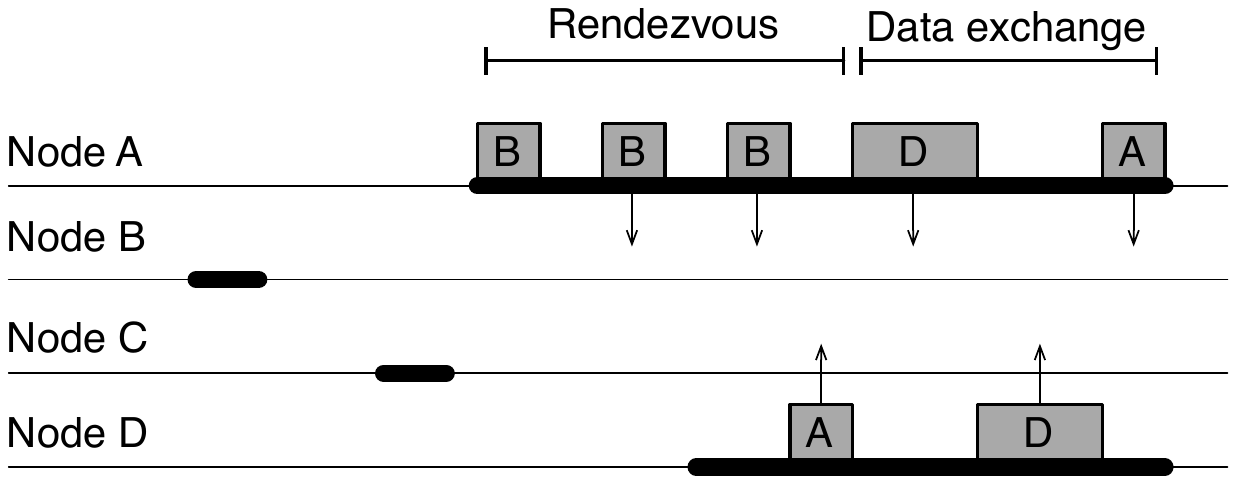} }\\ 
		\subfloat[Collision resolution (cf. Section~\ref{sec:sofa_implementation})]{ \label{fig:sofa_mechanism_collision} 
		\includegraphics[width=0.55 
		\textwidth]{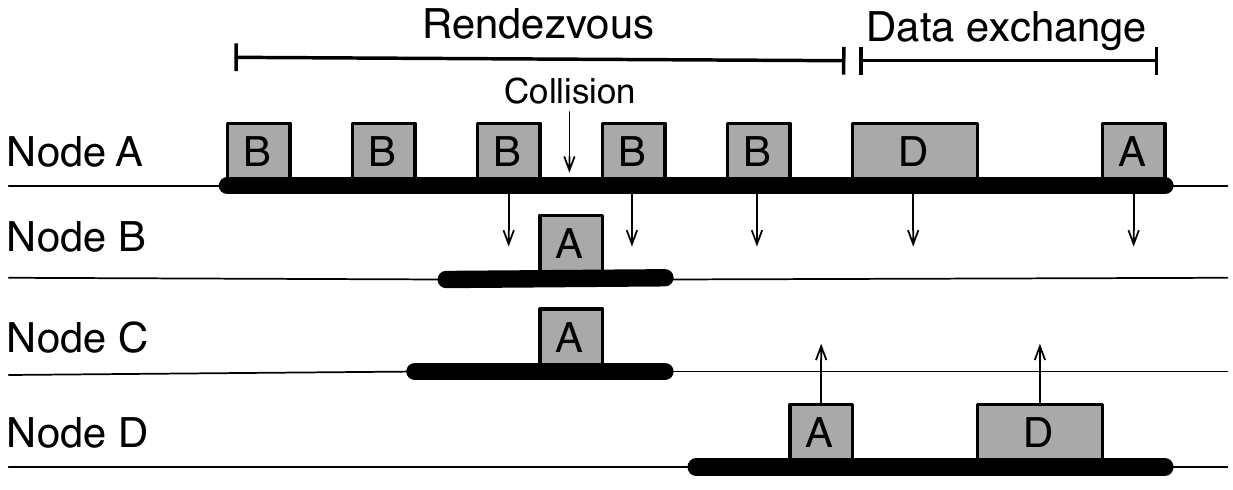} } 
	\caption{SOFA mechanism} 
	\label{fig:sofa_mechanism}
	\end{center}
\end{figure}
\subsection{The basic idea}
The general idea of SOFA can be applied to any asynchronous duty cycled MAC protocol. We focus our analysis on the LPL version of SOFA. The reason is that this implementation performs better in extreme densities, especially in terms of reliability. Nevertheless, in~\cite{Cattani2014}, we provide also insights on the LPP implementation.
 
\vspace{1em}\paragraph{Rendezvous phase.} In traditional LPL protocols~\cite{Buettner2006}, when a sender wakes up, it transmits a series of short packets --called beacons-- and waits for the receiver to wake up. When the intended receiver wakes up, it hears the latest beacon and sends an acknowledgement back. SOFA follows a similar mechanism: the sender, node~A in Figure~\ref{fig:sofa_mechanism_normal}, also broadcasts a series of beacons but only waits until \emph{any} neighbor wakes up. The main difference between the two mechanisms lays in the selection of the destination. While in LPL the destination is chosen by the upper layers in the stack, in SOFA the MAC protocol opportunistically chooses the destination that is most efficient to reach: the first neighbor to wake up. If nodes~B or C were to be chosen, node~A would need to send beacons (jam the channel) until these nodes wake up again. By sending its data to the first neighbor that wakes up (node~D), SOFA reduces the nodes' rendezvous time, allowing low-power MAC protocols to efficiently scale to EWSNs. We call this communication primitive \emph{opportunistic anycast}.

\paragraph{Data exchange phase.} Selecting the first (random) neighbor that wakes up as \emph{the} destination, has a strong relation with a family of randomized networking algorithms called gossiping~\cite{Boyd2005,Shah2009}. \emph{Gossip algorithms do not aim for traditional end-to-end communication} (where routes are formed and maintained ahead of time), instead they exchange information randomly with a neighbor (or subset of neighbors). 
The relation between SOFA and Gossiping is fundamental for the practical impact of our work. Unicast and broadcast primitives allow the development of a wide-range of algorithms and applications in WSNs such as routing, data collection, querying and clustering (to name a few). Unfortunately, under the stringent characteristics of EWSNs these basic primitives collapse. Our aim is to provide an alternative communication protocol for extreme conditions. We hope that this effort will allow the community to use SOFA as a basic building block for other gossip applications such as routing in delay tolerant networks~\cite{spyropoulos2005} and landscaping of data~\cite{Loukas2013}.

We will now describe the design of the three key characteristics of SOFA: \emph{short rendezvous phase, reliable push-pull data exchange,} and \emph{random peer sampling}. The design of a short rendezvous phase was influenced by the limitations of asynchronous duty cycled protocols. The push-pull data exchange and the random peer sampling were designed to satisfy the needs of general gossiping applications.

\subsection{Short rendezvous phase}
\label{sec:sofa_shortrendezvous} 
Stopping at the first encounter, instead of searching for a specific destination, has two important consequences on the performance of SOFA. First, and most importantly, it eliminates the main limitation that LPL has under extreme networking conditions: channel inefficiency. By drastically reducing the length of the rendezvous phase, the channel no longer gets easily saturated by medium/high traffic demands or medium/high node densities. A short rendezvous phase also reduces the duty cycle of the radio, which in turn, increases the lifetime of the node. Second, increasing the network's density (up to a point) improves the performance of SOFA. 
With more neighbors, the probability that one will soon wake up is higher.

To quantify the benefits of a short rendezvous phase, we present a simple model that captures the expected duration of the rendezvous phase as a function of the neighborhood size and the wake-up period (the time elapsed between two consecutive wake-ups of a node). Since nodes wake up periodically in a completely desynchronized way, we can model the inter-arrival times of the nodes' wake-ups as a set of independent random variables with uniform distribution. The first order statistic $U_1$ can then be used to estimate the length of the rendezvous phase. The expected length $E[U_1]$ of $N$ uniform random variables (neighbors) is given by the Beta random variable with parameters $\alpha$=1 and $\beta$=$N$ $$ U_1\sim B(1,N), \quad E[U_1]=\frac{1}{1+N} $$ 
Given a wake-up period $W$ and a neighborhood size $N$, the expected length of the rendezvous phase of SOFA can be computed as follows: 
\begin{align}
	\label{eq:sofa_rend_len}
	 E[s] = \frac{W}{1+N} 
\end{align}
Considering that the expected rendezvous time of unicast $E[u]$ is $W/2$~\cite{Buettner2006,Moss:2008} and that the time spent for the data exchange phase is negligible compared to the rendezvous (see Figure~\ref{fig:sofa_rv_testbed}), we can model the gain $G$ of SOFA compared to unicast as the following:
$$ G = \frac{E[s]}{E[u]} = \frac{W}{1+N} \frac{2}{W}= \frac{2}{1+N}$$
For a node with 99 neighbors, this means that the expected rendezvous times of SOFA is 50 times smaller than the one using unicast.  Figure~\ref{fig:sofa_rv_testbed} compares the expected length of the rendezvous phase using the proposed model with values observed in testbed experiments. In this example, $W$=1\,s and the neighborhood size ranges from 5 to 100 nodes.  The slight underestimation is mainly due to collisions, which delay the detection of the first node by the sender.
 
% It is important to highlight two key points about the impact of density on SOFA. First, since the performance of SOFA is not significantly affected by changes in medium/high densities, SOFA does not need to adapt to the density fluctuations of mobile networks. Second, to reduce the duration of the rendezvous phase in low density networks, the wakeup period can be reduced (at the cost of increasing the duty cycle). This trade-off is studied in more detail in Section~\ref{sub:exploreparameters}.

It is important to highlight three key points about the impact of density on SOFA. First, since the performance of SOFA is not significantly affected by changes in medium/high densities, SOFA does not need to adapt to this type of density fluctuations in mobile networks. 
Second, to reduce the duration of the rendezvous phase in low density networks, the wake-up period can be reduced (at the cost of increasing the duty cycle). This trade-off is studied in more detail in Section~\ref{sec:sofa_exploreparameters}.
Finally, in case the network switches from an extreme condition to a normal one (low density), the protocol stack can switch to the use of standard broadcast and unicast messages. To detect the density of the network, SOFA can exploit the tight correlation between the number of neighbors and the expected length of the rendezvous phase (\cf \chapref{chapter:estreme}). 
% This relation can be exploited, and the rendezvous can be used to decide which mechanism to use (SOFA or unicast). 
%Note that the relation between the network density and the length of the rendezvous phase can be exploited to detect when a network has no more extreme conditions.
% decide when to use SOFAthis switch We are currently working on a mechanism that estimates the cardinality of a node neighborhood based on the observed rendezvous times of SOFA. We argue that this mechanism could be effectively used by SOFA to decide when to switch to normal communication primitives (unicast and broadcast). 

\begin{figure}
	\begin{center}
		\subfloat[Rendezvous time of SOFA compared to the Beta model]{ \label{fig:sofa_rv_testbed} 
		\includegraphics[width=0.35
		\textwidth]{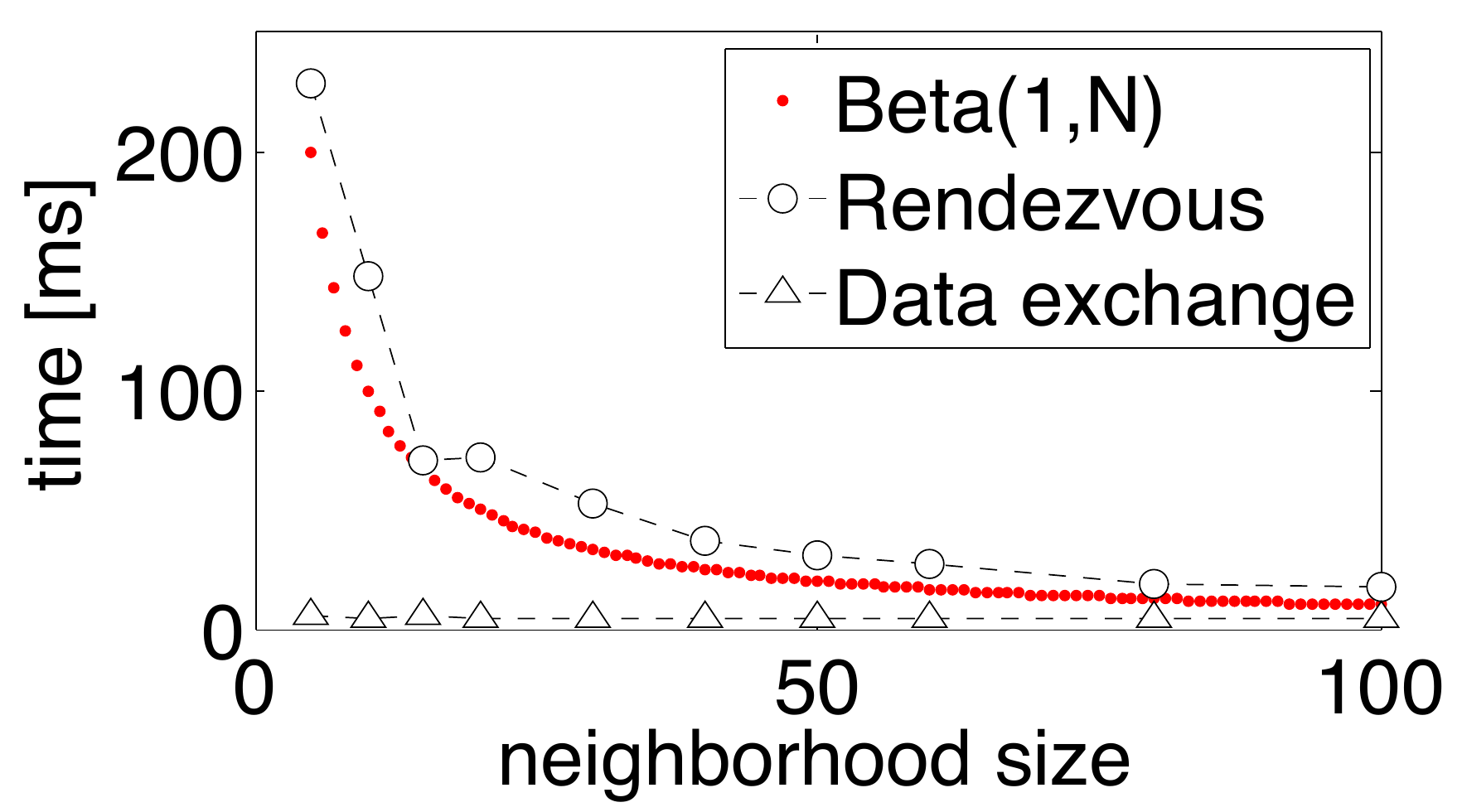} 
		} 
		\hspace{10px}
		\subfloat[Node selection. The score shows how many times each node is selected (cf. Section~\ref{sec:sofa_random})]{ \label{fig:sofa_random} 
		\includegraphics[width=0.35
		\textwidth]{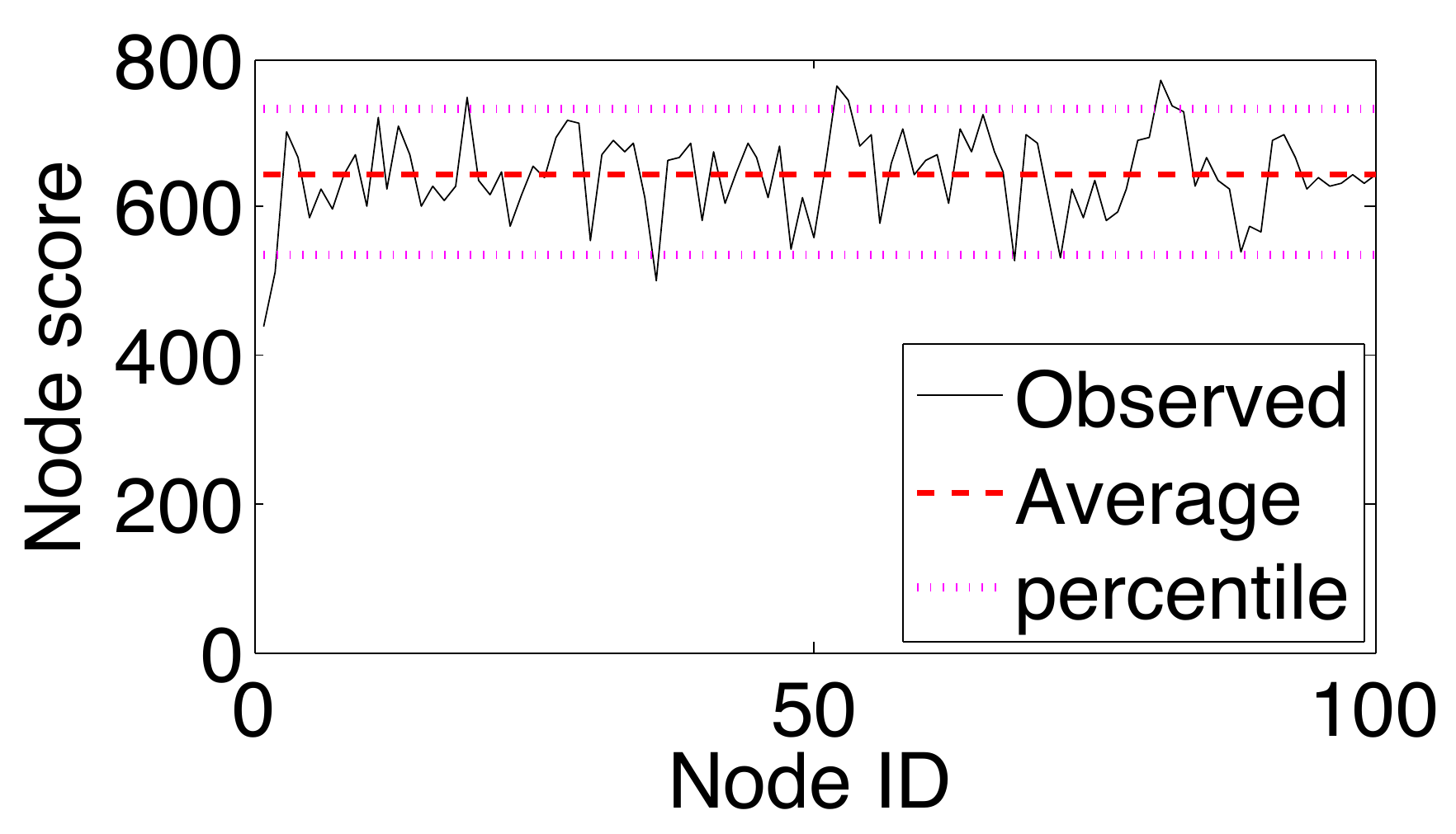} 
		} 
	\end{center}
	\caption{SOFA rendezvous phase (testbed results)} 
	\vspace{-1em}
\end{figure}

\subsection{Reliable push-pull data exchange} 

To exchange data efficiently and reliably, SOFA has two phases: a 2-way rendezvous phase and a 3-way data exchange phase. These phases are shown in Figure~\ref{fig:sofa_mechanism_normal} and their design is driven by two factors: (i) the high relative cost of the rendezvous phase compared to the data-exchange phase, and (ii) the effect of unreliable and asymmetric links on the \emph{constant mass} requirement of gossip's \emph{data-aggregation} algorithms. The effects
% of these factors 
are explained below.

\vspace{2mm}\noindent\textbf{Using a push-pull exchange amortizes the high relative cost of the rendezvous phase.} Gossiping algorithms have two types of data communication: push and push-pull. In the push method, only the sender transfers information to the receiver(s). In the push-pull method, two nodes exchange their information. Compared to the latter, push-pull allows gossip algorithms to compute more complex aggregates and converge faster~\cite{demers1987}.
% Both of these methods have advantages and disadvantages with respect to the way gossiping algorithms process the data, but this discussion is beyond the scope of this paper~\cite{demers1987}. 
Nevertheless, from our perspective what matters most is the relative cost of the rendezvous phase. Given that the cost of this phase is high compared to the data exchange phase, it is beneficial to exchange as much information as possible once two nodes rendezvous. For this reason, SOFA implements a push-pull approach. A push approach would double the overhead of the rendezvous phase, making SOFA less resilient to extreme conditions.
%(due to the less efficient use of the available bandwidth).

\vspace{2mm}\noindent\textbf{The 2-way rendezvous phase filters out asymmetric and unreliable links, while the 3-way handshake reduces the probability of losing ``gossip mass".} Losing messages has a particularly detrimental effect on the accuracy of gossiping. For example, when two nodes agree to swap half their value (mass), the loss of a message results in a too low value on the node that missed it, which influences the outcome of all the other nodes as the error propagates in consecutive rounds. The conservation of mass is, thus, an important issue in gossiping algorithms. From a design perspective, this means that we need to consider two important points. First, nodes should avoid the use of unreliable and asymmetric links (which have been shown to be commonplace in WSNs~\cite{Zamalloa2007}). Second, if a packet is lost, we have to reduce the chances of losing mass. 

The 2-way rendezvous phase reduces the chance of using unreliable and asymmetric links. Several studies have shown that unreliable links are usually asymmetric (and vice versa)~\cite{Zamalloa2007}. On the other hand, bidirectional links are usually characterized by being more reliable. By performing a 2-way exchange before transmitting the actual data, SOFA increases the chances of using a reliable link. It is important to remark that some LPL methods do not follow this approach~\cite{Moss:2008}. These methods piggyback the data on the  beacons and acknowledgement packets, that is, they transmit information without checking first if the link is reliable and symmetric or not. 
% To further guarantee the selection of reliable links, SOFA ignores beacons that are received with an RSS value below a certain threshold (-90 dBm in our implementation).  By tuning this value, we can trade a bit of energy efficiency (the rendezvous phase can be potentially longer) for more reliable links. 

The 3-way data exchange phase reduces the chance of losing mass in the event that a packet is lost. In spite of our efforts to filter out unreliable and asymmetric links during the rendezvous phase, the high temporal variability of low-power links can cause a reliable link to become momentarily unreliable. In the event that a packet is lost, the worst situation for two nodes is to disagree on the outcome of an event. That is, two nodes should either agree that the message exchange was successful (both nodes received the mass) or agree that no message was exchanged (aborting the exchange). If only one node deems the event as successful, then the mass of the other is lost. The latter situation happens when the last packet of an \emph{n}-way handshake is lost. This (dis)agreement problem is discussed in depth in~\cite{Boano2012}, and the authors prove that in WSNs the best strategy to reduce disagreements is to use a 3-way handshake. 

\subsection{Random Peer sampling}
\label{sec:sofa_random} 
Most gossip algorithms rely on the selection of a {\em random} neighbor (or subset of neighbors) at each round. Having a good random selection leads to a faster convergence. To ensure a proper random selection, SOFA introduces random values to the wake-up periods of each node. For a wake-up period of $W$ seconds, nodes wake up uniformly at random between [0.5$W$, 1.5$W$].%Since gossip uses fewer message exchanges to compute the data aggregate, it reduces the bandwidth utilization and energy consumption of the task at hand. While, in unicast and broadcast, gossip algorithms must first discover their neighborhood (by means of broadcasts), and then select a random neighbor to transmit information (by means of unicasts), with SOFA the random selection of nodes is performed by choosing the first awake node. 

To validate the effectiveness of our approach, we performed an experiment on a 100-node testbed. For 10 minutes nodes exchange messages and count the number of times they are selected by their neighbors (their \emph{score}). Figure~\ref{fig:sofa_random} shows that the distribution of the scores is close to uniform, with the [5, 95] percentiles close to the average value. It is important to remark that this evaluation was performed on a static testbed. Mobility would further randomize the selection of neighbors, facilitating the dissemination of data, and drastically reducing the convergence time of Gossip~\cite{Sarwate2012}. 

\section{Implementation}
\label{sec:sofa_implementation}

We implemented SOFA on the Contiki OS based on X-MAC~\cite{Buettner2006}. Nodes were configured to wake up every second for 10 ms. If a beacon is received within this 10 ms period, the node sends an acknowledgement and starts the data exchange phase. Otherwise, the node goes back to sleep. Notice that these parameters set a minimum duty cycle of 1\,\%, hence, any extra activity beyond this point is part of the overhead caused by the rendezvous and data exchange phases. Below we describe the implementation of the most important features of SOFA.

\paragraph{Transmit back-off.} In traditional MAC protocols, before sending a packet, a transmitter first checks the signal level of the channel (CCA check) to see if there is any activity. If no activity is detected the packet is sent. In SOFA, we do not perform a CCA check. Instead, a potential sender listens to the channel for 10\,ms acting, practically, as a receiver. If after this period no packet is detected, the node starts the rendezvous phase. If the node detects a packet that is part of an on-going data exchange, it goes back to sleep (collision avoidance). However, if the detected packet is a beacon, the node changes its role from sender to receiver.  By performing a transmit back-off instead of a CCA check, \emph{SOFA transforms a possible collision between two senders into a successful message exchange with a very low rendezvous cost}. %Figure~\ref{fig:backoff} shows the improvement in the duty-cycle and throughput due to this back-off.

\paragraph{Collision avoidance.} One of the key challenges of operating under extreme density conditions is the higher likelihood of collisions due to higher traffic demands. SOFA follows a simple guideline to reduce the frequency of collisions: if a sender detects a packet loss --for instance, by not receiving an ack--, instead of attempting a retransmission, the node goes back to sleep. This \emph{conservative approach} reduces the traffic in highly dense networks. The main caveat of this approach is when the lost packet is the last data ack. In this case, the two parties will disagree on the data delivery, causing an information (\emph{mass}) loss. Fortunately, our testbed results show that this is not a frequent event.

There is a collision event that is not avoided by the above mentioned approach and has a higher probability of occurrence in SOFA. When two or more active receivers detect a beacon, their ACKs are likely to collide (cf. nodes B and C in Figure~\ref{fig:sofa_mechanism_collision}). The sender will receive neither of the ACKs and will consequently continue transmitting beacons. Upon receiving a subsequent beacon (not the expected data packet), the two colliding receivers infer that a collision has occurred and both will go back to sleep. The first node to wake up after the collision (node D) will acknowledge the beacon and exchanges its data. Finally, randomizing the wake up periods of nodes helps in reducing the chances that this type of collisions occurs repeatedly among the same couples of nodes.  

\paragraph{Packet addressing.} SOFA uses two main types of data packets. For the rendezvous phase, the beacons have no destination address, any node can receive and process the information. For the data exchange phase, the packets contain the destination address of the involved parties. The beacon packets are as small as possible (IEEE 802.15.4 header + 1 byte to define the packet type and 1 byte when addressing is needed).

\section{Evaluation}
\label{sec:sofa_evaluation}

To evaluate the effectiveness of SOFA we ran an extensive set of experiments and simulations. Our testbed has 108 nodes installed above the ceiling tiles of an office building. The wireless nodes are equipped with a MSP430 micro-controller and a CC1101 radio chip. To reach the highest possible neighborhood size, we set the transmission power to +10\,dBm. With these settings, the network is almost a clique. For our simulations we used Cooja, the standard simulator for Contiki. We tested network densities of up to 450 nodes and different mobility patterns. Simulations beyond this density value are not possible with normal cluster computing facilities. 
For both experiments and simulations, the baseline scenario was configured to have a wake up period $W$=1\,s and a transmission period $T$=2\,s. That is, nodes wake up every second to act as receivers and every two seconds to act as senders. Considering that nodes listen for packet activity at each wake-up for 10\,ms, the baseline duty-cycle is $\approx$ 1\,\%. Any extra consumption beyond 1\,\% is caused by SOFA. The evaluation results presented in this section consider also other values for $W$ and $T$, but unless stated otherwise the experiments are carried out using the baseline parameters. The results are averaged over 20 runs of 10 minutes each.

\subsection{Performance metrics}
The evaluation of SOFA focuses on three key areas: energy consumption, bandwidth utilization and \emph{mass} conservation. To capture the performance of SOFA in these areas, we utilize the following metrics:  

\paragraph{Duty cycle.} The percentage of time that the radio is active. Duty-cycle is a widely utilized proxy for energy consumption in WSNs because radio activity accounts for most of the energy consumption in WSNs nodes.

\paragraph{Exchange rate.} The number of successful data-exchanges (3-way handshakes) in a second. This is a per-node metric. If, instead, we count the total number of data exchanges in a second over the entire network we refer to the \emph{global exchange rate}.

\paragraph{Mass delivery ratio.} The percentage of times that the data-exchange phase ends up without any information loss. Recall that if the ack of the data phase is lost, the receiver deems the exchange as successful, but the sender deems the exchange as unsuccessful and ignores the previously received packet (mass loss). The \emph{mass delivery ratio} is a metric focused on evaluating the viability of SOFA as a basic communication primitive to gossip algorithms. 
%Also, notice the subtle difference between the \emph{exchange rate metrics} and the \emph{mass delivery ratio}. The \emph{exchange rate metrics} only measure the percentage of successfully transmitted data packets (bare throughput), but not the percentage of successfully accomplished exchanges -- which is of interest to gossiping. 

\begin{figure}
\begin{center}
		\subfloat{  
		\includegraphics[width=0.35 
		\textwidth]{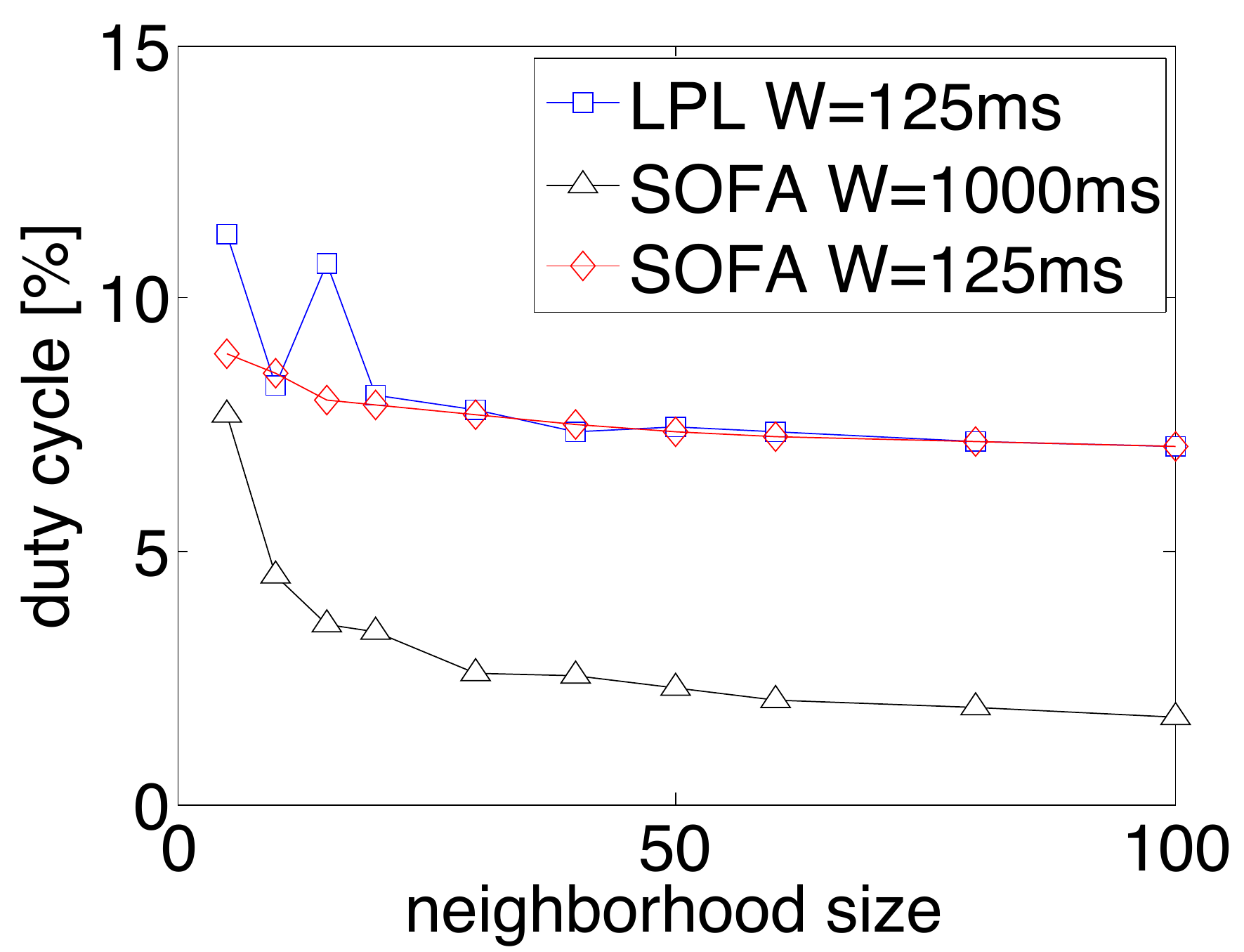} } 
		\hspace{10px}
		\subfloat{
		\includegraphics[width=0.35
		\textwidth]{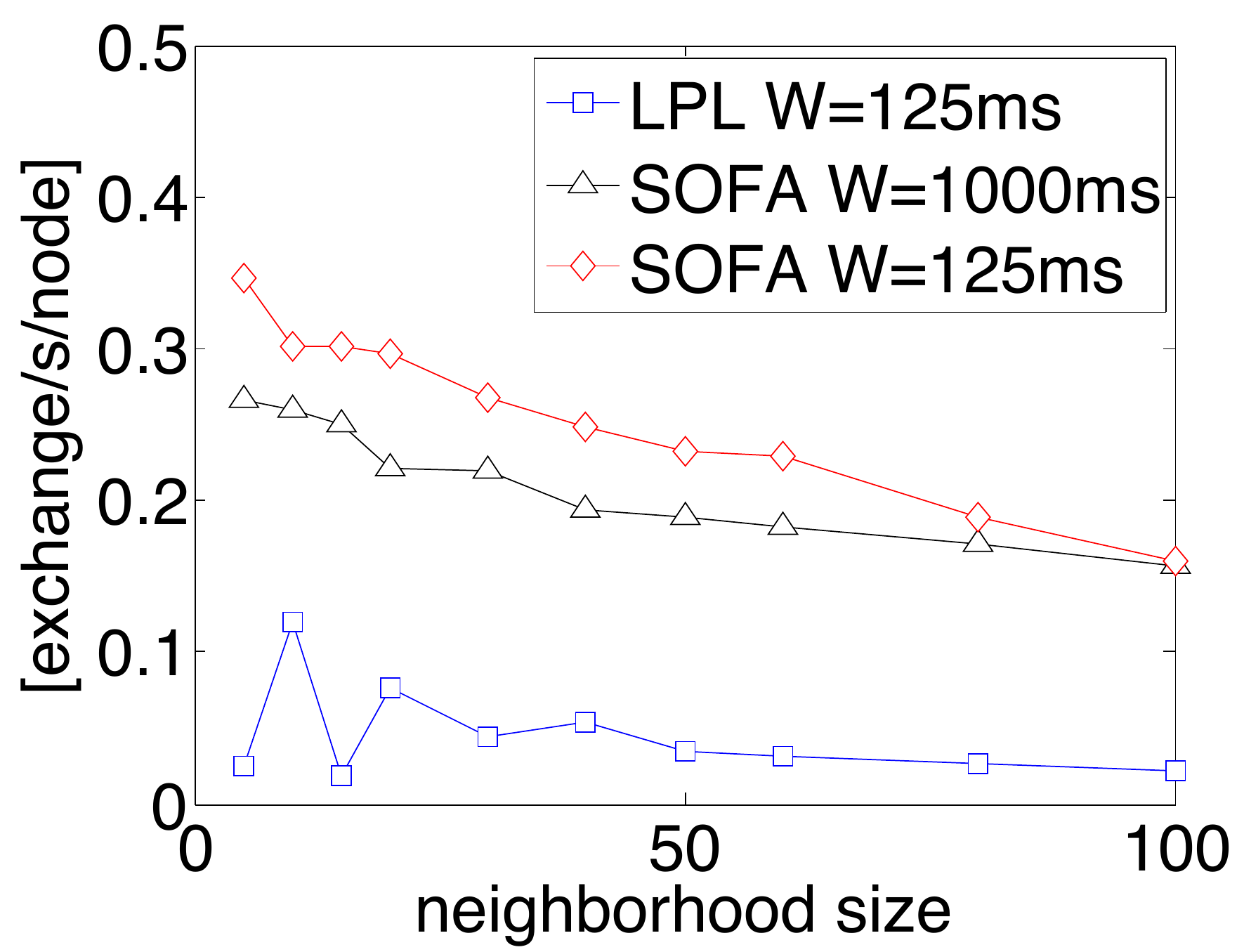} } 
\end{center} 
\caption{SOFA compared to LPL (testbed results).}%
\label{fig:sofa_sota}
\end{figure}%
     
\subsection{Results}
	 
%\subsection{The need for a novel approach}
Previously in this paper, we argued that traditional low-power methods collapse under the stress imposed by extreme networking conditions. This subsection quantifies this claim. We compare SOFA with the standard Contiki implementation of LPL (X-MAC) on our testbed. To provide a fair comparison, LPL chooses a random neighbor from a pre-computed list of destinations at every transmission request. That is, we do not enforce on LPL the necessary neighbor discovery process that would be needed to obtain the destination address (SOFA does not need an address to bootstrap the communication). 

Figure~\ref{fig:sofa_sota} compares the duty cycle and the exchange rate of SOFA and LPL in our testbed. For LPL, the evaluation shows only the result for $W$=125\,$ms$ because LPL collapses with the baseline $W$=1\,$s$. This collapsing occurs because, with $W$=1\,$s$, the rendezvous phase of LPL requires on average 0.5\,$s$. Hence, 5 nodes require on average a 2.5-seconds window to transmit their data, but the transmission period is 2\,$s$, which leads to channel saturation. Comparing the best parameter for LPL ($W$=125\,$ms$) with the best parameter for SOFA ($W$=1\,$s$) shows that SOFA widely outperforms LPL. For most neighborhood sizes (30 and above), SOFA uses four times less energy and delivers five times more packets for the same $T$.

It is important to remark that SOFA is not a substitute for traditional low power methods, as they aim at providing different services. SOFA cannot provide several of the functionalities required by applications relying on unicast and broadcast primitives. Most WSNs applications are designed for data gathering applications sending a few packets per minute. In these scenarios, the state-of-the-art in WSNs research performs remarkably well. The aim of our comparison is to highlight that traditional methods were not designed to operate under extreme conditions neither to efficiently support Gossip applications. We will now analyze the performance of SOFA based on different parameters and scenarios.

\begin{figure}
	\begin{center}
		\subfloat[Duty cycle, $W$=1000ms]{ \label{fig:sofa_dc_different_rates_1hz} 
		\includegraphics[width=0.35 
		\textwidth]{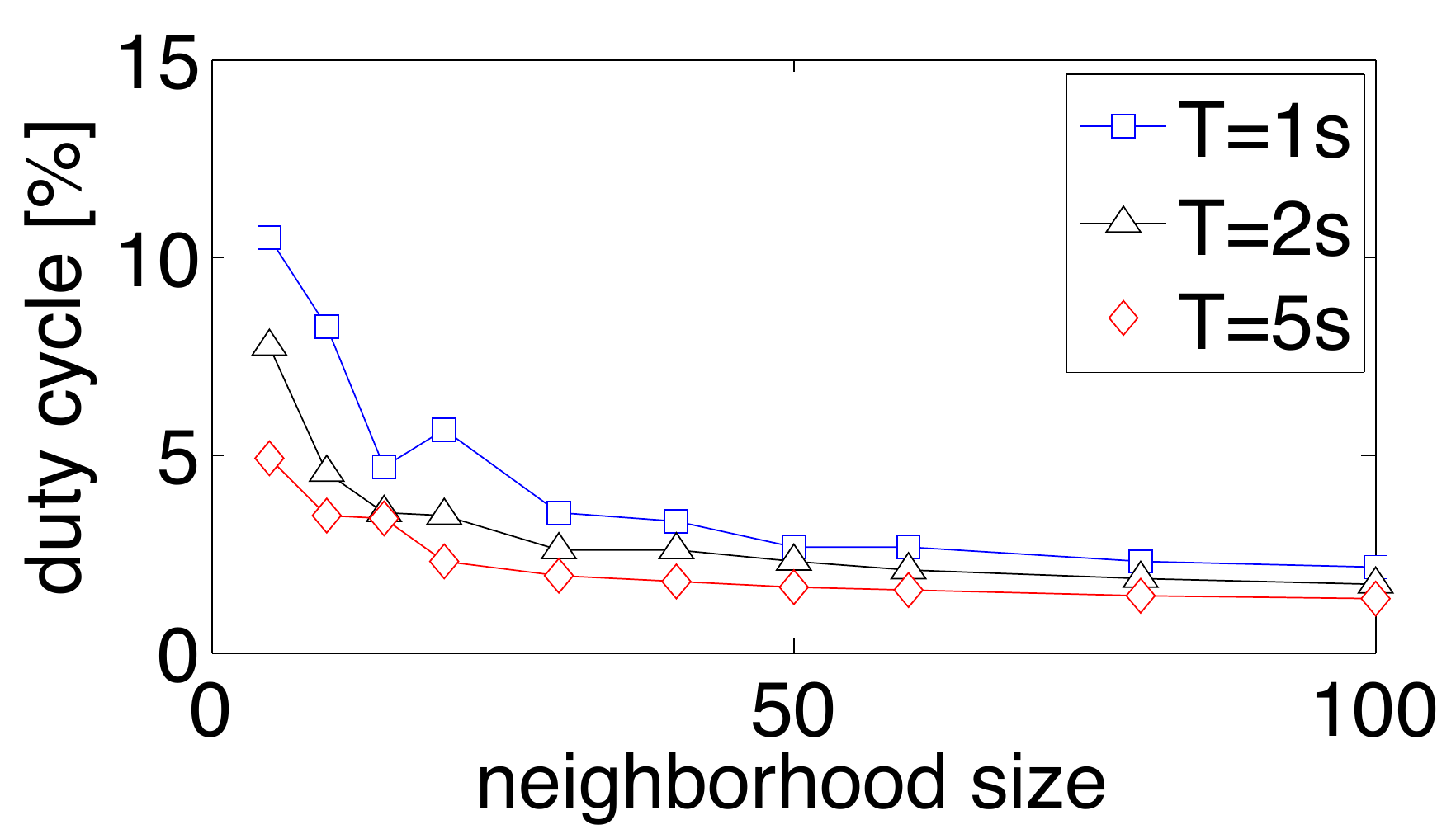} }
		\hspace{10px} 		
		\subfloat[Duty cycle, $W$=125ms]{ \label{fig:sofa_dc_different_rates_8hz} 
		\includegraphics[width=0.35 
		\textwidth]{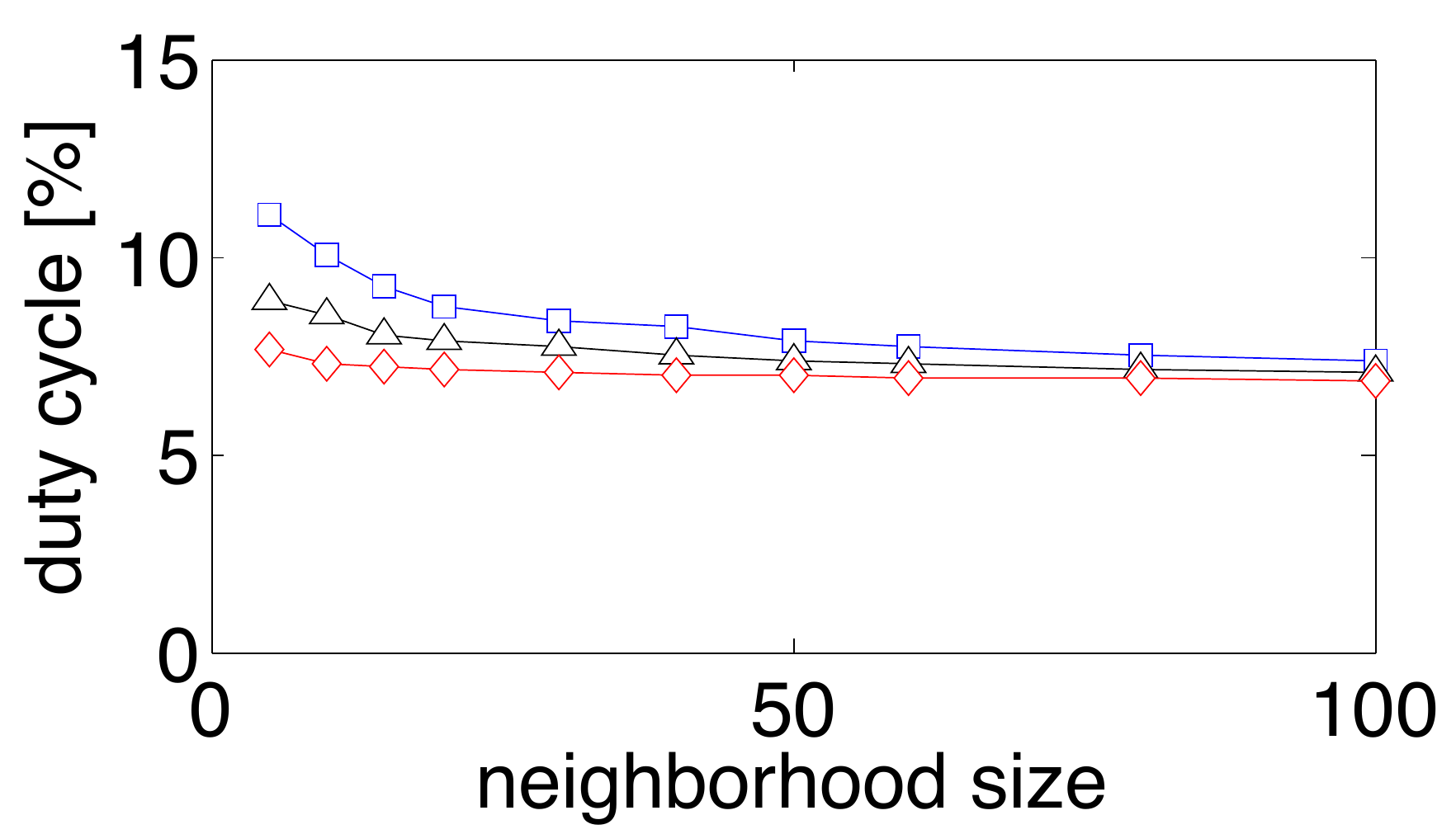} }\\ 
		\subfloat[Normalized Exchange rate, $W$=1000ms]{ \label{fig:sofa_tp_different_rates_1hz} 
		\includegraphics[width=0.35 
		\textwidth]{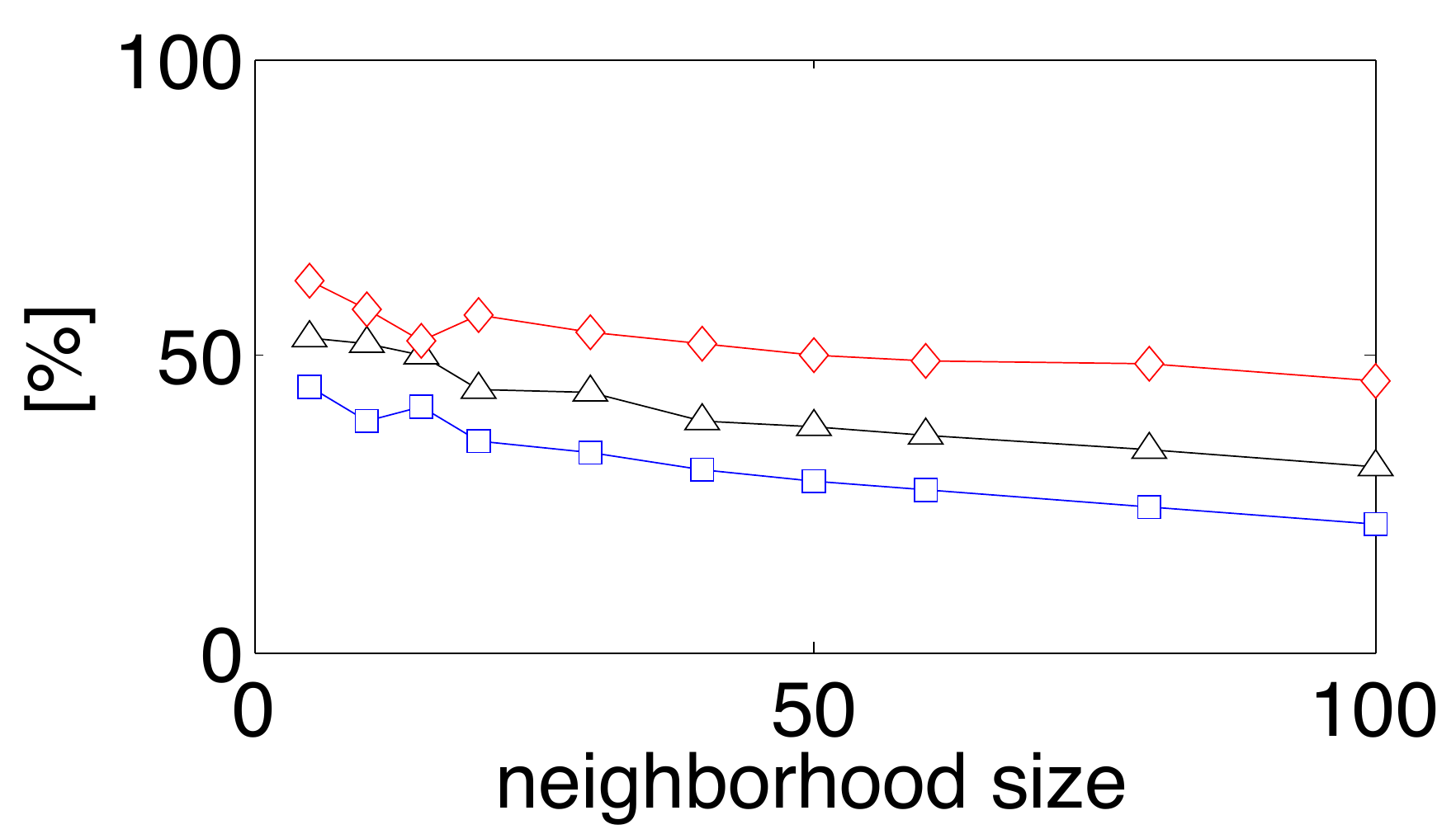} }
		\hspace{10px}
		\subfloat[Normalized Exchange rate, $W$=125ms]{ \label{fig:sofa_tp_different_rates_8hz} 
		\includegraphics[width=0.35 
		\textwidth]{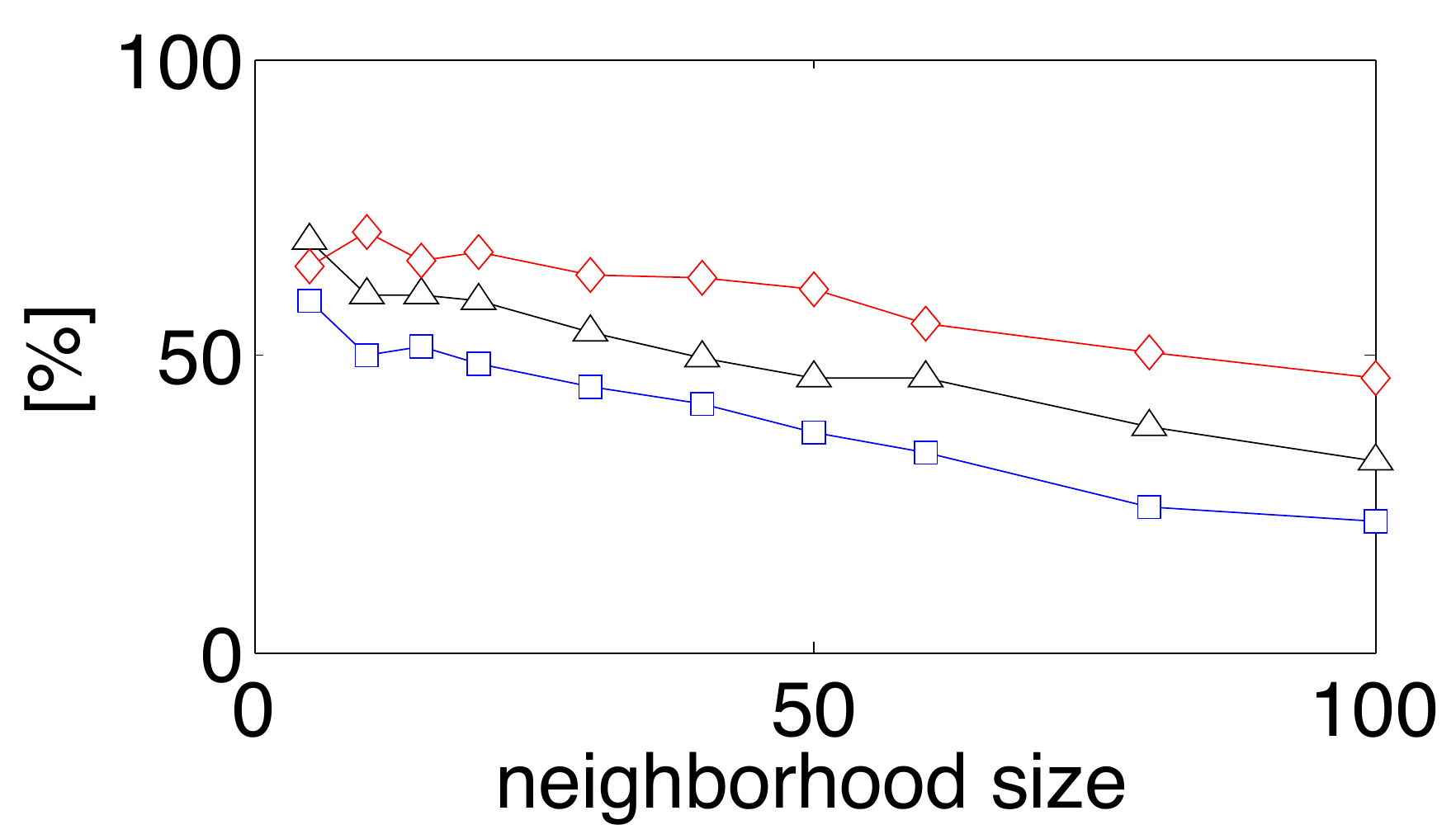} 
		}
	\end{center}
	\caption{Testbed performance for different wake-up times $W$ and transmission periods $T$. Note that the \emph{exchange rate} is normalized to $T$.} \label{fig:sofa_different_rates} 
\end{figure}

\subsection{Exploring SOFA parameters}
\label{sec:sofa_exploreparameters}
% --- EXPLORING SOFA PARAMETERS FIGURES ---
SOFA is a simple protocol with only two parameters available for fine-tuning: the wake-up period $W$ and the transmission period $T$. We now evaluate the performance of SOFA as a function of these parameters. The results of this subsection are all based on testbed experiments. Figure~\ref{fig:sofa_different_rates} shows the performance of SOFA for two  different wake-up periods (125 and 1000 ms), and for three different transmission periods (1, 2 and 5 seconds).

\vspace{1em}\paragraph{The impact of the transmission period T.} Let us start by analyzing the impact of $T$ on the duty cycle. Figure~\ref{fig:sofa_dc_different_rates_1hz} shows two important trends. First, beyond a certain neighborhood size ($\approx$\,30), $T$ does not have a significant impact on the duty cycle. Decreasing the transmission period certainly increases the duty cycle of the node, but not by much. Second, in low/medium dense networks (below 30 neighbors), increasing $T$ has a more significant effect on the duty cycle, but it is still a sub-linear relation. An increment of $T$ by a factor of five, increases the duty cycle by only a factor of two. The reason for the difference in duty cycle between low/medium and high density networks, is that at lower densities, SOFA spends more time on the rendezvous phase. This implies a higher overhead at each transmission attempt. Conversely, increasing the density increases the likelihood of finding a receiver sooner.

Note that, thanks to the \emph{transmit back-off mechanism} (which changes the role of senders to receivers to reduce collisions), increasing the transmission rate decreases the length of the rendezvous phase. With nodes sending data more often, the probability that two senders are active at the same time is higher. While in a normal MAC protocol this would lead to collisions, in SOFA it translates into an efficient message exchange (the rendezvous time is minimal) among the two senders.
As for the impact of $T$ on the exchange rate, SOFA behaves as most protocols do when they work under high traffic demands: the higher the traffic rate, the more saturated the channel, and the lower the probability to exchange information. This trend is observed in Figure~\ref{fig:sofa_tp_different_rates_1hz}. It is important to notice, however, that the exchange rate decreases in a gentle manner.

\paragraph{The impact of the wake-up period W.}  Intuitively, reducing the wake-up period should reduce the rendezvous time (because nodes wake up more frequently), which in turn should free up bandwidth and allow a higher exchange rate. However, the trade-off for a more efficient use of bandwidth would be a higher duty cycle. Figures~\ref{fig:sofa_dc_different_rates_8hz} and \ref{fig:sofa_tp_different_rates_8hz} show the performance of SOFA with a wake-up period $W$=125\,ms. With this value, the baseline duty cycle is 8\,\%. The figures show that reducing $W$ does increase the relative exchange rate, but mainly on low/medium dense networks (by $\approx$\,50\%). Therefore, it is possible to improve the performance of SOFA in low density networks at the cost of increasing the energy consumption. For high density networks, however, we have a similar throughput but with a duty cycle that is four times higher.

\subsection{SOFA under extreme densities and in mobile scenarios} 
\label{sec:sofa_extremedensities}

The previous testbed results show that SOFA performs well in densities as high as 100 neighbors. However, from a practical perspective it is important to determine (i) the saturation point of SOFA, i.e., how many nodes SOFA can handle before saturating the bandwidth, and (ii) the impact of mobility. Unfortunately, there are no large-scale mobile testbeds available in the community, and hence, we rely on the Cooja simulator to investigate these aspects.
\begin{figure}
		\begin{center}
			\subfloat[Duty cycle]{ \label{fig:sofa_dc_cooja_vs_testbed} 
			\includegraphics[width=0.35 
			\textwidth]{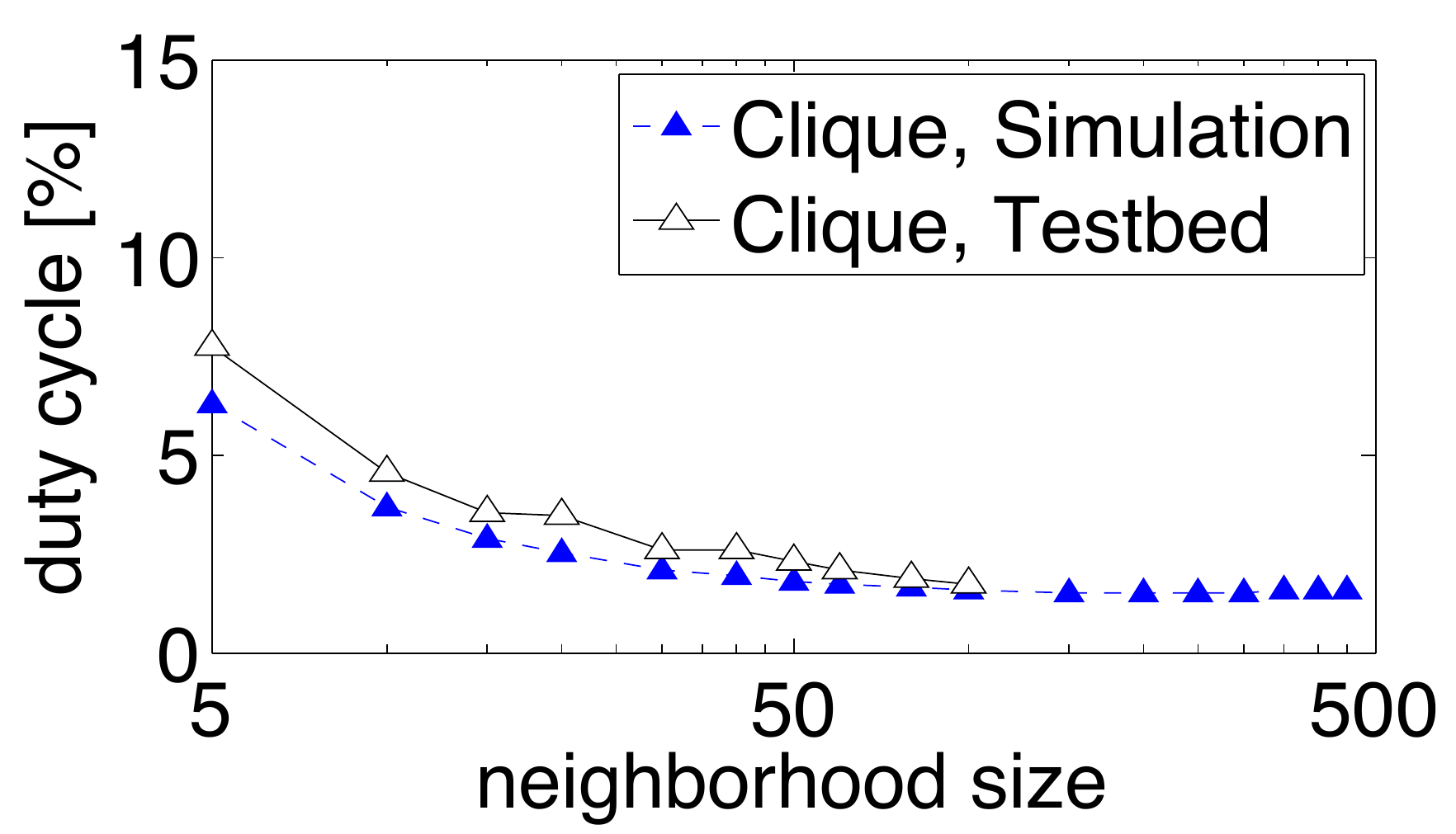} }
			\hspace{10px} 
			\subfloat[Mass delivery ratio]{ \label{fig:sofa_dr_cooja_vs_testbed} 
			\includegraphics[width=0.35 
			\textwidth]{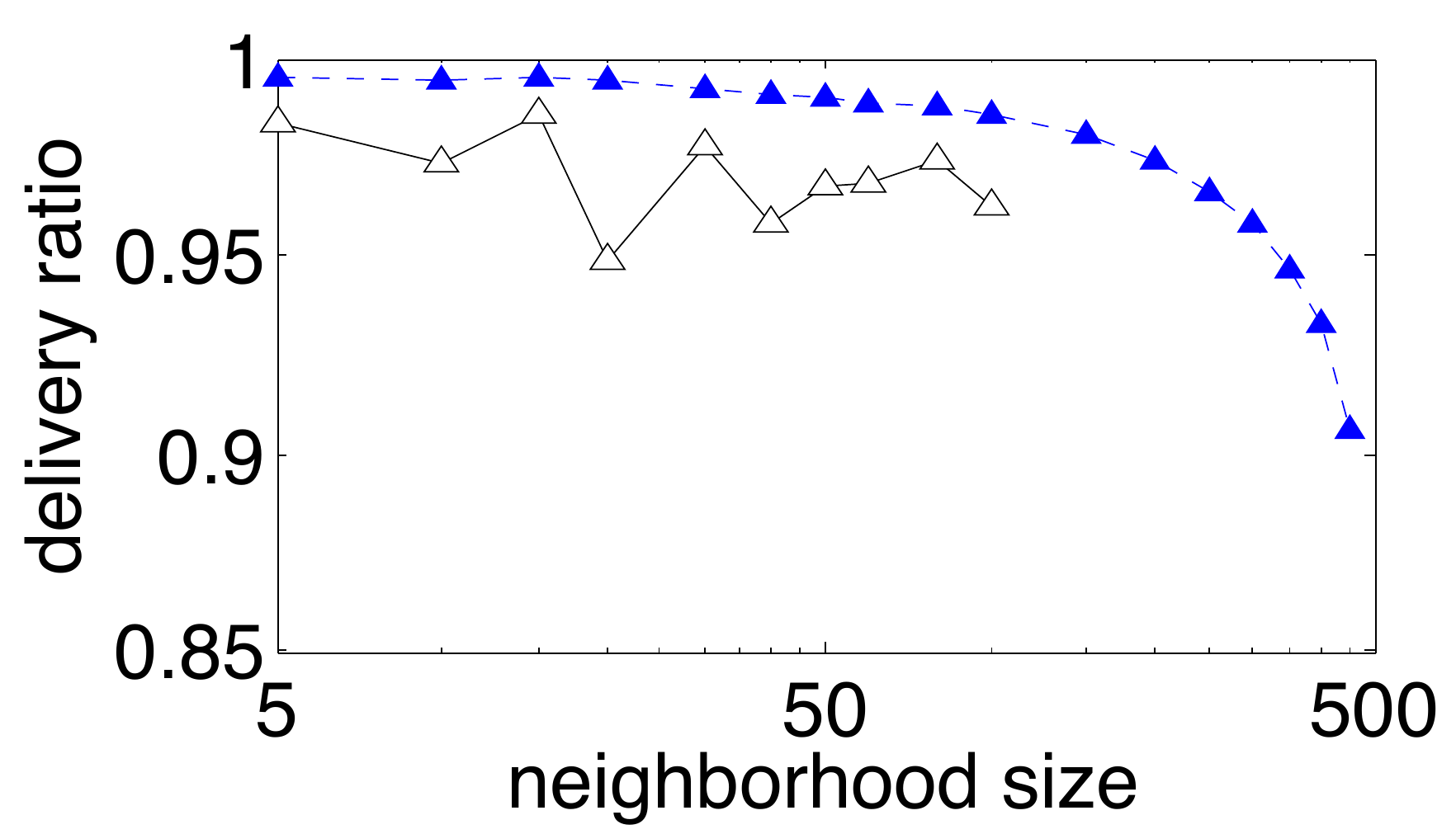} }\\
			\subfloat[Exchange rate]{ \label{fig:sofa_tp_compare} 
			\includegraphics[width=0.35 
			\textwidth]{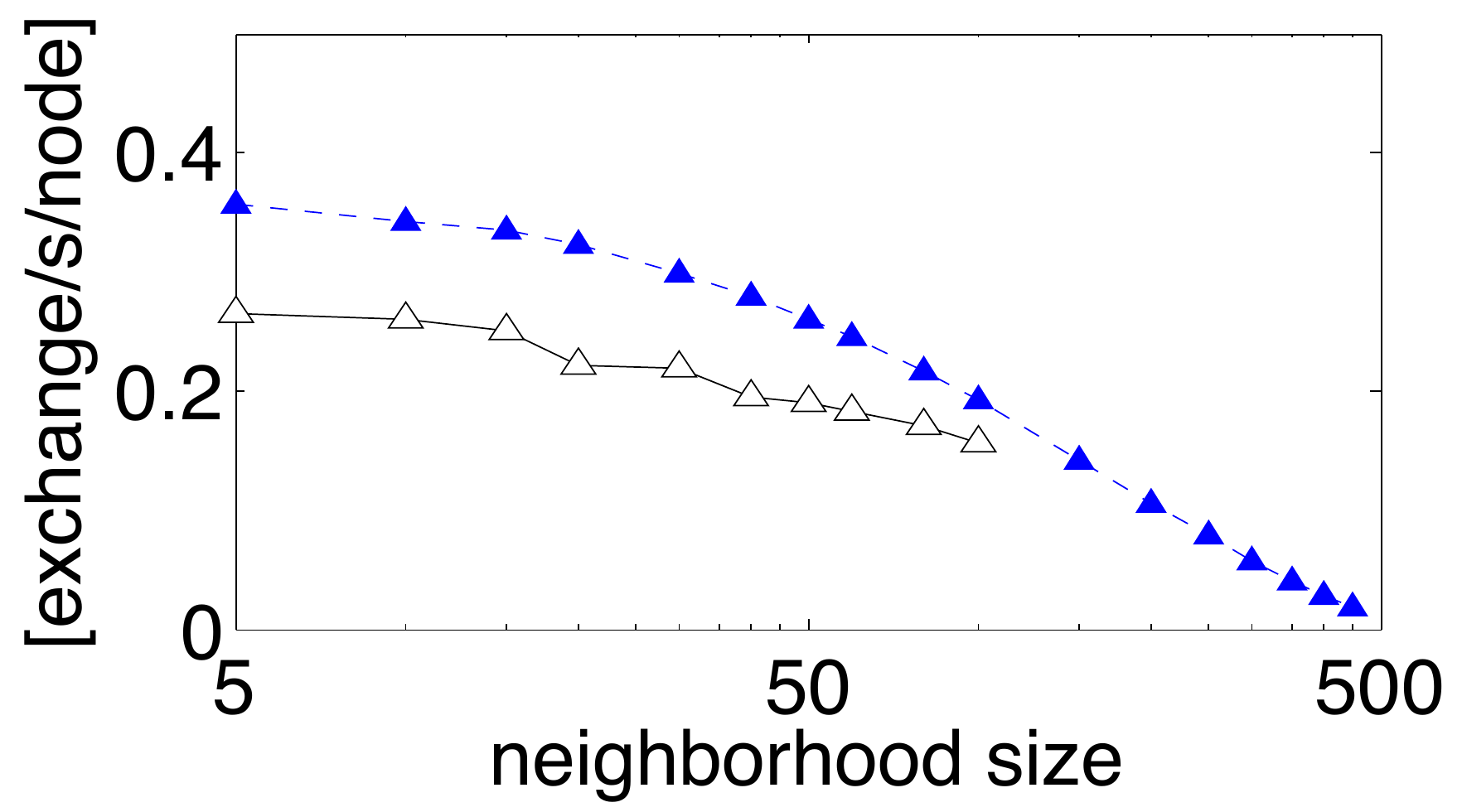} }
			\hspace{10px}
			\subfloat[Global exchange rate]{ \label{fig:sofa_tp_global} 
			\includegraphics[width=0.35 
			\textwidth]{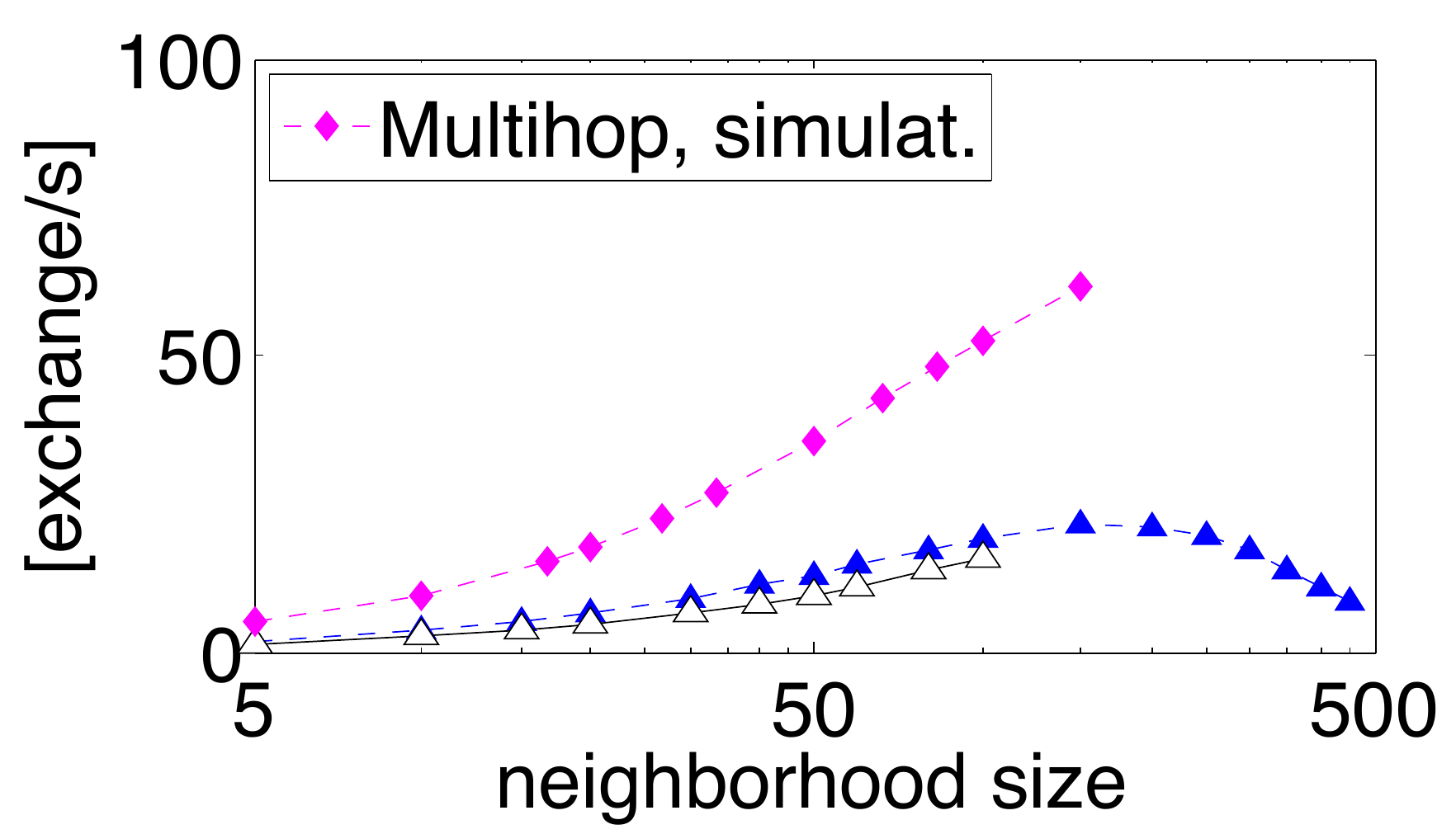} 
			}
		\end{center}
		\caption{SOFA's performance in extreme network conditions (testbed and simulation results).} \label{fig:sofa_extreme} 
\end{figure}

\paragraph{Insight 1.}\emph{SOFA shows a strong resilience to extreme densities.}  Figures~\ref{fig:sofa_dc_cooja_vs_testbed} and \ref{fig:sofa_tp_compare} show the prior testbed results together with the simulation results (notice that the x-axis is in a log scale). These results consider clique networks for both the testbed and simulation results. First, it is important to notice that Cooja captures, in a pretty accurate way, the trends observed on the testbed. Figure~\ref{fig:sofa_dc_cooja_vs_testbed} shows that the duty cycle continues to decrease (almost monotonically) and stabilize after a density of more than 100 neighbors. Figure~\ref{fig:sofa_tp_compare} shows that the exchange rate degrades monotonically but in a graceful manner. 

There is, however, a more important question to answer about SOFA: \emph{at what density does it saturate?} The clique curves (bottom two curves) in Figure~\ref{fig:sofa_tp_global} provide some insight into this question. In these experiments, we evaluated the global exchange rate at different densities. For the tested parameters, SOFA saturates when the density approaches 200 neighbors per node. Note that these are clique scenarios. In multi-hop networks, SOFA can exploit the well known spatial multiplexing effect (parallel data exchanges) and achieve higher global exchange rates. The top curve in \ref{fig:sofa_tp_global} depicts this behavior. The highest point represents a network with 450 nodes and an average density of 150 neighbors.

\paragraph{Insight2.}\emph{The performance of SOFA remains the same in static and mobile scenarios.} By being a stateless protocol, with nodes acting independently in an asynchronous and distributed fashion, SOFA does not require spending energy on maintaining information about the node's neighborhood and it is independent from the network topology and mobility. 

To test SOFA with dynamic topologies, we simulated an area of 150x150 meters where nodes moved according to traces generated by the BonnMotion's random waypoint model. We tested three speeds: 0\,m/s (static), 1.5\,m/s (walking) and 7\,m/s (biking). The radio range was set to 50 meters, with every node being connected, on average, to one third of the network. The maximum density was 150 nodes in a 450-node network. The resulting multi-hop networks had an effective diameter of just below three hops, which ensures that hidden terminal effects are taken into consideration.  Figure~\ref{fig:sofa_mobility} shows the duty cycle and the exchange rate of SOFA under the patterns \emph{static}, \emph{walking} and \emph{biking}. We can see that the speed of the nodes does not influence the energy consumption, the delivery ratio and the exchange rate. 

\begin{figure}
	\begin{center}
		\subfloat[Duty cycle]{ \label{fig:sofa_dc_mobility} 
		\includegraphics[width=0.3 
		\textwidth]{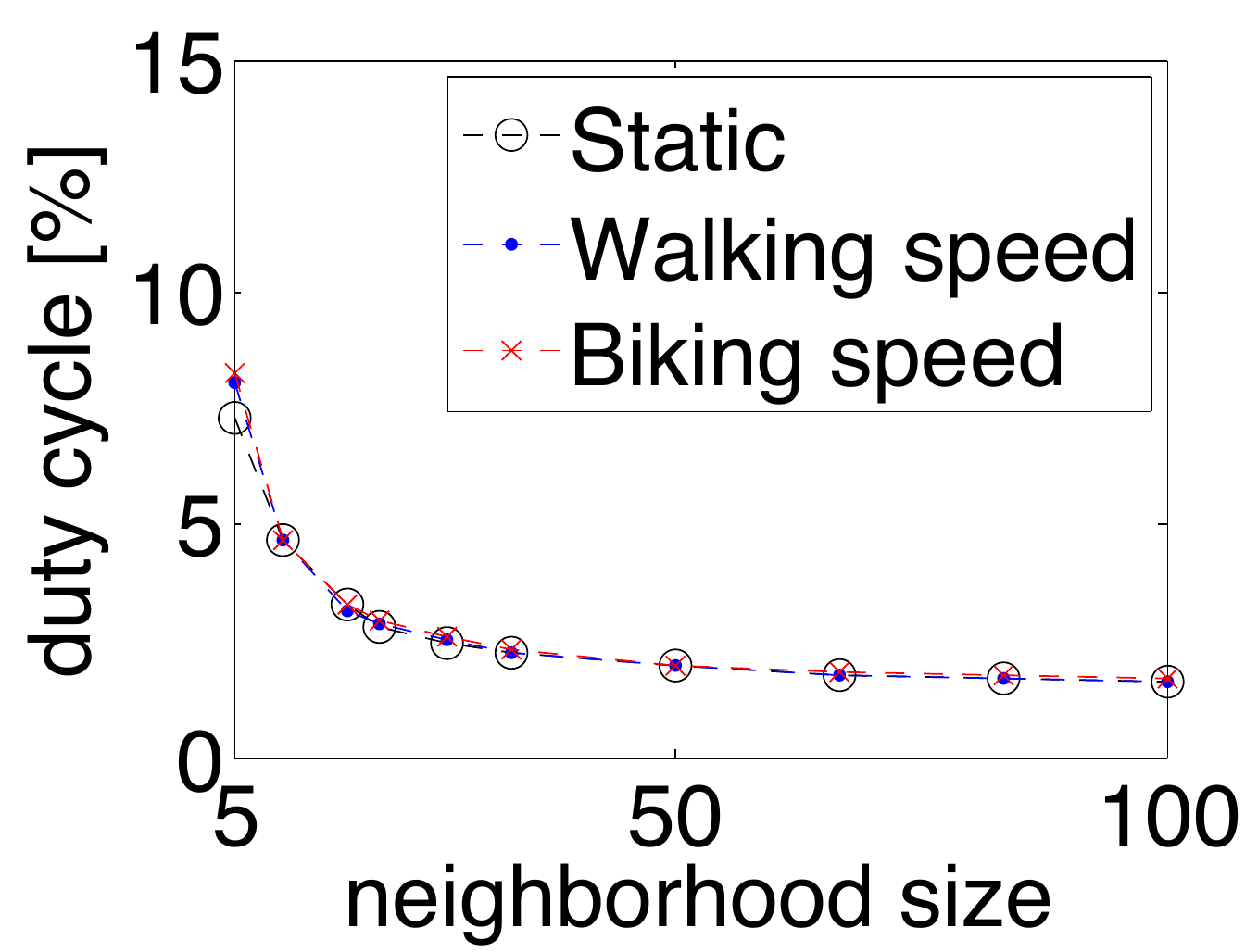} }
		\hspace{7px} 
		\subfloat[Mass delivery ratio]{ \label{fig:sofa_dr_mobility} 
		\includegraphics[width=0.3
		\textwidth]{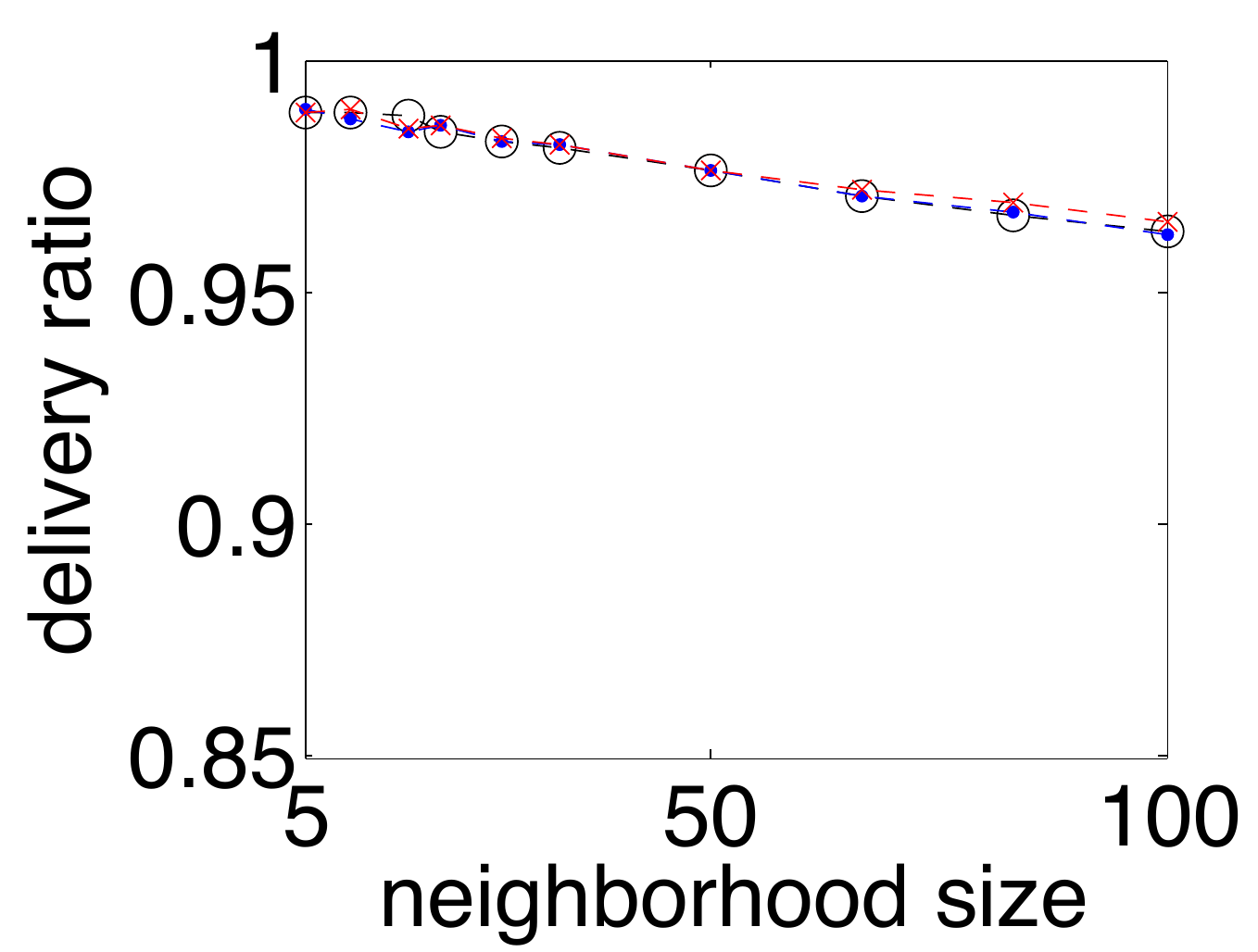} }
		\hspace{7px}
		\subfloat[Exchange rate]{ \label{fig:sofa_tp_mobility} 
		\includegraphics[width=0.3
		\textwidth]{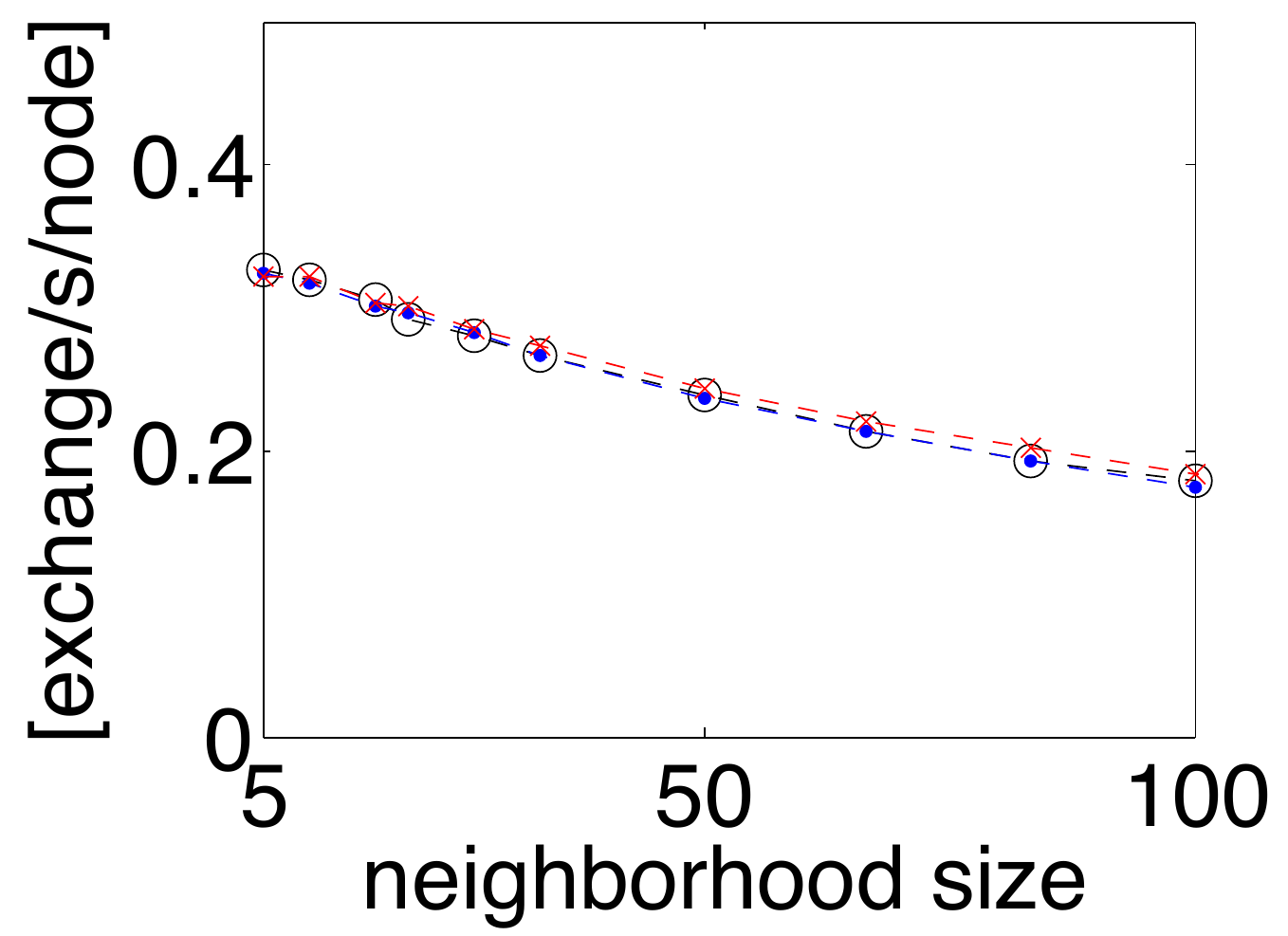} }
	\end{center}
	\caption{SOFA performance under different mobility scenarios (simulation results).} \label{fig:sofa_mobility} 
\end{figure}

\paragraph{Insight 3.}\emph{SOFA natively supports Gossip.} 
As mentioned before, one the goals of our study is to develop a communication primitive that is useful for general gossip applications. In gossip, it is important to conserve mass. Our 3-way handshake phase guarantees that, unless the last ack is lost, the two nodes will reach either a positive agreement (both nodes exchange their mass) or a negative agreement (both nodes keep their mass). Clearly, a positive agreement is the most desirable outcome, but both outcomes guarantee that no mass is lost during the exchange. The most important issue is to reduce the possibility of disagreements (when only one node, the sender, deems the transaction as successful).

To evaluate SOFA's ability for mass conservation, we compute the mass delivery ratio. This  metric represents the fraction of data exchanges that end up successfully. Figure~\ref{fig:sofa_dr_cooja_vs_testbed} depicts the \emph{mass delivery ratio} of SOFA under different densities. The figure shows that even under extreme densities (450 neighbors) SOFA is able to achieve a high percentage of successful exchanges (above 90\,\%). This is an important result. In the previous subsection, we found that SOFA saturates at approximately 200 nodes, beyond this point the exchange rate decreases monotonically. But,  Figure~\ref{fig:sofa_dr_cooja_vs_testbed} shows that the few exchanges that are able to occur beyond this point are able to be completed successfully. In other words, even under extremely demanding conditions SOFA has a remarkable ability to conserve mass. This feat is due to the careful design of SOFA aimed at (i) selecting reliable links (rendezvous phase),  (ii) implementing a transmit back-off instead of a CCA (to avoid sender-based collisions),  (iii) avoiding the use of retransmissions (which would jam the channel) and (iv) providing a method to reduce mass losses due to packet drops (3-way handshake).

\section{Conclusions}
\label{sec:sofa_conclusions}

In this chapter we made the first step towards communication in Extreme Wireless Sensor Networks, where nodes are mobile, density can reach hundreds of neighbors and bandwidth must be spared for the actual communication between nodes. 
With these challenges in mind, this chapter presented the design of SOFA, a medium access control protocol that, following the four design principles presented in \chapref{chapter:introduction}, is able to diffuse and process information in EWSNs in a robust and lightweight fashion.

By communicating with the first duty-cycling neighbor to wake up (opportunistic principle), SOFA is more efficient in extreme conditions (high densities), rather than in milder ones (anti-fragile principle). SOFA's asynchronous and distributed mechanism is immune to mobility and nodes' failures (state-less principle) and reaches a bandwidth saturation at densities close to 200 nodes. Over this point, it is still able to provide a reliable communication for densities of up to 450 nodes (robustness principle). 

%By following the characteristics of EWSN protocols defined in \chapref{chapter:introduction}, SOFA is able scale to networks where traditional WSNs mechanisms simply collapse. 
As we will see in the next chapters, SOFA forms the base communication layer for other network services and protocols for EWSNs, such as a lightweight density estimator (\cf \chapref{chapter:estreme}) and an opportunistic data collection protocol (\cf \chapref{chapter:staffetta}).
%
% \section{Old Conclusions. add something?}
%
% In this Chapter we proposed SOFA, a communication protocol that can operate efficiently and reliably under EWSN's stringent conditions. To the best of our knowledge, this is the first effort aiming at mobile sensor networks with high densities and high traffic rates. Traditional low-power protocols, which were not designed with these requirements in mind, simply collapse under such circumstances.
%
%
% %Third, SOFA  was not only designed to fulfill the requirements of our crowd monitoring scenario, but also to serve as a basic building block for general gossiping applications. With regards to the topic of crowd monitoring, we are currently evaluating to use the duration of the rendezvous phase as the means to infer the density of the network (a direct consequence of our rendezvous analytical model). This will allow SOFA to fulfill the needs for crowd monitoring without any extra overhead.
%
% Finally, it is important to stress that SOFA is not intended to replace traditional low-power methods, as these provide a different set of services. SOFA complements these traditional methods. If the network's conditions change (to a milder state), the network stack can switch to unicast and broadcast primitives.

%\references{library}

%% file: estreme/estreme.tex
%!TEX root = ../dissertation.tex
\chapter{Neighborhood Cardinality Estimation}
\label{chapter:estreme}

\blfootnote{Parts of this chapter have been published in IPSN'14, Berlin, Germany~\cite{Cattani2014b}.}

\epigraph{``When every logical course of action is exhausted, the only option that remains is inaction''}{Tudok}

%\newpage

%\section{Introduction}
%\label{sec:estreme_introduction}
% Link to previous chapter
\dropcap{K}{nowing} the neighborhood cardinality in Wireless Sensor Networks is an essential building block of many adaptive algorithms, such as resource allocation~\cite{Burri2007} and random-access control~\cite{Loukas2013a}. 
Cardinality estimation is also a valuable tool on itself. It can be used to monitor the surrounding environment, transforming the radio device into a smart sensor~\cite{Weppner2011}. 
In crowd-monitoring, for example, the number of (personal) devices in the communication range can be used as an indicator of the crowd density and serve as an alarm in case such density crosses a dangerous threshold.

Unfortunately, cardinality estimation in Extreme Wireless Sensor Networks is hard. Different from traditional studies, these networks are mobile, dense and require \emph{all} nodes to estimate the cardinality concurrently. Moreover, due to devices' limited capabilities, much of the already limited bandwidth is used for communicating. Therefore, we need an estimator that is not only accurate, but also fast, asynchronous (due to mobility) and lightweight (due to concurrency and high density).

% The SINGLE aim of this chapter
In this chapter we propose \emph{Estreme}, a neighborhood cardinality estimator with extremely low overhead that leverages the rendezvous time of low-power medium access control (MAC) protocols like SOFA (\cf \chapref{chapter:sofa}).
 
% Outline of the chapter
In \secref{sec:estreme_mechanism}, we model the rendezvous time using order statistics and derive a neighborhood cardinality estimator. The model permits us to (a) provide four rules that are necessary and sufficient for using Estreme in a wide family of communication protocols, (b) gain insights into the performance of the estimator, and (c) derive bounds for the estimation error.

In \secref{sec:estreme_implementation}, we implement Estreme on top of SOFA, the low-power MAC protocol presented in \chapref{chapter:sofa}. Even though Estreme is conceptually simple, implementing it on real nodes raises a number of challenges: observing the right event (first wake-up), measuring the rendezvous time accurately and exploiting spatial/temporal correlations in the sampling process. We thoroughly analyze these challenges and provide viable solutions.

Finally, in \secref{sec:estreme_results}, we extensively evaluate Estreme on a testbed consisting of 100 wireless nodes performing concurrent estimations.  Our implementation achieves an estimation error of $\approx$\,10\% with an overhead of just 4 bytes per estimation and a duty cycle of $\approx$\,3\%.  Moreover, Estreme provides a good trade-off between agility and precision. For example, when the neighborhood size changes abruptly from 30 to 60 nodes, the estimated cardinality of nodes converges within one minute to the 10\% error range.

% The ideas presented in this section form the basis of our approach for estimating the neighborhood cardinality. First, we derive an estimator based on order statistics and discuss its applicability to communication protocols (\secref{sec:estreme_mechanism}).
% %
% We then analyze the impact of timing errors on the estimation accuracy (\secref{sec:estreme_lemma}).

\section{Problem Statement}
\label{sec:estreme_problem_statement}

Neighborhood cardinality estimation is a broad and active research area, but for the purposes of our work, we are interested in studies that perform empirical evaluations on dynamic networks, \ie networks with frequent topology changes due to fluctuations in link quality and node movement. Within this scope, the state-of-the-art achieves an accuracy between 3\% and 35\% for 25 smartphones, relying on audio signals~\cite{Kannan2012}. Other studies, using bluetooth signals in smartphones~\cite{Weppner2011} and radio signals in sensor nodes~\cite{Zeng2013}, achieve comparable results for similar settings. Nevertheless, only a fraction of the nodes perform the estimation process.

We advance the state-of-the-art in two ways. First, we move from scales of a few tens of neighboring nodes to one hundred nodes. Second, we allow all nodes to perform the estimation concurrently. Solving this novel estimation problem is challenging. Network dynamics require estimations that are fast and asynchronous. These latter characteristics limit the number of samples (\ie information) that can be collected, which in turn decreases accuracy. On the other hand, higher density and concurrency pose an extra burden on the sampling process and necessitate an efficient use of bandwidth.  

To cope with these challenges we propose \emph{Estreme}, a low-overhead cardinality estimator that is robust to mobility and supports multiple estimators running simultaneously in an asynchronous manner. The key idea behind Estreme is simple: in networks where all nodes perform periodic but random events within a given period, the time difference between two consecutive events (rendezvous time) captures the density of the neighborhood. The shorter the rendezvous time, the higher the density, and vice-versa.

\section{Mechanism}
\label{sec:estreme_mechanism}

%\noindent\textbf{Problem statement.} 
The neighborhood set $V_u$ of a node $u$ in a wireless network consists of all nodes $v$ in the radio vicinity of $u$. Our objective is to compute $n_u = |V_u|$ \ie its cardinality. In the following, we use $n$ instead of $n_u$ when $u$ becomes implicit.

\vspace{2mm}\noindent\textbf{Cardinality estimator.}  We assume that, within a given period $t_w$, all nodes in the network perform an event in a random and desynchronized manner, as in Figure~\ref{fig:estreme_model}. Considering a random point in time, the time sequence of the subsequent events can be modeled as a set of independent random variables following a uniform distribution ($X_1 \dots X_n$). This random point in time represents the moment when a node, called \emph{initiator}, wishes to estimate the cardinality of its neighborhood. The rendezvous time with the first $k$ neighbors (events) captures the cardinality. Intuitively, the longer it takes to rendezvous, the lower the cardinality (because the distribution of random events during $t_w$ is sparser). 

The rendezvous time $T_r(k)$ with the first $k$ neighbors is a random variable and can be modeled using order statistics. The density function of $T_r (k)$ is known to follow the beta distribution~\cite{Balakrishnan2007}
\begin{equation} 
	T_r(k)\sim \text{Beta}(\alpha, \beta), 
\label{eq:estreme_beta} 
\end{equation}
and in our scenario $\alpha = k$ and $\beta = n + 1 -k$.  Considering the period $t_w$, the expected value of the rendezvous time, \ie the expected time it takes to observe $k$ events is\vspace*{-1ex}
\begin{equation} 
	\mathbb{E}[T_r(k)]= t_w \frac{\alpha}{\alpha+\beta} = t_w \frac{k}{n+1}.
\label{eq:estreme_rendezvous} 
\end{equation} 
Inverting the expectation, we obtain a simple-to-compute estimator for the neighborhood cardinality $n$ based on the average of the observed rendezvous times $\bar{t}_r$.
\begin{equation}
	 \hat{n} = t_w\frac{k}{\bar{t}_r}-1. 
\label{eq:estreme_cardinality} 
\end{equation}

In our work we consider $k=1$, that is, we estimate the cardinality by using the average rendezvous time with the first event only. As we will describe later, $k=1$ is chosen because it minimizes the amount of bandwidth used to collect a sample and well fits with the SOFA mechanism presented in \chapref{chapter:sofa}. 
%As mentioned before, an efficient use of bandwidth is a central requirement to cope with a high number of concurrent estimations. 

It is important to note that in theory,  for $k=1$ the expectation of the estimator diverges~\cite{Giroire2009}, because as $n \rightarrow \infty$, $t_r \rightarrow 0$ and $ \hat{n} \rightarrow \infty$. In practice, $t_r$ remains positive. However, for very large neighborhoods, as $n \rightarrow \infty$, Estreme could use $k \geq 2$. For these values of $k$ the expectation of the estimator becomes $n/(k-1)$~\cite{Chassaing2007a}.

\begin{figure}
	\centering 
	\includegraphics[width=0.6 
	\textwidth]{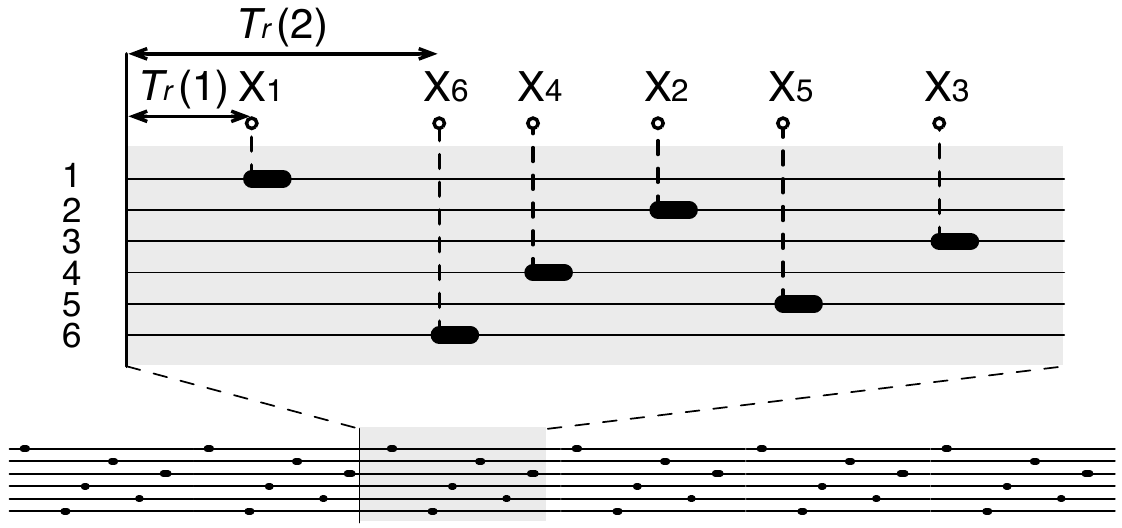} \caption{Modeling the rendezvous times with the first $k$ neighbors using the Beta random variable $T_r(k)$. The random variables $X_1 \dots X_n$ model the wake-up times of nodes.}
	\label{fig:estreme_model} 
\end{figure}

% --------------------------------------------------------------------------
\vspace*{2mm}\noindent\textbf{Applicability of Estreme}. Because of its simplicity, the proposed estimator is applicable to any network protocol that abides to three requirements:
% ----------------------------------------------------------------------------------
\begin{description} %,labelindent=7.0mm] 
	\item[R1.] Periodically (every $t_w$), nodes generate independent and random events.\label{rule:periodical}
	\item[R2.] The event of a node is observable by all neighbors.\label{rule:observable}
%	\item[R3.] Nodes can distinguish the order of events. \label{rule:order}
	\item[R3.] Nodes can accurately measure inter-event periods.\label{rule:time}
\end{description}
Requirement~1 guarantees that, given a random moment in time, we can model the occurrence of the first event using a uniform random variable. 
Requirement~2 ensures that an \emph{initiator} node appropriately identifies the first $k$ events. That is, it does not mistake a later event for an earlier one. For example, due to collisions, earlier events can be lost and Estreme would underestimate the cardinality. 
Requirement~3 ensures that the rendezvous time is measured accurately. Otherwise an overestimation (shorter rendezvous than factual) or underestimation (longer rendezvous than factual) of the cardinality would occur.

% ----------------------------------------------------------------------------------
\subsection{Timing inaccuracies}
\label{sec:estreme_lemma}
% ----------------------------------------------------------------------------------

As we will observe in the next section, measuring accurately the rendezvous time is challenging because several delays are introduced. Some of these delays can be measured and overcome, but others are difficult to track, such as computation and transmission delays. Analyzing the effects of these delays is central to understanding the performance of Estreme in real scenarios. In this section we derive a bound for the cardinality (estimation) error caused by a positive delay $\epsilon$ in the rendezvous time.
\begin{proposition}
\label{lemma:estreme_bound}
Given a timing error $\epsilon$ in the rendezvous time $t_r$, the expected cardinality error is
\begin{align*}
	\mathbb{E}[e_n] = \frac{\hat{n}_{\epsilon} - n}{n} = \Theta \left(-\frac{\rho}{1 +\rho}\right),
\end{align*}
where $\rho = \epsilon(n+1)/t_w$ , and $\hat{n}_{\epsilon}$ is the expected value of the estimated cardinality considering $\epsilon$.
\end{proposition}
\begin{proof}
The proof consists of two steps. First, we derive $\hat{n}_{\epsilon}$ based on an expected rendezvous time delayed by $\epsilon$. We then propose a bound on the resulting error $e_n$ and show that it is tight. %The bound is simpler and more intuitive that $e_n$. 
% Let us start with the derivation of $e_n$. Substituting the expression of $\mathbb{E}[T_r(1)]$ expressed in~\eqref{eq:rendezvous} into the definition of $e_n$, we have that
% 
Let us start with the derivation of $e_n$. Substituting the expression of $\hat{n}_{\epsilon}$ derived from \eqref{eq:estreme_rendezvous} and \eqref{eq:estreme_cardinality} into the definition of $\mathbb{E}[e_n]$, we obtain

\begin{equation}
	\mathbb{E}[e_n] = - \frac{\epsilon (n+1)^2}{n \left(t_w + \epsilon (n+1)\right)}.
	\label{eq:estreme_en}
\end{equation}
Observe that for $\epsilon>0$, the above is negative; any positive delay causes the underestimation of $n$. Unfortunately, even though exact, Equation~\eqref{eq:estreme_en} is not intuitive. It is easy to see that, the simpler and more insightful expression 
\begin{align}
	\Phi &= - \frac{\epsilon (n+1)^2}{(n+1) \left(t_w + \epsilon (n+1)\right)} \nonumber \\
	     &= -\frac{\epsilon (n+1)/t_w}{1 + \epsilon (n+1)/t_w} = -\frac{\rho}{1 +\rho}
	     \label{eq:estreme_phi}
\end{align}
tightly bounds $\mathbb{E}[e_n]$. It is sufficient to show that, for each $n$, positive $k_1$, $k_2$ exist such that
\begin{equation}
\label{eq:estreme_bound}
  k_1 \Phi \leq \mathbb{E}[e_n] \leq k_2 \Phi, \text{ for all } n \geq 1. \nonumber 
\end{equation}
Indeed, the above inequality holds for
\begin{equation}
\label{eq:estreme_k1k2}
k_1 \leq {(n+1)}/{n} \quad \text{and} \quad k_2 \geq {(n+1)}/{n}. \nonumber 
\end{equation}
Setting $\rho = \epsilon(n+1)/t_w$ in Equation~\eqref{eq:estreme_phi} concludes our proof.
\end{proof}

Proposition~\ref{lemma:estreme_bound} leads to two observations for practical implementations of Estreme.
\emph{First, given a fixed measurement error, longer periods $t_w$ are preferable in terms of estimation error.} As $\rho \rightarrow 0$ the estimation error tends to $0^{-}$. This can be caused both by $\epsilon$ being very small (this is obvious, since a small measurement error implies small estimation errors), but also by long periods $t_w$ and small neighborhood sizes. While we cannot choose the number of neighbors, we should prefer longer periods $t_w$ to shorter ones. 
\emph{Second, to run Estreme in EWSNs, a platform should be able to measure the rendezvous time with sub-millisecond accuracy.} If we apply the Proposition to a neighborhood of 100 nodes with a period $t_w=1000\,ms$, the estimation error introduced by a measurement delay $\epsilon = 1\,ms$ is  
$$ \Theta(-\frac{0.101}{1+0.101} = - 0.09).$$
That is, with a measurement error of just 1\,ms, the estimation error is 9\%.
Obtaining an accurate estimation of the rendezvous time at sub-millisecond accuracy is therefore central for the correct operation of Estreme.

\section{Implementation}
\label{sec:estreme_implementation}

Estreme is a general framework that can be implemented over many communication protocols. In
our work, we build Estreme on top of SOFA, the low-power listening MAC protocol presented in \chapref{chapter:sofa}. We first describe a \emph{Naive} implementation of Estreme and highlight its limitations (Section~\ref{sec:estreme_naive}). Then, we analyze these limitations and provide solutions for them (Sections~\ref{sec:estreme_observable}, \ref{sec:estreme_accurate}, and~\ref{sec:estreme_average}).

\subsection{Naive implementation}
\label{sec:estreme_naive}
Similar to other low-power listening MAC protocols, in SOFA nodes wake up periodically and with a fixed frequency $t_w$. When a node, called \emph{initiator}, wants to communicate (node 1, \figref{fig:estreme_mechanism}), it sends a strobe of beacons (B) until the first neighbor (node 3) wakes up. The wake-up of this neighbor is the observable event required by Estreme. To adhere to the uniform and random distribution of requirement~1, nodes wake up deterministically at every $t_w$, plus a random delay in the range $[{-t_w}/{_2}, {t_w}/{_2}]$. 
Thus, the probability density function of the rendezvous times adheres to the Beta distribution required by Estreme to perform the estimation. 

%With SOFA's anycast primitive, the rendezvous time with the \emph{first} neighbor to wake up can be obtained quite simply. Whenever an initiator receives the first response from a neighbor, it computes the rendezvous time as the difference between its first beacon (B{\small1}) and the first received acknowledgement (A{\small1}).

%
\begin{figure}
	\begin{center}
		\includegraphics[width=0.75
		\textwidth]{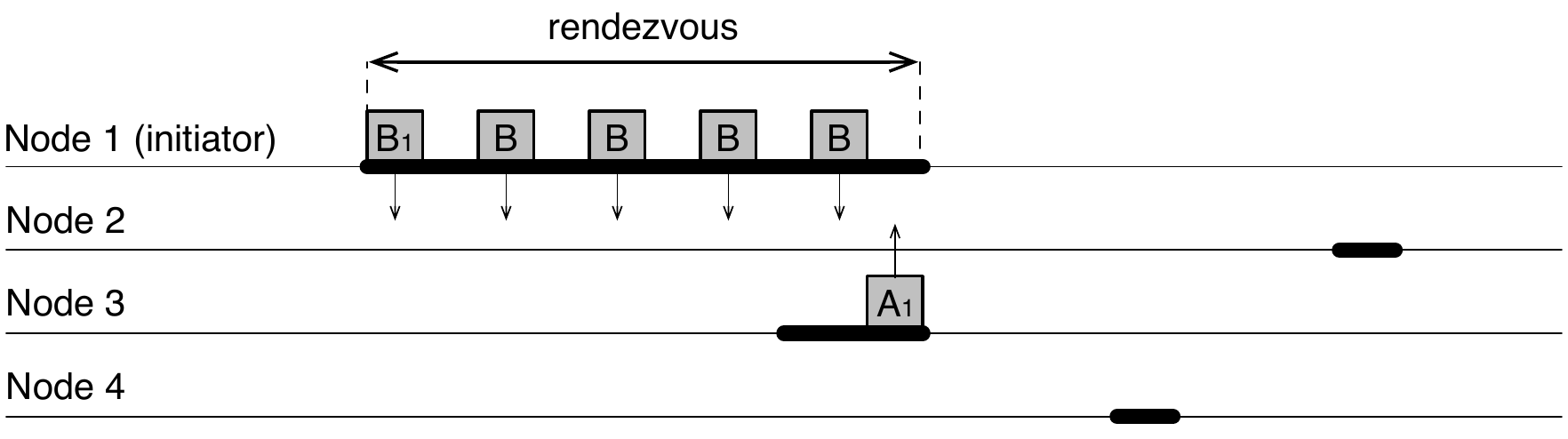} 
	\end{center}
	\caption{Estreme mechanism applied to SOFA} 
	\label{fig:estreme_mechanism} 
\end{figure}
At first glance, the rendezvous time is simple to measure: the initiator starts
a timer before sending its first beacon (B{\small1}, \figref{fig:estreme_mechanism}) and stops the timer upon
receiving the acknowledgement (A{\small1}, \figref{fig:estreme_mechanism}). We evaluated this \emph{Naive}
method on our testbed for various neighborhood sizes, from 10 to 100. The
cardinality estimation is based on the mean rendezvous time ($\bar{t_r}$) of
the last 50 samples observed by the initiator. The initiator sampled its
neighborhood at a rate of 1\,Hz.

Figure~\ref{fig:estreme_naive_rendezvous} compares the mean rendezvous times with (i) our analytical model ($\mathbb{E}[T_r]$ according to Equation~\ref{eq:estreme_rendezvous}), and (ii) Monte Carlo simulations performed in \textit{Matlab}, which capture an ideal networking environment. Note that for all neighborhood cardinalities, the \textit{Naive} implementation of Estreme over-measures the rendezvous time. 
Figure~\ref{fig:estreme_naive_error} shows that, as a consequence, the estimation error of the neighborhood  cardinality (black bars) is significant, between 10\% and 20\%. Based on the \emph{maximum} difference ($\epsilon$) between the \emph{Naive} measurements and the expected rendezvous time $\mathbb{E}[T_r]$ (Proposition~\ref{lemma:estreme_bound}), we observe that the estimation error could reach values between 20\% and 40\%.

The problem with the ``\emph{Naive}'' implementation is that it does not adhere to Estreme's 2nd and 3rd requirements. Sometimes an initiator misses the first event, and mistakes subsequent events as being the first (non-compliance with requirement 2), and the calculation of the rendezvous time includes delays that are not part of the rendezvous itself (non-compliance with requirement 3). 
In the next sections we describe and solve these two  problems, and explain a method to accelerate the estimation convergence of Estreme.
\begin{figure}
	\begin{center}
		\subfloat[Rendezvous times]{ 
		\label{fig:estreme_naive_rendezvous}  
		\includegraphics[width=0.5
		\textwidth]{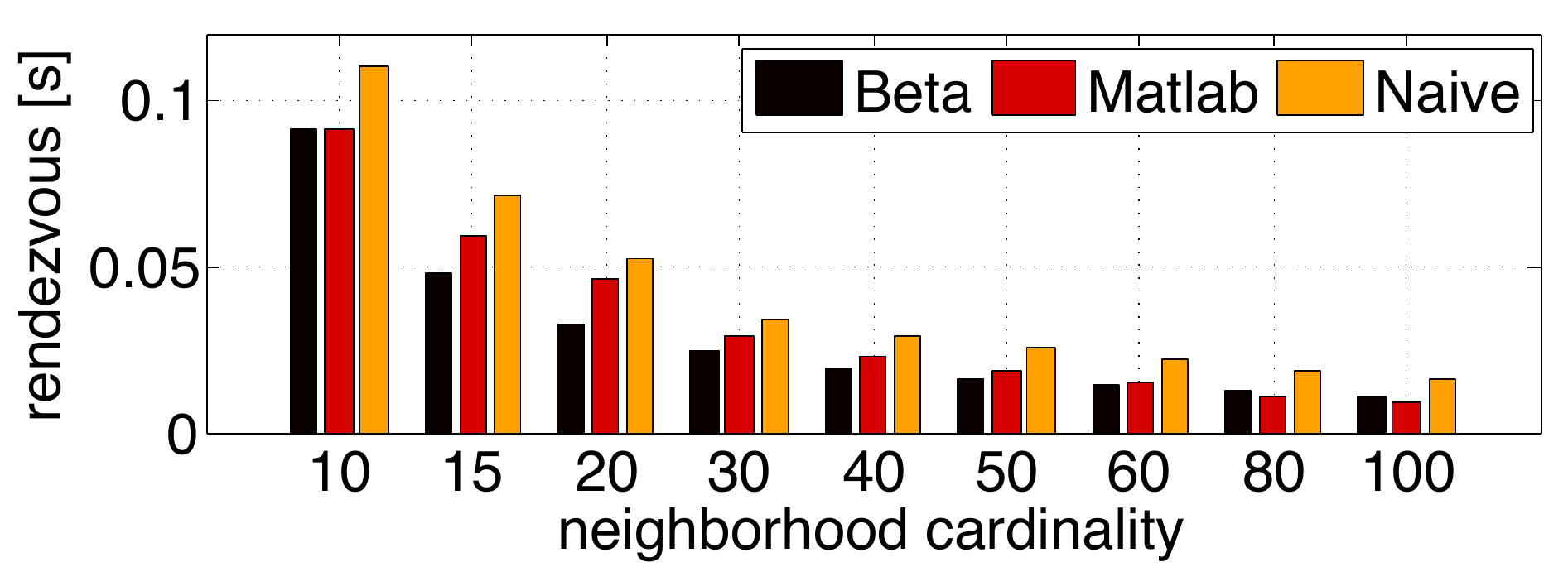} } 
		\subfloat[Estimation error and bound]{
		\label{fig:estreme_naive_error}
		\includegraphics[width=0.5
		\textwidth]{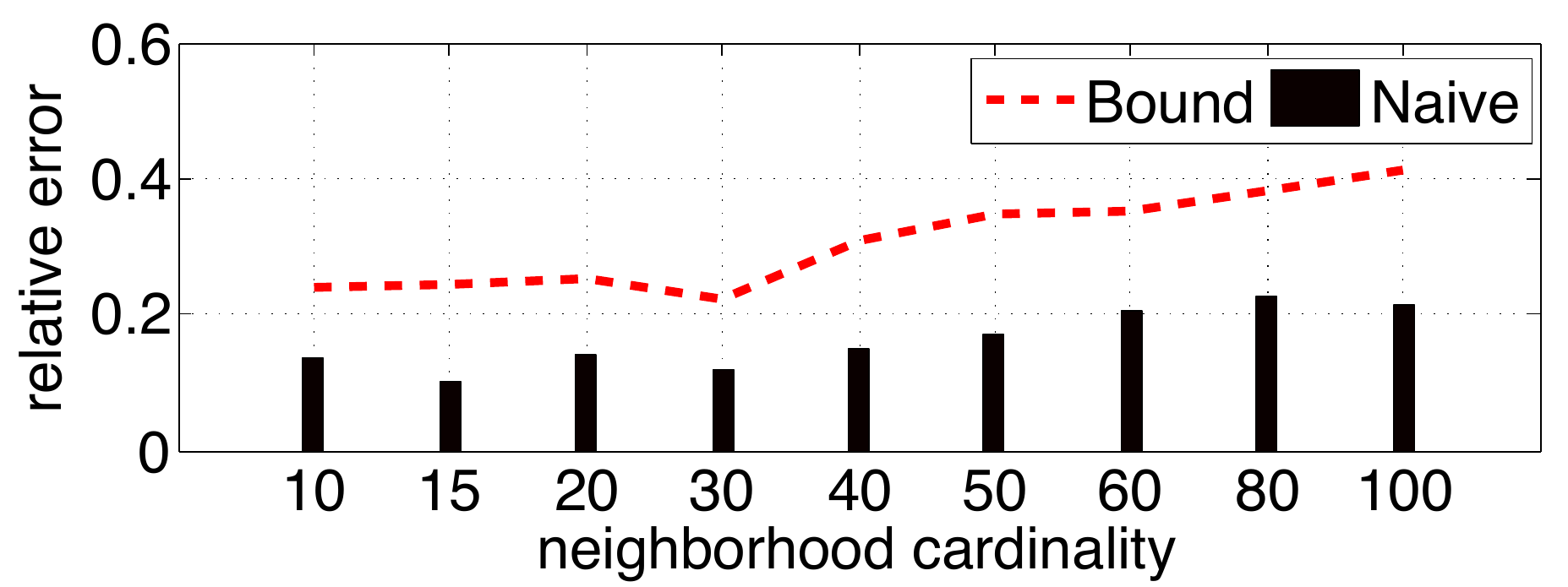} }
	\end{center}
	\caption{Naive implementation of Estreme compared to the Beta model.} \label{fig:estreme_naive} 
\end{figure}
\subsection{Correct observations}
% \subsection{Adhering to requirement~3: collisions}
\label{sec:estreme_observable}

Since Estreme's estimations are based on the wake-up sequence of nodes, it is essential that such observations are correct. But collisions can affect this requirement, see Figure~\ref{fig:estreme_delays}. If two or more neighbors wake up between two consecutive beacons, a collision will occur. 
%because all the neighbors will only receive the latter beacon and acknowledge it. 
Denoting $t_b$ as the inter-beacon interval, the collision probability is

\begin{align*}
P(\text{collision}) & = \sum_{j} \sum_{i=2}^{n} P\left(
\begin{array}{c}
  \text{first } i \text{ nodes wake up at } \\
  \text{the } j^{\textit{th}} \text{ interval}
  \end{array}
\right)\\
    & = \sum_{j=1}^{ \lfloor t_w/t_b \rfloor} \sum_{i=2}^n \binom{n}{i} \left( \frac{t_b}{t_w} \right) ^i \left( 1-\frac{t_b}{t_w}j \right)^{n-i}
\end{align*}
where $t_b/t_w$ is the probability that a node wakes up at each inter-beacon duration, and $1-(t_b/t_w)j$ is the probability that a node wakes up after the $j^{th}$ interval. The likelihood of a collision is not high. For example, if $n$ is 10, 50 and 100, the collision probability is 0.02, 0.11 and 0.22, respectively. Note that, even though this probability captures the collision of any number of nodes, most of the probability mass is concentrated on the case when only two nodes collide. 

To reduce the chances of missing the first event(s), Estreme implements the following conflict resolution mechanism: if a node detects that its acknowledgment is lost (by receiving the beacon again), the node will retransmit its acknowledgement with probability $p$ and go back to sleep with probability $(1-p)$. Denoting $n_c \geq 2$ as the number of colliding nodes, the conflict resolution mechanism leads to three outcomes:
 
(i) \emph{starvation}, with probability $(1-p)^{n_c}$, all nodes go to sleep and the first event(s) is lost; 

(ii) \emph{completion}, with probability $n_c\,p\,(1-p)^{n_c-1}$, only one node remains awake and sends successfully the acknowledgement; 

(iii) \emph{contention}, more than one node remain awake and the contention 
process restarts with those nodes, with probability $1-(1-p+ n_c\,p)(1-p)^{n_c-1}$. 

\noindent Since we want to maximize the probability of \emph{completions},
\begin{equation}\label{eq:estreme_optimal_p}
\frac{\partial}{\partial p} n_c\,p\,(1-p)^{n_c-1} = 0 \implies p =\frac{1}{n_c}.
\end{equation}

In reality, however, Estreme will not know the number of contending nodes. In our implementation we settle for $p=0.5$, because the higher $p$, the lower the probability of \emph{starvation}, but according to Equation~\ref{eq:estreme_optimal_p}, $p\leq 0.5$. As an example, if $n = 100$ the probability that exactly two nodes collide and reach starvation is  $P(\text{collision)} P(\text{starvation}) = 0.19 \times 0.\overline{3} = 0.06$.

To avoid the (unlikely) possibility of infinite \emph{contentions}, nodes will back off and go to sleep after a maximum number of unsuccessfully retransmissions.

Note that this conflict resolution mechanism introduces a delay in the measurements of the rendezvous time that badly affects Estreme's accuracy. In the next section, we provide a comprehensive solution to overcome this delay.

\subsection{Accurate measurements}
\label{sec:estreme_accurate}
\begin{figure}
	\centering 
	\includegraphics[width=0.55 
	\textwidth]{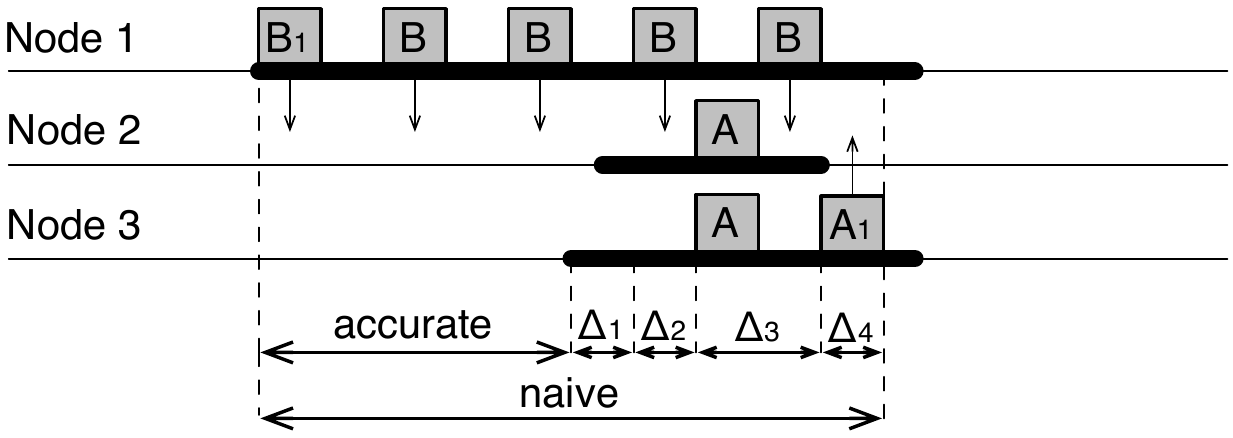} \caption{Delays of the rendezvous measure.}
	\label{fig:estreme_delays} 
\end{figure}

Encountering the first node that wakes up is a necessary, but not sufficient, step to properly estimate the neighborhood cardinality. \emph{The rendezvous time of nodes needs to be measured accurately}.  
 
\paragraph{Overcoming delays.} In principle, the measuring of the rendezvous time should stop at the moment when the first neighbor wakes up. But, as Figure~\ref{fig:estreme_delays} shows, the \emph{Naive} mechanism includes different kinds of delays; namely, collisions ($\Delta${\small3}), the time required to transmit the radio packets ($\Delta${\small2}, $\Delta${\small4}), and the listening time between the actual wake-up and the moment when the beacon is sent ($\Delta${\small1}). This initial listening period $\Delta${\small1} is a typical feature of LPL protocols such as SOFA and X-MAC~\cite{Buettner2006} and can be as long as the inter-beacon interval $t_b$. 

To obtain an accurate measurement, we need to subtract these delays from the naive measurements. In our implementation, this is done through the receiver (node 3), by starting a timer once the receiver wakes up and piggybacking the elapsed time on the acknowledgments ($\Delta${\small1} + $\Delta${\small2} + $\Delta${\small3}). Assuming a fixed bit rate and a fixed acknowledgement size, both reasonable assumptions, $\Delta${\small4} can be computed off-line as a constant and be systematically subtracted by the initiator after the acknowledgement is received. In our case $\Delta${\small4} = 1.1\,ms.

\paragraph{Accurate timing in Contiki OS.}
In Contiki, the clock library allows measurements with a maximum precision of 2.83\,ms. This coarse-grained measure is not suitable for our estimation purposes. To ensure the maximum possible accuracy, we measured the rendezvous times using Contiki's real-time module (\emph{rtimer}). Its precision depends on the processor's clock frequency (32 KHz in our devices) and allows to measure time-ranges with sub-milli\-second accuracy. 
%Note that, due to its precision, this timer can only be used to measure very short period of times ($\approx$~2 seconds).

\subsection{Improving the estimation process}
\label{sec:estreme_average}

As any other estimator based on order statistics, Estreme is bound to the law of large numbers: the larger the number of samples, the closer the mean gets to the expected value and the more accurate the estimation. Unfortunately, in mobile networks nodes cannot collect many samples because they only have a limited amount of time to capture the current status of their neighborhoods. The central question is hence, \emph{how can we facilitate a fast gathering of samples?} Estreme's design tackles this problem in the temporal (T-Estreme) and spatial (S-Estreme) domains.

\paragraph{Averaging samples in time.}  Intuitively, given a certain period $T$, we would like nodes to gather as many samples as possible during that period. The key characteristic to achieve this goal is an efficient use of bandwidth, that is, to spend as little time as possible gathering each sample. By using SOFA's primitive (opportunistic anycast), \emph{Estreme reduces the use of bandwidth to the minimum time required to observe an event}. To process these samples, T-Estreme utilizes a simple moving window average (MWA) filter, where the last $w$ samples of the rendezvous time are averaged to estimate the cardinality. We tried other robust statistical methods based on the median and on alpha-trimmer filters, but there was not much difference with regards to MWA (because requirements 2 and 3 already remove most outliers). Like most MWA estimations, the tradeoff in the time domain is between accuracy and adaptability: the bigger $w$, the more time it takes to adapt to changes in cardinality.

\paragraph{Averaging samples in space and time.} Different from most cardinality estimators in the literature, Estreme is designed from inception to allow \emph{all} nodes to perform concurrent estimations. This characteristic is not only good because some practical applications require it, but perhaps more importantly because it can quadratically increase the convergence of the estimation. In S-Estreme, besides sending the delay information described in Section~\ref{sec:estreme_accurate}, a receiver also sends the average of its own time window. \emph{Every node is hence able to process $w^2$ samples in the same time required to locally collect $w$ samples}. Note that in S-Estreme, it is preferable to choose $w$ based on the expected cardinality of the neighborhood $n$. With $w > n$, some nodes' samples will be counted multiple times. With $w < n$, on the other hand, only part of the neighbors' samples will be taken into account by the estimator. While in our experiments $w$ is kept constant, it is also possible to dynamically adapt $w$'s size based on the cardinality estimation provided by T-Estreme.
It is important to highlight that while in most cases S-Estreme will significantly improve the estimation process, in case of drastic (spatial) changes in network density, errors can be introduced due to the spatial smoothing that is inherent to S-Estreme. This phenomenon is analyzed in Section~\ref{sec:estreme_results} and, later, in \chapref{chapter:nemo}.

\section{Evaluation}
\label{sec:estreme_results}

\begin{figure}
	\begin{center}
		\includegraphics[width=0.6
		\textwidth]{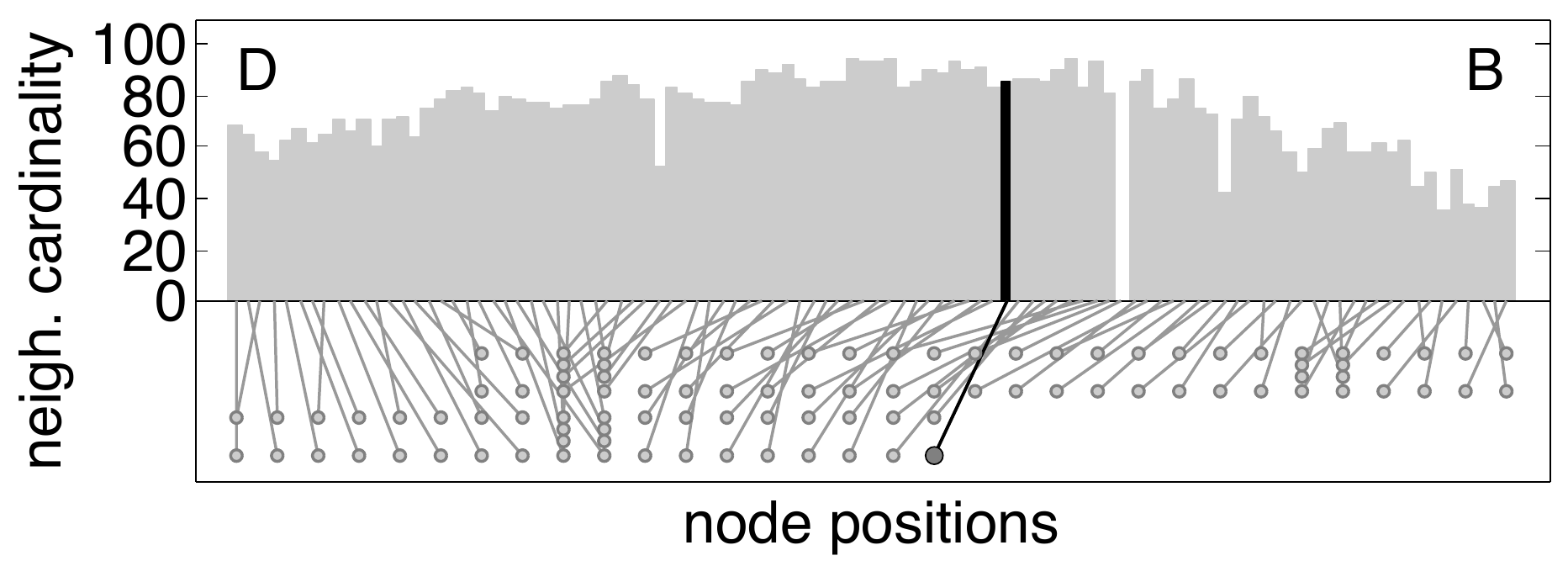} 
	\caption{Relation between the topology of the testbed (bottom) and neighborhood cardinalities (top). For specific experiments, only the results from a representative node (in black) are shown.} \label{fig:estreme_testbed} 
	\end{center}
\end{figure}

To evaluate Estreme, we ran an exhaustive set of experiments on a testbed
consisting of 100 sensor nodes equipped with an MSP430 processor and a CC1101
radio. The testbed is deployed in the ceiling of our offices and covers
an entire floor of the building. As the testbed was operated throughout
the week, Estreme was subject to all kinds of environmental conditions (e.g.,
human activity during office hours and temperature swings during the weekend)
inducing quite a variety of link dynamics, see~\cite{Woehrle2012b}. 

To test
network dynamics in a more controlled manner, we follow a twofold approach:
(i) we systematically turn on/off various nodes in the network, and (ii) we provide a few people with nodes and have them walk around our floor and
building to test transitions between areas with different node densities.

\paragraph{Ground truth}: In testbed experiments, it is
difficult to define the ground truth of the neighborhood cardinality. This
occurs because the quality of links is highly variable, and cardinality
changes significantly over time. This high variability is not particular to
our testbed. In~\cite{Zuniga2010}, the authors report that in the Twist testbed,
the cardinality of some nodes oscillate between 12 and 17 in periods of 30
minutes. 
%For reference, our experiments typically last for an hour per run Burri2007(and we report averages over XXX runs).

For each experiment, we compute the ground truth cardinality of a node
as the number of neighbors that have been \emph{observed} at least once
during the duration of an experimental run. A node is \emph{observed} if it
successfully exchanges at least one acknowledgement with the corresponding
initiator. Note that the ground truth cardinality is an upper bound that
is seldom reached. The fact that a node is observed during a period $t_w$
does not imply that it will be observed in the next period(s). Hence, all the
estimation errors reported in this section are \emph{worst case errors}.

Figure~\ref{fig:estreme_testbed} shows the network topology (bottom part) together with
the neighborhood cardinality (upper part) of each node in a \emph{typical}
experiment. Due to the shape of the building (narrow and long), when the
maximum power is used (+10dBm), the central nodes communicate with most
of the network, while the nodes on the two far-ends reach approximately
half of the other nodes. 
This setup creates a network with a minimum diameter of 2 and offers a wide range of
cardinalities, approximately from 35 to 85 neighbors. These network cardinalities allow us to understand the behavior of Estreme in case of \emph{un-even densities}. 
Also, to provide a more accurate
comparison of different methods, in some cases we use a representative node
with low cardinality variability. This representative node is marked black
in Figure~\ref{fig:estreme_testbed}.
%
%sefwefwefw bhgkj ergergegrr tgrtgh rtrt

\paragraph{Metrics and parameters}. To evaluate the accuracy
of Estreme, we computed the \emph{relative error} of the estimations as $\left|{(\hat{n}-n})/{n}\right|$, where $\hat{n}$ is the estimated neighborhood
size and $n$ is the ``ground truth'' neighborhood cardinality. With regards
to the parameters, unless stated otherwise, our experiments use the default settings listed in Table~\ref{tab:estreme_parameters}.
\begin{table}[!t]
\renewcommand{\arraystretch}{1.3}
\small
\centering
\begin{tabular}{c|c|c}
\hline
Symbol & Parameter & Value\\
\hline \hline
$t_w$ & wake-up period &  1\,second\\
$t_s$ & sampling period & 1\,second\\
$w$ & window size & 50 samples\\
%$s$ & space window & 50 samples\\
$n$ & network size & 100 nodes\\
\hline
\end{tabular}
\caption{Default parameter settings.}
\label{tab:estreme_parameters}
\end{table}

\vspace{2mm}\noindent\textbf{Window size exploration.} We evaluated the performance of T-Estreme with different window sizes (from 10 to 100) and different network sizes (10, 50, and 100 nodes). 
Figure~\ref{fig:estreme_time_window_exploration} shows the results for
a central node (the black node in Figure~\ref{fig:estreme_testbed} in the full testbed
(100 nodes)).  T-Estreme reaches a plateau at $w=50$. For $n=10$ and $n=50$
(results not shown) we obtained similar results. Since bigger windows will
reduce the estimator's agility to adapt to changes in cardinality, we chose
$w=50$ as the default value. 
%As we will show later, the reason why $w=50$ holds for all $n$ is a natural consequence of the fact that the rendezvous samples follow a gaussian distribution for all $n$.
When both time and spatial information is taken into account, the
cardinality can be estimated as follows: 
$$ \hat{n} = \hat{n}_T\,\alpha +
\hat{n}_S\,(1-\alpha),$$ 
where $\hat{n}_T$ is the cardinality estimated
by T-Estreme, i.e. considering only the node's own samples, and
$\hat{n}_S$ is the cardinality estimated by S-Estreme, i.e. considering
only the node's neighbors averages. In our work, we only report the
extremes. When T-Estreme is used  $\alpha=1$, when S-Estreme is used
$\alpha=0$. Clearly, both estimations, temporal and spatial, can be combined by fine-tuning $\alpha$ according to the needs of the application. 
Figure~\ref{fig:estreme_space_window_exploration} shows the results for
S-Estreme. Due to the gains in spatial correlation, S-Estreme with $w=10$
provides similar results to T-Estreme with $w=50$. 

Note that for network
sizes beyond 70, the performance of S-Estreme decreases. This phenomenon is
due to the smoothing behavior caused by spatial averaging and it is evaluated
in more detail later in this section.  
\begin{figure}
	\begin{center}
		\subfloat[Exploring different window sizes ($w$) for T-Estreme]{. %Note that with $w$=10, few outliers (not shown) reach a relative error up to 2.]{ 
		\label{fig:estreme_time_window_exploration}  
		\includegraphics[width=0.5
		\textwidth]{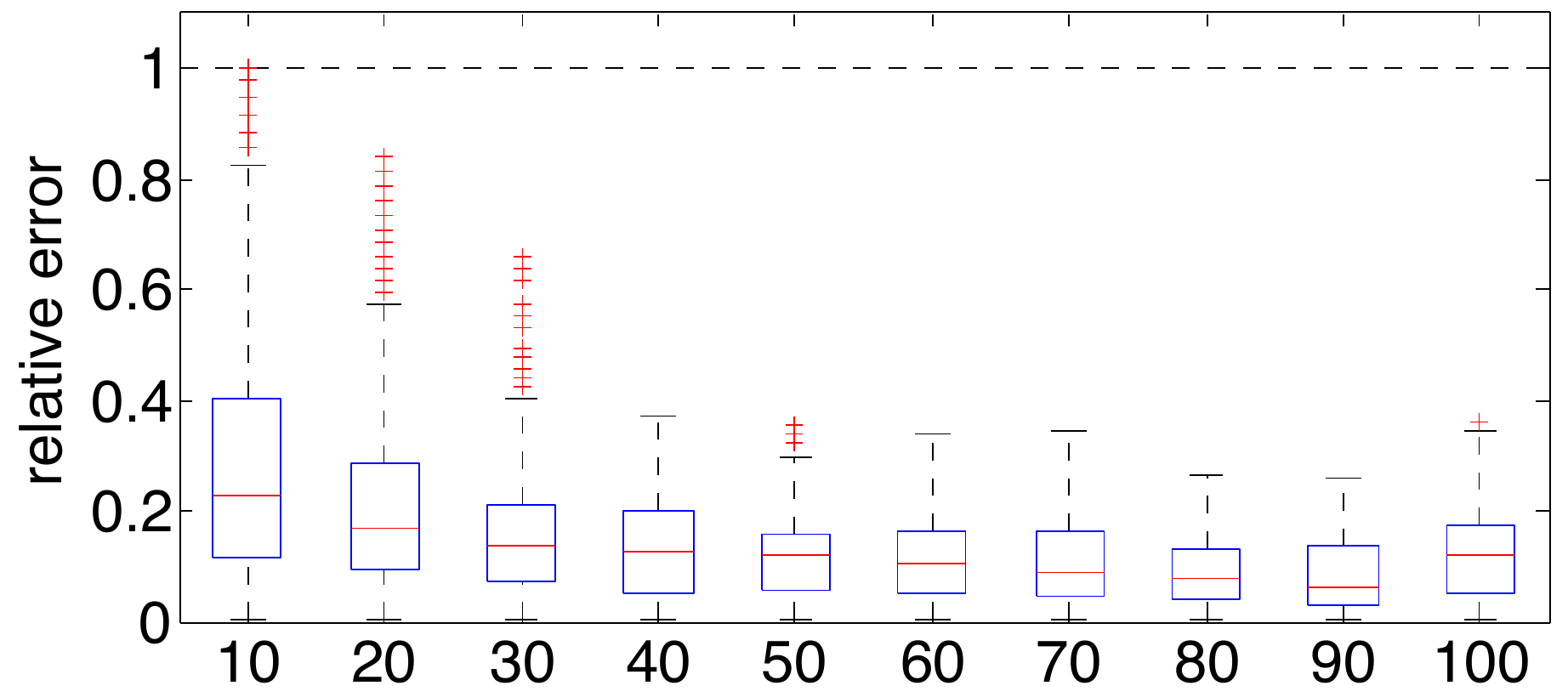} }\\
		\subfloat[Exploring different window sizes ($w$) for S-Estreme.]{
		\label{fig:estreme_space_window_exploration}
		\includegraphics[width=0.5
		\textwidth]{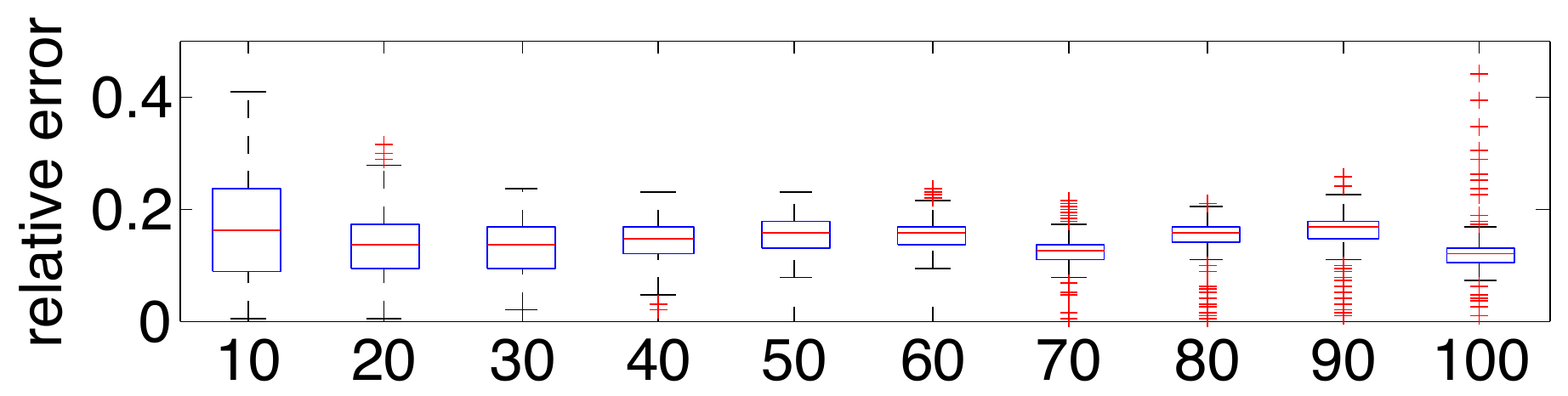} }
		\end{center}
	\caption{Exploring the error and variance of Estreme for different sizes of the sampling windows.} \label{fig:estreme_window_exploration} 
\end{figure}

\paragraph{Comparing with the state of the art.} To evaluate the performance of Estreme, we implemented a \emph{baseline} estimator combining the ideas of three studies~\cite{Ammar1991,Demirbas2011,Zeng2013}. The specific difference of Estreme with each one of these studies is described in the related work section (Section~\ref{sec:estreme_relatedwork}). The baseline estimator works as follows. The initiator broadcasts a request and the subsequent time is divided into 10 time slots. At each slot $m$, nodes jam the channel with probability $\frac{1}{2^m}$. The idea of jamming the channel is taken from \cite{Zeng2013}. Sending raw radio signal strengths (jamming), instead of sending packets, permits estimations in the order of a few ms (like Estreme). The idea of using time slots is borrowed from~\cite{Demirbas2011}. This time-slotted method permits coarse-grained but fast estimations of large cardinalities (which is also the focus of Estreme). Note that the baseline method requires all nodes to be awake at the moment of the estimation, while our method is asynchronous and requires only two active nodes at a time.

\begin{figure}[ht]
	\centering 
	\includegraphics[width=0.96 
	\textwidth]{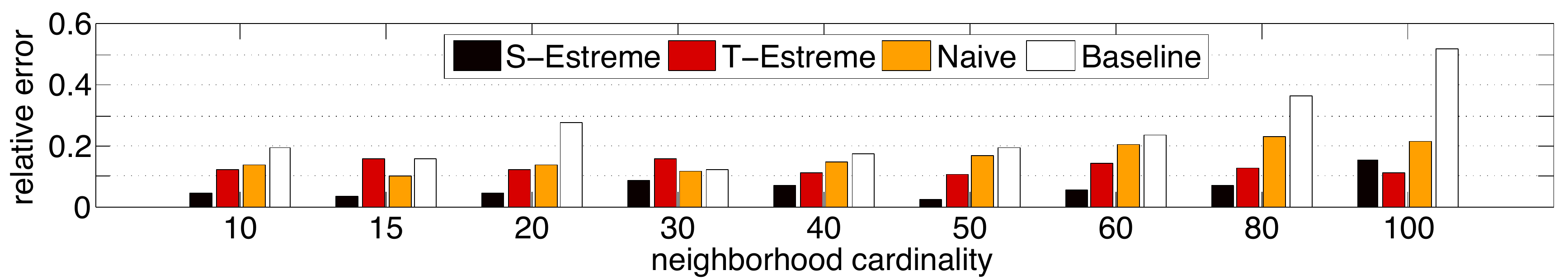} \caption{Estreme's relative error compared to the baseline.}
	\label{fig:estreme_baseline} 
\end{figure}

\subsection{Results}
This section tackles two key aspects of Estreme:  (i) the ability to provide accurate cardinality estimations for large neighborhood sizes while performing concurrent estimations, and (ii) the agility to adjust to network dynamics. 

\paragraph{Insight 1.} \emph{In the temporal domain, T-Estreme provides an accuracy that is similar to those of current solutions for low cardinalities, but at high cardinalities the performance of T-Estreme is many times better}. Figure~\ref{fig:estreme_baseline} depicts the accuracy for various neighborhood sizes and various methods. For each tuple $<$method, neighborhood size$>$, the experiment ran for an hour. For T-Estreme, S-Estreme and Naive, the estimation process runs concurrently on all nodes. For Baseline, the estimation runs only on the \emph{representative} node (the black node in Figure~\ref{fig:estreme_testbed}), because running Baseline on all nodes requires synchronizing their requests, otherwise collisions occur. At low cardinalities (under 40), there is no clear gain for T-Estreme. Even further, while Baseline takes $\approx$10 ms for each estimation, T-Estreme takes on average 1000/$n$ ms, e.g. 100 ms for 10 nodes. Under these conditions, Baseline would have sufficient bandwidth to perform concurrent communications as well (in spite of requiring an extra synchronization step). At higher cardinalities however, the accuracy of T-Estreme remains around 10\%, while Baseline's performance deteriorates significantly. At $n=100$, Baseline and Estreme spend a similar amount of time on each estimation. The reason for the poor performance of Baseline is its coarse-grained logarithmic approach. The estimation could be made more accurate by following a linear approach, as in~\cite{Zeng2013}, but this would require a high level of synchronization and longer times, which may not be permissible in dynamic settings.

\begin{figure}
	\begin{center}
		\subfloat[Averaged estimation accuracy.]{
		\label{fig:estreme_error_bound}
		\includegraphics[width=0.5
		\textwidth]{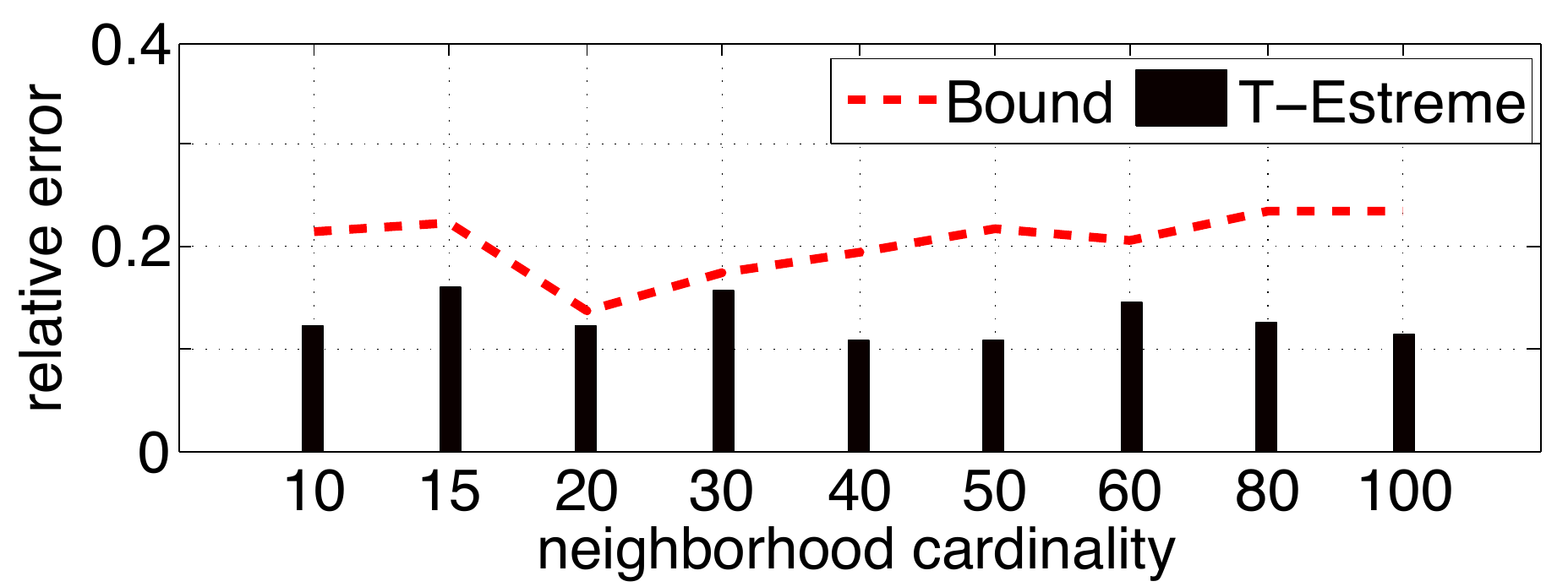} }\\
		\subfloat[Underlying distributions (at the node level) of Estreme's observation errors.]{ 
		\label{fig:estreme_error_distribution}  
		\includegraphics[width=0.5
		\textwidth]{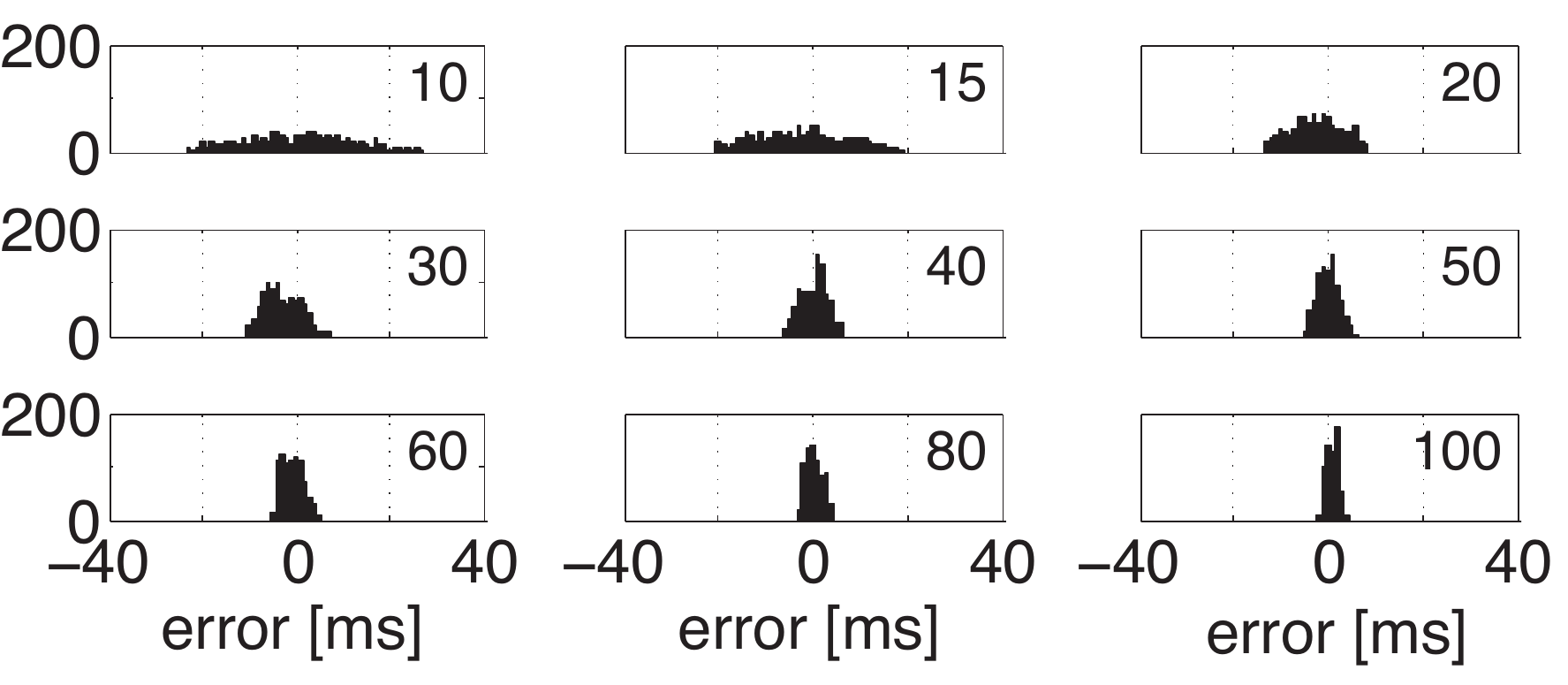} } 
		\end{center}
	\caption{Estimation and observation accuracies of T-Estreme across different cardinalities.} \label{fig:estreme_error_exploration} 
\end{figure}

An important characteristic of T-Estreme is that its accuracy remains rather stable for various cardinalities, between 10\% and 15\%, and the error bound validates this trend, see Figure~\ref{fig:estreme_error_bound} . None of the other estimators shown in Figure~\ref{fig:estreme_baseline} share this characteristic. Intuitively the error should increase as $n$ increases because 1\,ms of error at $n=10$ should matter less than the same error at $n=100$. This rather stable behavior occurs because when $n$ increases, not only the expected rendezvous time decreases linearly but so does the range of errors, see Figure~\ref{fig:estreme_error_distribution}. These trends cancel each other out. At this point it is also important to notice that the distribution of the rendezvous samples follows a Gaussian-like distribution. 

\begin{figure}
	\begin{center}
		\subfloat[T-Estreme in a static scenario with un-even cardinalities.]{ 
		\label{fig:estreme_space_avgtime}  
		\includegraphics[width=0.48
		\textwidth]{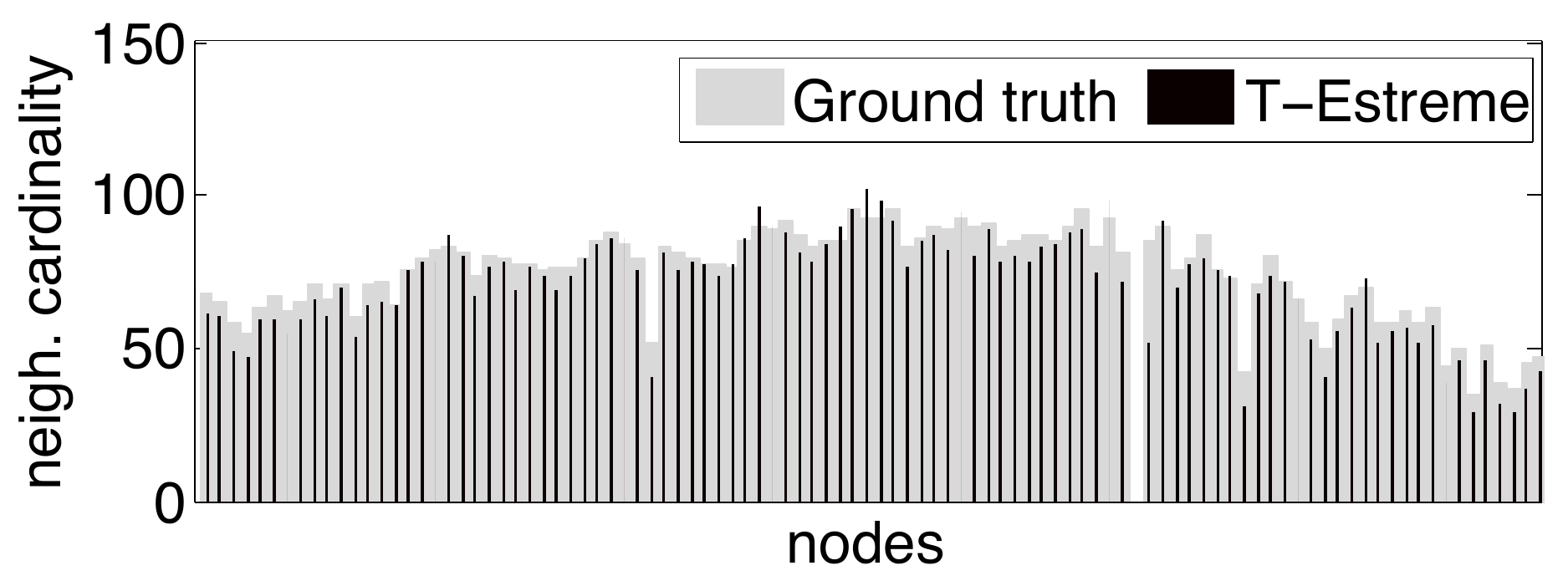} } 
		\subfloat[S-Estreme in a static scenario with un-even cardinalities.]{
		\label{fig:estreme_space_avgspace}
		\includegraphics[width=0.48
		\textwidth]{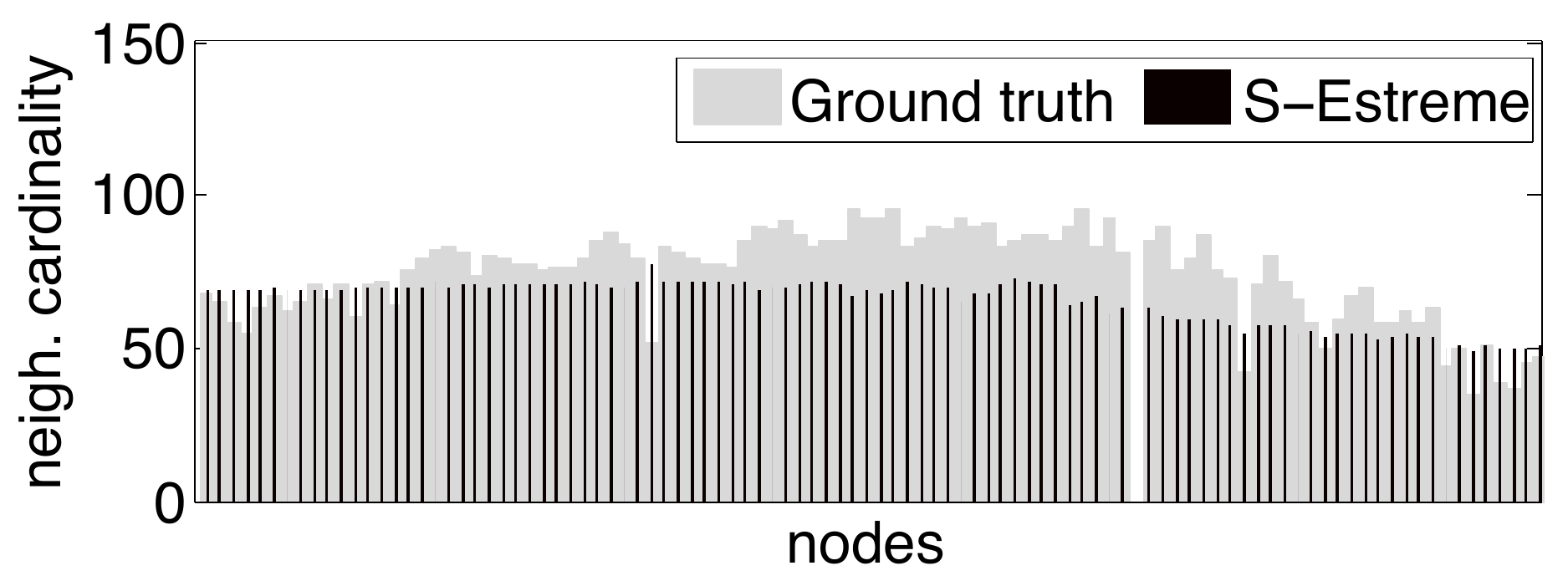} }
		\\
		\subfloat[T-Estreme with cardinality that grows over time (30 nodes every 30 minutes).]{ 
		\label{fig:estreme_time_avgtime}  
		\includegraphics[width=0.48
		\textwidth]{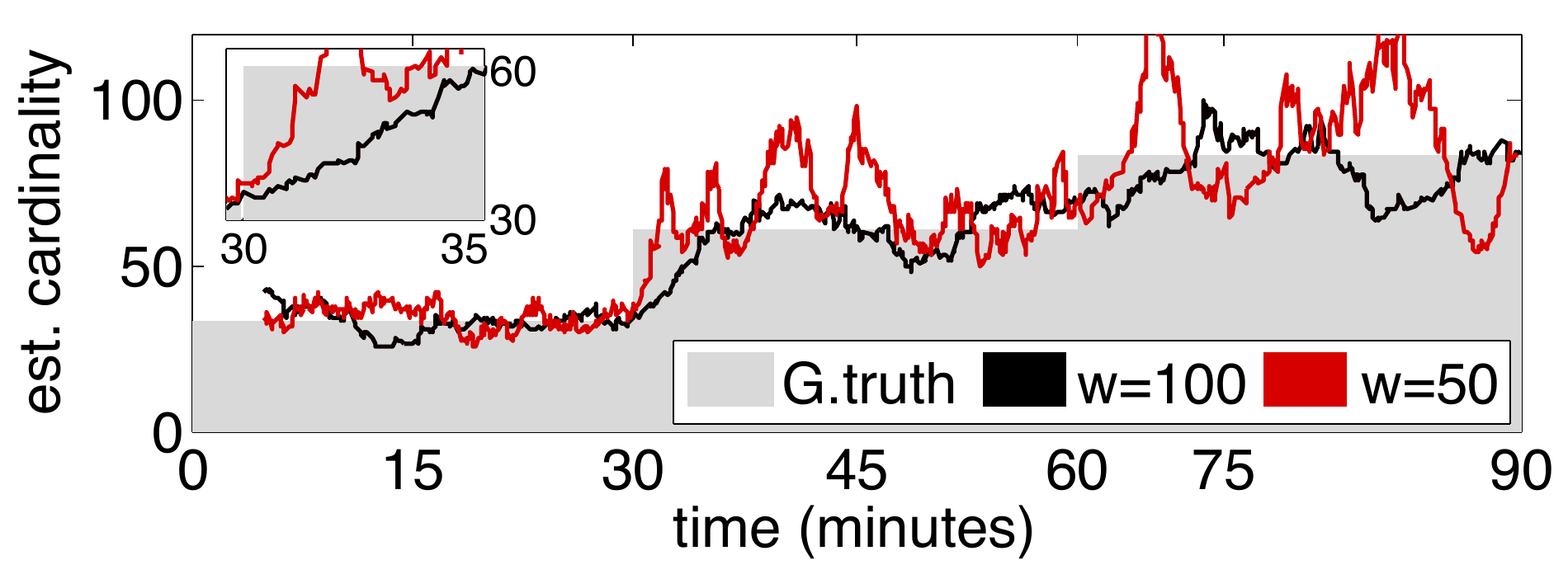} } 
		\subfloat[S-Estreme with cardinality that grows over time (30 nodes every 30 minutes).]{
		\label{fig:estreme_time_avgspace}
		\includegraphics[width=0.48
		\textwidth]{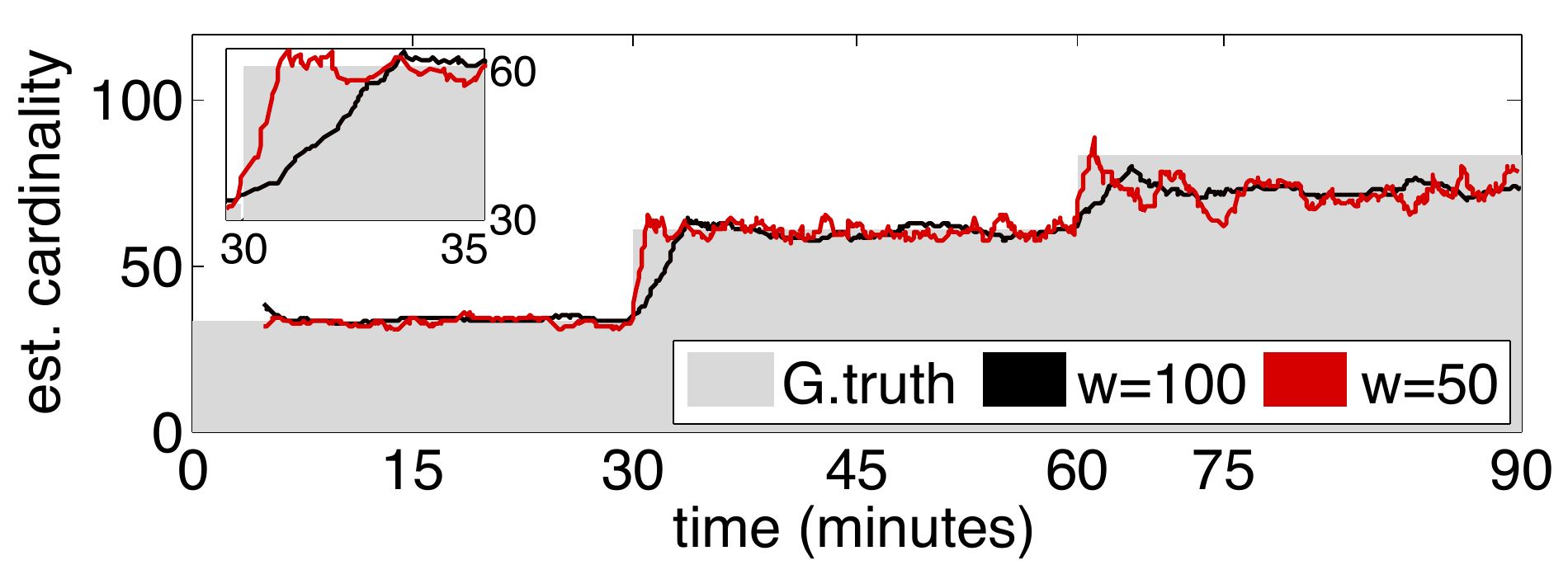} }
		\end{center}
	\caption{Estreme estimations in networks with spatial (top) and temporal (bottom) variations in neighborhood cardinality.} \label{fig:estreme_timespace_exploration} 
\end{figure}
\paragraph{Insight 2.} \emph{In the spatial domain, S-Estreme outperforms current solutions at low and high cardinalities, but in scenarios with un-even deployments the estimation accuracy is compromised due to spatial averaging effects}. The ability of S-Estreme to collect data in a quadratic manner allows for a remarkable accuracy in a short period of time. Figure~\ref{fig:estreme_baseline} shows that for most cardinalities S-Estreme has an error under 5\%. The main exception occurs at $n=100$, and it is due to the spatial averaging of un-even densities. Figures~\ref{fig:estreme_space_avgtime} and~\ref{fig:estreme_space_avgspace} capture the spatial averaging effect. In T-Estreme, each node has a precise view of its own neighborhood, as shown by the accurate mapping between the individual node estimations and the ground truth (Figure~\ref{fig:estreme_space_avgtime}). In S-Estreme however, by averaging the views of its neighbors, a node with a low cardinality --with regards to its neighbors-- will overestimate its density, and vice versa (Figure~\ref{fig:estreme_space_avgspace}). Recall that we are showing the extremes of the parameter $\alpha$ (0 and 1), an intermediate $\alpha$ can be used to balance the benefits of T-Estreme and S-Estreme. It is also important to highlight that, while in principle spatial correlations can be applied to any estimator that exchanges packets with its neighbors, in mobile scenarios this spatial correlation requires a very efficient use of bandwidth (to be fast and to allow concurrency), which is one of the key characteristics of Estreme.

\paragraph{Insight 3.} \emph{Under network dynamics, Estreme adapts to sudden cardinality changes in a few minutes.} We run a series of 90-minute experiments in which the neighborhood cardinality grows. Starting from 30 nodes, every 30 minutes, 30 nodes are added. Figure~\ref{fig:estreme_time_avgtime} shows the estimations of the \emph{representative} node for T-Estreme. With twice the window size $w: 50\rightarrow 100$, the estimator has a lower variance (more accurate), but it takes three times longer (five minutes) to adapt to the change in the neighborhood cardinality (zoom-in on the top left corner of Figure~\ref{fig:estreme_time_avgtime}).  
%The plot in the top left corner of Figure~\ref{fig:time_avgtime} shows a zoom of the moment in wich the neighborhood cardinality grows from 30 to 60 nodes. With $w$=100 the estimator requires approximately 5 minutes to adapt. This is twice of the time of its counterpart with $w$=50. 

S-Estreme, on the other hand, achieves a more accurate and faster convergence with $w=50$, see Figure~\ref{fig:estreme_time_avgspace}. In less than one minute it is able to adapt to the new cardinality. On the last jump however ($n: 60 \rightarrow 90$), S-Estreme suffers from the spatial averaging effect: the \emph{representative} node starts receiving averages from nodes at the far-ends of the testbed (lower density), resulting in an underestimation of the neighborhood size. Note that a 1-minute convergence implies that a person at walking speed (1\,m/s) covers approximately 60 meters while sampling the current neighborhood. Assuming a device with a transmission range of 50\,m, Estreme should be able to cope with the dynamics of practical environments.

As a final experiment, we equipped 3 colleagues with a sensor node and asked them to move according to a predefined path. The experiment lasted 50 minutes. In the first 15 minutes, we asked our colleagues to have a chat at the ground floor of the building (location A, no testbed coverage). After taking the elevator, they reached the 9th floor where the testbed is located and moved towards one far-end of the testbed where the cardinality is lower (location B, cf.\ Figure~\ref{fig:estreme_testbed}). After standing there for 15 minutes, they were asked to move slowly to the other end of the floor (location D). The slow movement (in section C) was required to get an accurate measurement of the ground truth: at each step we required approximately 10\,s to get a snapshot of the cardinality of the testbed node that was closest to the mobile node. 

Figure~\ref{fig:estreme_mobility} shows the estimated neighborhood cardinality of one of the mobile nodes. This figure highlights the tradeoff of T-Estreme and S-Estreme. If a quick estimation is required, S-Estreme is the best solution. On the other hand, if a more accurate, but longer, measurement is needed T-Estreme should be used. 
By considering both estimators is possible to combine the agility of S-Estreme with the accuracy of T-Estreme (see the trade-off estimator in Figure~\ref{fig:estreme_mobility} with $\alpha =0.5$).

\begin{figure}
	\centering 
	\includegraphics[width=0.5 
	\textwidth]{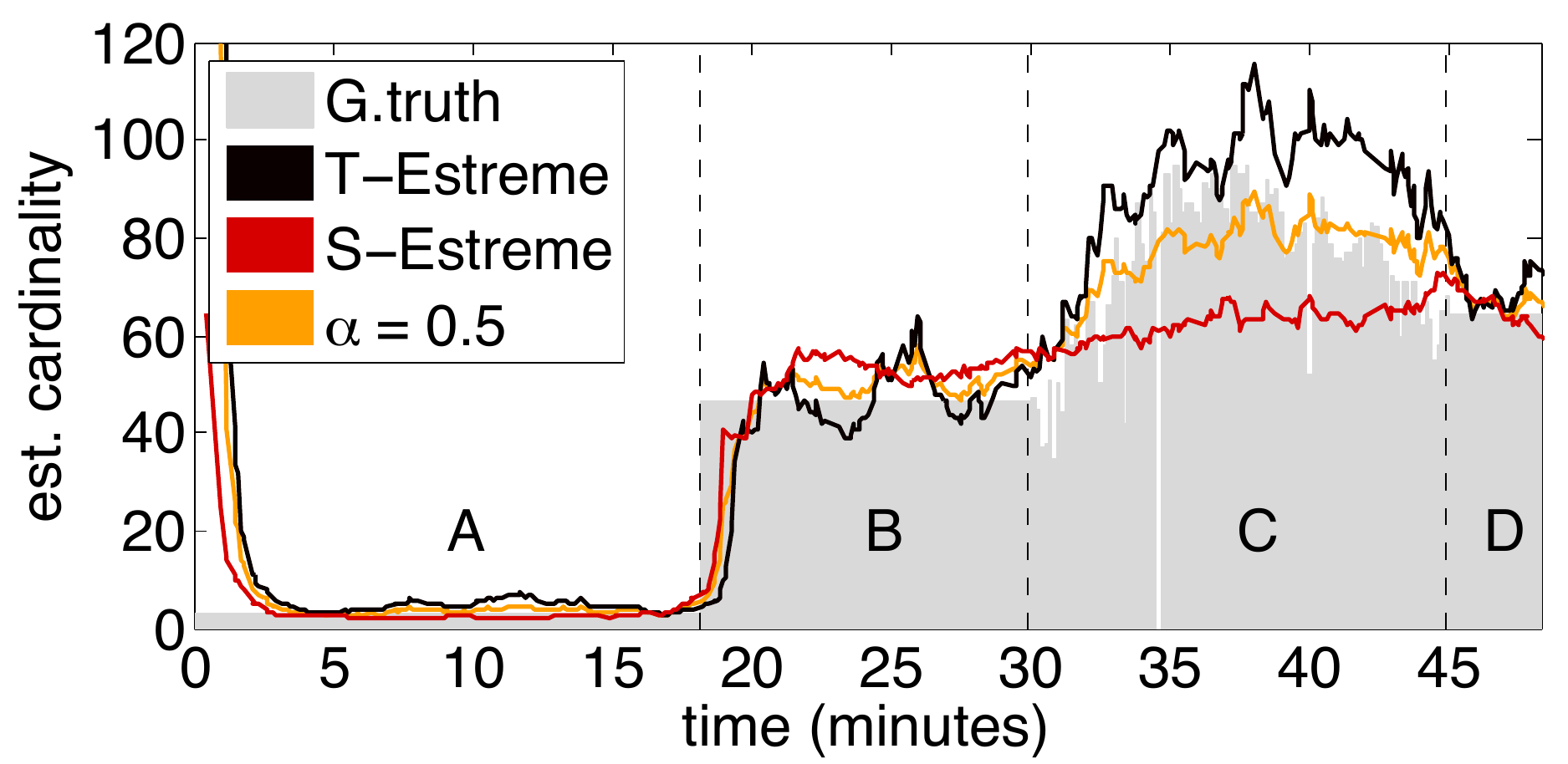} \caption{Estreme's estimations in a 50-minute experiment with a group of 3 mobile nodes.}
	\label{fig:estreme_mobility} 
\end{figure}

\paragraph{Energy efficiency and bandwidth utilization.}
Estreme is not only a simple method, its overhead is very limited. T-Estreme
only needs to piggyback the extra-delays incurred by the receiver (2 bytes).
S-Estreme adds another 2 bytes to convey the current cardinality average of the
receiver. In our implementation, the final data overhead is thus 4 bytes per
estimation. This low overhead together with the short rendezvousing of SOFA,
not only leads to a very efficient utilization of bandwidth, but also to a very
low duty cycle (energy consumption). Figure~\ref{fig:estreme_energy} shows the average
duty cycle of nodes running Estreme on our testbed. The overhead of Estreme is
quite limited, the average duty cycle is between 5\% and 2\%, and decreases
with growing neighborhoods (as the time spent on rendezvousing decreases).
\begin{figure}
	\centering 
	\includegraphics[width=0.5 
	\textwidth]{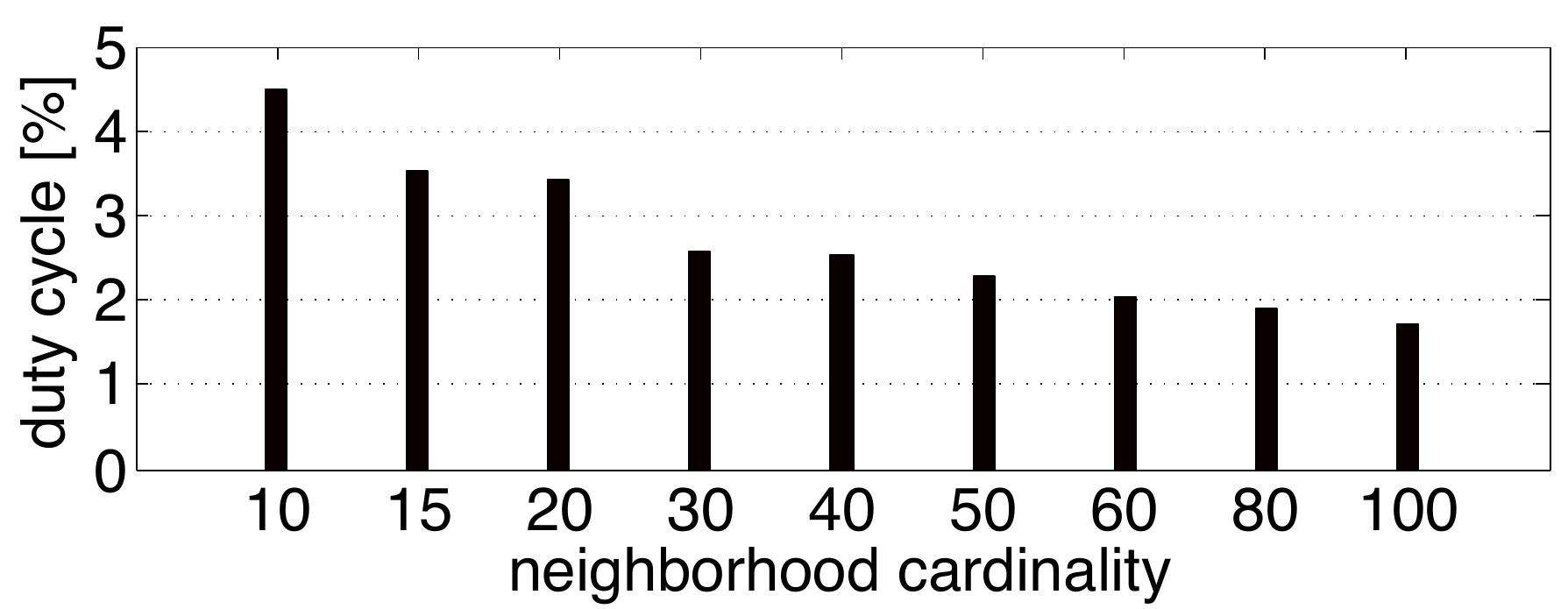} \caption{Average duty cycle of nodes with different cardinalities.}
	\label{fig:estreme_energy} 
\end{figure}

\paragraph{Exploring the parameter space.} Throughout our experiments, we used some default values for the parameters of both the underlying MAC and Estreme itself, namely, the wake-up period ($t_w$) and the sampling period ($t_s$). These parameters influence the performance of Estreme and, hence, it is important to assess them more thoroughly. We tested Estreme with two different wake-up periods, 0.5 and 1 second, and with five different sampling periods, 0.5, 1, 2, 5 and 10 seconds. Each $<$$t_s, t_w$$>$ tuple was run for an hour using 100 nodes. Table~\ref{tab:estreme_limit} show the result of our experiments. The metrics of interest are the relative error and the relative sampling rate. The relative sampling rate captures the percentage of sampling requests that are successfully completed. 

\paragraph{Relative error.} As predicted by our model, increasing $t_w$ decreases the error. Changing the sampling period, instead, does not seem to have a clear consequence on the estimation error, except for very low $t_s$ and $t_w$, which may indicate some channel saturation.

\paragraph{Relative sampling rate}. Ideally, we would like every request from an initiator to be completed, but as $t_s$ decreases this is not possible due to channel saturation issues. For the default setup in our experiments $<$$t_w=1s, t_s=1s$$>$ only approximately 30\% of the requests are successful. This relative low rate occurs because, in Estreme, an initiator that observes and on-going estimation goes back to sleep (according to SOFA's back-off mechanism). As the channel saturates, more and more requests are silently canceled. It is important to highlight however that the fraction of requests that are successful provides the same rendezvous accuracy as of those  requests that are performed on lightly-loaded channels. In general, the \emph{relative sampling rate} can be raised by reducing $t_w$, but results in larger estimation errors, so is to be avoided.

\section{Related work}
\label{sec:estreme_relatedwork}
We start by providing a comparison of Estreme within the (mobile) sensor network domain. We then broaden our scope and overview solutions proposed in RFID systems and mobile phone networks. For a concise overview of the related work see Table~\ref{tab:estreme_sota}. 

\paragraph{Mobile sensor networks.} Although neighborhood cardinality estimation is an essential building block of various protocols, it has been little explored within the context of (mobile) sensor networks. We follow by comparing to the three overarching approaches: 

% these are our strongest opponents in the field
(i) \emph{Polling and group testing mechanisms are the closest to Estreme}. Their goal it to efficiently estimate the cardinality of the $n$ neighbors that hold a given property. Different from Estreme, these mechanisms do not specifically target mobile networks and, in some cases, do not estimate the \emph{exact} number of neighbors. Backcast~\cite{Dutta2008} tests if $n>0$, while~\cite{Demirbas2011} tries to asses if $n$ crosses a certain threshold. If $n>0$, Strawman~\cite{Osterlind2012} is able to efficiently identify one of the responders. Among the various polling mechanisms, the one that is the most related to our work is the study by Zeng \etal~\cite{Zeng2013}, where the authors propose two different methods. LogPoll is able to estimate the logarithmic cardinality of a neighborhood, \ie $\log{n}$, in just $4\,ms$, while LinearPoll is able to recognize all the identities of the responders at the cost of spending more energy and time. This remarkable performance comes at a high cost and relies on several assumptions. LogPoll and LinearPoll model and periodically calibrate the signal strength of neighbors, effectively limiting the applicability of these estimators to static networks with minor link dynamics. Additionally, LinearPoll can only count a limited number of neighbors, because it requires neighbors to send radio signal strengths that are $\Delta$ apart. LogPoll, on the other hand, can only estimate $\log{n}$, which gets increasingly coarse-grained as $n$ grows. 
%To obtain more accurate estimations, it would require multiple estimation rounds, frustrating it main feature: efficiency. 

\begin{table}
\renewcommand{\arraystretch}{1.3}
\small
\centering
\begin{tabular}{c|ccccc}
\hline
\multicolumn{6}{c}{Relative error ($\left|\hat{n}-n\right|/n$)} \\	
\hline\hline
$t_w$ &  \multicolumn{5}{c}{Sampling period $t_s$ (s)} \\
(s) & 10 & 5 & 2 & 1 & 0.5 \\	
\hline
0.5 & 0.16 & 0.16 & 0.29 & 0.31 & 0.24\\
1 & 0.13 & 0.11 & 0.14 & 0.11 & 0.07\\
\hline
\end{tabular}
%\caption{Relative error for different $t_w$ and $t_s$.}
%\label{tab:estreme_limit}
%\end{table}
%\begin{table}
%\renewcommand{\arraystretch}{1.3}
%\small
%\centering
\quad
\begin{tabular}{c|ccccc}
\hline
\multicolumn{6}{c}{Relative sampling rate (\%)} \\	
\hline\hline
$t_w$ &  \multicolumn{5}{c}{Sampling period $t_s$ (s)} \\
(s) & 10 & 5 & 2 & 1 & 0.5 \\	
\hline
% 0.5 & 0.09 & 0.15 & 0.29 & 0.40 & 0.47\\
% 1 & 0.08 & 0.14 & 0.24 & 0.31 & 0.35\\
0.5 & 93.4 & 76.9 & 58.7 & 40.2 & 23.9\\
1 & 81.6 & 73.3 & 49.3 & 31.8 & 17.7\\
\hline
\end{tabular}
\caption{Relative error and sampling rate for different $t_w$ and $t_s$.}
\label{tab:estreme_limit}
\end{table}

\begin{table}[!t]
\renewcommand{\arraystretch}{1.3}
\small
\centering
\begin{tabular}{@{}c|l@{\hspace{1em}}c@{\hspace{1em}}c@{\hspace{1em}}c@{\hspace{1em}}l@{\hspace{1em}}l@{\hspace{1em}}l@{\hspace{1em}}c@{}}
\hline
Work & Evaluation & Mobile & Scale & Error & Technique & Estimator & Device & Concurrent\\
\hline\hline
\cite{Kannan2012} & Real world & yes & 25 & 3-35\% & Audio tones & LoF~\cite{Qian2011} & Phone & few \\

\cite{Weppner2011}  & Real world & yes & 50 & 20-70\% & Bluetooth & Classifier & Phone & few \\

% \cite{Yuan2013} & dens. est. & real world & yes & 16 & 5-15\% & RSSI meas.& Classifier & TelosB & no \\

\cite{Zeng2013} & Testbed & no & 32 & 1\% & RSSI meas.& Classifier & Mote & no \\

\cite{Iyer2011}  & Simulation & yes & 100 & 25\% & Radio obs.& MLE est. & Mote & yes \\
% \cite{Cichon2012} & card. est. & simul. & yes & 10k & 20\% & RSSI & 2-phase est.\footnote{order statistic. and bernoulli trials} & Sensor node & 1 \\

\cite{Qian2011} & Simulation & yes & 65K & 10\% & Data packets & LoF~\cite{Qian2011} & RFID & no \\

\cite{Han2010a} & Simulation & yes & 10K & 1\% & Data packets & Order stat. & RFID & no \\
\hline
\hline
Estreme & Testbed~\footnote{In \chapref{chapter:nemo} Estreme is also evaluated in real-world conditions.} & yes & 100 & 10\% & Radio obs.& Order stat. & Mote & yes \\
\hline
\end{tabular}
\caption{Comparison of the various methods to estimate neighborhood cardinality. In the ``Technique'' column, ``meas.'' stands for measurements and ``obs.'' for observations. In the ``Estimator'' column, ``est.'' stands for estimator and ``stat.'' for ``statistic''.}
\label{tab:estreme_sota}
\end{table}

% one phrase on neighbor discovery
(ii) \emph{Cardinality estimation can also be provided through neighbor discovery mechanisms like Disco~\cite{Dutta2008}, WiFlock \cite{Purohit2011} and U-Connect~\cite{Kandhalu2010}}. Unfortunately, even with the help of acceleration techniques such as ACC~\cite{Zhang2012}, these mechanisms take in the order of 10 minutes or more, which is too slow to cope with the mobile networks addressed by Estreme. 

% the strange case of netdetect
(iii) \emph{Last, the NetDetect~\cite{Iyer2011} algorithm, while similar in spirit to Estreme, makes stronger assumptions and achieves lower estimation accuracy.} Similar to Estreme, NetDetect relies on the underlying distribution of packet transmissions to estimate neighborhood cardinality. Nevertheless, with the same default window as Estreme (50 samples), it estimates the cardinality of a 100-node network with a relative error of 0.25 (in simulations). Furthermore, unlike Estreme which is built on top of duty-cycling MAC protocols, NetDetect uses the Aloha protocol and assumes that the radio is always on.

\paragraph{RFID}. An area in which cardinality estimation plays a central role is Radio-Frequency Identification (RFID). Different from sensor networks, RFID systems are designed to track and monitor thousands of goods in storage facilities. These systems assume a centralized initiator, called \emph{reader}, and a set of cheap, passive devices called \emph{tags}. To estimate the neighborhood cardinality (the number of tags), the reader starts a \emph{response collection} process, which is usually based on a TDMA scheme. This process consists of multiple trials, each one having multiple slots. Within each trial, every tag selects a slot in which it sends a response to the reader. This slot is chosen based on the tag's state, a random number, and a command sent by the reader at the beginning of each trial.

Based on the number of observed responses, collisions, and empty slots the initiator is able to estimate the population size~\cite{Qian2011}. This process is repeated iteratively multiple times (trials) to improve the accuracy of the estimators. While usually the reader cannot terminate a trial until the last assigned slot is observed, in~\cite{Han2010a} the authors propose to use order statistics to drastically reduce the trials' duration (an initiator needs to observe only the first $k$ responses).
Compared to Estreme, this work is able to achieve an accuracy of 1\% and scale to thousands of nodes. On the other hand, it has only been tested in simulations and requires two orders of magnitude more samples than Estreme. 

Adapting such a mechanism (and others from the RFID community) to wireless sensor networks is possible but very challenging. For a correct observation, neighbors will need to be synchronized, so that responses in adjacent slots do not collide. Moreover, a global scheduling mechanism should be in place to guarantee that at any moment in time only one node in the neighborhood can act as a reader (initiator).

Recently, Chen \etal~\cite{Chen2013} discovered that having a 2-phase estimation (a rough one followed by a more refined one) is the key to improve the estimation accuracy, independently of the chosen technique. In light of this, we plan to extend Estreme with a second estimation phase that will use the IDs of the encountered neighbors to improve the estimation accuracy.

\paragraph{Mobile phones.}
Initial studies have used mobile phones to estimate the density of crowds. 
The most relevant work uses audio tones to count neighbor devices~\cite{Kannan2012}. The main challenge involved is to successfully transmit data packets using low-quality speakers/microphones, as well as to cope with the presence of environmental auditive noise. Thanks to the richer computational capabilities of mobile phones, techniques such as Fast Fourier Transform can be used to code the signal. Energy efficiency is also addressed, but on a different scale than Estreme. Here the comparison is against the typical consumption of a WiFi network interface. Finally, similar to Estreme, this system is able to support multiple, concurrent estimations. While in theory the method scales to hundreds of devices by using multiple frequencies, the system was only tested with two concurrent estimators.

An interesting alternative to cardinality estimation is proposed in \cite{Weppner2011}. The phones scan for discoverable bluetooth devices and, based on the number of unique identifiers discovered, an estimation of the crowd density is performed. Different from Estreme and the other related works, this work focuses on classifying the density into four classes instead of providing a cardinality estimation. Even though the estimation task is simpler, unrealistic assumptions on the distribution of bluetooth-enabled devices results in estimation errors that vary from 20 to 70\%.

\section{Conclusion}
\label{sec:estreme_conclusions}
In this chapter we addressed the issue of determining the neighborhood cardinality in Extreme Wireless Sensor Networks, where link quality fluctuations and
node mobility ask for a robust and agile approach. An additional and important focus of our work is to handle high densities and concurrent estimations, requirements that are necessary in public safety applications such as crowd monitoring. Traditional approaches cannot meet these stringent requirements, so we developed Estreme: a cardinality estimator based on monitoring the inter-arrival times of (randomized) events raised by neighboring nodes \ie SOFA's wake-ups (\chapref{chapter:sofa}). 

To minimize channel usage, which is essential in EWSNs, Estreme
leverages on the short rendezvous time offered by SOFA's opportunistic anycast. 
To obtain high accuracy, Estreme employs a bag of tricks to account for various overheads, (transmission) latencies, collisions, and other factors that all distort
the true rendezvous time with the first node to wake up. We derived a
theoretical model underlying the rendezvous times, and used it to gain
insight into the accuracy of Estreme as well as to derive bounds on the estimation error.

As a proof of concept, we implemented Estreme on the Contiki OS, and evaluated it on a testbed with node densities up to 100 nodes. Estreme achieves solid performance results with typical estimation errors
below 10\%, which compares favorably to state-of-the-art solutions. Estreme was
also demonstrated to handle abrupt changes in density exemplified through an
experiment involving few nodes moving through our testbed.

Later in \chapref{chapter:nemo} we will also evaluate Estreme on several hundreds of devices moving in a large-scale science museum, showing that Estreme is able to scale to extreme mobile environment as soon as the spatial changes in density are not too drastic.

%As part of our future work, we are currently planning a more ambitious deployment consisting of hundreds of nodes, with a large fraction of them being mobile. We are also working on extending our analytical framework to provide bounds on delay and to characterize the operational regions of Estreme based on the channel saturation.

%% file: staffetta/staffetta.tex
\chapter{Opportunistic Data collection}
\label{chapter:staffetta}

\blfootnote{Parts of this chapter have been published in SenSys'16, Stanford, USA~\cite{Cattani2016}.}

\epigraph{Kal-toh is not about striving for balance. It is about finding the seeds of order, even in the midst of profound chaos.}{Lieutenant Tuvok}

\dropcap{I}{n} Wireless Sensor Networks, data collection is a solved problem. That is, from the packet-reception
perspective, with protocols reported to collect more than 99.9\,\% of all data. 
However, data collection is still a challenge when viewed from the real-world perspective, where networks dynamics and energy conditions (see \chapref{chapter:introduction}) play a major role. 
In that light, ORW is the first protocol to effectively use opportunistic anycast in WSNs~\cite{Landsiedel2012} to perform data collection. 
ORW shows that compared to traditional routing schemes, such as CTP~\cite{Gnawali2009}, opportunistic forwarding decisions can significantly improve data collection performance in asynchronously duty-cycled networks in terms of end-to-end packet latency and energy efficiency, while being more resilient to topology changes. The latter characteristic is essential for data collection mechanisms to scale to EWSNs. 

As seen in \chapref{chapter:sofa}, the idea of opportunistic routing is simple: instead of first making the routing decision and then waiting for the destination to wake up, nodes forward packets opportunistically to the first neighbor that wakes up. The more potential forwarders a node has, the shorter the time it needs to wait before it communicates, and hence the more efficient it can operate~\cite{Landsiedel2012}
%\footnote{While in SOFA all neighbors are potential forwarders, in ORW nodes filter-out the forwarders that do not provide enough routing progress towards the sink.}.

Opportunistic routing (ORW, SOFA, \etc) typically assumes that all nodes have the same duty cycle, and thus, they all end up having the same forwarding costs in terms of energy and delay. But nodes play different roles. The closer a node is to the sink, the more packets it has to forward, and hence, the lower its forwarding costs should be. Duty cycles should be adjusted to the local needs of each node. This is indeed the idea of various adaptive duty-cycling mechanisms.
However, these have only been studied analytically~\cite{Kim2011,Kim2010}, thus lacking validation against the real-world dynamics of low-power wireless, or have been designed for static networks and traditional communication primitives such as unicast~\cite{Jurdak2007,Meier2010,Zimmerling2012}, which renders them unusable in an opportunistic setting and extreme network conditions.

Inspired by the promising results of opportunistic protocols, we propose a tight integration of duty cycling and opportunistic routing with a mechanism called \emph{Staffetta}. 
Staffetta is not another routing protocol, but a duty-cycle scheduling mechanism for opportunistic data collection that adapts the wake-up frequency of individual nodes such that nodes closer to the sink wake up more often than nodes at the edge of the network. The resulting \emph{activity gradient} drastically reduces routing costs
%in terms of energy, bandwidth, and latency, 
and automatically ``steers'' packets in the right direction, as the probability of finding an awake neighbor is highest in the direction of the sink.
Since Staffetta operates at the medium access control~(MAC) layer, it can be combined with different opportunistic data collection protocols that operate at the network layer, targeting real-world applications featuring dynamic multi-hop networks with data rates on the order of one packet every few seconds~\cite{Ceriotti2009,Dutta2010}.
To the best of our knowledge, Staffetta is the first mechanism that applies the idea of dynamic duty cycling to opportunistic data collection. 

We implement Staffetta in Contiki and run extensive experiments on the FlockLab and Indriya testbeds to quantify the improvements of Staffetta underneath traditional opportunistic protocols, which typically do not adjust the nodes' wake-up frequency.
We find that, despite its simplicity, Staffetta considerably improves data collection performance, most notably in terms of packet latency and network lifetime, while incurring minimal overhead.
From a broader perspective, we show that Staffetta's activity gradient can also be used as a routing metric by itself, achieving performance similar to other routing mechanisms empowered by Staffetta, but without the need for maintaining a complex routing structure.

\vspace{-.5ex}
\paragraph{Contributions.} This chapter makes three main contributions:
\begin{itemize}
\item \secref{sec:staffetta_mechanism} presents the design of Staffetta, a lightweight, generic mechanism that adapts the duty cycle of individual nodes to improve opportunistic data collection. By studying the relation between duty cycling and opportunistic communication, we show how Staffetta's activity gradient reduces packet latency, while extending network lifetime.

\item \secref{sec:staffetta_implementation} describes how one can efficiently implement and combine
Staffetta with existing opportunistic routing protocols. As a paradigmatic case study, we choose three opportunistic mechanisms inspired by state-of-the-art collection protocols and their routing metrics: Expected Duty Cycle~\cite{Landsiedel2012}, Queue Backlog~\cite{Moeller2010}, and SOFA's Random Walk~\cite{Cattani2014a}.

\item \secref{sec:staffetta_results} evaluates the above described routing metrics in combination with Staffetta on two testbeds. Our results show that Staffetta reduces end-to-end packet latency by an order of magnitude and halves energy consumption across the board.
%, while incurring a negligible overhead in terms of memory and processing
A case in point is the substantial improvement Staffetta brings to ORW in one of the testbeds: two to six times less energy consumption, 80 to 450 times shorter latencies and faster adaptation to mobile sinks. These gains are obtained with a negligible overhead in terms of memory and processing.  
\end{itemize}

\vspace{1ex}
\noindent
In light of our contributions, we review related work in \secref{sec:staffetta_sota} and conclude in \secref{sec:staffetta_conclusions}.

\section{Problem Statement}
\label{sec:staffetta_problem}

\begin{table}
\centering

\begin{tabular}{lc}
\hline
Description & Symbol\\
\hline
Energy budget & $\text{DC}_{max}$\\
Number of forwarders & $n$\\
Wake-up frequency & $\omega$\\
Hop distance to sink & $h$\\
Forwarding delay & $\delta_F$\\
Rendezvous time & $\delta_R$\\
Transmission time & $\delta_{tx}$\\
\hline
\vspace{-1.5ex} \\
\multicolumn{2}{p{.68\columnwidth}}{\small When appropriate symbols will be indexed by node number (i) or hop distance from the sink (h).} \\
\end{tabular}
\caption{Mathematical notations}
\label{tab:staffetta_notations}
\end{table}

%\think{m.zu: we need a sentence to state that a high delivery rate is the most important, but it is already guaranteed by most protocols. After that we can mention that latency and overhead are important. In the intro we should also mention delivery rate somehow.}
In duty-cycled networks, two main metrics determine the efficiency of data collection. First, the time until a packet is collected at the sink, known as \emph{packet latency}. 
Second, the \emph{routing overhead}, which accounts for the resources (messages, time, energy, \etc) spent by each node to forward messages.
%\think{reception rate also account in this overhead, since a lost packet can always be retransmitted by higher network layer at the cost of additional overhead}.
%
The goal of many data collection protocols (and their routing metric) is to \emph{minimize} the packet latency $\Delta$ with the least possible amount of additional resources. More formally, consider a collection protocol and denote by $P = \{ u_0, u_1, \ldots, u_{h}\}$ the sequence of nodes forming the path of length $h$ (in hops) from node $u_0$ to the sink $u_{h}$. 
In an asynchronous duty-cycled network, the packet latency over path $P$ is given by
\begin{align} 
	\Delta_P = \sum_{i = 0}^{h-1} {\delta_F(i)},
\label{eq:staffetta_min_goal}
\end{align} 
where the \emph{forwarding delay} $\delta_F(i)$ is the time node $u_{i}$ needs
to send a packet to $u_{i+1}$. This delay can be subdivided into two different
components: the time node $u_i$ waits for its destination to wake up (\emph{rendezvous time}, $\delta_R(i)$) and the time required to forward the packet (\emph{transmission time}, $\delta_{tx}$):
\begin{align}	
 	\delta_F(i) = \delta_R(i) + \delta_{tx}.
	\label{eq:staffetta_fwd_delay}
\end{align}

When $\delta_R$ and $\delta_{tx}$ are constant across all nodes\footnote{This is the case for most collection protocols for WSNs that have fixed wake-up frequencies and communicate with unicast primitives.}, the efficiency of data collection is improved by minimizing the path length $h$.
%, can be used for data collection with minor modifications (\eg keeping into account the channel quality/conditions). 
However, this is not the case for \emph{opportunistic routing}, where $\delta_R(i)$ can drastically change (shorten) depending on the number of forwarders $n_i$ and their wake-up frequencies.
%($\delta_{tx}$ is often a constant that is implementation and platform dependent).
%
This is because in opportunistic routing, messages are routed to the \emph{first} available neighbor that provides enough routing progress. The more forwarders a node has, the sooner one of them will wake up.
%
%Concretely, consider a node $u_i$ whose $n_i$ neighbors $N = \{u_1, \ldots, u_{n_i}\}$ wake-up frequencies $\omega_1 \dots \omega_{n_i}$.
%
Concretely, when all forwarder wake-up frequencies are identical and equal to
$\omega$, it is known~\cite{Cattani2014b,Ghadimi2014} that the rendezvous time $\delta_R(i)$ has an expected value of
\begin{align}
	 \E{\delta_R(i)} = \frac{1}{(n_i + 1)\,\omega}.
	\label{eq:staffetta_expected_rendezvous}
\end{align}
%

%From \eqref{eq:expected_rendezvous} and \eqref{eq:uneven_rendezvous_approx},
This rendezvous process has two key problems.
First, the following trade off arises: increasing the number of forwarders $n_i$ reduces the expected $\delta_R(i)$, but increases the path length (due to non-optimal routing paths). On the other hand, with few forwarders the path length is reduced, at the cost of longer forwarding delays. 
%Finding the right number of forwarders that balances efficient forwarding and path lengths is the key to efficient data collection in opportunistic routing\footnote{In ORW, the best trade off was experimentally found to be 5 forwarders.}.
Therefore, finding the right number of forwarders is a complex balancing act~\cite{Dubois-ferriere2011},
%between path length and forwarding delay, 
which is key to efficient data collection in opportunistic routing\footnote{\mbox{In ORW, the best trade off was empirically found to be 5 forwarders~\cite{Landsiedel2012}}}. 
Second, the expected rendezvous (eq.~\ref{eq:staffetta_expected_rendezvous}) is load-agnostic. A node closer to the sink has the same rendezvous cost as a leaf node, but it needs to rendezvous way more often (due to its higher forwarding load). Having the same expected rendezvous time across the network is like designing a city using only streets (no avenues or highways): latency and energy consumption increase unnecessarily. To solve these two problems, we propose a low-overhead and fully distributed duty-cycle scheme that adapts the rendezvous time of each node according to its local needs. \emph{The basic idea behind our solution is very simple. Instead of setting a constant wake up frequency for all nodes, we set a constant energy budget. All nodes will die at the same time. This means that nodes need to make local decisions regarding their use of power: if their observed rendezvous time is long (short), they should wake up less (more) often.} As we will see, this adaptive scheme is general enough to run underneath various opportunistic routing methods, and improves the performance of these methods on all fronts: the delivery rate is maintained or even improved slightly, the lifetime is increased and the latency is reduced. 

% !TEX root = staffetta.tex

\section{Staffetta}
\label{sec:staffetta_mechanism}
\begin{figure}
\begin{center}
	\subfloat[Fixed (static duty cycle)]{
	\includegraphics[width=0.6\linewidth]{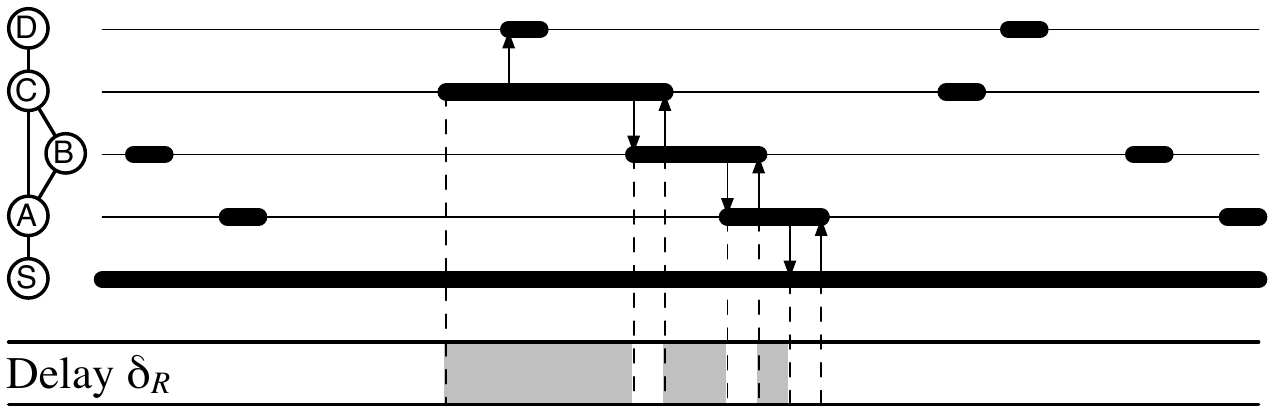}
	\label{fig:staffetta_opp}
	}
	\\
	\subfloat[Staffetta (dynamic duty cycle)]{
	\label{fig:staffetta_opp_staffetta}
	\includegraphics[width=0.6\linewidth]{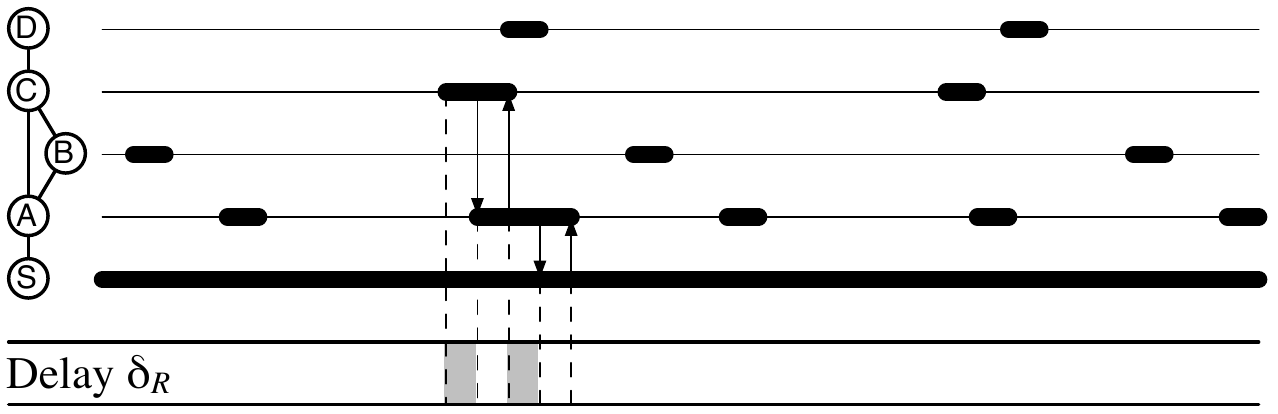}
	}
	\end{center}	
	\caption{Opportunistic routing of a message from node C to the sink S (see topology on the left). By letting the node closer to the sink (node A) wake up twice as often, the rendezvous times (in grey) and the number of hops are reduced.} 
	\label{fig:staffetta_effects}
\end{figure}

Staffetta builds upon the intuition that, by 
letting nodes closer to the sink wake up more frequently than others, %(\emph{activity gradient}), 
one can drastically shorten the forwarding delay of messages, while reducing the average path length. 
As an example, Figure~\ref{fig:staffetta_effects} compares an opportunistic collection mechanism with fixed wake-up frequency (\ref{fig:staffetta_opp}) against one where the node closer to the sink \ie node A, wakes up twice as frequently (\ref{fig:staffetta_opp_staffetta}). Since in Figure~\ref{fig:staffetta_opp_staffetta} A increases its activity, it wakes up before B and gets selected as forwarder by C. This not only reduces the forwarding delay of C, but also decreases the number of hops needed to deliver C's message to the sink.

In opportunistic routing, the forwarding decisions of nodes are based on two aspects: i) the rendezvous time, and ii) the routing metric. 
Nodes forward their message to the \emph{first} node that wakes-up and provides \emph{enough} routing improvement (trading optimal routing decision for lower forwarding delays). 
Thus, by making nodes closer to the sink wake up more often, it is possible to improve opportunistic data collection in two ways: rendezvous with nodes closer to the sink becomes i) \emph{more probable} and ii) \emph{more efficient}. %Next, we explain how Stafetta forms the activity gradient.

\subsection{Activity Gradient.}
\label{sec:staffetta_activity_gradient}
Staffetta adjusts the wake-up frequency of nodes in order to create an \emph{activity gradient}, where nodes closer to the sink wake up more often than the ones further away. Staffetta achieves this in a fully distributed fashion and with minimal overhead, as explained next. 

\vspace{2mm} \noindent i) \emph{Fix an energy budget.} Given a network operational lifetime, as required by the application, Staffetta imposes a maximum duty cycle $DC_{max}$, $0 < DC_{max} < 1$, that is equal for all nodes in the network. This value determines the maximum fraction of time each node can keep its radio on. Since the radio is often the most power-hungry component on a typical low-power sensor node, $DC_{max}$ effectively ensures that nodes stay within a fixed energy budget.

\vspace{2mm} \noindent ii) \emph{Adjust the wake-up frequency to meet the energy budget.} Given an observed forwarding delay $\hat{\delta}_F(i)$, each node computes its wake-up frequency as 
\begin{equation} %\[
\omega_i = \left\{ 
\begin{array}{ll}
  \infty & \text{sink (always on)}\\
  DC_{max} \, / \, \hat{\delta}_F(i) & \text{otherwise}.
\end{array} 
\right.
\label{eq:staffetta_wakeup_frequency}
\end{equation} 
Note that, different from the model in Equation~\eqref{eq:staffetta_fwd_delay}, $\hat{\delta}_F$ is based on the nodes' observations, thus it takes into account implementation specific delays such as packet retransmission and channel sensing. Whenever the observed $\hat{\delta}_F$ changes, nodes update their wake-up frequency.

\vspace{2mm} \noindent \emph{How does this mechanism create an activity gradient?}
In the presence of a node that is always on, \eg a sink, the activity gradient emerges as follows. All nodes start with a fixed wake up frequency, as in traditional opportunistic routing. However, since the sink does not duty-cycle its radio, its neighbors experience extremely short forwarding delays that are approximately equal to $\delta_{tx}$ (the rendezvous with the sink is almost instantaneous). 
%Assuming there is enough energy budget to send multiple messages, 
Thus, any node residing 1-hop from the sink adapts its wake-up frequency to
\begin{equation}
	\omega_1 \approx {DC_{max}} \, / \, {\delta_{tx}}.
	\label{eq:staffetta_onehop_rendezvous}
\end{equation}
The higher activity of 1-hop nodes will, in turn, reduce the forwarding delay of 2-hop nodes, which will then increase their wake-up frequency. 
This \emph{cascading effect} quickly propagates (with a decaying factor) throughout the network, creating what we call 
the \emph{activity gradient} of Staffetta. 

\subsection{Analysis}

While the behavior of a single entity (e.g., adjusting the wake-up frequency) is easy to describe and understand, the emergent behavior of the system (\eg the resulting activity gradient) is more complex and difficult to predict.
To gain a further understanding about Staffetta's gradient formation, we now present a numerical evaluation and a simple model. \emph{This analysis is not meant to define the assumptions or guidelines used for the practical implementation of Stafetta, but to use a clean setup to capture its macro properties.} A thorough description of Staffetta's practical implementation is presented in the next section.

Figure~\ref{fig:staffetta_matlab} shows a numerical evaluation of opportunistic routing with and without Staf\-fetta. The evaluation is based on the simple network shown at the bottom of the figure. Except for the sink's neighbors, each node has 3 forwarders, an initial wake-up frequency of $\omega = 5$\,Hz, a data rate of 1\,Hz and $DC_{max}$=0.15. %which is a reasonable wake-up frequency for opportunistic routing schemes~\cite{Landsiedel2012}. 
%We observe that, with Staffetta's activity gradient, 
For the case without Staffetta (white bars), the wake up frequency is constant and so is the forwarding delay, as captured by eq.~(\ref{eq:staffetta_expected_rendezvous}). With Staffetta (orange bars), the short rendezvous of the sink's neighbors leads to a high wake up frequency, as shown in eq. (\ref{eq:staffetta_wakeup_frequency}). The neighbors that are 2-hops away observe a higher forwarding delay than the one observed by the sink's neighbors and adjust their wake up frequency accordingly. This process continues further down the chain and has two important consequences. First, the packet delivery delay is reduced significantly, because end-to-end latency is determined by the sum of the forwarding delays at all hops, as described in eq.~(\ref{eq:staffetta_fwd_delay}). Second, the energy consumption (duty cycle) is also reduced drastically, because the nodes with the heaviest loads spend less time and energy forwarding packets (the duty cycle can be roughly estimated as the forwarding delay times the forwarding load).

Another important trend is that the wake-up gradient decreases exponentially. We now derive a simple model to understand how the gradient's shape and steepness are controlled by the forwarding topology and the energy budget. 
For mathematical tractability, we assume that i) no collisions or message retransmissions occur, and ii) all nodes have the same number of forwarders $n$. 
Given that the forwarding delay perceived by nodes at $h$ hops from the sink is $\E{\delta_R(h)} + \delta_{tx}$, we can describe their activity gradient as
\begin{equation} 
	\omega_h \approx \frac{DC_{max}}{\E{\delta_R(h)}+\delta_{tx}} = \frac{DC_{max}}{1/((n+1)\omega_{h\text{-}1})+\delta_{tx}}.
\end{equation} 
Since usually in duty-cycled networks $\delta_R \gg \delta_{tx}$, we can further simplify the above expression by setting $\delta_{tx} = 0$ for all nodes with $h>1$. Thus, we obtain
\begin{equation} 
	\omega_{h} \approx DC_{max}(n+1)\omega_{h\text{-}1} = (DC_{max}(n+1))^{h\text{-}1} \omega_1.
	\label{eq:staffetta_gradient_model}
	\end{equation}

Even though the resulting model is a simple approximation it captures the characteristics of the wake-up gradient.
According to this model, the activity gradient attains the maximum frequency at the sink's neighbors ($\omega_1$) and decreases with geometric rate \mbox{$\Eb(n+1)$}. This geometric rate maps the trend observed in Figure~\ref{fig:staffetta_matlab}.
Overall, eq.~\ref{eq:staffetta_gradient_model} conveys two important points: 
i) smaller energy budgets $\Eb$ result in steeper activity gradients, and
ii) increasing the number of forwarders $n$ flattens the gradient.

Thus, in case of dense networks, it is possible to increase the number of forwarders and lower the energy budget $\Eb$ without affecting the resulting activity gradient. Lowering the energy budget $\Eb$ will extend the overall network lifetime, since in Staffetta the energy budget drives the maximum energy consumption of all nodes in the network (\cf duty cycle in Figure~\ref{fig:staffetta_matlab}).

\begin{figure}
	\centering
	\includegraphics[width=0.8\linewidth]{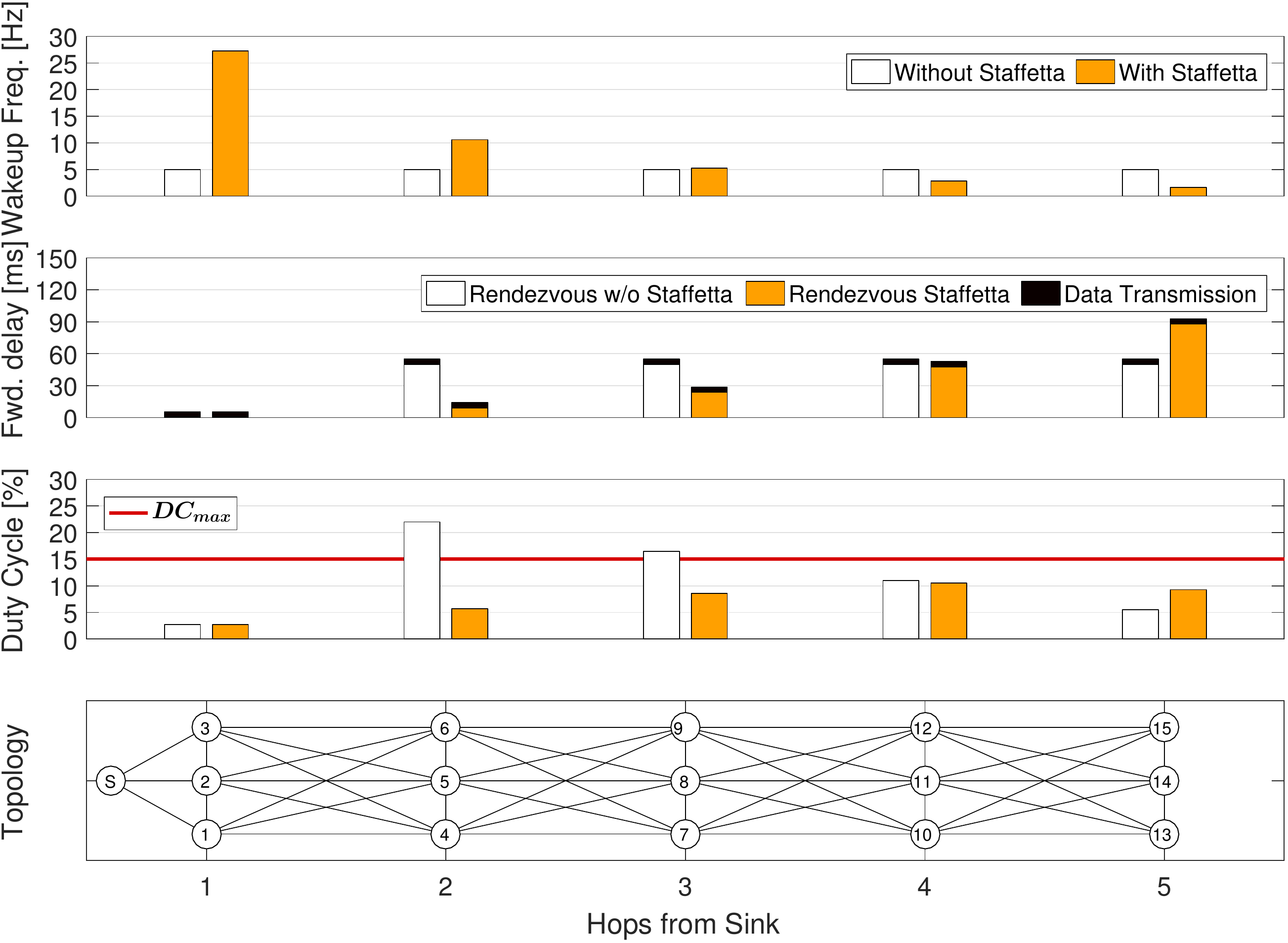}
	\caption{Staffetta's activity gradient in a simple topology. The closer nodes are to the sink (left side), the higher the wake-up frequency, and thus, the smaller the forwarding delay. Because the wake-up frequency is inversely proportional to the forwarding delay, the resulting energy consumption is equal for all the nodes in the network ($DC_{max}$).} 
	\label{fig:staffetta_matlab}
\end{figure}

\subsection{Understanding Staffetta's Performance Gains.}
Next, we analyze in more detail the benefits of Staffetta's activity gradient on packet latency, routing overhead, and network lifetime, the three key metrics of data collection applications~\cite{Gnawali2009}.

\paragraph{Packet Latency.}
The end-to-end packet latency depends on the \emph{forwarding delays} and the \emph{path length} between source and sink. In the previous subsection we noted that by drastically reducing the forwarding delay of nodes closer to the sink, the overall end-to-end latency is also reduced, \cf Equation~\eqref{eq:staffetta_fwd_delay}. Now we will discuss how Staffetta further reduces end-to-end latency by implicitly choosing short paths and exploiting temporal links.

While traditional collection protocols minimize latency by routing along a shortest-path routing tree, opportunistic protocols reduce the forwarding delay by selecting the first available node among a set of \emph{candidates} that improve the routing metric over a certain threshold.
Therefore, the wake-up frequency plays an important role in reducing packet latency of opportunistic mechanisms. First, it affects the forwarding delay of nodes; the higher the frequency, the lower the delay (\cf Equation~\eqref{eq:staffetta_expected_rendezvous}). 
Second, it biases the forwarding choices of nodes. In particular, given a node~$i$ with $n_i$~forwarders, each with wake-up frequency $\omega_1 \dots \omega_{n_i}$, the probability of selecting node $j$ as a forwarder is%\think{isn't confusing with the $w_h$?}
\begin{align}
	\Prob{ij} = \frac{\omega_j}{\omega_1 + \dots + \omega_{n_i}}.
 	\label{eq:staffetta_p_fwd}
\end{align}
When forwarders wake up with a fixed frequency, they have the same probability of being selected. 
Thus, as mentioned in Section~\ref{sec:staffetta_problem}, the size of the set of potential forwarders set must be chosen carefully. 
Too many forwarders and the performance will perform similar to a random walk (long paths, but short forwarding delays); 
too few forwarders and the mechanism will resemble deterministic routing (short paths, but long forwarding delays). 

With Staffetta, on the other hand, the number of forwarders becomes a secondary concern: forwarding choices are inherently biased towards the nodes that are closer to the sink \ie the ones that wake up more often (\cf equations~\eqref{eq:staffetta_gradient_model} and \eqref{eq:staffetta_p_fwd}). 
Therefore, \emph{Staffetta improves opportunistic protocols by biasing the forwarding choices towards the candidate nodes that provide the highest routing progress, effectively reducing the path length of messages}.

Even if a ``bad'' node (\ie a forwarder that offers little progress) is included in the forwarding set, the selection probability will be low as the wake-up frequency decreases with the distance to the sink. 
The effect is that Staffetta ``routes'' along paths whose length is comparable to mechanisms based on shortest-path routing trees, while maintaining the benefits of opportunistic protocols, namely, shorter forwarding delays, better load balancing, and exploitation of temporal links, that is, short-lived, long-distance connections~\cite{Srinivasan2010}.

The exploitation of temporal-links is particularly interesting, because in Staffetta temporal links that provide better progress towards the sink, will have a higher probability of being selected. 
Consider for example node 8 in Figure~\ref{fig:staffetta_matlab}, with three forwarders (nodes 4, 5 and 6) each waking up at frequency $\omega_4=\omega_5=\omega_6=11$\,Hz. According to Equation~\eqref{eq:staffetta_p_fwd}, the forwarding probabilities are $\Prob{84}=\Prob{85}=\Prob{86}=0.\bar{3}$. Now assume that, for a short period of time, node 8 is able to communicate with node 2, waking up at frequency $\omega_2=27$\,Hz. The forwarding probabilities are now heavily biased towards the temporally-available new node, with $\Prob{82}=0.45$, and $\Prob{84}=\Prob{85}=\Prob{86}= 0.18\bar{3}$. 
Note that, with a constant wake-up frequency (\ie no Staffetta) the forwarding probabilities will be $\Prob{82}=\Prob{84}=\Prob{85}=\Prob{86}=0.25$, almost halving the probability of exploiting the temporal link between nodes 8 and 2. Therefore,
\emph{Staffetta is able to drastically increase the probability of exploiting temporal links.} % when high routing progress is provided}.

\paragraph{Routing Overhead.} 
\label{sec:staffetta_routing-overhead}
% First, we explain how Staffeta does not introduce any overhead (except for measuring the rendezvous time)
In static networks, routing overhead is usually a minor problem. Since nodes must build the routing structure only once, they can spend a significant amount of resources (energy, bandwidth) and amortize the costs over time. Unfortunately, in practice it is difficult to have real static conditions. Even if nodes do not physically move, the topology continuously changes due to link-quality fluctuations and uncontrolled sources of interference~\cite{Ceriotti2009}.

When the topology changes more often, nodes must constantly update their routing structures. Amortizing the overhead over time becomes less and less possible, up to the point that the mechanism is continuously using part of the channel bandwidth to maintain the routing structure. This itself limits the data rate of collection mechanisms, increases the chances of packet collisions, and incurs high energy overhead.

Staffetta overcomes these limitations as follows. 
First, Staffetta introduces only a minimal overhead to the existing protocol stack. 
Staffetta's mechanism requires just local observations (of the forwarding delay). The overhead thus consists of a few bytes of memory for storing the latest observations and trivial computations to determine the wake-up frequency, whereas it utilizes \emph{no} additional bandwidth at all.
Therefore, \emph{empowering collection protocols with Staffetta comes with negligible overhead}.

Second, Staffetta allows data collection protocols to reduce their own overhead.
Since Staffetta's activity gradient biases the opportunistic selection of a forwarder already towards the sink, it is possible to exploit Staffetta in a cross-layer fashion. 
In particular, the continuous neighbor discovery, needed by collection protocols to maintain their routing structure, can be avoided and substituted with the biased opportunistic forwarders provided by Staffetta.
For example, to select the best forwarders in ORW~\cite{Landsiedel2012}, nodes must keep an up-to-date list of their neighbors' routing metric. 
While the current metric can often be piggybacked (e.g., on beacons and acknowledgments), link dynamics could force nodes to actively request their neighbors for updated metric information, increasing the routing overhead.
Using Staffetta, the best forwarders are often the ones that wake up first, so nodes can avoid maintaining their neighbor's routing metric and simply select the first node provided by the opportunistic mechanism, which is simpler and more efficient.

\paragraph{Network lifetime.}
Network lifetime, defined as the time until the first node in the
network depletes its battery, is a very important metric for real-word
deployments. This metric depends on the battery capacity of nodes and the average power consumption of the most energy-demanding
node, the one that will deplete its energy first. To this end, it
is more rewarding to reduce the maximum energy consumption among all nodes rather than the
average energy consumption in the network. This is precisely what Staffetta achieves by imposing a fixed energy budget.
However, note that although all nodes run with the same budget, they spend
it differently. Nodes next to the sink can forward messages without any rendezvous delay, and spend the saved energy on waking up more frequently; nodes at the edge of the network do not need to forward any data and can spend their complete budget on rendezvous with parent nodes that wake up at a low frequency.

\section{Implementation}
\label{sec:staffetta_implementation}

To evaluate the performance gains and generality of Staffetta on real devices, we implemented it within a custom MAC protocol based on SOFA atop the Contiki operating system as shown in Figure~\ref{fig:staffetta_stack}.
Taking inspiration from state-of-the-art opportunistic data collection protocols, we consider three different routing metrics: Expected Duty Cycle~(EDC)~\cite{Landsiedel2012}, Queue Backlog~(QB)~\cite{Moeller2010}, and SOFA's Random Walk~(RW)~\cite{Cattani2014a}.
Using one of these metrics, we parametrize a generic data collection mechanism, called \emph{baseline}, to create three different opportunistic collection protocols.
On the MAC layer, we implemented Staffetta's adaptive duty-cycle scheduling mechanism.
For comparison, we also use a non-adaptive mechanism with a fixed wake-up frequency.

The protocols, from which we take inspiration, are designed for different operating systems and platforms, with many implementation details that may hide the real benefits of their core mechanism, which is in our opinion the routing metric.
Using a common implementation of the opportunistic data collection mechanism (baseline) allows us to perform a fair, controlled comparison, where differences in performance are solely due to the routing metric and the Staffetta mechanism. 
\emph{Our goal is not to state that Staffetta is better than the state of the art, but rather that it makes the state of the art better.} That said, our implementations provide competitive performance, as it will be shown in the next section.

\begin{figure}
\begin{center}
	\includegraphics[width=0.6\linewidth]{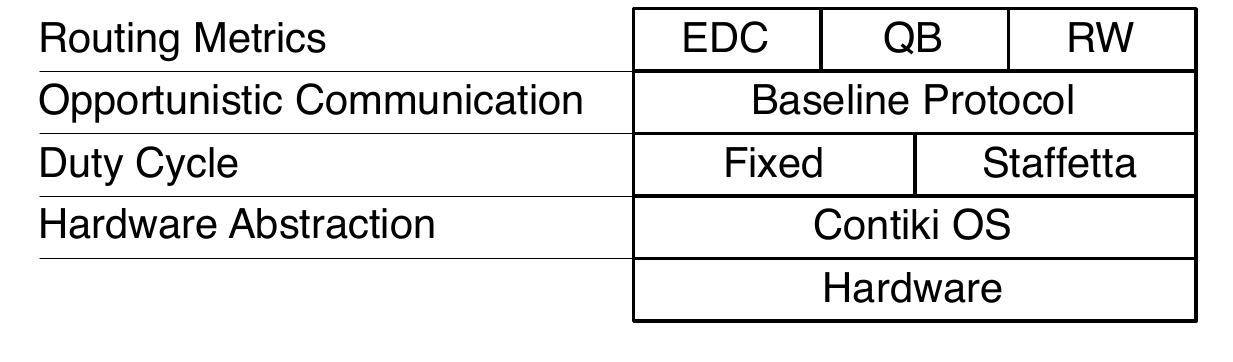}
	\caption{Protocol stack of our system. By Combining three routing metrics with two possible duty-cycling mechanisms we obtain six data collection protocols.}
		\label{fig:staffetta_stack}
	\end{center}
\end{figure}

In the following, we detail our baseline data collection mechanism, followed by a description of the three routing metrics and a few implementation details of Staffetta.

\subsection{Baseline} 
\label{sec:staffetta_baseline}

% Baseline mechanism
Inspired by opportunistic protocols such as SOFA and ORW, our baseline mechanism is based on low-power listening~\cite{Polastre2004} and opportunistic anycast, and works as follows (see Figure~\ref{fig:staffetta_mechanism}): Given a wake-up frequency as decided by the duty cycling mechanism, nodes wake up periodically to listen for incoming messages (\emph{listening}). If no message is received (channel is clear), nodes start probing the channel with beacons (B) until a viable forwarder wakes up (\emph{forwarding}). %in search of a destination. 
Each beacon contains the source's routing metric so that, whenever a node receives a beacon, it can compare the received metric with its own and send an acknowledgement (A) only in case of routing progress towards the sink.
In case of multiple colliding acks, which is a known problem of opportunistic mechanisms~\cite{Landsiedel2012}, nodes may be able to decode only the stronger ack due to power capture~\cite{Whitehouse}. If capture does not occur, 
all receivers would forward the data packet, causing duplicates. To ameliorate this problem, ORW relies on overhearing and a probabilistic back-off ($P=0.5$) that receivers follow when retransmitting an ack. 

Even with these mechanisms, ORW suffers from duplicate packets, in particular in scenarios with high data rates and high densities. In the worst case the duplicates get duplicated themselves en route to the sink, leading to an exponential growth in the number of duplicates and posing serious scaling limitations. Considering that Staffetta works by (smartly) increasing the activity of nodes in the network, the `duplicated packets' problem is expected to be exacerbated. To increase the resilience of our baseline protocol to packet duplicates, we implement a 3-way handshake similar to the one used in the SOFA mechanism. A third `selection' packet (S) is added during the rendezvous process. If the sender cannot decode a single ack, it does not send S. Potential receivers will then re-transmit their acks but with a back off probability of 0.5. Upon decoding an ack, the sender transmits S. As long as a selection packet is received, the duplication problem is completely avoided.  
If a selection packet is lost, the nodes that sent an ack can proceed in two ways: i) they discard the received beacon containing the data payload, or ii) they proceed forwarding the packet. 
In the first case, a data packet is lost; in the second case, a duplicate is created. We observed that the number of duplicates due to a lost S was not that high, thus, to maintain a high delivery ratio, our implementation follows the second approach. 
% Implicit link quality estimation (instead of packet overhearing in ORW)
Note that the 3-way handshake (beacon-ack-select) is also useful to implicitly filter out low-quality links without the need for a link-quality estimator.

\begin{figure}[t]
\begin{center}
	\includegraphics[width=0.6\linewidth]{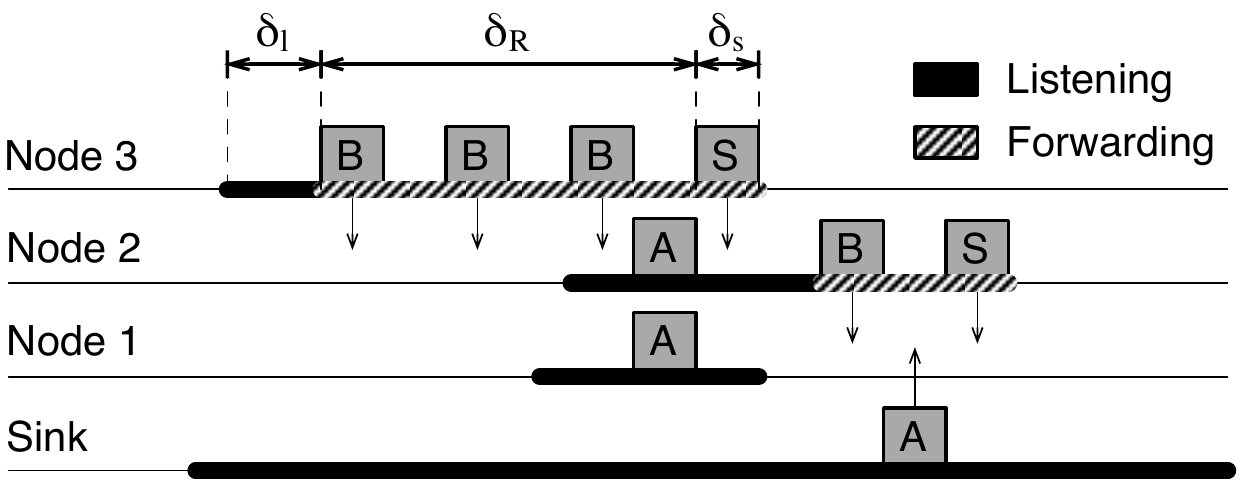}
	\caption{Baseline opportunistic mechanism with timing delays for listening ($\delta_l$), rendezvousing ($\delta_R$), and transfer of the select packet ($\delta_s$).}
	\label{fig:staffetta_mechanism}
	\end{center}
\end{figure}

\subsection{Routing Metrics}
\label{sec:staffetta_routing_metrics}
Routing metrics are at the core of data collection protocols and, together with the MAC protocol used, largely determine their performance. 
For this reason, we implemented three routing metrics from the state-of-the-art protocols on top of our baseline mechanism (see Figure~\ref{fig:staffetta_stack}). 

\paragraph{Expected Duty Cycle (EDC).} First proposed in ORW~\cite{Landsiedel2012}, EDC adapts ETX~\cite{DeCouto2005} to opportunistic, duty-cycled networks. In essence, EDC describes the expected duration until a packet reaches its intended destination (\ie the packet latency) by summing the expected forwarding delays at each hop. 
In ORW, the forwarding delay (\ie single-hop EDC) of a node \emph{i} is computed as the inverse of the sum of the link quality of the forwarders plus a weight $\delta_{tx}$ that reflects the transmission costs 
$$\emph{EDC}_{i} = (1/\sum{q_{ij}}) + \delta_{tx}.$$
$\emph{EDC}_{i}$ approximates the expected forwarding delay of node~$i$~\cite{Ghadimi2014} by (partially) using probes. In our implementation, we avoid the overhead of neighbor discovery and link quality estimation and compute this metric directly using the average of the last observed delays 
$$\emph{EDC}_{i} = \hat{\delta}_{F} = \hat{\delta}_R + \delta_{tx},$$
where $\hat{\delta}_R$ is computed by averaging the last 20 observed rendezvous times. 

\paragraph{Queue Backlog (QB).} QB is the metric used by the Backpressure Collection Protocol \cite{Moeller2010} (BCP). 
%based on a stochastic network optimization technique called Utility Optimal Lyapunov Networking. 
In BCP, nodes select the neighbor with the shortest queue backlog as the forwarder. Considering that the sink absorbs all packets, this simple mechanism leads to a gradient that decreases the size of nodes' queues towards the sink and can be used to make per-packet forwarding decisions. 
As demonstrated in the original paper, this agile metric shows superior delivery performance especially in dynamic networks experiences external interference or featuring a mobile sink. 
Different from the implementation in BCP, where the metric is mixed with ETX, our implementation solely uses the queue backlog for the opportunistic forwarding decisions. Low-quality links, which in the original implementation are discarded by ETX, are filtered out by the baseline's 3-way handshake.

\paragraph{SOFA's Random Walk (RW).} In SOFA, nodes choose their forwarders opportunistically, independently of their routing progress. Since no routing metric is used, its random walk mechanism is lightweight and performs great in highly dynamic networks, where the overhead of maintaining routes is too high. On the other hand, in static networks, random walk is too dispersive and suffers from long path lengths. 
Because in SOFA all nodes wake up with the same frequency, they have the same probability of being opportunistically selected (\cf Equation~\eqref{eq:staffetta_p_fwd}). With Staffetta underneath, this probability will not be uniform anymore, but biased towards the better links. 

%Koen: some white space here to make sure the reader notices that we are no longer in the RW paragraph. 
\vspace{2mm}Note that, for the first two routing metrics (EDC and QB), nodes forward their information to the first neighbor that wakes up \emph{and} provides routing progress. For SOFA's Random Walk, instead, nodes forward their message to the first neighbor that wakes up, independently of its distance to the sink.

\subsection{Staffetta}
\label{sec:staffetta_mechanism2}
Staffetta's mechanism consists of adapting the wake-up frequency of nodes based on the observed forwarding delay $\hat{\delta}_F$. For this reason, implementing Staffetta is straightforward. 

\paragraph{Measuring the forwarding delay.} We empirically measure the fixed communication delays $\delta_{tx}$. 
For the baseline mechanism, depicted in Figure~\ref{fig:staffetta_mechanism}, $\delta_{tx}$ includes the time needed to listen to the channel to check if it is clear for transmission ($\delta_l$), and the time needed to transmit the selection packet ($\delta_{s}$). 

To quantify the actual rendezvous time ($\hat{\delta}_R$), each node measures the period between the first beacon sent and the first acknowledgment received.
Due to the stochastic nature of the rendezvous delay, our implementation averages the last 20 measurement using a simple moving average (SMA) to produce a stable, yet reasonably agile activity gradient that is robust to small disturbances such as fast link fluctuations.

\paragraph{Computing the wake-up frequency.}
\label{sec:staffetta_implementation_approaches} 
Once $\delta_{tx}$ and $\hat{\delta}_R$ are determined, the wake up frequency is computed according to Equation~\eqref{eq:staffetta_wakeup_frequency}. In particular, after each successful message exchange, our implementation of Staffetta updates the average rendezvous time and the wake-up frequency $\omega$, and schedules the next wake-up in $1/\omega$ seconds. 

However, setting the wake-up frequency based on local information could lead to some complications.  
Different from traditional duty-cycling techniques, in which the activity rate
depends on the application's traffic requirements, Staffetta bases the wake-up
frequency solely on the forwarding delay, which depends on the network topology. %, that is, the size of the forwarder set and their wake-up frequency, as detailed in Section~\ref{sec:activity_gradient}. 
In effect, the wake-up frequency is a function of the distance to the sink and the (average) number of forwarders at each hop. 
This could be problematic in case of weakly connected nodes, located at the edge of a large network. 
For these nodes, communication could be so expensive (due to long rendezvous times) that maintaining the required data rate will make the node consume more energy than budgeted. 
There are two approaches to handling this issue.

\vspace{2mm}\noindent\emph{i) Give priority to the application
requirements}. This approach is based on the idea that Staffetta may
only set a wake-up frequency that is above the packet injection rate
required by the application. Capping the minimum frequency has the
consequence that Equation~\eqref{eq:staffetta_wakeup_frequency} will be violated,
causing the energy budget to be exceeded as more data will be sent than
can be afforded (given the long rendezvous time with the forwarders). 
This effect is exacerbated by a poor selection of forwarders, as the gradient
is effectively flattened (negligible bias), leading to longer
paths and extra traffic, hence, raising energy consumption even more.

\vspace{2mm}\noindent\emph{ii) Give priority to the energy
budget}. Using this approach, nodes can adapt their wake-up
frequency to a value that is lower than the application requirements. The
obvious consequence is that packets will be queued at the source,
leading to longer latencies and eventually packet drops, likely violating
latency and delivery requirements. 
However, 
%the gradient is properly maintained so packets will travel along good paths. Moreover, 
if the application is delay-tolerant, the problem can be mitigated by aggregating packets at the source as the energy costs are dictated by the rendezvous costs, not the actual data transfer. 
An even more important observation is that only the nodes at the edge of the network suffer from this; most nodes will run on a higher wake-up frequency serving application data at the right pace.

In our implementation we decided to test Staffetta under the most stringent conditions, that of violating our core principle: Equation~\eqref{eq:staffetta_wakeup_frequency}. Thus, we adopt the first approach and use a minimum frequency (0.1\,Hz by default), which can be set by the application. 

\section{Evaluation} 
\label{sec:staffetta_results}

In this section, we use measurements from two testbeds to evaluate the performance of Staffetta and the effectiveness of using the activity gradient to guide the forwarding decisions.

\subsection{Methodology}

\setlength{\tabcolsep}{12.0pt}
\begin{table}[!tb]
\begin{center}
\begin{tabular}{l|c|c}
\hline
Testbed & FlockLab & Indriya \\
\hline
Number of nodes & 28 & 139 \\
Tx power [dBm] & 0 & 0  \\
Network diameter [hops] & 5 & 5 \\
Node degree & 8.5 & 22 \\
Radio channel & 26 & 26 \\
Node ID of sink & 1 & 2 \\
\hline
\end{tabular}
\end{center}
\caption{Testbed characteristics and settings}
\label{tab:staffetta_testbed_characteristics}
\end{table}

\paragraph{Testbeds.}
We use the FlockLab~\cite{Lim2013} and Indriya~\cite{Doddavenkatappa2012} testbeds.
Out of the 28 TelosB on FlockLab, 25 nodes are deployed in several offices and hallways on one floor, while three nodes are deployed outside on the rooftop of an adjacent building.
The resulting network is quite sparse: each node has on average 8.5 neighbors within direct radio range, and the diameter is 5 hops.
The 139 TelosB on Indriya are spread across three floors in an office building.
With 22 neighbors on average, the network is much denser than FlockLab; the diameter is 5 hops.
In all experiments, nodes transmit at the highest power setting of 0$\,$dBm, using channel 26 to limit the impact of external interference from co-located Wi-Fi networks.
Table~\ref{tab:staffetta_testbed_characteristics} lists all testbed settings, including the node ID of the sink.
 
\paragraph{Compared schemes.}
We compare the following schemes:
\begin{itemize}
 \item \textbf{ORW:} This is the original TinyOS-based implementation of ORW~\cite{Landsiedel2012}, the current state-of-the-art opportunistic routing protocol for low-power wireless. ORW runs on top of the standard LPL layer in TinyOS with a wake-up frequency of 0.5$\,$Hz, which matches the configuration used by Landsiedel et al.~\cite{Landsiedel2012}.
 \item \textbf{EDC, QB, RW:} These are implementations of different opportunistic schemes based on the baseline mechanism described in Section~\ref{sec:staffetta_baseline}. The schemes are named after the specific routing metric they employ: Expected Duty Cyle~(EDC), Queue Backlog~(QB), and Random Walk~(RW).  As detailed in Section~\ref{sec:staffetta_routing_metrics}, EDC is the metric underlying ORW~\cite{Landsiedel2012}, and QB is similar to the metric used by BCP~\cite{Moeller2010}, which is the state-of-the-art routing protocol for mobile sink settings. All three schemes run on top of the default LPL layer and opportunistic anycast in Contiki with a wake-up frequency of 1$\,$Hz.
 \item \textbf{ST.EDC, ST.QB, ST.RW:} These schemes are the same as EDC, QB, and RW above except that instead of using a fixed wake-up frequency Staffetta's dynamic duty-cycle mechanism is employed to adjust the wake-up frequency at runtime, as described in Sections~\ref{sec:staffetta_mechanism} and~\ref{sec:staffetta_mechanism2}.
 \item \textbf{DIRECT:} This scheme, implemented atop the baseline mechanism, fully exploits Staffetta by using the activity gradient \emph{directly} as a ``routing metric:'' a node forwards a packet to the neighbor that wakes up first and has a wake-up frequency higher than itself. In other words, instead of using an explicit routing metric, such as EDC or QB, to decide whether a neighbor provides routing progress, DIRECT makes this decision using the gradient of wake-up frequencies inherent in Staffetta.
\end{itemize}

As discussed earlier, existing dynamic duty-cycle mechanisms are theoretical~\cite{Kim2011,Kim2010} or expressely designed for traditional non-opportunistic routing schemes~\cite{Meier2010,Zimmerling2012}, rendering an experimental comparison against Staffetta impossible.

\paragraph{Energy budget.} When Staffetta is enabled, we set the energy budget $DC_{\mathit{max}}$ to 7.5$\,$\% on FlockLab and to 6$\,$\% on Indryia.

\paragraph{Metrics.}
Our evaluation uses the following three key performance metrics of real-world applications~\cite{Ceriotti2009}:
\begin{itemize}
 \item \emph{Latency:} The time between when a packet is generated by the source node and when that packet is successfully received by the sink. To measure latency, we leverage the time synchronization among nodes on FlockLab and the serial logging of Indriya to timestamp both the generation and the successful reception of each packet.   
 \item \emph{Duty cycle:} The fraction of time the radio is turned on. We measure duty cycle in software using Contiki's accurate energy profiler.
 \item \emph{Delivery ratio:} The fraction of packets generated by the source nodes that is delivered to the sink. We determine delivery ratio based on the unique sequence numbers of packets successfully received at the sink.
\end{itemize}

\noindent
We also measure \emph{path length}, the number of times a message is relayed before reaching the sink, to explain the behavior and performance of a given routing scheme.
We compute these metrics based on three independent runs with the exact same settings.
To allow each protocol to bootstrap its operation, we start measuring after an initial delay of 1 minute.
We also give each protocol sufficient time (\ie 1--5 minutes) to empty their packet queues before the end of each run.
Unless stated otherwise, we report the results as median, 25th and 75th percentiles, and minimum and maximum values across all nodes in a testbed using box plots; statistical outliers are represented by crosses.

\subsection{Comparison against the State of the Art}
\begin{figure}[!tb]
	\centering
		\subfloat[Latency {[s]} - FlockLab]{ 		
		\label{fig:staffetta_latency_sota} 
		\includegraphics[width=0.7
		\textwidth]{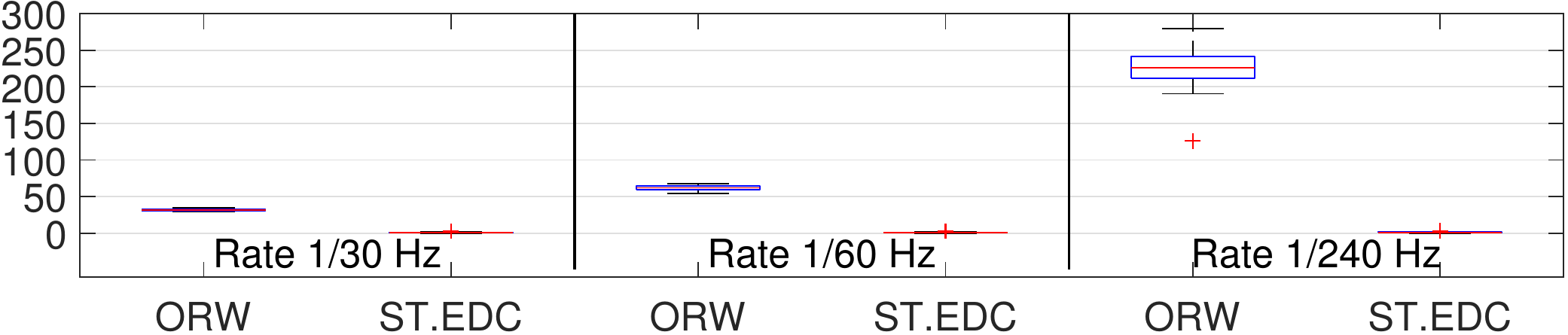}}\\
		\subfloat[Delivery Ratio - FlockLab]{ 		
		\label{fig:staffetta_deliveryratio_sota} 
		\includegraphics[width=0.7
		\textwidth]{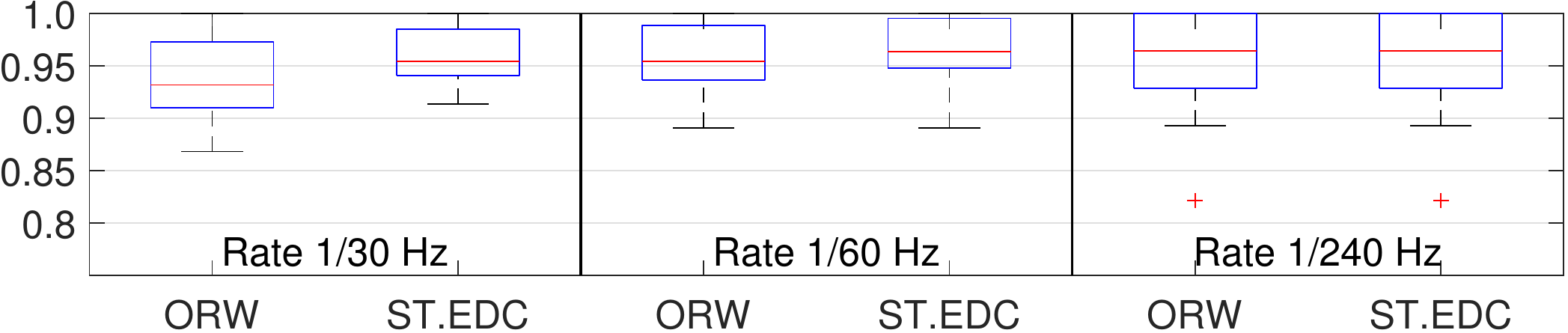}}\\
		\subfloat[Duty Cycle {[\%]} - FlockLab]{ 		
		\label{fig:staffetta_dutycycle_sota} 
		\includegraphics[width=0.7
		\textwidth]{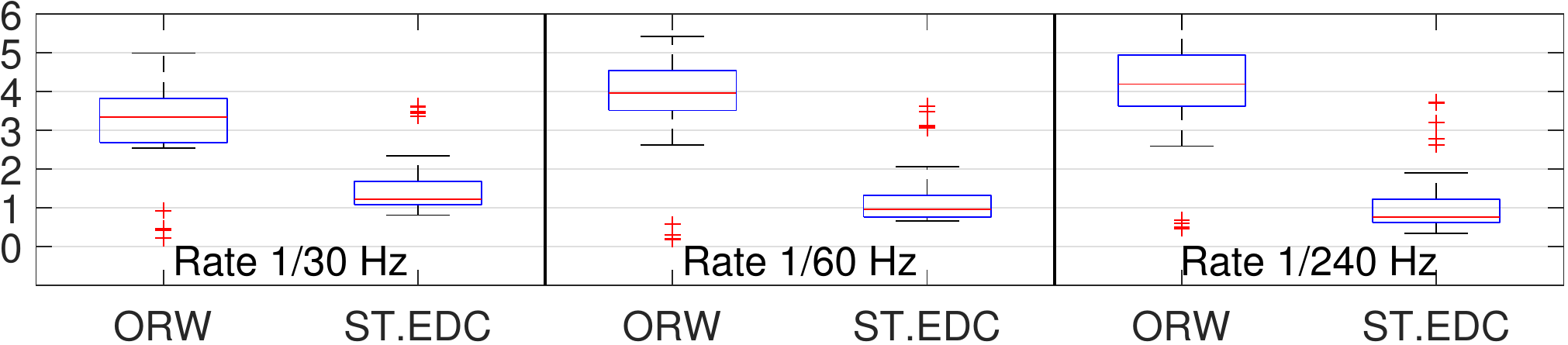}}
	\caption{Latency, delivery ratio, and duty cycle of ST.EDC and ORW for different data rates. Using Staffetta's dynamic duty-cycle adaption, ST.EDC outperforms ORW in all metrics, achieving several-fold improvements in latency and duty cycle.} 
	\label{fig:staffetta_flocklab_sota_comparison}
\end{figure}

We begin by comparing ST.EDC against ORW to quantify the performance gains of Staffetta over the state of the art.

\paragraph{Settings.}
We perform 1-hour experiments on FlockLab.
All nodes except the sink generate packets at the same rate.
We test three different data rates in different runs: 1 message every 30 seconds (1/30$\,$Hz), 1 message every minute (1/60$\,$Hz), and 1 message every 4 minutes (1/240$\,$Hz).
The latter equals the data rate used in the ORW paper~\cite{Landsiedel2012}.
Nodes are given 5 minutes to empty their message queues at the end of a run.

\paragraph{Results.}
Figure~\ref{fig:staffetta_flocklab_sota_comparison} shows latency, delivery ratio, and duty cycle of ST.EDC and ORW for different data rates.
We observe that ST.EDC outperforms ORW across all metrics regardless of the data rate.
Using Staffetta to adapt the nodes' wake-up frequencies, ST.EDC delivers packets on average 79--452$\times$ faster than ORW, while achieving a 2.75--9$\times$ lower duty cycle.
ST.EDC also provides a higher delivery ratio, especially as the traffic load increases. 
Moreover, ST.EDC reduces the variance in all metrics across nodes compared with ORW.

The results demonstrate the significant performance gains enabled by Staffetta's dynamic duty-cycle adaption.
The basic mechanism is conceptually simple and lightweight to implement, but also highly effective in practice. 
The next section investigates the reasons for these improvements.

\subsection{Benefits across Diverse Routing Metrics}
\begin{figure}[!tb]
	\begin{center}
		\subfloat[Latency {[s]} - FlockLab]{ 
		\label{fig:staffetta_latency_flocklab} 
		\includegraphics[width=0.45 
		\textwidth]{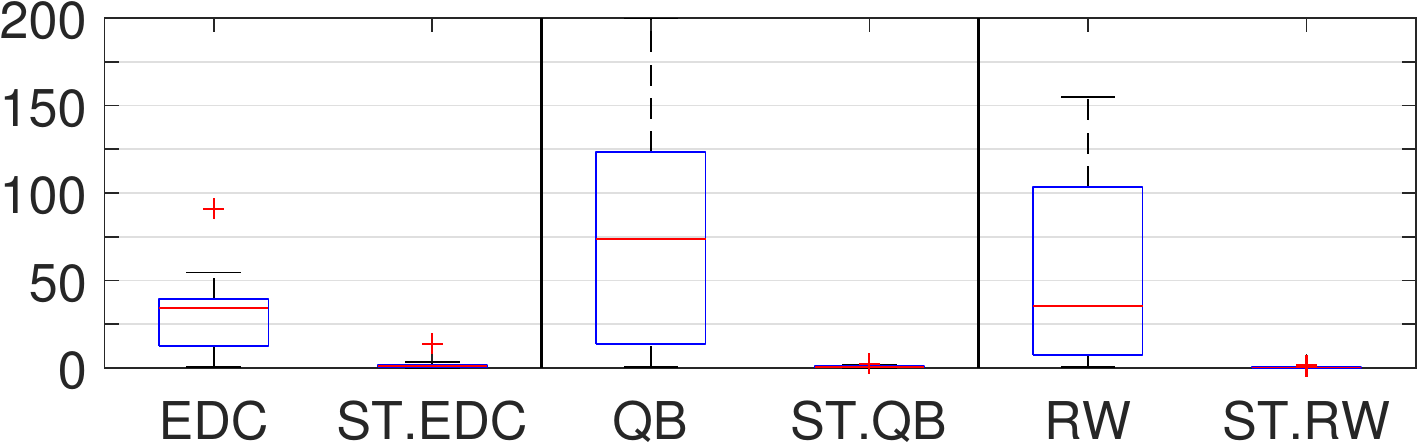} }
		\subfloat[Latency {[s]} - Indriya]{ 
		\label{fig:staffetta_latency_indriya} 
		\includegraphics[width=0.45 
		\textwidth]{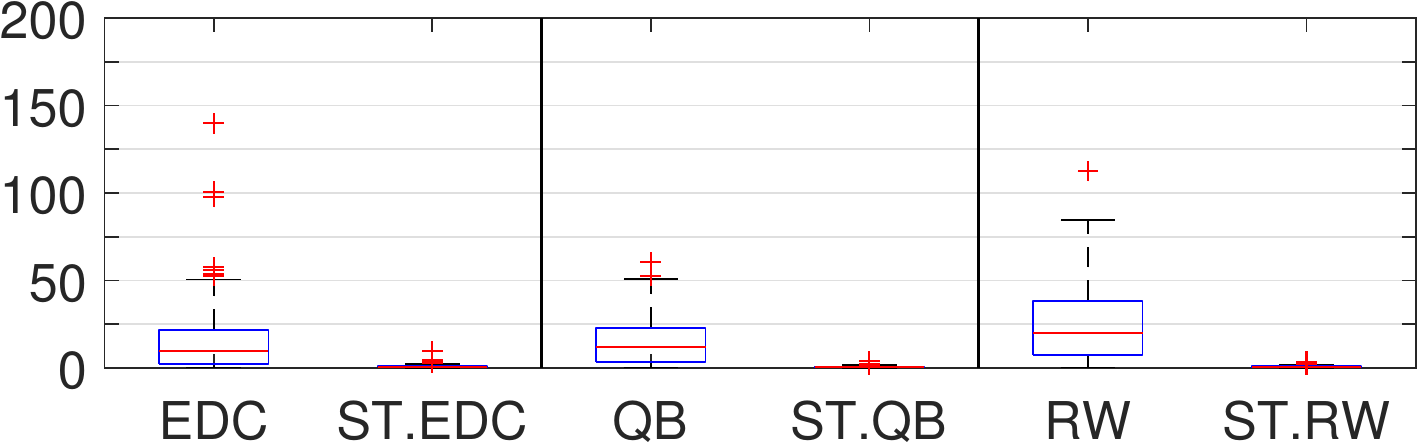} }\\
		\subfloat[Delivery Ratio - FlockLab]{ 		\label{fig:staffetta_deliveryratio_flocklab} 
		\includegraphics[width=0.45 
		\textwidth]{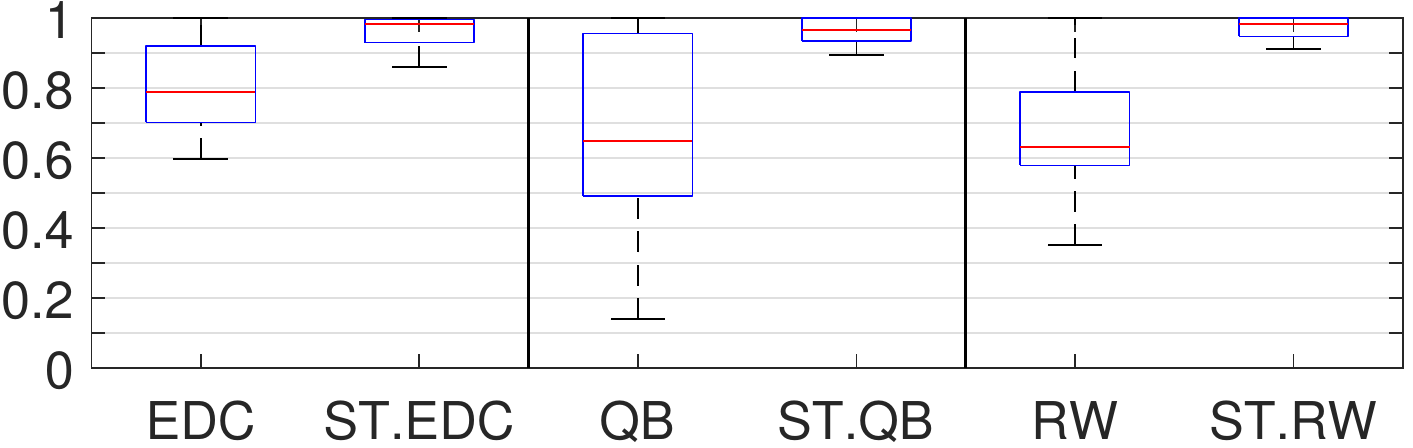} } 		
		\subfloat[Delivery Ratio - Indriya]{ 		\label{fig:staffetta_deliveryratio_indriya} 
		\includegraphics[width=0.45 
		\textwidth]{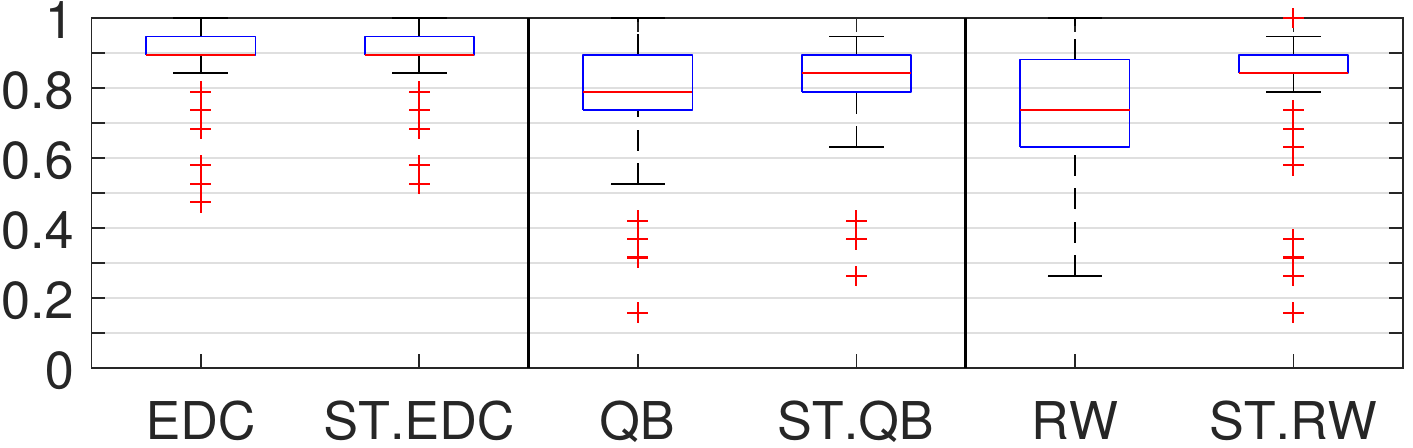} }\\
		\subfloat[Path Length - FlockLab]{ 		\label{fig:staffetta_pathlength_flocklab} 
		\includegraphics[width=0.45 
		\textwidth]{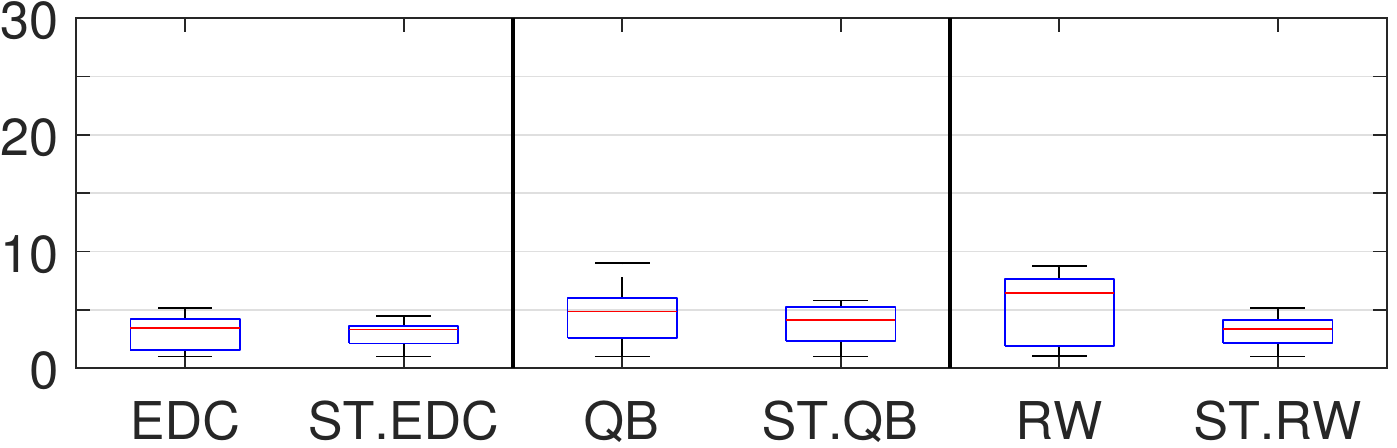}}
		\subfloat[Path Length - Indriya]{ 
		\label{fig:staffetta_pathlength_indriya} 
		\includegraphics[width=0.45 
		\textwidth]{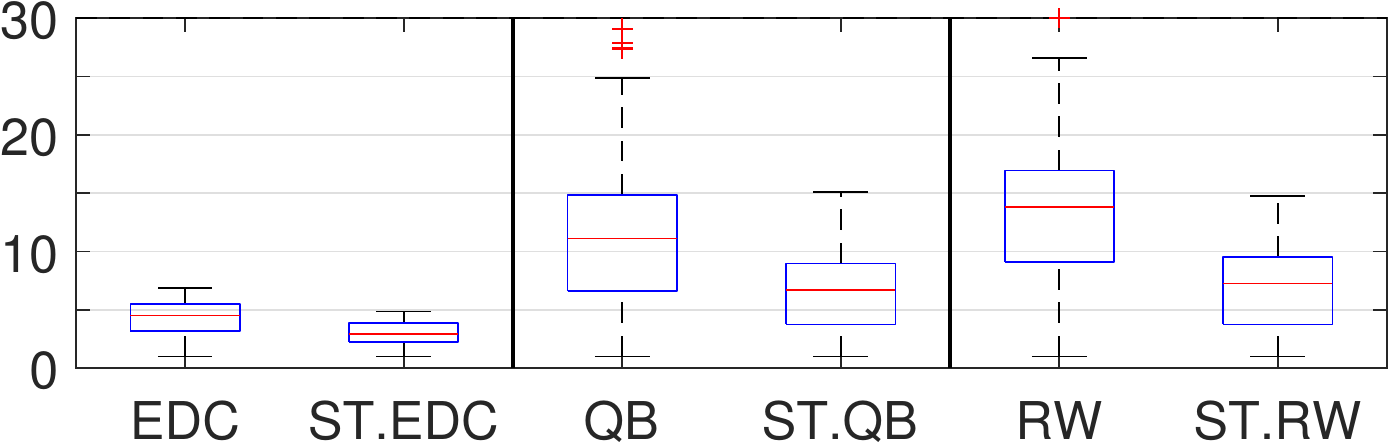}}\\
		\subfloat[Duty Cycle {[\%]} - FlockLab]{ 		\label{fig:staffetta_dutycycle_flocklab} 
		\includegraphics[width=0.45 
		\textwidth]{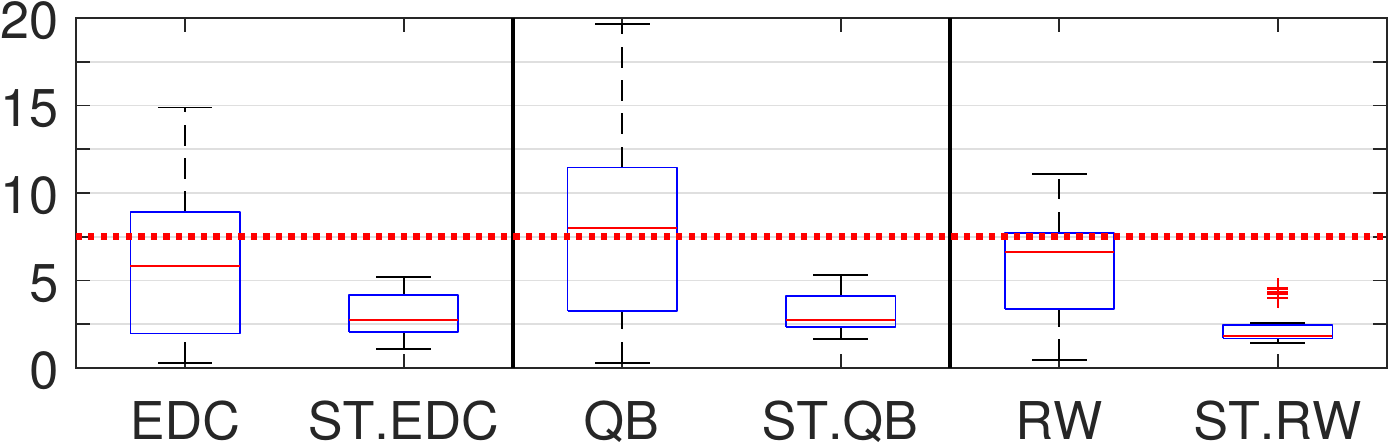} }	
		\subfloat[Duty Cycle {[\%]} - Indriya]{ 		\label{fig:staffetta_dutycycle_indriya} 
		\includegraphics[width=0.45
		\textwidth]{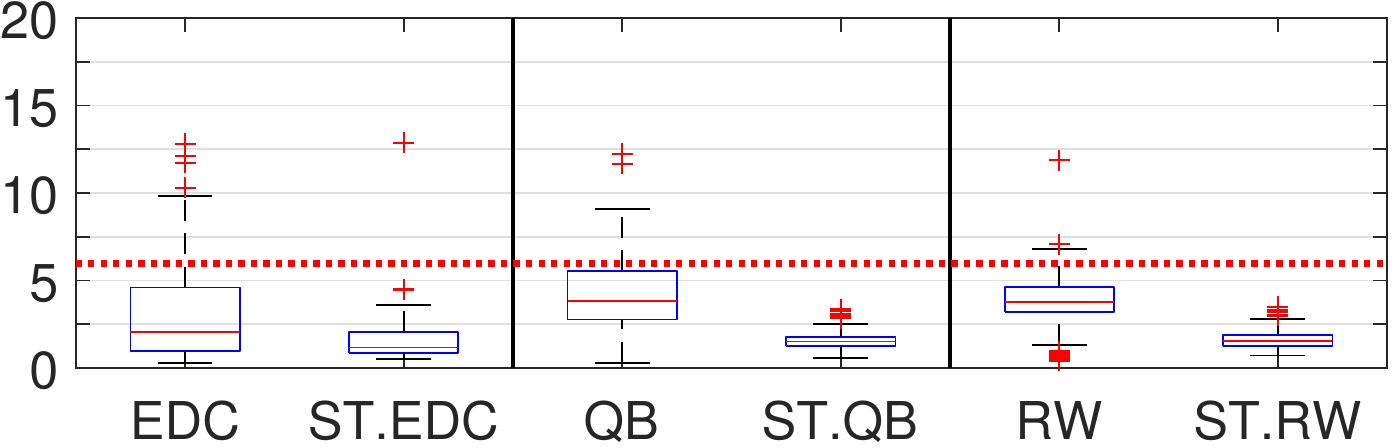}}
	\end{center}
	\caption{Performance metrics and path length with and without Staffetta for different routing metrics on the FlockLab and Indriya testbeds. Using Staffetta significanly boosts performance in all metrics compared with a fixed wake-up frequency.}
	\label{fig:staffetta_flocklab_indryia} 
\end{figure}

We now explore in detail the benefits of Staffetta in terms of performance for different well-known routing metrics.

\vspace{2em}\paragraph{Settings.}
To this end, we conduct a set of 10-minute experiments on FlockLab and Indryia with EDC and ST.EDC, QB and ST.QB, RW and ST.RW.
Except the sink, all nodes generate packets at the same fixed rate.
We test two data rates: 1 packet every 10 seconds (1/10$\,$Hz) on FlockLab and 1 packet every 30 seconds (1/30$\,$Hz) on Indriya.
Nodes have 1 minute to flush their packet queues at the end of each experiment.

\paragraph{Results.}
Figure~\ref{fig:staffetta_flocklab_indryia} depicts for each scheme the measured performance and path length.
We see that using Staffetta instead of a fixed wake-up frequency results in superior performance across the board, while reducing the variance among nodes.

\setlength{\tabcolsep}{5.8pt}
\begin{table}[!tb]
\begin{center}
\begin{tabular}{c|c|c|c}
 \hline
 Testbed & Latency & Delivery Radio & Duty Cycle \\
 \hline
 FlockLab  & 37.5--75.0$\times$ & 1.3--1.5$\times$ & 2.4--3.1$\times$ \\
 Indriya   & 12.2--20.1$\times$ & 1.0--1.2$\times$ & 1.6--2.2$\times$ \\
 \hline
\end{tabular}
\end{center}
\caption{Performance gains of Staffetta (cf.~Figure~\ref{fig:staffetta_flocklab_indryia})}
\label{tab:staffetta_relative_improvements}
\end{table}

Table~\ref{tab:staffetta_relative_improvements} lists the significant improvements in median performance enabled by Staffetta.
Looking at both testbeds and all routing metrics, Staffetta reduces latency by 12.2--75.0$\times$ and increases delivery ratio by 1.0--1.5$\times$, while reducing the energy costs, measured in terms of duty cycle, by 1.6--3.1$\times$.

The improvements are generally higher on Flocklab than on Indriya.
The reason is that opportunistic routing schemes naturally benefit from a higher node density: the larger number of potential forwarders on Indriya leads to, for example, a better load balancing and reduced forwarding delays.
Thus, Staffetta is particularely beneficial in networks that are fairly sparse in general or contain low-density node clusters through which the bulk of the traffic must be funneled.

\begin{figure}[!tb]
	\begin{center}
		\includegraphics[width=0.7\linewidth]{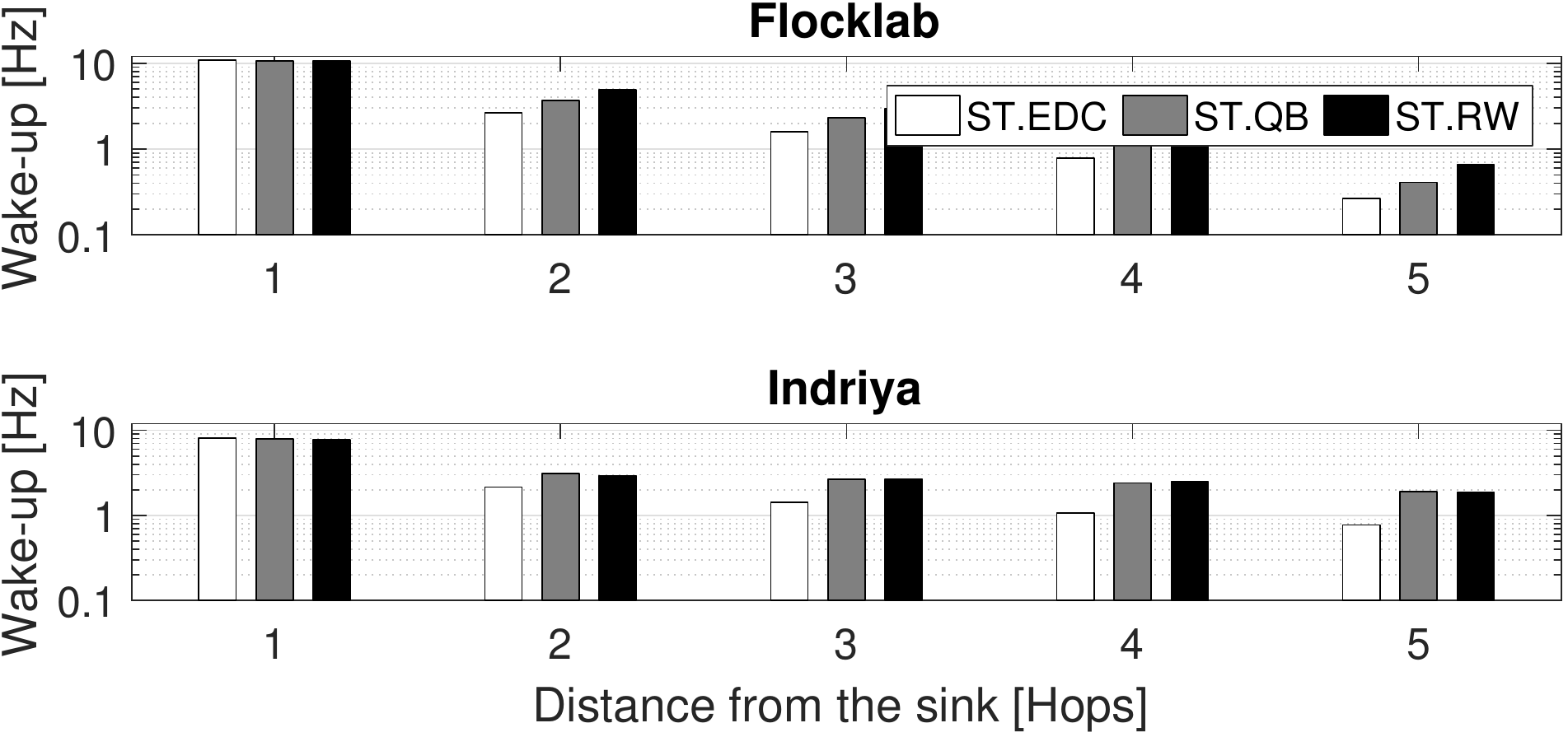}
	\end{center}
	\caption{Activity gradient on FlockLab and Indriya.}
	\label{fig:staffetta_gradient_testbed} 
\end{figure}

\emph{But how does Staffetta achieve these performance gains?}
The key lies in the activity gradient.
By letting each node dynamically adapt its wake-up frequency to the observed forwarding delay, nodes closer to the sink are more active than those at the fringe of the network.
This is visible in Figure~\ref{fig:staffetta_gradient_testbed}, which plots the average wake-up frequency of nodes against their hop distance from the sink on both testbeds.
We clearly see the geometrical decay for increasing hop distance.
This way, the bottleneck typically found in duty-cycled networks around the sink is largely eliminated, allowing for faster and more reliable packet delivery.
At the same time, the nodes' wake-up frequencies are proportional to their respective traffic loads, which results in a more effective use of energy.
We now take a closer look at each performance metric.

\paragraph{Understanding lower latencies with Staffetta.}
Latency is mainly determined by three factors: forwarding delay, path length, and message backlog (\ie the accumulation of packets in a node's queue).
As discussed below, Staffetta reduces all of them, which explains the overall reduction in latency.

Looking at Figures~\ref{fig:staffetta_pathlength_flocklab} and \ref{fig:staffetta_pathlength_indriya}, we can observe the reduction in path length for the different routing metrics.
RW and QB tend to select diverse paths that are not necessarily directed towards the sink, whereas EDC chooses forwarders that definitely provide high routing progress.
Thus, Staffetta reduces the path length significantly for RW and QB, while EDC leaves little room for further shortening the paths.
 
\begin{figure}[!tb]
	\centering
		\subfloat[FlockLab]{ 		\label{fig:staffetta_wakeup_flocklab} 
		\includegraphics[width=0.45 
		\linewidth]{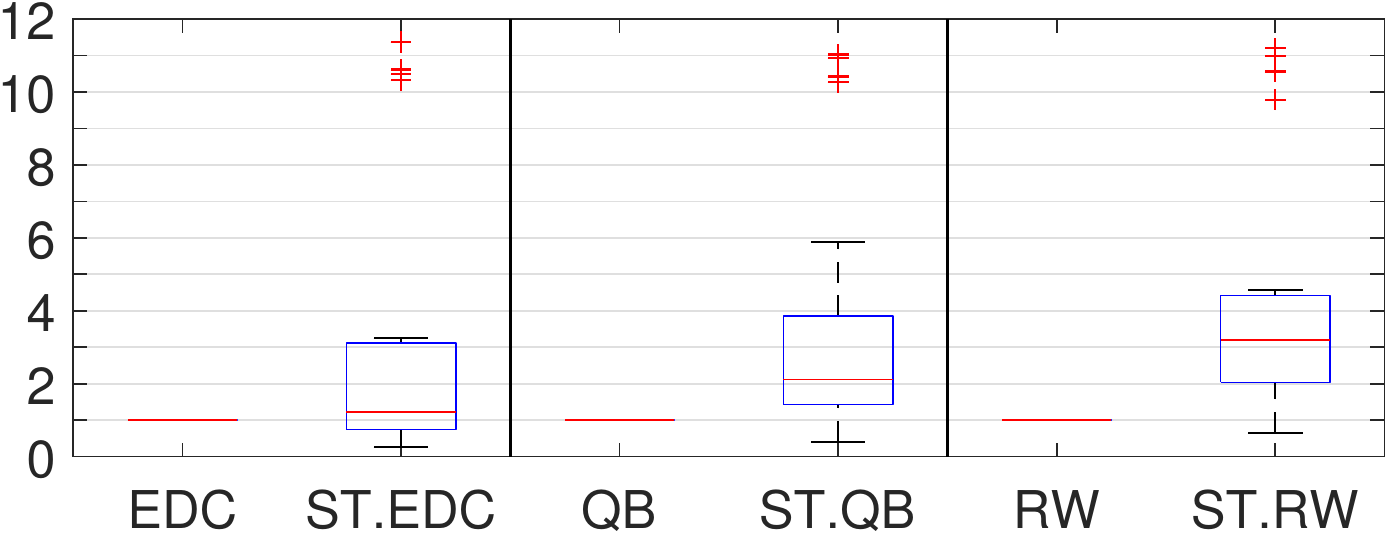} } 
		\subfloat[Indriya]{ 		\label{fig:staffetta_wakeup_indriya} 
		\includegraphics[width=0.45 
		\linewidth]{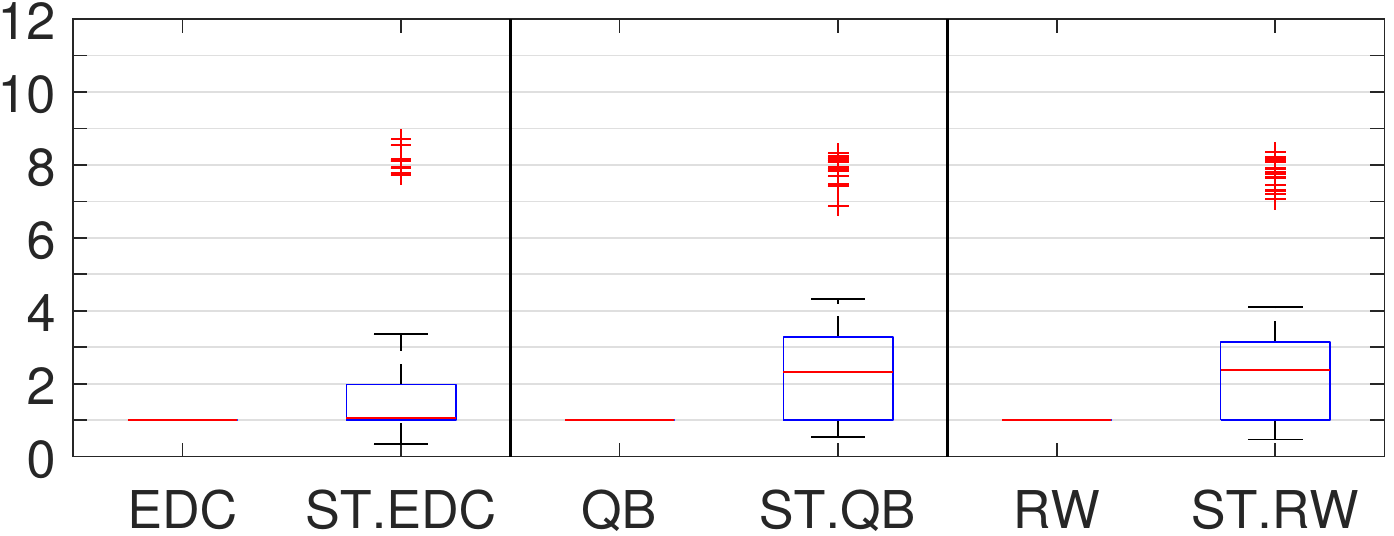} } 
	\caption{Wake-up frequency (in Hz) on the two testbeds.} \label{fig:staffetta_wakeup_testbed} 
\end{figure}

As stipulated by Equation~\eqref{eq:staffetta_expected_rendezvous}, the higher the wake-up frequency of forwarders, the lower the forwarding delay.
Figure~\ref{fig:staffetta_wakeup_testbed} shows the wake-up frequency of nodes on both testbeds for all routing metrics.
We see, for example, that the median wake-up frequency with Staffetta is consistently above the fixed wake-up frequency of 1$\,$Hz.
This indicates that Staffetta reduces the forwarding delay for more than half of the nodes on both testbeds.
A few nodes experience a higher forwarding delay due to a lower wake-up frequency, but this has a negligible impact on latency since these nodes are located at the fringe of the network and hence carry only very little traffic.

Finally, we find that message backlog is a major reason for the high latency when a fixed wake-up frequency is used.
This is because most packets are funneled through a small set of nodes actually delivering them to the sink.
These nodes are effectively a bottleneck whenever the fixed wake-up frequency (here 1$\,$Hz) is too low to sustain the aggregate load load around the sink (2.7$\,$pkts/s on FlockLab and 4.5$\,$pkts/s on Indriya).
As a result, packets are backlogged at these bottleneck nodes, severely increasing latency.
By contrast, with Staffetta nodes close to the sink are highly active, which increases the capacity and removes the bottleneck completely.

\paragraph{Understanding higher delivery ratios with Staffetta.}
Delivery ratio is affected by packets lost en route and packets stuck indefinitely in some node's queue.
We tackle the former by the beacon-ack-select handshake of the baseline protocol, which ensures that a link is only used when communication (temporarily) succeeds in both directions. 
This 3-way handshake not only reduces packet loss, it also serves to keep packet duplicates under control, which is a common problem in opportunistic protocols.
However, without Staffetta, duplicates still fill the nodes' queues, thereby delaying other packets.
This is also confirmed by comparing latency in Figures~\ref{fig:staffetta_latency_flocklab} and~\ref{fig:staffetta_latency_indriya} with delivery ratio in Figures~\ref{fig:staffetta_deliveryratio_flocklab} and~\ref{fig:staffetta_deliveryratio_indriya}: the longer the latency, the lower the delivery ratio.
Staffetta, instead, reduces the message backlog, which improves delivery ratio except for a few badly connected nodes.

\paragraph{Understanding lower duty cycles with Staffetta.}
In a nutshell, Staffetta increases the activity of nodes where needed, while it reduces the energy consumption wherever possible. 
Figures~\ref{fig:staffetta_dutycycle_flocklab} and~\ref{fig:staffetta_dutycycle_indriya} indicate that Staffetta reduces the duty cycle significantly, but also reduces the variance.
This shows that Staffetta sets an appropriate wake-up frequency for every node.
Ideally, all nodes would have the same duty cycle equal to the budget~$DC_{\mathit{max}}$.
Instead, Staffetta undershoots the budget.
This is because when a node wakes up and fails to receive a packet, it goes to sleep immediately.
However, the forwarding delay is \emph{not} updated, making Staffetta overestimate this parameter and hence select a lower wake-up frequency.
%Accounting for the undershoot would make Staffetta select higher wake-up frequencies, leading to channel saturation and wrong activity gradients.
%Thus, we did not do so.
This policy is beneficial as the maximum duty cycle in the network stays below the budget, which is important in scenarios where \emph{all} nodes need to be operational and lifetime is determined by the first node depleting its battery.
Staffetta significantly boosts the network lifetime in these scenarios.

\subsection{Adaptiveness to Mobile Sink Dynamics}
\label{sec:staffetta_mobile_sink}

Support for mobile sinks is a key requirement in participatory sensing, Internet of Things (IoT) and EWSNs applications.
We evaluate Staffetta's performance and ability to cope with the highly dynamic network topology in such scenarios.

\begin{figure}[!tb]
\begin{center}
	\includegraphics[width=0.65\linewidth]{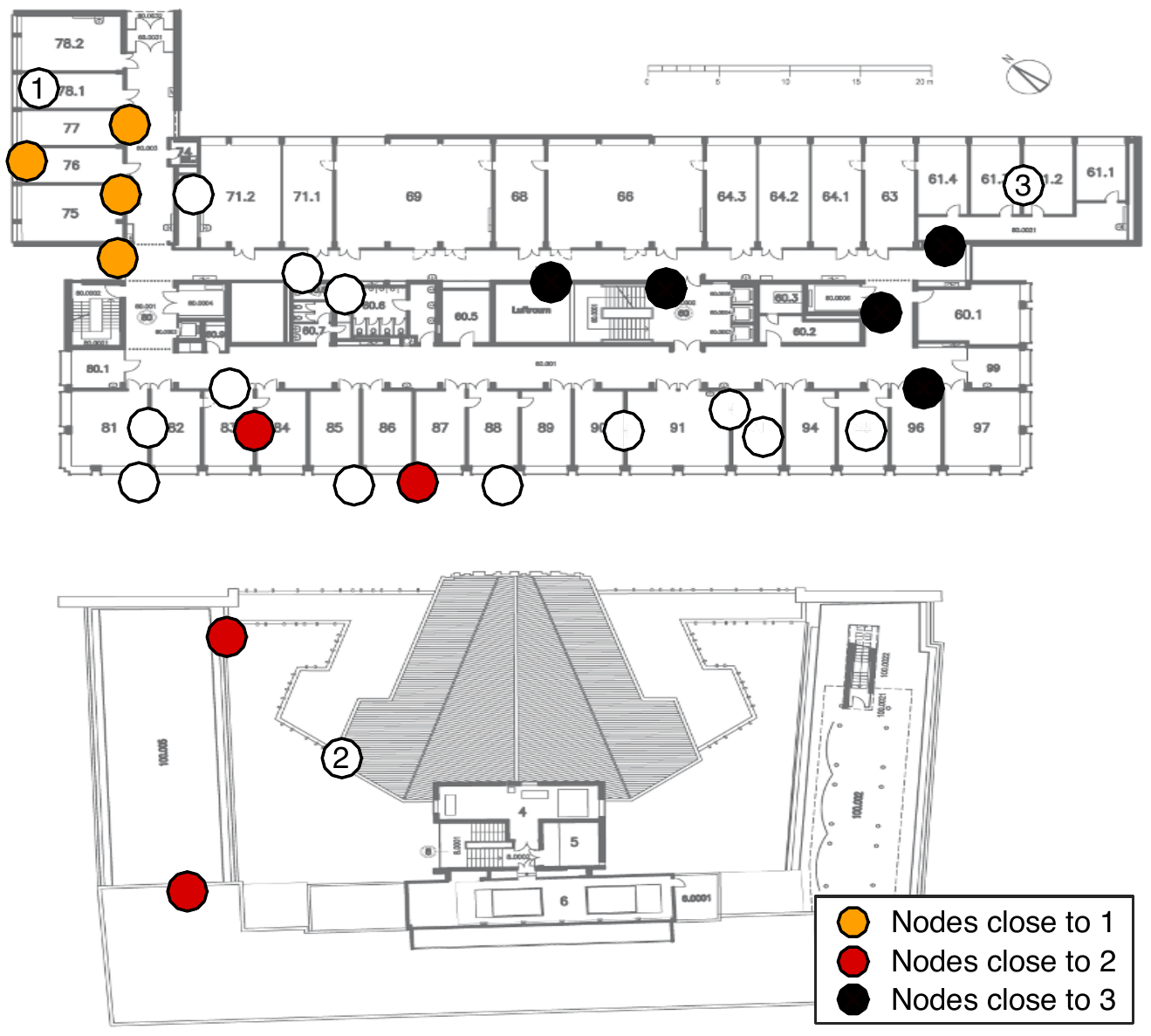}
	\caption{Nodes 1, 2, and 3 acting as a sink on FlockLab in the experiments of Section~\ref{sec:staffetta_mobile_sink}, and the corresponding clusters of nodes in their vicinity used to plot Figure~\ref{fig:staffetta_flocklab_movingsink}.}
	\label{fig:staffetta_flocklab_map}
	\end{center}
\end{figure}

\paragraph{Settings.}
To enable repeatable mobility patterns, we emulate a mobile sink that wanders about the FlockLab testbed area.
Specifically, we select a sequence of three nodes, labeled 1, 2, and 3 in Figure~\ref{fig:staffetta_flocklab_map}, that passes through several offices, hallways, and even several buildings.
One node at a time acts as the sink for 200 seconds.
Every change in the sink designation triggers a drastic change in the forwarding decisions.
We repeat the experiment multiple times, using different routing metrics both with and without Staffetta in different runs.

\paragraph{Results.}
Figure~\ref{fig:staffetta_flocklab_movingsink} demonstrates that Staffetta adapts the activity gradient promptly and correctly as the sink designation changes.
The figure charts the wake-up frequencies of nodes over time using ST.RW, grouped into three disjoint clusters: nodes in the vicinity of sink 1 (top), in the vicinity of sink 2 (middle), and in the vicinity of sink 3 (bottom).
The results for the other Staffetta-enabled schemes are very similar.
We see that Staffetta creates a distinct activity gradient with respect to the current sink designation.
Depending on a node's hop distance to the preceding sink, it takes 5--30 seconds to adapt its wake-up frequency after a change. 
Staffetta quickly forms a new gradient, driving messages toward the new sink.

\begin{figure}[!tb]
	\begin{center}
	\includegraphics[width=0.9\linewidth]{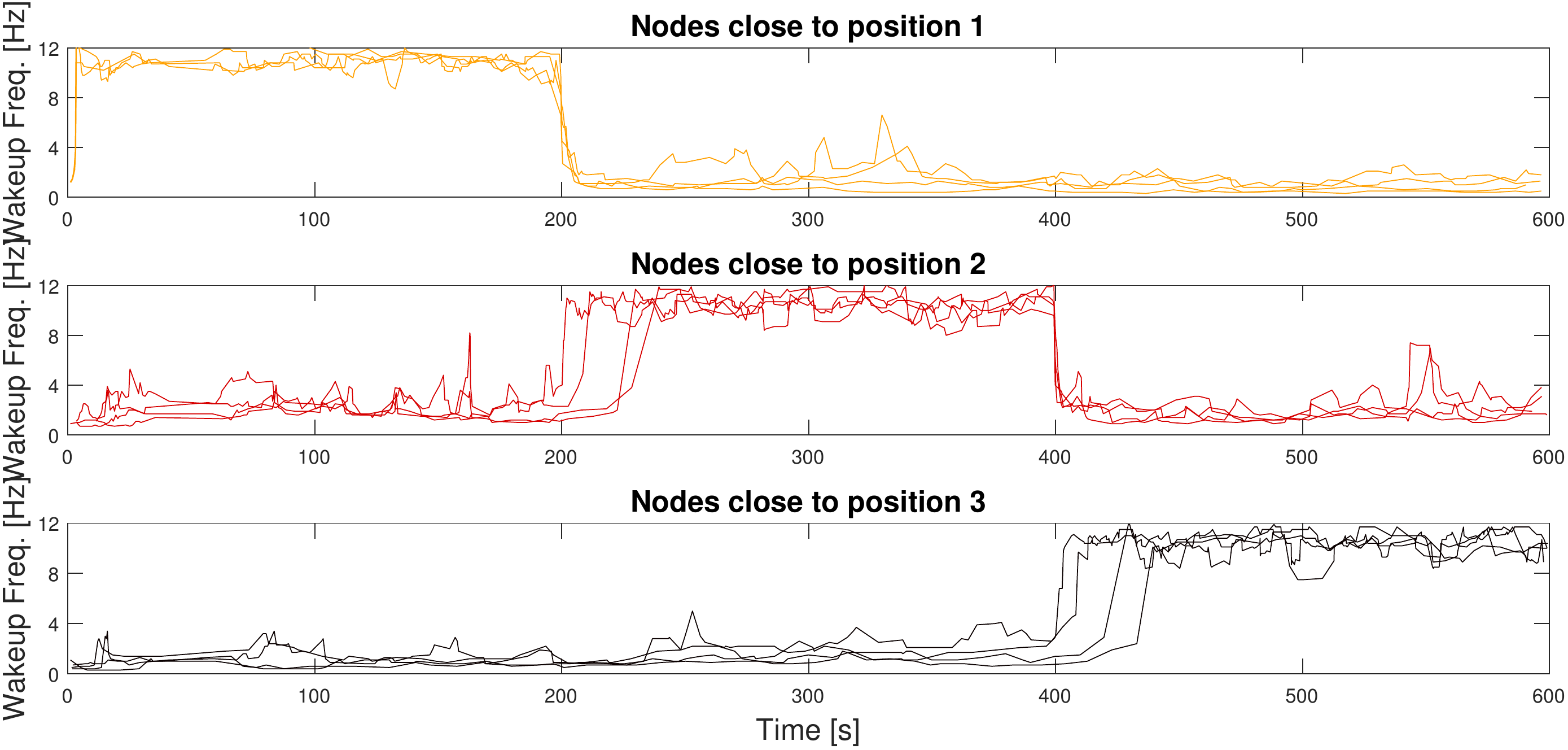}
	\end{center}
	\caption{Wake-up frequency of nodes running ST.RW with an emulated mobile sink on FlockLab (cf.~Figure~\ref{fig:staffetta_flocklab_map}). Staffetta quickly adapts the activity gradient, which helps deliver messages efficiently and reliably as the sink moves.} 
	\label{fig:staffetta_flocklab_movingsink}
	\vspace{-5mm}
\end{figure}

The net result is an improved resilience against network dynamics.
This is confirmed by looking at Figures~\ref{fig:staffetta_flocklab_movingsink_detail} and \ref{fig:staffetta_flocklab_movingsink_staffettaless_detail}, which display latency and throughput (packets received per second) at the current sink with~(ST.RW) and without~(RW) Staffetta as the sink designation changes from 1 to 2.
Despite a significant reduction in latency with Staffetta, we note that Staffetta provides a consistent and high throughput, whereas without Staffetta the throughput fluctuates widely, especially after the change in sink destination at time 200 seconds.
The spikes are due to nodes with a considerable backlog that suddenly become neighbors of the (new) sink and hence get the chance to empty their queues.  
Instead, Staffetta does not suffer from this problem, delivering packets reliably and consistently two orders of magnitude faster as the sink moves.

\subsection{Using Staffetta's Activity Gradient for Efficient Routing}

To forward messages efficiently and reliably, nodes typically maintain a list of their neighbors' routing costs with respect to some metric (\eg ETX~\cite{DeCouto2005}, EDC~\cite{Landsiedel2012}, or QB~\cite{Moeller2010}). 
Keeping this state up to date requires to execute tasks such as neighbor discovery and information sharing via periodic beaconing that cost energy, bandwidth, and time. 
We show in the following that Staffetta avoids this overhead by using its inherent activity gradient to effectively guide the forwarding decisions without maintaining explicit neighbor tables. 

\paragraph{Settings.}
We run another set of experiments with DIRECT, ST.EDC, ST.QB, and ST.RW on FlockLab.
However, unlike previous experiments, we use a bursty traffic pattern typical of applications that respond to external events.
To this end, we initially let each node accummulate 20 packets in its local queue before transmitting.
Afterwards, we measure until all nodes have transmitted their 20 packets, plus 1 minute to allow route-thru packets to leave the network.
This results in a throughput of 2.4--4.6 packets per second at the sink.

\vspace{1em}\paragraph{Results.}
Figure~\ref{fig:staffetta_direct_flocklab} shows for each scheme the measured performance and path length.
We see that DIRECT achieves a performance that is comparable or even superior to the other three schemes.
This occurs without the additional overhead with regards to implementation complexity and resource usage incurred in maintaining an explicit routing metric.
Using DIRECT, a node simply forwards to the neighbor that wakes up first and has a wake-up frequency higher than itself.
In this way, packets essentially ``surf on Staffetta's activity gradient,'' which is both efficient and highly effective. 

\begin{figure}[!tb]
	\begin{center}
	\includegraphics[width=0.9\linewidth]{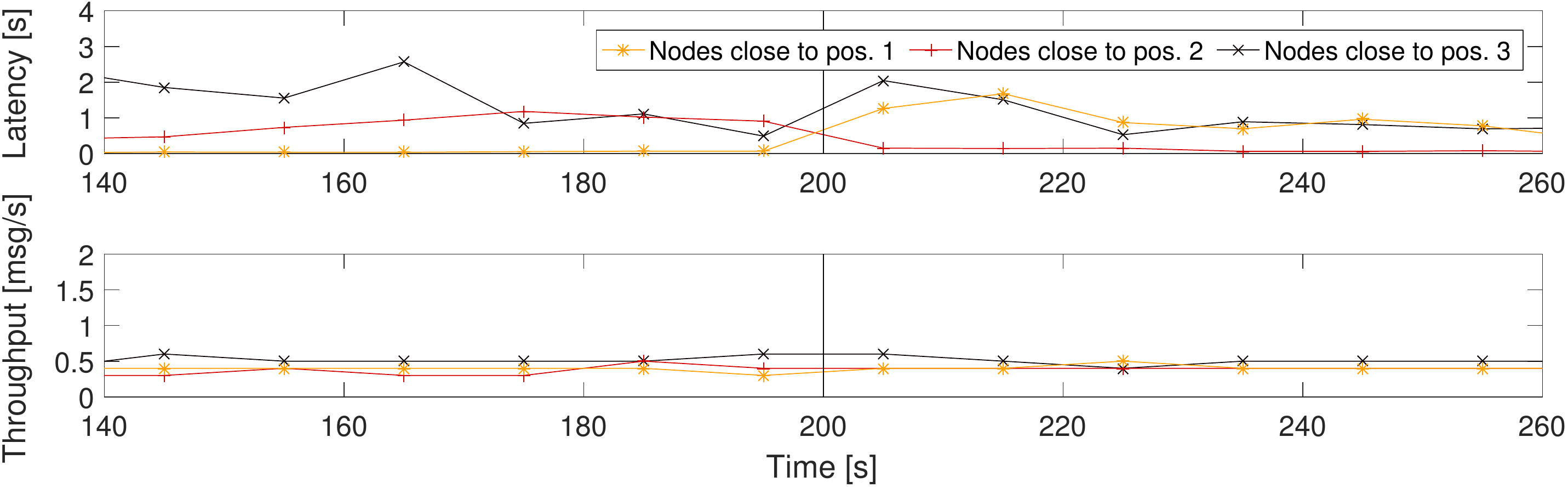}
	\end{center}
	\caption{Latency and throughput when Staffetta is enabled as the sink designation changes from 1 to 2 (cf.~Figure~\ref{fig:staffetta_flocklab_movingsink}). Staffetta reduces latency by two orders of magnitude, while providing a consistently, high throughput.} 
	\label{fig:staffetta_flocklab_movingsink_detail} 
	\vspace{-0mm}
\end{figure}

\begin{figure}[!tb]
	\begin{center}
	\includegraphics[width=0.9\linewidth]{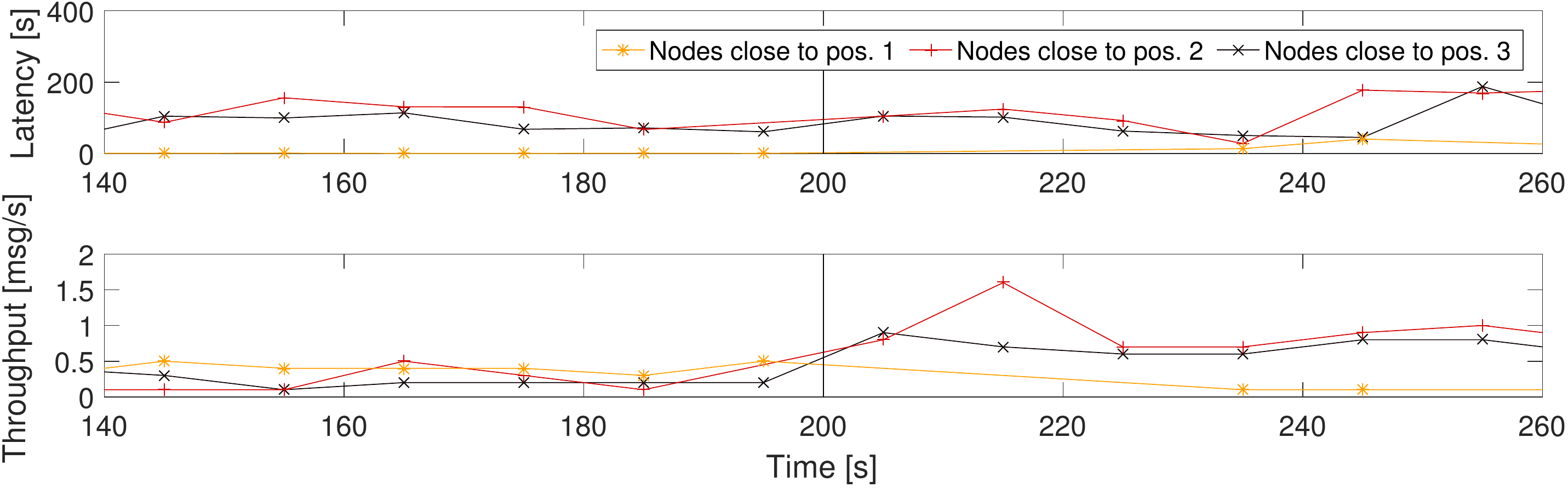}
	\end{center}
	\caption{Latency and throughput without Staffetta as the sink designation changes from 1 to 2 (cf.~Figure~\ref{fig:staffetta_flocklab_movingsink}).} 
	\label{fig:staffetta_flocklab_movingsink_staffettaless_detail} 
	\vspace{-3mm}
\end{figure}

\begin{figure}[!tb]
	\begin{center}
		\subfloat[Latency {[s]}]{ 
		\label{fig:staffetta_latencydirect_flocklab} 
		\includegraphics[width=0.45 
		\linewidth]{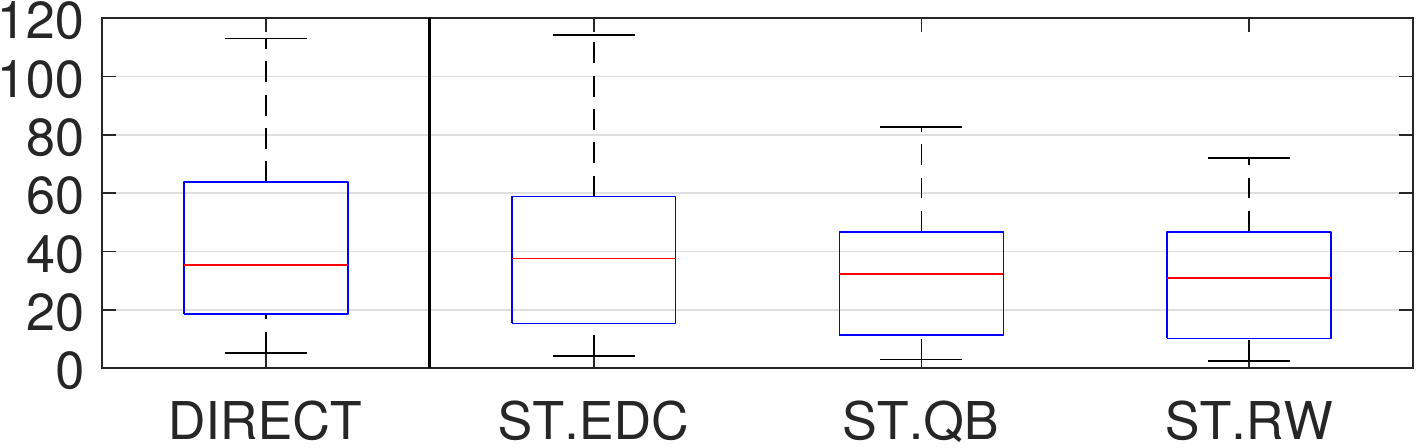} }
		\subfloat[Delivery Ratio]{ 		\label{fig:staffetta_deliveryratiodirect_flocklab} 
		\includegraphics[width=0.45 
		\linewidth]{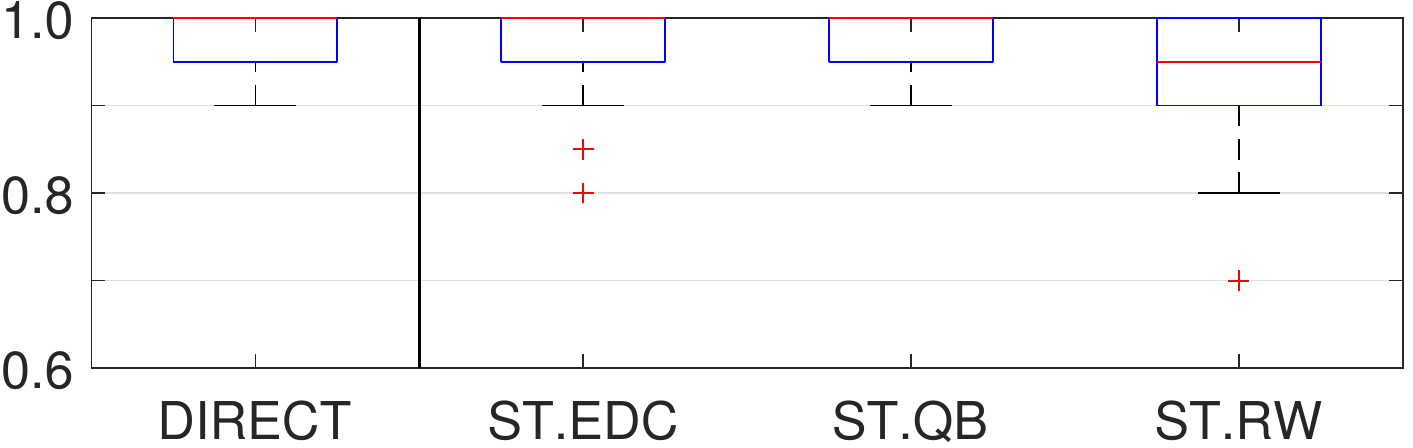} }\\	
		\subfloat[Path Length {[hop]}]{ 		\label{fig:staffetta_pathlengthdirect_flocklab} 
		\includegraphics[width=0.45
		\linewidth]{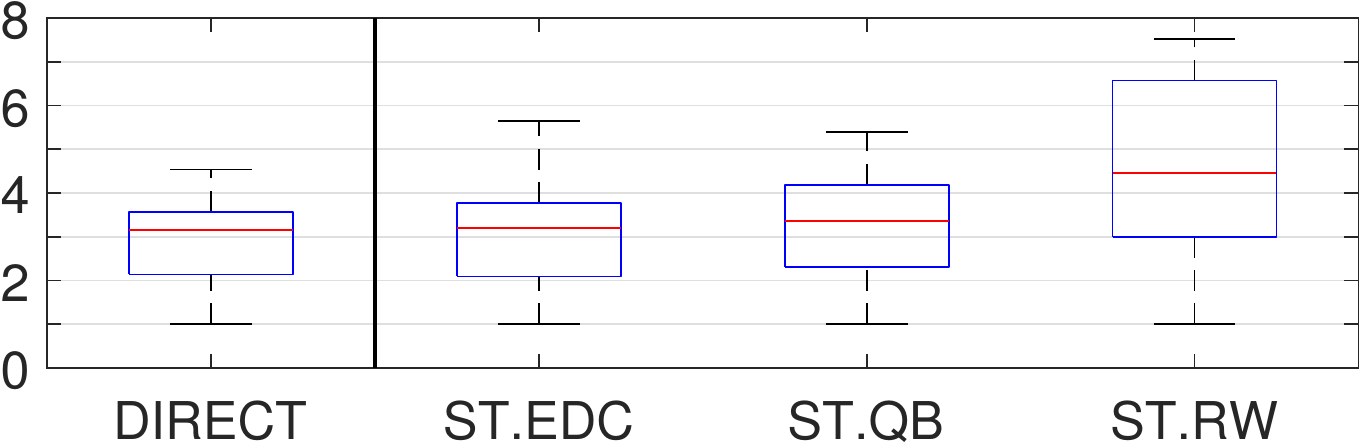}}
		\subfloat[Duty Cycle {[\%]}]{ 		
		\label{fig:staffetta_wakeupdirect_flocklab} 
		\includegraphics[width=0.45 
		\linewidth]{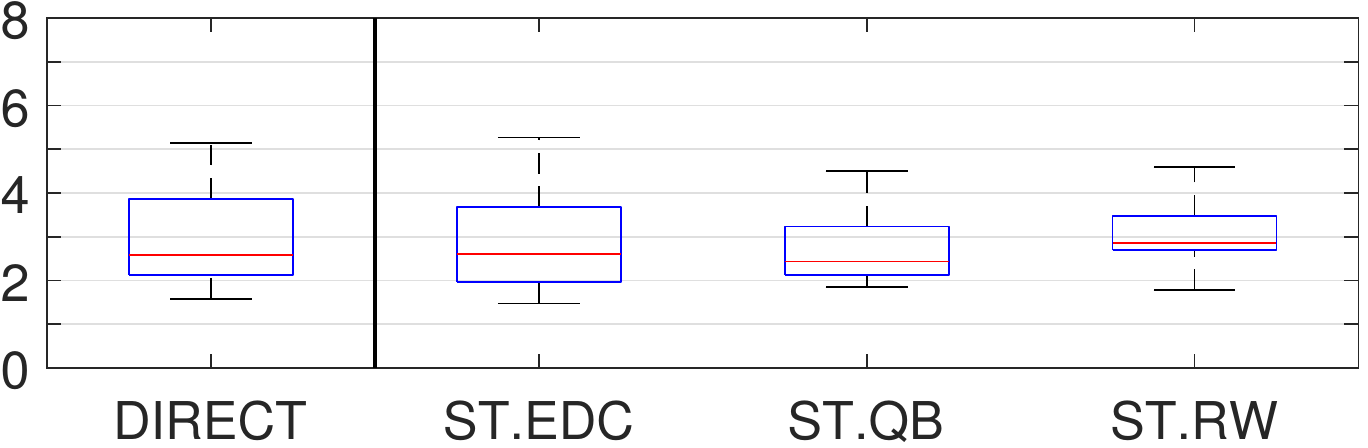} }
 	\end{center}
	\caption{Performance metrics and path length with Staffetta for DIRECT and different explicit routing metrics on the FlockLab. DIRECT uses Staffetta's activity gradient to effectively guide its forwarding decisions without the additional overhead in terms of implementation complexity and resource usage incurrent when maintaining explicit routing costs.} \label{fig:staffetta_direct_flocklab} 
\end{figure}

\section{Related work}
\label{sec:staffetta_sota}
% What is the problem

The amount of work related to energy-efficient communication protocols that
duty cycle the radio is vast.  In this section we focus on the two categories
that are most related to our work on Staffetta: protocols for data collection
over duty-cycled networks, and mechanisms that dynamically schedule the duty
cycle to improve performance.

\paragraph{Data collection protocols.}
The first generation of data collection protocols for sensor networks saw
duty-cycling as an issue that had to be overcome. The Collection Tree Protocol
(CTP)~\cite{Gnawali2009}, for example, was built on top of the fact that
broadcast is a primitive that should be used rarely, because it consumes a
lot of energy under duty-cycling techniques. Therefore, CTP uses unicast for
all data transmissions and broadcast only for route discovery~\cite{Levis2004}.
The Broadcast Free Collection protocol (BFC)~\cite{puccinelli2012} took this
strategy one step further, and avoids the use of broadcast altogether; routes
are discovered by eavesdropping unicast transmissions of neighboring nodes.
Although this approach saves significant amounts of energy, the price of
overhearing in duty-cycled networks is rather high leaving room for further
improvements.

Glossy~\cite{Ferrari2011} and its successors~\cite{Ferrari2012}
completely remove the need for route discovery by using an efficient
flooding mechanism based on constructive interference. Every message is
delivered to every node in the network, making data collection trivial.
In this case, duty cycling happens between synchronized active periods,
in which every nodes utilizes its radio. The downsides of these techniques
are the need of a central node to coordinate the network and the dependency on very tight time synchronization.

Finally, opportunistic protocols like ORPL, ORW, and
ORiNoCo~\cite{Duquennoy2013,Landsiedel2012,Untersch2012a} enhance the efficiency
of routing in duty-cycled networks by exploiting \emph{opportunistic anycast}
to reduce forwarding delays, balance load more evenly, and take advantage of
temporal links. As demonstrated in the previous section, Staffetta improves
the performance of these kind of protocols by biasing the forwarder selection
to nodes closer to the sink. 

\paragraph{Dynamic duty-cycling mechanisms.}
No matter how efficient data collection protocols strive to be, their
performance will always be bounded by the intermittent activity of the radio
underneath (usually with a fixed frequency). Therefore, several works adapt
the duty cycle to values that still meet the requirements of (traditional)
data collection.

Theoretical works~\cite{Kim2011,Kim2010}, for example, have analyzed optimal
duty-cycling parameters (like wake-up frequency) to minimize the packet
latency of opportunistic collection protocols in energy-constrained networks.
However, results are generally derived based on assumptions that are made
more for mathematical tractability than to resemble typical WSNs settings
(Poisson arrivals vs.\ strict periodic reporting) thereby ignoring the
intricacies and overheads of practical implementations.

From a more realistic viewpoint, the pTunes framework~\cite{Zimmerling2012}
models communication efficiency as an optimization problem with parameters
such as latency, energy consumption and reliability.  Through an efficient
network dissemination technique (based on Glossy), pTunes periodically collects
(centrally) network information and, based on the application requirements,
computes the optimal duty-cycling parameters and disseminates them to the
nodes. Staffetta complements this centralized approach by providing a good
but non-optimal solution that is fully distributed, has minimal overhead
and addresses the unique characteristics of opportunistic protocols.

In a similar way, ZeroCal~\cite{Meier2010} balances the energy consumption of
nodes by adapting their wake-up frequency. Compared to ZeroCal, we believe
that Staffetta has three main differences. First, ZeroCal requires radio
and data traffic parameters to solve a distributed optimization problem,
while Staffetta requires only one parameter (the energy budget) and a simple
operation to adapt the wake-up frequency locally.  Second, ZeroCal incurs a
higher overhead since it requires to monitor the state of nodes' children,
while Staffetta does not require any such explicit input. Finally, ZeroCal
focuses on the unicast primitive, while Staffetta addresses opportunistic
protocols and their anycast primitive.

\section{Conclusions}
\label{sec:staffetta_conclusions}

% What we do
In this chapter we addressed the fundamental issue of efficient data
collection in dynamic networks, where protocol designers must balance
flexibility with communication performance and energy efficiency.
% Why sota is not good enough?
While opportunistic collection protocols successfully trade off 
these two antithetical goals, they do not exploit the full potential of
the network, which is essential to scale to the demanding conditions of EWSNs (see \chapref{chapter:introduction}).
We have demonstrated that, instead of \emph{coping} with the
challenges of low-power communication, collection protocols can \emph{tame}
the duty-cycle mechanism and use it to their advantage (opportunistic principle), raising performance and
reducing (average) energy consumption at the same time.

% What is our solution? How does it work?
Our low-level Staffetta mechanism sets the wake-up frequency according to the forwarding latency observed by the nodes; the less time is spent on forwarding, the more time can be spent on servicing incoming traffic (i.e.\ waking up more frequently). The net effect is that Staffetta sets up a gradient with nodes close to the sink waking up more frequently than nodes at the edge of the network. This gradient automatically steers opportunistic anycast traffic towards the sink as the probability of finding an awake neighbor is highest in the direction of the sink.

% How good it works? can we give some general remarks?
We implemented Staffetta in Contiki, and evaluated it with three opportunistic collection protocols on two different testbeds. The extensive set of experiments showed that Staffetta significantly reduces both packet latency and energy consumption, while maintaining high -- or even improving -- packet delivery ratios. 
%A large part of this success can be attributed to solving the traffic bottleneck around the sink, which occurs when all nodes use the same duty cycle. 
As Staffetta does not need to maintain complicated routing metrics spanning the entire network (state-less principle), it can also handle network dynamics like link-quality fluctuations and node mobility really well. We found that Staffetta adapts its gradient in just a matter of seconds.

%% file: nemo/nemo.tex
%!TEX root = ../dissertation.tex
\chapter{Crowd Monitoring in The Wild}
\label{chapter:nemo}

% \blfootnote{Parts of this chapter have been published in IPSN \textbf{324}, 289 (1906) \cite{sofa}.}

\epigraph{``Infinite diversity... in infinite combinations.''}{Vulcan Master and Young Tuvok}

\dropcap{T}{he} achievements in the field of Wireless Sensor Networks have been steadily increasing in the last years, leaving less and less time to researchers for properly developing, testing and publishing their ideas before they get obsolete. 
For this reason, many protocols for WSNs (extreme or not) have only been evaluated under limited and controlled conditions. From high-level simulations to mildly-mobile testbeds. 
Even though simulations and testbeds are highly valuable tools for the initial deployment of wireless systems, many researchers stop their evaluation at these (controlled) conditions, often missing the variety of challenges that only real-word scenarios can exhibit. This is particularly true for crowd monitoring applications and EWSNs.
For this reason, the following chapter completes this thesis by testing some of the previously presented mechanisms under real world, extreme conditions. 

During the Christmas break of 2015, 
%In particular, 
we seized the opportunity to use a large-scale science museum in Amsterdam (\textit{NEMO}) to evaluate a bunch of crowd monitoring systems, one of which was based on SOFA and Estreme. These systems had to track the density of the crowd around Points of Interest (PoI) to infer their popularity. 
Due to the size and complexity of the museum, the curators had difficulties  quantifying the dynamics and interests of thousands of daily visitors on the many activities and exhibits offered by the museum.

%\section{Introduction}
% The goal. 
%In this chapter we report our experience with a system based on SOFA (\chapref{chapter:sofa}) and Estreme (\chapref{chapter:estreme}), and designed to monitor the popularity of different exhibits in a large science museum. Due to the size and complexity of the museum, the curators had a difficult time quantifying the interest of thousands of daily visitors on the many activities and exhibits offered by the museum.

% The setup: 
% Describe the museum focusing on unique/challenging features
\paragraph{The NEMO museum.\footnote{This research was part of Science Live, the innovative research programme of Science Center NEMO that enables scientists to carry out real, publishable, peer-reviewed research using NEMO visitors as volunteers.}} NEMO is the 5$^{th}$ most visited museum in the Netherlands, with half a million visitors per year. What makes this museum interesting as a case-study is not only its numbers, but rather three characteristics that are seldom found in indoor environments:
%\begin{compactitem}
% \item \emph

\subpar{A big open space}. Unlike other museums, or buildings in general, that consist of clearly differentiated areas (rooms), the layout of this museum could be seen as an open 3D space deployed over six stories with few boundaries (walls), \cf Figure~\ref{fig:nemo_nemo_3d}. With more than 5,000\,m$^2$ of exhibition space, the lack of boundaries makes it difficult to estimate the number of people participating, or interested, in each particular exhibit. 
% \item \emph

\subpar{Many random dynamics}. The challenges posed by the open space layout are aggravated by the large number of visitors and their ``random'' movement patterns. The museum is targeted mainly for kids and it does not have suggested routes. Visitors are free to roam around, and we had many: more than three thousand visitors per day. The average visit lasted approximately three hours. At peak hours the museum had more than two thousand visitors. 
% \item \emph

\subpar{Heterogeneous events}. The museum has different types of exhibitions. Some attract tens of visitors, other hundreds. Most of them are open experimentation areas where visitors spend variable amounts of time. Other events are scheduled at particular times of the day and have a rather constant duration. These scheduled events can take place in either open or closed areas. 
%\end{compactitem}

\paragraph{The monitoring system.} The monitoring system at our disposal was based on radio-frequency beacons and had three main components: bracelets given to visitors, a network of anchor points placed at various points in the museum and a network of sniffers to collect the data. We now describe briefly each one of these components.

%\begin{compactitem}
%\item \textit

\subpar{Bracelets.} Visitors were asked to wear a bracelet equipped with a 2.4\,GHz transceiver and a microcontroller. The bracelets sent periodic beacons and were powered by a super capacitor that could be fully charged in less than 2 minutes. The lifetime of the bracelets was approximately 8 hours with a duty cycle of 2\%.  The bracelets were charged and programmed just before check-in, when visitors were informed of the experiment and decided whether to participate or not (about 25\,\% did).

%\item \textit
\subpar{Anchor Points.} The curators asked us to monitor 30 Points of Interest (PoIs) in the museum. At each point of interest, we placed an anchor node. These anchor nodes had the same hardware as the bracelets, but were powered by a larger battery (since they were static). This extra source of energy enabled anchor nodes to send beacons at a higher frequency, which increased the probability of being discovered by nearby bracelets and speed up the density-estimation process.

%\item \textit
\subpar{Sniffers.} Upon a successful encounter between a bracelet and an anchor, a packet describing this exchange was sent on a dedicated channel to the system's backbone: the sniffers. Each sniffer consisted of bracelet hardware, to receive the packets, and a single-board computer that used either WiFi or Ethernet to commit the received information to a central database. The sniffers were placed uniformly in the museum in order to cover all areas.
%\end{compactitem}

%In terms of privacy, the only requirement was for the data to be anonymous. Bracelets could have IDs but the IDs could not be mapped to the identity of visitors. 
\fakeparagraph In terms of privacy, the only requirement was for the data to be anonymous.
Bracelets have IDs, but as we purposely refrained from registering visitor details, these IDs could not be used to identify specific people.

Considering the needs and requirements of the museum's curators and the monitoring system at hand, we proposed the use of energy-efficient neighbor discovery mechanisms. Having each anchor point monitor periodically the surrounding bracelets' ID provides sufficient information to track the popularity of an exhibit.

Being aware, however, of the growing concerns users have with ID tracking and the challenges of performing neighbor discovery in EWSNs (\cf \chapref{chapter:introduction}), the museum setup gave us a unique opportunity to also assess the performance of ID-less methods \eg Estreme in a highly unstructured scenario. 
Our initial hypothesis was simple: the community has developed lightweight methods to estimate node density without the need of IDs, given that the popularity of an exhibit is determined by the number of people in its surroundings, density estimation should be an accurate proxy for popularity. Our evaluation, however, shows that density estimators are ill-equipped to monitor popularity in scenarios with open spaces, high dynamics and heterogeneous types of events. Our key findings are three-fold:

\vspace{-.2mm}
%\item[1)] \emph
\subpar{1) Density estimators are not accurate in many cases.}
Neighborhood cardinality estimators (based on order statistics) lead to high
error rates. About 80\,\% of anchor points failed to detect the number of surrounding people within 20\,\% of the true value. This low accuracy differs from testbed results and is due to a mixture of having open, unstructured spaces and high crowd dynamics.

\vspace{-.2mm}
%\item[2)] \emph
\subpar{2) Density classifiers do not work in most cases.} Density classifiers, which infer the size of a crowd based on RSS signatures, do not work in our setup. We found that the underlying assumptions of these methods (more people imply lower mean RSS and higher RSS variance) are not universal. Many points of interest in the museum provide contradicting RSS trends. This occurs because prior studies assume uniform densities, while the museum has a highly heterogeneous distribution.  

\begin{figure}
	\centering 
	\includegraphics[width=0.6\linewidth]
	{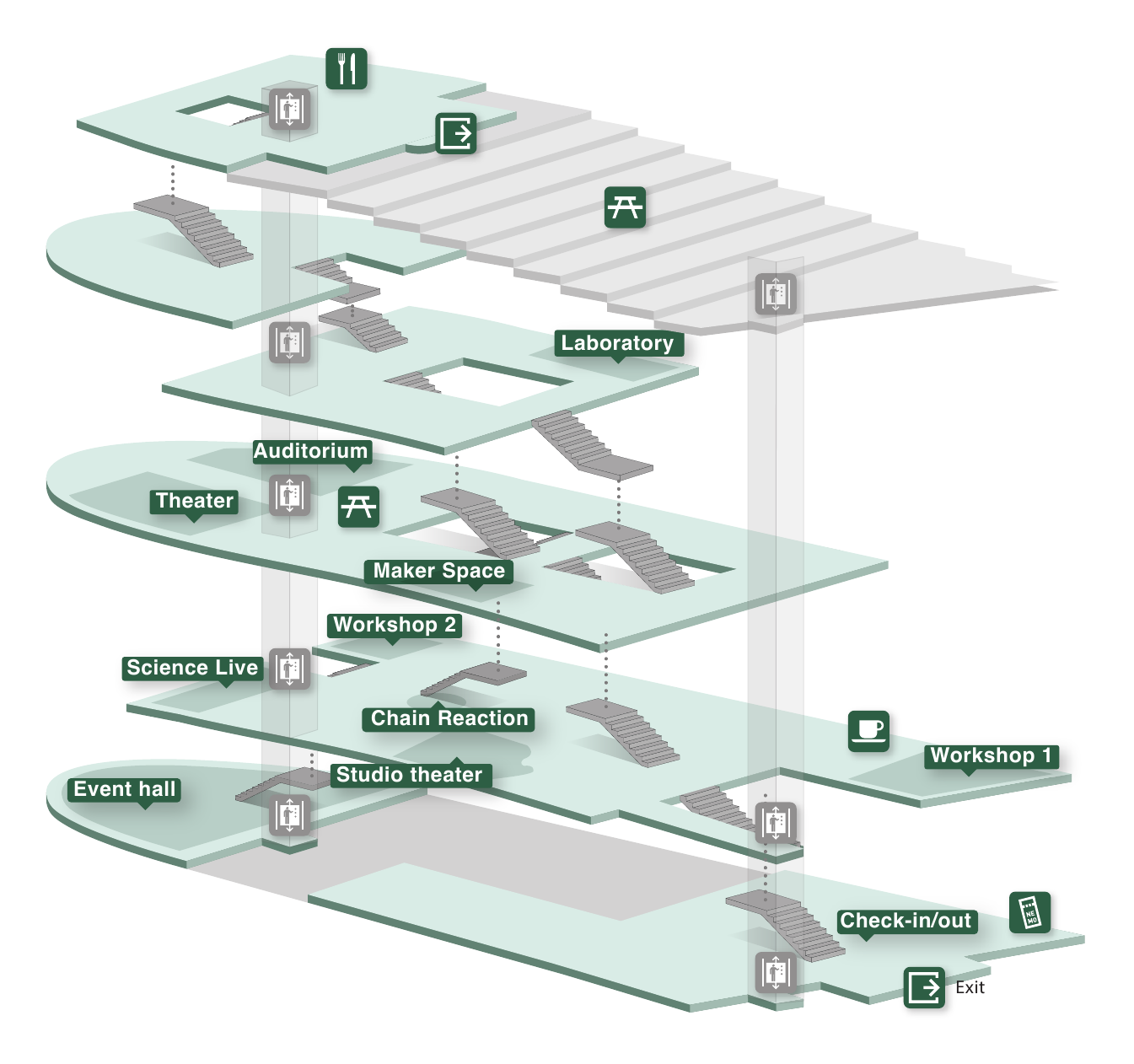} \caption{Floor plan of the science museum.}
	\label{fig:nemo_nemo_3d} 
\end{figure}

\vspace{-.2mm}
%\item[3)] \emph
\subpar{3) Even if density estimators would be accurate, they cannot monitor popularity in scenarios with open spaces having a mixture of mobile and static people.} Most museums have clear structures, thus, it is fair to assume that the people in front of a exhibit, say a painting, are looking at it. But our museum is an open space, and the people around an exhibit are a mixture of people passing by and people actually using or queuing for the exhibit. ID-less methods cannot differentiate the interested crowd (popularity) from the by-passers. More work must be done to understand and integrate the effect of crowd dynamics on density estimators.

\section{Related Work}

The analysis of human mobility is of significant importance to gain fundamental insights about people's activities, needs, and interests. 
% With the spread of technologies to track the mobility of individuals and city dynamics, recently more attention has been dedicated to opportunities, methods, and applications to manage trace data~\cite{6525604}. 
To position our work within this body of knowledge, we first divide existing efforts into two main categories: wide- and local-area scale, then we describe work pertaining the monitoring of museums.

%On the one hand, human mobility can be tracked at 
\subpar{Wide-area scale}. At this scale, the focus is put on data collection within city, regional or global boundaries. To track this large body of data it is necessary to utilize various sources of information and techniques, such that the rich semantics of the underlying human behavior are captured accurately. Technologies used to track the location of individuals at this scale span from GPS, WiFi, Cell Tower, to Bluetooth, and the data can be collected from different sources, such as mobile phones~\cite{isaacman2012human}, buses~\cite{bhattacharya2013gaussian}, subways~\cite{lathia2011smart}, cars~\cite{Giannotti2011}, and taxi vehicles~\cite{zhang2013coride}. To avoid biases in the collected data, due to a strong focus on a single or a small set of sources, a multi-source approach is necessary~\cite{Desheng2014}.

\subpar{Local-area scale}. These approaches are characterized by a stronger focus on a particular domain or application, and within the smaller boundaries of one or few buildings, or sometimes just a room. For example, studies focus on the analysis of visitors interactions, flows, and the popularity of points of interest within buildings~\cite{Martella2016,you2011}, fairs~\cite{Weppner2011}, festivals~\cite{Weppner2013}, and conferences~\cite{martella14}. As these scenarios are more constrained, the infrastructure is usually owned and controlled by the user of the data (\eg the museum, the fair, the conference), and it is often designed for the task at hand. Compared to Wide-area scale scenarios, Local-area efforts aim at uncovering the unique properties of the scenario at hand. Our work falls in this category and focuses on museums.

\subpar{Museum monitoring.} For the case of monitoring and analyzing the behavior of visitors in a museum, the visualization of metrics such as popularity, attraction, holding power and flows has been explored to support the work of museum staff~\cite{lanir2014visualizing,Strohmaier2015}. Many studies have also utilized sensors to measure the behavioral patterns of museum visitors. Earlier works focus on localizing visitors at coarse-grained levels (\eg room level) through technologies like Bluetooth~\cite{yoshimura2012new} to support multimedia guides~\cite{bruns2007enabling,wilson2004}. A recent work utilizes entrance and exit surveys together with physiological sensors and RFID-based positioning to classify the behavior of visitors into three categories: ``the contemplative'', ``the enthusiast'', and ``the social experience''~\cite{kirchberg2015museum}. Similarly, an RFID-based positioning sensor together with a compass have been used to study pairs of visitors, classify their behavior early in their visit and provide social-aware services to increase their engagement~\cite{dim2014automatic}. 
These works are conducted on traditional compartmentalized museums, focusing on classifying the visitor's experience. Our deployment provides \mbox{new insights based on a challenging} scenario with an open floor plan and very dense crowds.

An alternative, and potentially complementary, approach to mobile and wearable sensors is to use cameras. Computer-vision techniques are a common approach to tracking human mobility~\cite{ steffen2010methods, yaseen2013real, zhan2008crowd} and detecting anomalies in a crowd~\cite{krausz2012loveparade, mehran2009abnormal,Wijermans2016142}. However, cameras can suffer from poor lighting, and temporary or permanent obstructions~\cite{zhan2008crowd}. %, that characterize highly dynamic indoor scenario such as the museum subject of our deployment~\cite{zhan2008crowd}. 
Moreover, fusing the views from multiple cameras together in a highly dynamic indoor scenario is still challenging~\cite{song2011wide}. Finally, collecting large-scale footage of visitors can potentially generate privacy issues.

Once visitors mobility has been tracked, it is suitable for a number of applications. In the case of monitoring and analyzing the behavior of the visitors of a museum, visualization of metric such as popularity, attraction, holding power and flows has been explored with the purpose of supporting the museum staff~\cite{lanir2014visualizing,Strohmaier2015}. In the case of city-wide events, crowd behavioral patterns such as crowd density and pedestrian flows can help the monitoring of the crowd by the institutions~\cite{martino2010ocean,wirz2013coenosense} as well as support the individual decision-making of the members of the crowd~\cite{blanke2014capturing, wirz2010user}.

\subsection{Evaluated Methods}
Given the monitoring system available at the museum, we looked for methods that are amenable to RF beacons. We consider three methods: neighbor discovery, density estimators and density classifiers. 

\paragraph{Neighbor Discovery.}
% demand more resources but also provide more information 
Neighbor discovery methods monitor periodically the set of devices in their radio range. 
%For crowd monitoring applications, neighbor discovery has been proved useful to provide so-called \emph{crowd textures} abstractions to detect different crowd dynamics, like pedestrian lanes, congestions, and social groups~\cite{martella14b}, as well as to uniquely associate mobile devices with static anchor points and reconstruct the complete sequence of locations visited by a person~\cite{Martella2016}.
%
By keeping track of the encountered IDs over time, these mechanisms provide to crowd monitoring applications rich and accurate information about the so-called \emph{crowd texture}, usually modeled through a series of \emph{proximity graphs}~\cite{martella14}. Using this information, it is indeed possible to detect different crowd dynamics, like pedestrian lanes, congestions, and social groups, as well as to uniquely associate mobile devices with static anchor points and reconstruct the complete sequence of locations visited by a person~\cite{Martella2016}.

\subpar{Efficient discovery.} Neighbor discovery provides high quality information but is resource-demanding and can have non-trivial implementations depending on the requirements of the application~\cite{Dutta2008,Kandhalu2010,martella14b,Purohit2011} and the network conditions (\cf \chaprefs{chapter:introduction}{chapter:estreme}).
Our deployment, however, allows for a simple implementation. Instead of making the PoIs discover their neighborhood (potentially hundreds of bracelets), we made the neighbors (bracelets) discover the PoIs. Taking advantage of their extra energy, anchor points have longer duty cycles and are more frequently active. Upon waking up, a bracelet swiftly discovers its PoI. The information is then immediately sent to the sniffers, allowing an off-line reconstruction of the anchor point's neighborhood. 

\subpar{Data filtering.} 
% Because of the rich information provided by neighbor discovery mechanisms, it is possible to filter the neighborhood data, improving its quality.
In the event of a bracelet discovering multiple PoIs during a period, we disambiguate the information by uniquely assigning the bracelet to the PoI with the strongest signal strength. RSS signals are known to be noisy, but this simple heuristic allows us to eliminate many false positives, as demonstrated later. 
% ~\footnote{even though RSSI is not good to estimate distances, it is good enough for distance comparisons (RSSI(a)>RSSI(b) implies distance(a)<distance(b)). Can we put a citation here? For Claudio: can we do better than selecting the strongest?}.\think{Claudio: the results of my analysis of particle filters at NEMO show that they give the same performance as choosing the strongest AND smoothening with a sliding window. If, however, I also consider mobile-to-mobile proximity, than the particle filter smashes it by far.}\think{Later on in this paper, we will see how these ID disambiguation is useful to remove false positives in the density estimation.}
% \think{This way, for each point in time, a bracelet can be only assigned to 0 or 1 anchors, solving many problems with the variable range of radio devices.}

% estimation based on order statistics, any type of beacon will do
\paragraph{Neighborhood Cardinality Estimators.}
Compared to neighbor discovery mechanisms, neighborhood cardinality estimators trade information richness for speed and efficiency. 
To perform their estimations, these mechanisms often model the lower communication layers, searching for features that correlate to device density. 
% In Wireless Sensor Networks (WSN), 
The underlying principle is based on order statistics and is simple to follow: the more neighbors beaconing, the shorter the perceived inter-beacon interval and the higher the estimated density. 
For example, NetDetect exploits the underlying distribution of packet transmissions~\cite{Iyer2011}, while Qian \etal exploit the packets' arrival patterns in RFID systems for a fast estimation of tag cardinality~\cite{Qian2011}. 
For our deployment, we opted for Estreme (\cf \chapref{chapter:estreme}), which estimates the neighborhood cardinality by measuring the average time between periodic, asynchronous beacons. Compared to the other techniques, Estreme was preferable because of its simplicity, high reported accuracy and minimal parameter tuning.

% features based on RSSI characteristics 
\paragraph{Density Classifiers.}
%In a way that is similar to cardinality estimator, 
Density classifiers try to infer the density of a crowd by exploiting correlations with signal strength statistics~\cite{Nakatsuka2008,Weppner2011,Weppner2013,Yuan2011,Yuan2013}. These techniques require placing a set of wireless devices within the crowd, analyzing the changes of several radio features as the crowd density changes, and training a classifier. 
%These techniques require to place in the crowd a set of wireless devices that continuously communicate with each other. 
For example, a common premise is to assume that as the density of people
increases, the mean RSS decreases (due to bodies blocking radio communication) and the RSS variance increases (due to the added multipath effects).
In our work we do not implement these classifiers, but focus on assessing if the RSS trends reported in the literature hold in all scenarios.

\section{Density analysis} 
% change ground truth with best possible estimation (off-line, resource intensive, information-rich)

As stated in the introduction, in some scenarios, resource constraints or privacy issues may prevent the use of neighbor discovery methods. Under these circumstances, density could be used as a proxy to assess the popularity of an exhibit. Thus, the first question that needs to be answered is: \emph{how accurate are the density estimators reported in the literature in the heterogeneous setup of our museum?} 
% Evaluating the feasibility of density estimators is important, since efficient neighbor discovery mechanisms are not always an option. This is the case, for example, when the information provided by devices cannot be controlled (smartphones) or is limited due to resource constrains and privacy issues.
% Compare density estimators against neighbor discovery
% Fortunately, our deployment presented only few of these constrains, allowing us to run both neighbor discovery and density estimator mechanisms. 
% Therefore, we exploited the availability of neighbor discovery to serve as a comparison for the evaluation density estimation techniques. We start by evaluating the precision of these latter techniques.

\subsection{Estimation accuracy}

\begin{figure}
	\centering 
	\includegraphics[width=\textwidth]
	{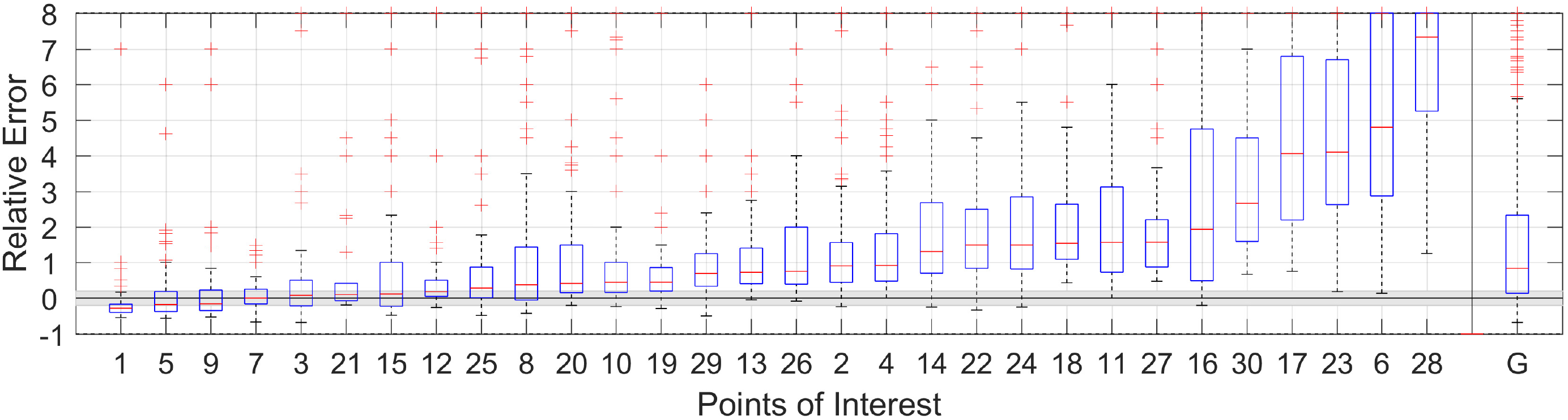} \caption{Density estimation error for each PoI (1$\dots$30) and for the global system (G).}
	\label{fig:nemo_estimation_error} 
\end{figure}

% Moving from cardinality to density is challenging
In our evaluation, we use the information obtained with the neighbor discovery method as ground truth\footnote{Recall that the participation rate of visitors was around 25\%, thus the actual number of visitors was four times what is shown as ground truth. Note that this sub-sampling rate does not affect the validity of our analysis.}. The error of the density estimator is computed as 
%Figure~\ref{fig:nemo_estimation_error} shows the density estimation error computed as
$$
\frac{D_{est} - D_{true}}{D_{true}},
$$
where $D_{est}$ is the density provided by the cardinality estimator and $D_{true}$ is the ground truth.
Figure~\ref{fig:nemo_estimation_error} shows the results per exhibit during one day of deployment, 09:00 - 17:00. The plot captures the min, max, median (red line) and the 25th and 75th percentile errors. The red plus signs depict outliers, which we define as points deviating by more than two times the standard deviation from the mean. PoIs are ordered by increasing median error. The last box represents the aggregate accuracy of the global system (G).
%Except for the last box that represents the global system (G), each other box represents the estimation errors of each point of interest (PoI) during one day of deployment (09:00 - 17:00). PoIs are ordered by increasing median error.

As can be seen, only a quarter of the PoIs show a median error within the range observed in testbed experiments of \chapref{chapter:estreme} ($<$ 20\,\%), grey horizontal bar in Figure~\ref{fig:nemo_estimation_error}. 
Except for three PoIs, which are underestimating their density, the remaining ones are all over-estimating their density, with median errors ranging from 20\% to 700\%.
The resulting global estimation error shows a median over-estimation of 83\%.

\subpar{Why is the estimator so inaccurate?}
The problem lies on the fact that previous work was evaluated in a standard building, with clearly demarcated areas (rooms) and with homogeneous densities. Conditions that in our deployment only occur for a small fraction of PoIs, as discussed next.

\begin{figure}
	\centering 
	\includegraphics[width=\linewidth]
	{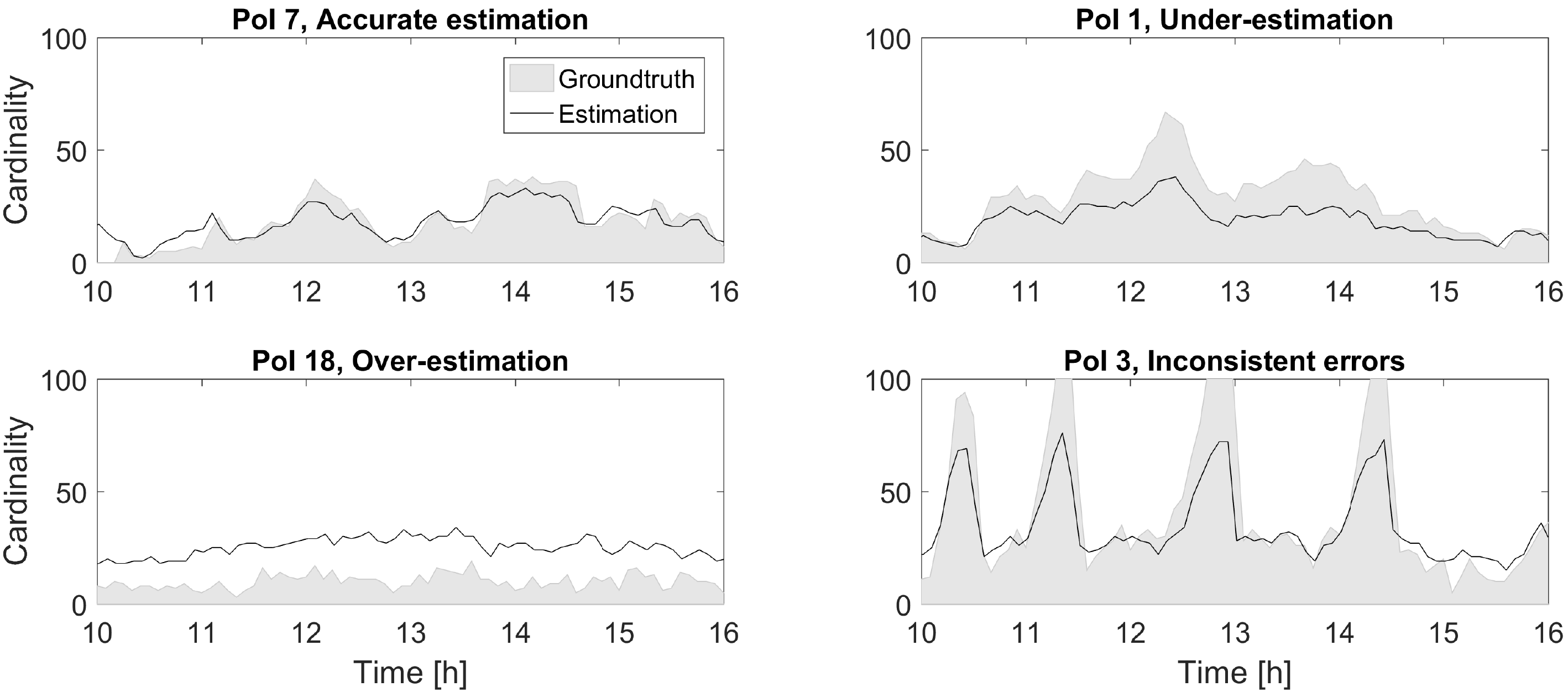} \caption{Density estimation errors over time.}% for a representative selection of PoIs.}	
	\label{fig:nemo_density_comparison} 
\end{figure}

\subsection{Coverage versus Overlap}
\label{sec:nemo_coverageoverlap}
To better understand the causes of inaccuracies of the density estimator, we analyzed the individual estimations at each PoI. Figure~\ref{fig:nemo_density_comparison} depicts a selected representative subset of PoIs. We distinguish four common patterns.

\subpar{Accurate Estimation.} 
PoI 7, corresponding to the museum's ``theater'', represents an enclosed room where visitors join a workshop that runs continuously. In this scenario, the conditions are similar to the ones in the original work. The density estimator performs accurately throughout the day, with errors that are within the expected  boundaries ($<$ 20\%).

\subpar{Consistent Under-estimation.} 
PoI 1 corresponds to the museum's ``event hall'' and represents a large open-space area where different workshops are given. 
Compared to the previous PoI, this is a more isolated area that is connected to the rest of the museum via a set of long stairs. 
This large area is covered by one anchor point that was not capable of achieving stable communication with some of the bracelets in the area (\emph{false negatives}). 
Due to this poor coverage the densities were consistently under-estimated, especially when the crowd increased, such as at 12:30.
 
\subpar{Consistent Over-estimation.} 
As stated before, most of the PoIs suffer from over-estimation. PoI 18 is a representative point that corresponds to an exhibit that is centrally located in the museum.
Due to the open 3D layout (no barriers), the high number of people (bracelets) and dynamics (movement), we face two problems: overlapping coverage regions among PoIs and \textit{long links}\footnote{Wireless channels with high node densities and high dynamics are known to form temporal long links. Long links are a desirable phenomena for routing but not for density estimation and localization.}. These overlapping regions and long links lead to many false positives (duplicated IDs), and the high number of false positives leads in turn to gross over-estimations. Under normal circumstances --buildings with rooms and relatively low and homogeneous densities-- false positives are less common and usually filtered out by walls. 

A popular method used to discard false positives is to implement a virtual filter (virtual wall). In this way estimators would not consider signals below a given RSS. Unfortunately this filtering method would not work in our setup. We varied the RSS threshold of the neighborhood discovery mechanism and analyzed the resulting neighborhood cardinalities 
%\footnote{When an RSSI threshold is applied to a neighbor discovery mechanism, nodes are discovered only if their signal is stronger than the threshold.} 
(\cf Figure~\ref{fig:nemo_coverageoverlap}).
For each moment in time, \emph{coverage} was computed by counting the number of unique IDs in the whole system, while \emph{overlap} was computed by summing the number of unique IDs at each PoI: 
$$
\text{coverage}(t) = \left|\bigcup\limits_{i=1}^{30}n_i^t\right|, \quad \text{overlap}(t) = \sum\limits_{i=1}^{30}\left|n_i^t\right|,
$$
where $n_i^t$ is the neighborhood discovered by PoI $i$ at time $t$. 
As we can see in Figure~\ref{fig:nemo_coverageoverlap}, the more aggressive the RSS thresholding is, the smaller the overlap problem becomes. Unfortunately, at the same time the coverage problem aggravates (which would lead to under-estimations), making it impossible to find an optimal thresholding value that balances coverage with overlap. 

Another method used to overcome false positives is to tailor the coverage of each anchor point to the specific area of interest of the exhibit, as in~\cite{Martella2016}, which uses directional antennas. But this option would be cumbersome to adjust since the exhibits in the museum do not have clear predetermined areas.

\subpar{Inconsistent Estimation Errors.} 
While for the previous two types of PoI (consistent errors) it may be possible to solve the coverage/overlap problem with a careful and time consuming deployment setup, for exhibits with large coverage areas and highly heterogeneous densities even a perfect positioning and filtering system would not reduce estimation inaccuracies.

This is the case for PoI 3, which corresponds to the ``chain reaction'' show. 
This is the most famous exhibit in the museum and spans from the first to the third floor of the building. This show is repeated 5 times a day, lasts 15 minutes and attracts a large fraction of the visitors around the internal balcony of the museum. 
As we can see in Figure~\ref{fig:nemo_density_comparison}, the estimation is quite accurate during the time where there is no show (low homogeneous density composed by people passing by), but during the show the density is heavily under-estimated.

This estimation error is not due to a coverage problem, but is an
artefact of how estimators cope with sudden changes in density. The
estimation process is sped up by averaging local estimates with
neighboring ones, effectively increasing the number of samples, in turn
boosting convergence.  An adverse effect is that this approach smoothes
the spatial density peaks, leading to under-estimations in denser areas
(PoI 3) and an over-estimation in sparser ones (see Section~\ref{sec:estreme_results}).

\begin{figure}
	\centering 
	\includegraphics[width=\linewidth]
	{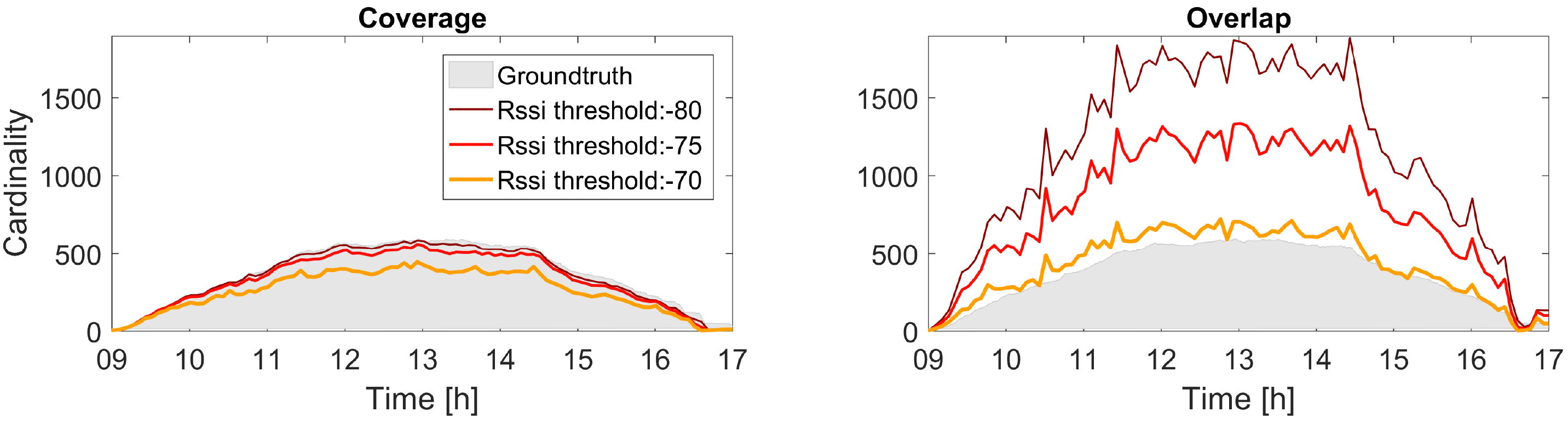} \caption{Coverage and overlap for different RSS thresholds. }
	\label{fig:nemo_coverageoverlap} 
\end{figure}

\subsection{Effects of Crowd Dynamics}
% analyze in deep effects of a crowd on signal
We now analyze the effect of crowds on the propagation of radio signals. This analysis highlights the limitations of density classifiers.
Our deployment has 11 static bracelets. We used these bracelets, together with anchor nodes (also static), to monitor the changes in radio signals throughout the day.

Figure~\ref{fig:nemo_signal_characteristics} depicts the effect that density has on radio coverage and signal strength. This data was obtained as follows. Every 300\,s, each static point (bracelet and anchor) calculated (i) the average density in its surroundings, (ii) the average RSS from other static points and (iii) the average `radio range'. The average radio range was obtained as follows: upon receiving a packet from a static point, the receiver calculates its distance to the sender (all static points have x,y coordinates), then, the receiver averages the distances observed during the 300\,s period. The boxplots represent the statistics during a day for the accumulated 300-second-long periods.

\begin{figure}
	\centering 
	\includegraphics[width=\linewidth]
	{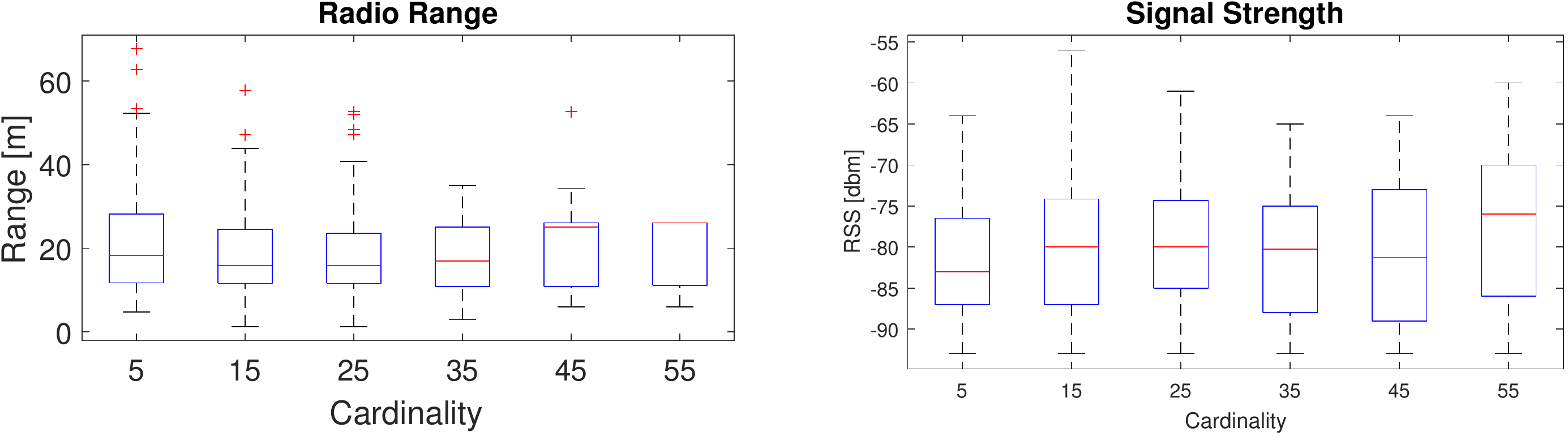} \caption{Effect of the surrounding density on static nodes.\vspace{-.8ex}}
	\label{fig:nemo_signal_characteristics} 
\end{figure}

% range vs. density
\paragraph{Effects on Range.}
The results in Figure~\ref{fig:nemo_signal_characteristics} (left) indicate that an increase in crowd density reduces the maximum range of devices, but does not affect significantly the median values. 
This is not too surprising, since people are known to cause signal attenuation due to body effects, which filters out weaker radio signals sent by distant nodes. 
This result indicates that cardinality estimators, such as Estreme, are ill-equipped to calculate densities (number of devices over a given area), since the radio range can be highly variable.

% RSSI vs. density
\paragraph{Effects on Signal.}
According to many crowd density classifiers, when the crowd density increases, signal strength characteristics should get worse. 
Similar to what is hypothesized for radio range, the assumption is that an increasing number of people should reduce the average RSS (due to body effects) but increase its variance (due to multi-path effects).

Surprisingly, this relation does not hold for the RSS values observed in our deployment. Figure~\ref{fig:nemo_signal_characteristics} (right) shows no clear correlation between crowd density and key features in the RSS (median and variance).
To better understand why our observations diverge from previous works, we analyzed the relation between density and signal strength for each PoI. Figure~\ref{fig:nemo_rssi_density} depicts a representative subset. We can observe that each PoI in this subset showcases a different relation.

\subpar{Negative correlation.} 
PoI 9 represents an open-space exhibit close to other PoIs. This PoI has a rather constant and homogeneous density throughout the day, cf. Figure~\ref{fig:nemo_estimation_error}. 
In this scenario, the relation between density and RSS follows the assumption made in prior studies, with a decreasing median and a slightly increasing variance (except for density 55, for which we do not have enough data points to accurately compute the variance).

\subpar{Positive correlation.} 
Surprisingly, even though PoI 20 represents a scenario that has crowd dynamics similar to those of PoI 9, its RSS values show the opposite (!) relation with density. We do not know the exact reasons for this difference, but we hypothesize that it is caused by the different levels of interference perceived by the two PoIs.
By looking at Figure~\ref{fig:nemo_estimation_error}, we can see that PoI 20 is affected by more false positives than PoI 9. This implies that PoI 20 is exposed to more overlapping areas and long links, which increase interference. According to \cite{Behboodi2015}, this increased interference could lead to RSS values that are up to 8\,dBm higher than normal, an offset that is almost as big as the range observed in Figure~\ref{fig:nemo_rssi_density}.

\begin{figure}
	\centering 
	\includegraphics[width=\linewidth]
	{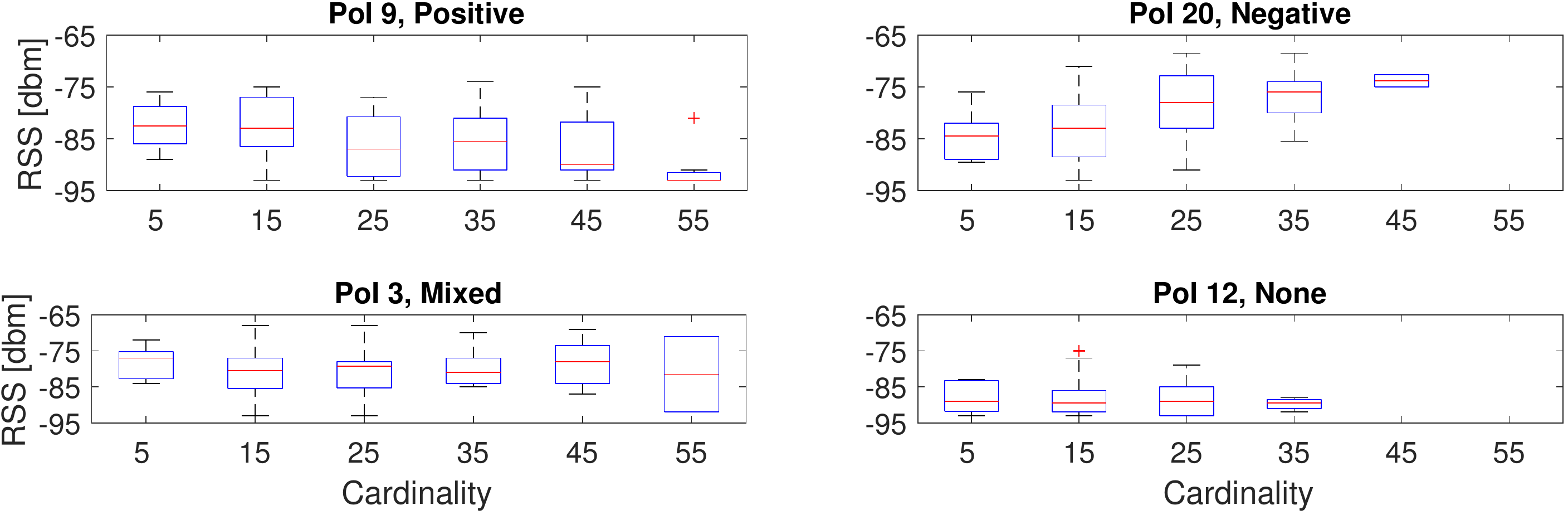} \caption{Effect of the surrounding density on RSS characteristics.}
	\label{fig:nemo_rssi_density} 
\end{figure}

\begin{figure}
	\centering 
	\includegraphics[width=\linewidth]
	{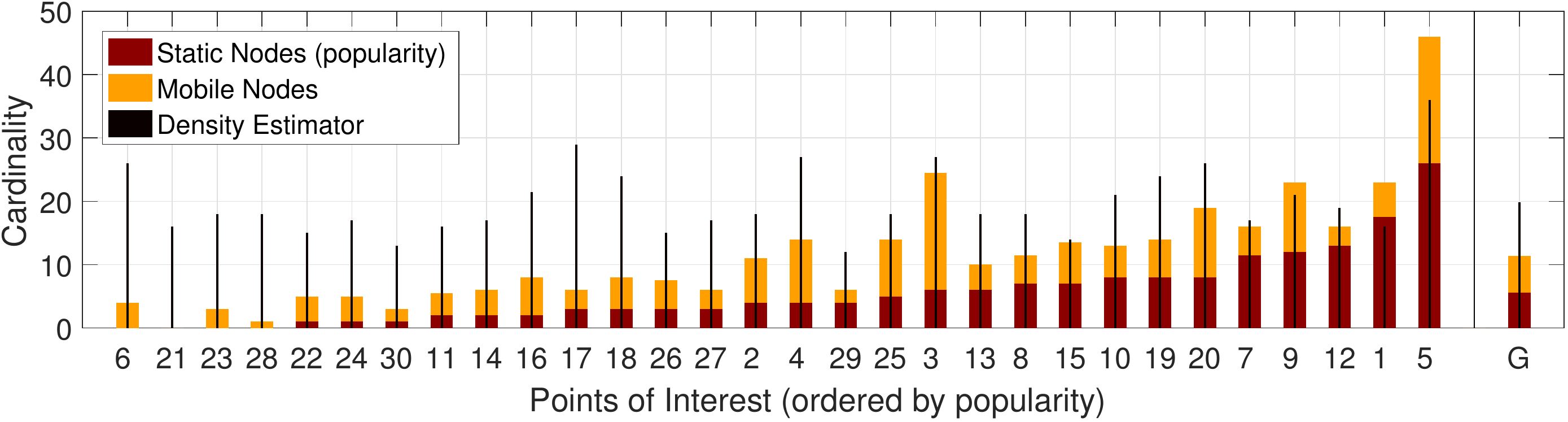}
	\caption{Breakdown of static and mobile nodes (neighborhood discovery)
	and the density estimates.}
	\label{fig:nemo_popularity_ordering} 
\end{figure}

\subpar{Inconsistent correlation.} 
PoI 3 corresponds to the open-space ``chain reaction'' show and presents inconsistent behaviors, with median RSS values that both increase and decrease with density. 
This bimodal behavior is probably caused by the peculiar characteristics of this PoI, which periodically morphs from a crowded point of transit to a large point of attraction (and vice versa).

\subpar{No correlation.} Finally, PoI 12 corresponds to a laboratory with limited number of participants, all seated and well spaced from each other. In this scenario the crowd is so sparse and small that it is not able to affect the perceived signal strength.

\fakeparagraph{}%
As can be seen, the relation between signal and density drastically changes depending on the scenario. 
We argue that for such classifiers to work in heterogeneous scenarios, such as our museum, PoIs would need to undergo a cumbersome training process to account for the peculiar radio features of each area. 

\vspace{1em}\section{Popularity Analysis} 
As was shown in the previous section, density estimation and
classification techniques struggle with the heterogeneous set of
conditions found in the museum, producing inaccurate results in many
cases. What is worse, in this section we will see that density in itself is a poor proxy for the popularity
of an exhibit. The reason is that density does not discriminate visitors
engaged with the exhibit's activity from ``passers by'' that happen
to be in the vicinity for a short while. This is especially a concern in the
museum as its lack of fixed routes (e.g., corridors) in combination with the
open-space layout makes it hard to physically separate visitors and passers by,
which rules out the use of RSS thresholding and directional antennas to
limit the area of interest surrounding an exhibit. (The concept of an area of
interest is also dynamic as at busy times people queue up in unpredictable
ways.)

\begin{figure}
	\centering 
	\includegraphics[width=\linewidth]
	{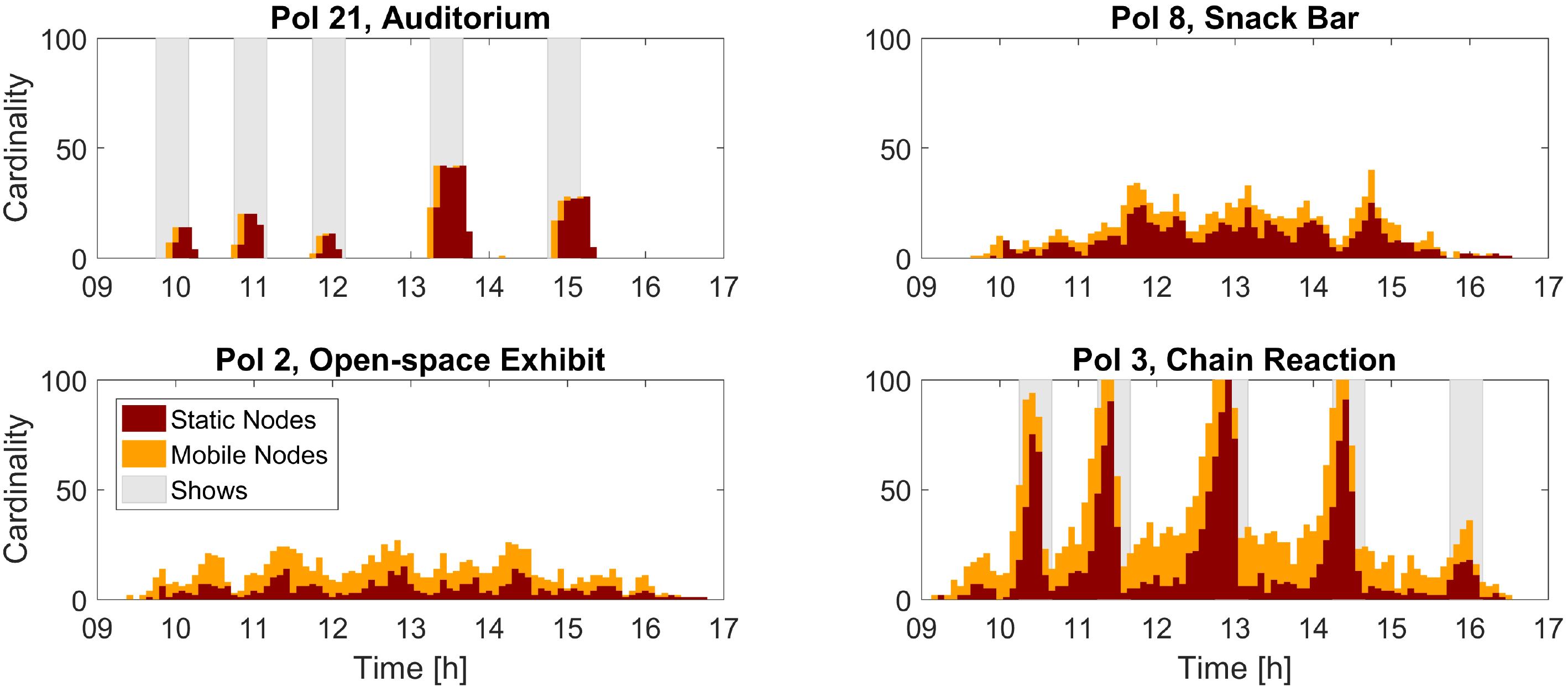} \caption{Density and popularity over time.}
	\label{fig:nemo_popularity_density} 
\end{figure}

\vspace{1em}The definition of popularity that we will be using throughout the
remainder of this chapter is defined as the number of visitors that stays in the
vicinity of an anchor point for some time window $W$. More formally,
we distinguish Static (S) and Mobile (M) nodes as follows:
$$
\text{S}(i,t) = \bigcap\limits_{w=0}^{W}n_i^{t-w},\quad \text{M}(i,t)= n_i^t - S(i,t),
$$
where $n_i^t$ is the neighborhood discovered by PoI $i$ at time $t$.  
The popularity of exhibit $i$ at time $t$ is then simply given by $|S(i,t)|$.
We experimented with lengths of window $W$ in the range of 10 - 300
seconds. As longer windows provide a smoother display of changes over
time, we will present the results for the largest window of 300\,s in the
sequel of the chapter, but it must be noted that with shorter windows similar
results are obtained.

The importance of properly discriminating static and mobile nodes
can be seen in Figure~\ref{fig:nemo_popularity_ordering}, which plots the
median number of static (red bars) and mobile (yellow bars) nodes observed with our
neighborhood discovery data during a day.  The higher the red bar, the more popular the PoI is. Note that the fraction of mobile nodes
differs considerably between PoIs, and is uncorrelated with popularity. Overall, about half the nodes are mobile, and should be
disregarded. For example, looking at the aggregated cardinality,
PoI 3 would be the second most popular exhibit, while in fact it is only the 11$^{th}$
most popular one (PoI 1 is the second most popular). For reference,
Figure~\ref{fig:nemo_popularity_ordering} also includes the cardinalities
produced by the density estimator (black lines), showing that using
them would result in an even-more incorrect ordering of the popularity
of the PoIs.

\begin{figure}
	\centering 
	\includegraphics[width=\linewidth]
	{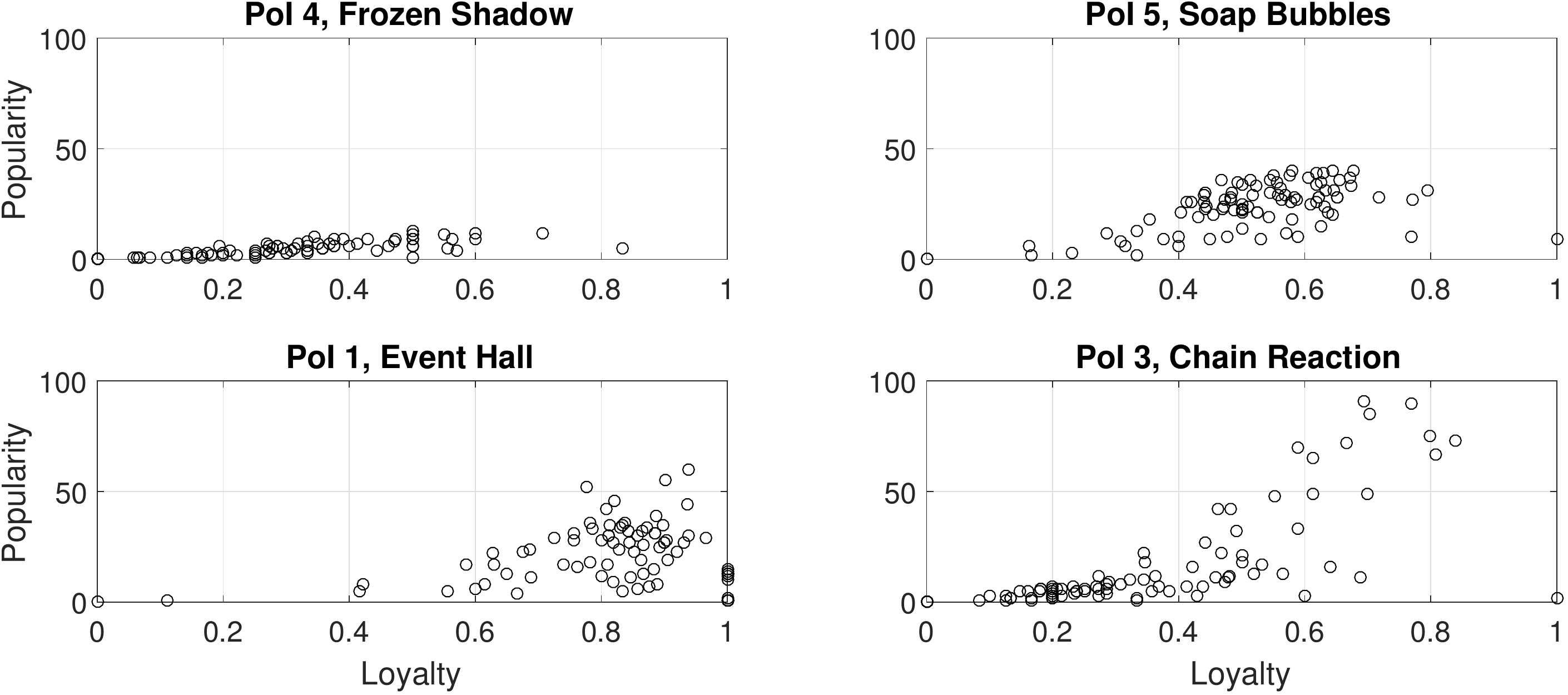}
	\caption{Attractiveness of point of interests.}
	\label{fig:nemo_attractiveness} 
\end{figure}

While Figure~\ref{fig:nemo_popularity_ordering} presents the median cardinalities
across a whole day, it is also instructive to look at how the fraction of
mobile nodes changes over time as interesting patterns can be observed.
Figure~\ref{fig:nemo_popularity_density} plots the static/mobile breakdown of four
exemplary scenarios.

\subpar{Points of attraction.} 
PoI 21 corresponds to an auditorium enclosed by walls in which periodic
shows attract tens of visitor for a short, fixed period of time (grey
bars in Figure~\ref{fig:nemo_popularity_density}). In these conditions the
flow dynamics are minimal as people enter the room shortly before a show
starts; note the minimal yellow shading at the left of each show.
Consequently, density estimations maps almost perfectly to popularity for
the auditorium.
PoI 8 corresponds instead to a small bar outside the previously mentioned
auditorium. Here density captures the visitors having a short snack (intended)
as well as the ones passing by (unintended) clobbering the real popularity.

\subpar{Point of passage.} 
PoI 2, on the other hand, represents an exhibit located in an open space, close
to a point of passage \ie an area that is similar to a corridor but does not
have walls. In this case most of its perceived density it is due to people
passing by. To understand the popularity of PoIs like this, it is therefore
essential to distinguish between static and dynamic neighboring bracelets,
especially since the fraction of mobile nodes varies across time.

\subpar{Temporal hybrids.} 
Finally, PoI 3 represents an exhibit (the ``chain reaction'') that is
normally a popular point of passage (amidst several staircases,
see Figure~\ref{fig:nemo_nemo_3d}), but periodically becomes a point of
attraction for short periods of time (show times are highlighted with
a grey background). The error, \ie the amount of moving nodes taken into account, 
in this case drastically fluctuates, with rather accurate
estimations during the periodic shows (hardly any yellow on top) and bad
ones during the remaining periods (lots of yellow).

\subsection{Attractiveness}
Besides differentiating static and mobile bracelets, which allowed for
establishing the true popularity of exhibits, we computed several other
metrics to help the museum curators obtain a deeper understanding of the
dynamics behind each PoI's popularity. 

The most interesting of these metrics is the \emph{loyalty} of visitors at
each PoI, defined as the fraction of static nodes (or ratio of popularity
over density).  The rationale for the loyalty metric is that the more
engaging a PoI is, the more people passing by will stop, and become
static (engaged).  A loyalty of 0 denotes no interest as apparently all
visitors in the vicinity of the PoI are ``passers by''. 
On the other hand, a loyalty of 1 denotes a PoI that engages all visitors in range.

Figure~\ref{fig:nemo_attractiveness} compares the loyalty against the
popularity for a few selected PoIs, showing the remarkable diversity in
scenarios found at the science museum. Each dot captures a five-minute interval.

\subpar{Unpopular attractions.} 
PoI 4 represents the ``frozen shadow'' room, a PoI with low loyalty and
popularity. This is not surprising, since the engagement in this PoI is usually
short as the attraction consists of a room with phosphorescent walls that
captures the shadows via a strobing light (flashing once every 30\,s). Even the most enthusiastic kids do not freeze their shadow more than ten times, hence, limiting the popularity to near zero as that maximum time equals the window size (300\,s).

\subpar{Popular attractions.} 
Compared to PoI 4, the ``soap bubbles'' attraction (PoI 5) consists of an open area were kids are able to make huge soap bubbles using a wide set of tools. 
As can be seen from Figure~\ref{fig:nemo_attractiveness} this PoI shows both a higher loyalty and attractiveness \ie it is able to engage passing by visitors.

\subpar{Workshop areas.} 
A similar popularity, but a higher loyalty is achieved by the ``event hall'' (PoI 1), a slightly isolated open-area where several workshop are taught. The main difference between this PoI and the previous one, is that these workshops are very attractive and take
much longer to complete, hence, capturing a particularly `long' loyalty from visitors.%, effectively increasing the measured loyalty.

\vspace{2em}\subpar{Periodic attractions.}
Finally, PoI 3 represents the periodic ``chain rection'' show. In this case, one can observe (again) the bimodal behavior of this PoI: most of the time it serves as a point of passage (low loyalty and popularity) and during shows as a point of attraction (high loyalty and popularity). Note that compared to the event hall, we have fewer dots on the high loyalty region of the plot. This is because the chain reaction lasts for 15 minutes (3 dots) and it is scheduled only 5 times a day, while the workshops run continuously.%the loyalty is lower, due to the shorter duration of the chain reaction in comparison to an event workshop.

\section{Conclusions}
In this chapter we reported our experience in monitoring the popularity
of exhibits at a modern, open-space science museum. During some of the
museum's most crowded days, we equipped visitors with bracelets
emitting RF beacons, allowing us to study existing crowd-monitoring
mechanisms such as neighbor discovery, cardinality estimators, and
density classifiers in a diverse set of challenging real-world settings.

An in-depth analysis of the collected data (proximity detections and
density estimations) provided a better understanding of the limitations
of the existing mechanisms. In particular, the use of ID-based neighbor
discovery allowed us to univocally associate each bracelet to the point of
interest, and distinguish static from mobile nodes. The latter proved key
to understanding the relation between density and popularity. 
Cardinality estimators, such as Estreme, cannot provide such differentiation, resulting in estimation errors far higher than originally observed in the testbed experiments of \chapref{chapter:estreme}. 
%They simply failed to provide reasonable information in 80\,\% of the cases. 

An important downside of neighbor discovery is that it relies on identity
information, which may not always be available because of privacy issues~\cite{Freudiger2015}, network conditions and limited resources (\cf \chapref{chapter:introduction}). 
An important line of future research is thus to devise lightweight and
privacy-preserving algorithms that can differentiate static and mobile
nodes, allowing a precise estimation of popularity.

%\section{Fun Facts}
%
%Although experimentation at the scale of a thousand devices comes with a lot of
%hard work and frustration, the museum experience also revealed some fun facts
%that kept us going.
%
%i) Independently of the battery capacity, the lifetime of a bracelet is just
%two days. We argue that this is due to the fact that kids like to use bracelets as swords.
%
%ii) Having a beaconing mechanism is useful when recovering bracelets abandoned
%at the most unusual places \ie plant pots and urinals.
%
%iii) Ethernet is not as reliable as WiFi, since wall sockets can easily
%be powered down by technicians leaving for home at 5PM.

%% file: conclusions/conclusions.tex
%!TEX root = ../dissertation.tex
\chapter{Conclusions}
\label{chapter:conclusions}

Inspired by the need of monitoring large crowds with wireless wearable devices, this thesis tackled the challenges of communication in Extreme Wireless Sensor Networks, where thousands of mobile nodes gather together in highly confined areas and periodically need to exchange their sensed information. 
As we argued, in these extreme conditions, traditional low-power mechanisms simply do not scale and collapse. Therefore, we proposed a novel communication stack that is based on opportunistic communication.

\section{Regard challenges as opportunities}
The mechanisms presented in this thesis originated from a simple idea %hypothesis 
presented in \chapref{chapter:introduction}. 
Instead of adapting existing mechanisms to the challenges of EWSNs, it is better to design new mechanisms that a) benefit from harsh network conditions  (anti-fragile principle) and b) are able to spontaneously adapt  when the same conditions change (opportunistic principle). 
The resulting mechanisms, presented in chapters~\ref{chapter:sofa}, \ref{chapter:estreme} and \ref{chapter:staffetta} not only scale to EWSNs, but do so with a remarkable simplicity (state-less principle) and flexibility (robustness principle).

SOFA (\chapref{chapter:sofa}) is the cornerstone of the communication stack presented in this thesis and is an example of how opportunistic mechanisms can exploit extreme networks conditions to their own advantage. 
In particular, SOFA's communication primitive exchanges data with the first available neighbor, independently of its ID and topology location. By giving priority to efficiency, rather than the certainty of the routing path, nodes are able to efficiently communicate throughout a wide range of network conditions (density, size and load) with low energy consumption ($\approx$ 2\,\% duty cycle) and high delivery rates (above 90\,\%). 
In other words, SOFA achieves performance that is similar to synchronous protocols, but with the simplicity and flexibility of the asynchronous ones.
Moreover, because SOFA's mechanism is stateless and does not require any information from the neighbors, it shows a remarkable resilience to  mobility and link failures, which are common characteristics of EWSNs. 
SOFA's flexibility has been particularly important during the real-world experiments at the NEMO science museum, where the mechanism had to reliably communicate throughout a wide set of network conditions.

%two characteristic  and is able to reliably exchange information even beyond the saturation point of the wireless channel.

% This rate of diffusion fits well crowd monitoring applications, which need to monitor the status of the network with minute-accuracy. SOFA fulfills such monitoring task by fully supporting data dissemination mechanisms such as Gossiping, which are able to compute many aggregates both at network and neighborhood level.

\section{Building on top of opportunistic primitives}
Even though SOFA allows to communicate in EWSNs, 
%it also make adapting existing 
not all network services can be adapted to its novel communication primitive in a straightforward way. 
%Some, like Gossip, are based on random sampling and are naturally supported by SOFA. Others, which are based on deterministic primitives such as unicast and broadcast, requires instead in fact to rethink from the ground up many network services such as neighbors discovery and routing
Opportunistic anycast requires in fact to rethink -- from the ground up -- many network services such as neighbors discovery and routing, which are normally based on deterministic primitives such as unicast and broadcast.
To this end, in \chaprefs{chapter:estreme}{chapter:staffetta} we proposed Estreme and Staffetta, two alternatives to traditional neighbor discovery and data collection, respectively.

%\paragraph{Estreme.}In WSNs, neighbor discovery is both essential to many other protocols and challenging to achieve. For this reason it has been widely studied in literature.
%Unfortunately, under high densities and mobility, existing solutions require too much time to discover all the neighbors within the limited amount of time that is available \ie too many neighbors, too short time.
%
%To tackle this problem,
Estreme (\chapref{chapter:estreme}) proposed to estimate, rather than count, the number of neighboring devices, thus removing the need of explicitly discovering each device. 
Following the guidelines of \chapref{chapter:introduction}, Estreme performs its estimation by exploiting the unique correlation present in SOFA between the rendezvous times and the neighborhood cardinality (opportunistic principle). The denser the network, the shorter these times are. Because Estreme is based only on passive observations (of the rendezvous time), it can perform its estimations concurrently (up to 100 concurrent estimators) and can exploit the spatial correlation of density to speed up its estimation convergence. This way, Estreme is able to handle abrupt changes in the network (robustness principle), with errors that are below 10\,\% and stable across a wide set of conditions (from 10 to 100 neighbors).

%\paragraph{Staffetta}
%Traditional data collection mechanisms, 
%Staffetta tackles the problem of data collection in EWSNs by building a routing mechanism based on SOFA's opportunistic anycast. Unfortunately, routing over opportunistic anycast is quite challenging since, for sake of efficiency, the destination of the forwarded data is not known in advance.

Staffetta (\chapref{chapter:staffetta}), on the other hand, is able to bias the opportunistic communication of SOFA in a desired direction, effectively providing data collection over opportunistic anycast. 
Staffetta achieves this bias due to the following phenomenon: 
because in SOFA messages are forwarded to the first available neighbor, it is possible to bias the forwarding of data by making some nodes to wake up more often than others. 
In the case of data collection, Staffetta forms an \textit{activity gradient} in which nodes are the more active, the closer they are to the sink, \ie the node collecting all the sensors' data. 
Similar to the other mechanisms proposed in this thesis, Staffetta's activity gradient forms spontaneously from local nodes' observations and, thus, can easily adapt to network changes with minimal overhead (state-less principle).   
The extensive set of experiments presented in \chapref{chapter:staffetta} showed that Staffetta significantly reduces both packet latency and energy consumption of data collection over opportunistic primitives while maintaining high packet delivery ratios (robustness principle). As Staffetta does not need to maintain complicated routing metrics spanning the entire network, it can also handle network dynamics like link-quality fluctuations and node mobility really well. We found that Staffetta adapts its gradient in just a matter of seconds.

%Estreme, for example, performs its estimation by modelling SOFA's performance against the neighborhood cardinality. This model is valid when all nodes wake up with the same frequency and clearly could not be immediately applied in the case of an activity gradient (Staffetta), where the nodes' wakeup frequency depends on the distance from the sink. 
%Thus, an integration between Estreme and Staffetta will require to find a more flexible model of SOFA's performance or to develop a new mechanism that allows to differentiate between different kind of wakeup events.

\section{Extreme Wireless Sensor Networks in practice}
\chapref{chapter:nemo} concluded this thesis by testing the previously-proposed mechanisms in a wide range of real-world conditions. 
In particular, we tested both SOFA and Estreme as a tool to monitor the crowd characteristics in a six-stories science museum. The deployment data, collected over three of the most crowded days of the year with close to a thousand visitors, showed that both mechanisms are able scale to EWSNs even though some expected and unexpected challenges arose.
First, similar to what we showed in \chapref{chapter:estreme}, our density estimator Estreme suffered when the network density drastically changed in space. We plan to address this accuracy problem in future work.
Second, using the neighborhood cardinality to estimate the crowd density provided inaccurate results. This is due to the fact that controlling the range of radio devices in heterogeneous environments (with walls, rooms and open-spaces) is very challenging. More research must be done to efficiently solve the coverage/overlap problem described in \secref{sec:nemo_coverageoverlap}. 
Third, because our and others estimators do not detect if a neighboring device is moving, they are not able to distinguish between visitors who are interested, and stand, at an exhibit, and visitors who are not interested and are just passing by. This differentiation is essential, for example, to compute exhibit-related metrics such as the popularity of a point of interest.

\section{Future Work}
The four principles presented in \chapref{chapter:introduction} served as guidelines to develop the solutions presented in this thesis and resulted in mechanisms that are both simple and scalable to EWSNs. Nevertheless, we recognize that these principles could be improved to make EWSNs' mechanisms even more flexible and efficient.

\subpar{Towards versatile mechanisms.} Designing anti-fragile mechanisms \ie mechanisms that perform better in hard conditions, is the first step toward communication in EWSNs. 
Nevertheless, we argue that the next step is to make these mechanisms \textit{versatile} to any network condition, from mild to extreme.
Unfortunately, creating such mechanisms is difficult and presents three main challenges. First, versatile mechanisms for EWSNs still needs to excel in extreme network conditions. Second, under mild conditions, these mechanisms needs to operate with performance that are comparable to traditional mechanisms -- which were specifically designed for the mild case. Failing in solving this challenge could impede the adoption of EWSNs technologies. In other words: ``\textit{why should you use a technology that does not work in non-challenging conditions?}''. Third, versatile mechanisms should be able to adapt rapidly to the abrupt changes between mild and extreme conditions. As we showed in this thesis, opportunistic behaviors could be a solution to this latter challenge.

As an example, according to the anti-fragile principle, SOFA is very efficient when the neighborhood cardinality crosses $\approx$~20 neighbors. 
Under this density threshold, on the contrary, SOFA's rendezvous gets longer and longer, resulting in very high energy consumption. 
This inefficiency could have been problematic during the NEMO experiment (\cf \chapref{chapter:nemo}) where low density conditions could have been temporarily present. In this case, to make SOFA more versatile, we borrowed some ideas from Staffetta and imposed an energy budget on each node. 
%Thus, node were waking up once a second in order to be detected by neighbors, but were actively strobing the channel with a frequency that depended on the observed rendezvous length and their available budget. The more capable the battery, the higher the budget\footnote{Note that Estreme requires all nodes to wake up with a fixed frequency. No constrains are instead imposed on the number of times each node can actively rendezvous with its neighbors.}. 
%We were therefore able to adapt the sampling frequency of estreme based on the given energy budget, without the need of adapting its estimation model (which requires all nodes to wakeup with a fixed frequency).

\subpar{Integrating opportunistic mechanisms.}
Even though Estreme and Staffetta provide viable alternatives to their traditional protocol counterparts, their coexistence is not straightforward and must be further analyzed. Estreme, for example, relies on the fact that all nodes in the network wake up with the same frequency. Staffetta, instead, modifies this same frequency to bias the direction of the forwarded messages. 
A possible solution to this conflict, which should be further analyzed in future work, could consist of distinguishing between different wake-up events (the ones of Estreme from the ones of Staffetta). Alternatively, Estreme's model could be extended in order to support an uneven set of wake-up frequencies. Nevertheless, we argue that similar integration issues could arise also in future mechanisms for EWSNs that opportunistically exploit the base communication layer (SOFA) to efficiently fulfill their tasks. 

\subpar{Estimate the right type of information.}
\chapref{chapter:nemo} suggests that, for museum and crowd monitoring applications, estimating the neighborhood cardinality is not enough to obtain meaningful insights. 
%To this end, we plan to further explore the right type of information needed for these application and, based on the results, add the corresponding estimators to our protocol stack.
In this case, for example, knowing the nodes' identity proved to be essential for the accurate estimation of the crowd density and the popularity of Points of Interest. Unfortunately, as argued in \chaprefs{chapter:sofa}{chapter:estreme}, discovering the neighbors' identities in EWSNs is very challenging and leaves the problem of crowd density estimation in EWSNs still open to the research community.
Nevertheless, the results presented in \chapref{chapter:nemo} should inspire future work to explore which type of information is required by each application and, based on the results, develop the corresponding estimators for our EWSNs' stack.

%explore which information (mobility, node identifier, \etc) is important for each specific application and, based on the results, add the corresponding estimators to our protocol stack.
%If such information should prove to be essential, the proposed mechanisms will need to be extended in order to provide it in an opportunistic and lightweight manner. 

%TODO add conclusion regarding the design principles
%TODO which design principles failed in chapter 5? why?

%% file: summary/summary.tex
\chapter*{Summary}
\addcontentsline{toc}{chapter}{Summary}
\setheader{Summary}

Sensor networks can nowadays deliver 99.9\% of their data with duty cycles below 1\%. This remarkable performance is, however, dependent on some important underlying assumptions: low traffic rates, medium size densities and static nodes. In this thesis, we investigate the performance of these same resource-constrained devices, but under scenarios that present extreme conditions: high traffic rates, high densities and mobility: the so-called Extreme Wireless Sensor Networks (EWSNs). 

From a networking perspective, communicating in these extreme scenarios is very challenging. The combined effect of high network densities and dynamics makes the network's characteristics fluctuate drastically both in space and time. Traditional mechanisms struggle to cope with these sudden changes, resulting in a continuous exchange of information that saturates the bandwidth and increases the energy consumption. 
Once this saturation threshold is reached, mechanisms take decisions based on wrong, outdated information and soon
%usually break.
stop working. 
Even flexible mechanisms have difficulties adapting their settings to the fickly conditions of EWSNs and result in poor performance.

To efficiently communicate in EWSNs, mechanisms must therefore comply to a set of requirements \ie design principles, which are explained next. First, they need to be resilient to local and remote failures and operate as independent as possible from the status of other nodes (\textit{state-less principle}). Second, because in EWSNs bandwidth is a scarce resource, it should be mainly used for the transmission of the actual \textit{data}. Mechanisms should  not be artificially orchestrated and should exploit each other in a cross-layer fashion to reduce as much as possible their communication overhead (\textit{opportunistic principle}). Third, mechanisms should support extreme network conditions from their inception. Adapting traditional mechanisms, which are designed for milder conditions, would otherwise result in complex and fragile mechanisms (\textit{anti-fragile principle}). Fourth, in the case the resources saturate, mechanisms should operate in a conservative fashion, so that performance degrades gracefully without drastic disruptions (\textit{robustness principle}).

Inspired by these four principles, this thesis detaches from traditional communication primitives -- which are deterministic and based on rigid structures --  and proposes a novel communication stack based on \textit{opportunistic anycast}. According to this primitive, nodes communicate with the first available neighbor, independently of its location and identity. The more neighbors, the more efficient the communication. 

At the foundation of this communication stack lays \textit{SOFA}, a medium access control (MAC) protocol that exploits opportunistic anycast to handle extreme densities in an efficient manner. Its implementation details are presented in \chapref{chapter:sofa}. On top of the SOFA layer, this thesis builds two essential network services: neighborhood cardinality (density) estimation and data collection. The former service is provided by \textit{Estreme}, a mechanism presented in \chapref{chapter:estreme}, which exploits the \textit{rendezvous time} of SOFA to estimate the number of neighbors with almost zero overhead. The latter service is provided by \textit{Staffetta} in \chapref{chapter:staffetta}, a mechanism that adapts the wake-up frequency of nodes to bias the opportunistic neighbor-selection of SOFA towards the desired direction \eg towards the sink node collecting all data.

%Anycast is therefore opportunistic because it exploits one of the main challenges of EWSNs (extreme densities) to improve the communication performance. Following a similar opportunistic approach, the other mechanisms presented in this thesis exploit the presence of opportunistic anycast, in a cross-layer fashion, to minimize their overhead. 

Finally, this thesis presents an extensive evaluation of a complete opportunistic stack on simulations, testbeds and a challenging real-world deployment in the form of the NEMO science museum in Amsterdam. Results show that opportunistic behavior can lead to mechanisms that are both lightweight and robust and, thus, are able to scale to EWSNs.

\chapter*{Samenvatting}
\addcontentsline{toc}{chapter}{Samenvatting}
\setheader{Samenvatting}

{\selectlanguage{dutch}

Sensor netwerken kunnen tegenwoordig 99.9\% van de gemeten data afleveren
terwijl de duty cycle onder de 1\% blijft. Deze opmerkelijke prestaties
kunnen echter alleen gerealiseerd worden onder gunstige voorwaarden: lage data
snelheden, gematigde dichtheden (aantallen buren) en statische nodes. In dit
proefschrift onderzoeken we de prestaties van dezelfde apparaten, maar dan
in uitdagende scenario's met hoge data snelheden, hoge dichtheden en mobiele
nodes, die bekend staan onder de naam Extreme Wireless Sensor Networks (EWSN).

Communiceren in zulke extreme scenario's is zeer uitdagend. De combinatie van
hoge dichtheden en mobiliteit maakt dat de netwerkkarakteristieken sterk
fluctueren in zowel ruimte als tijd. Klassieke protocollen
hebben het moeilijk met zulke plotselinge veranderingen en vervallen in het
continue uitwisselen van (topologie) informatie, wat kostbare bandbreedte
in beslag neemt en het energieverbruik nodeloos verhoogt. Als de netwerk
belasting teveel stijgt resulteert dit in het nemen van beslissingen op
basis van verouderde informatie, en komen deze protocollen effectief tot
stilstand. Zelfs meer flexibele mechanismen hebben problemen met het zich
aanpassen aan de snel veranderende omstandigheden in EWSN en kampen met
verminderde prestaties.

Om effici\"ent te kunnen communiceren in EWSN's is het noodzakelijk 
de volgende ontwerp principes in acht te nemen. Ten eerste moet men
bestand zijn tegen storingen zowel in de nabijheid als verafgelegen, en
zo min mogelijk afhangen van de toestand waarin andere nodes verkeren
(toestandsonafhankelijkheid). Ten tweede, omdat de netwerk capaciteit
(bandbreedte) beperkt is, moet deze zoveel mogelijk gebruikt worden
voor het verzenden van {\em data}, en zo min mogelijk voor controle
berichten. Protocollen dienen daarom zo min mogelijk kunstmatig georchestreerd
te worden, en zoveel mogelijk met elkaar samen te werken (cross layer)
teneinde de communicatie-overhead tot een minimum te beperken (opportuniteit).
Ten derde, mechanismen moeten per ontwerp om kunnen gaan met de extreme
omstandigheden van EWSN's. Het aanpassen van traditionele concepten, die
gemaakt zijn voor mildere scenario's, leidt te vaak tot complexe en broze
oplossingen (anti-breekbaarheid). Ten vierde, als de limieten van het systeem
bereikt worden moet het niet abrupt maar geleidelijk terug schalen zodanig
dat er geen grote prestatieschommelingen optreden (robuustheid).

Op basis van deze vier principes breekt dit proefschrift met de traditionele
aanpak, die is gebaseerd op deterministische en vaste structuren, en
introduceert een nieuwe protocol stack op basis van {\em opportunistic anycast}.
Trouw aan deze communicatie vorm, communiceren nodes met de eerst beschikbare
buurnode, onafhankelijk van diens locatie en identiteit. Hoe meer buurnodes,
hoe effici\"enter de communicatie.

Aan de basis van deze protocol stack light SOFA, een medium access control
(MAC) protocol dat opportunistic anycast benut om effectief met de hoge
dichtheden van EWSN om te gaan. De implementatie details worden in
hoofdstuk~\ref{chapter:sofa} behandeld. SOFA wordt gebruikt om twee essenti\"ele
netwerk services mee te implementeren: dichtheidsschatting (neighborhood
cardinality) en datacollectie. De eerste service wordt geleverd door {\em
Estreme}, zoals behandeld in hoofdstuk~\ref{chapter:estreme}, die het {\em rendez-vous}
mechanisme van SOFA gebruikt om het aantal buurnodes te schatten met
minimale extra kosten. De tweede service wordt geleverd door {\em Staffetta},
zoals beschreven in hoofdstuk~\ref{chapter:staffetta}, die de wake-up frequentie van de
nodes instelt zodanig dat nodes dichter bij de uitgang van het netwerk (de sink
node) vaker wakker worden, en dus eerder gekozen worden door SOFA.

Het proefschrift besluit met een uitgebreide 
evaluatie van de complete opportunistische protocol stack zowel in simulatie,
een gecontroleerde testomgeving (met 100 nodes), als een uitdagend experiment in
het NEMO science museum te Amsterdam. De resultaten tonen aan dat
een opportunistische benadering leidt tot mechanismen die zowel lichtgewicht als
robuust zijn, en derhalve uitermate geschikt zijn om de schaalgrootte van
EWSN's te adresseren.

}